\documentclass[11pt,a4paper]{article}
\usepackage{jheppub}
\usepackage{amsmath}
\usepackage{amsthm}
\usepackage{multicol, bbm}
\usepackage{enumerate}
\usepackage{bibentry}
\usepackage{epstopdf}
\usepackage{float}
\usepackage{subfig}
\usepackage{tikz}
\usepackage{graphicx}

\usepackage{amssymb}
\newtheorem{theorem}{Theorem}
\newtheorem{conjecture}[theorem]{Conjecture}
\newtheorem{heuristic}[theorem]{Heuristic}

\newtheorem{corollary}[theorem]{Corollary}
\newtheorem{proposition}[theorem]{Proposition}

\def\L{{\mathcal L}}
\def\id{{\rm id}}

\numberwithin{theorem}{section}

\theoremstyle{remark}
\newtheorem{example}[theorem]{Example}

\newcommand{\be}{\begin{eqnarray}}
\newcommand{\ee}{\end{eqnarray}}
\newcommand{\M}{\mathcal{M}}

\newcommand{\B}{\mathcal{B}}
\newcommand{\C}{\mathbb{C}}
\newcommand{\Q}{\mathbb{Q}}

\def\E{\mathcal{E}}

\def\Fl{{\rm Fl}}

\def\isom{\stackrel{\sim}{\longrightarrow}}
\newcommand\defn[1]{{\em  #1}}
\def\Cone{{\rm Cone}}
\def\Vol{{\rm Vol\;}}
\def\deff{:=}
\def\Image{{\rm Image}}
\def\Conv{{\rm Conv}}
\def\sp{{\rm span}}
\def\tPhi{\tilde \Phi}
\def\aOmega{\underline{\Omega}}

\def\ttheta{{\tilde \theta}}
\def\ta{{\tilde a}}
\def\oPi{{\mathring \Pi}}
\def\oX{{\mathring X}}
\def\k{\underline{k}}
\def\tgeq{\tilde \geq}

\newcommand{\lb}{\left <}
\newcommand{\rb}{\right >}
\newcommand\ip[1]{\langle #1 \rangle}
\newcommand\remind[1]{{\color{red} #1}}

\def\ZZ{{\mathbb Z}}

\def\M{{\mathcal{M}}}
\def\G{{\mathcal{G}}}

\def\R{{\mathbb R}}
\def\A{{\mathcal A}}
\def\P{{\mathbb P}}
\def\Res{{\rm Res}}
\def\Int{{\rm Int}}

\def\GL{{\rm GL}}
\def\dlog{d\log}

\def\ZZ{\mathbb{Z}}
\def\Spec{{\rm Spec}}
\def \tPi{{\tilde \Pi}}

\begin{document}

\title{Positive Geometries and Canonical Forms}
\author[a]{Nima Arkani-Hamed,}
\author[b]{Yuntao Bai,}
\author[c]{Thomas Lam}
\affiliation[a]{School of Natural Sciences, Institute for Advanced Study, Princeton, NJ 08540, USA} 
\affiliation[b]{Department of Physics, Princeton University, Princeton, NJ 08544, USA}
\affiliation[c]{Department of Mathematics, University of Michigan, 530 Church St, Ann Arbor, MI 48109, USA} 

\abstract{
Recent years have seen a surprising connection between the physics of
scattering amplitudes and a class of mathematical objects--the
positive Grassmannian, positive loop Grassmannians, tree and loop
Amplituhedra--which have been loosely referred to as ``positive
geometries". The connection between the geometry and physics is
provided by a unique differential form canonically determined by the
property of having logarithmic singularities (only) on all the
boundaries of the space, with residues on each boundary given by the
canonical form on that boundary. The structures seen in the physical
setting of the Amplituhedron are both rigid and rich enough to
motivate an investigation of the notions of ``positive geometries" and
their associated ``canonical forms" as objects of study in their own
right, in a more general mathematical setting. In this paper we take
the first steps in this direction. We begin by giving a precise
definition of positive geometries and canonical forms, and introduce
two general methods for finding forms for more complicated positive
geometries from simpler ones--via ``triangulation" on the one hand,
and ``push-forward" maps between geometries on the other. We present
numerous examples of positive geometries in projective spaces,
Grassmannians, and toric, cluster and flag varieties, both for the
simplest ``simplex-like" geometries and the richer ``polytope-like"
ones. We also illustrate a number of strategies for computing
canonical forms for large classes of positive geometries, ranging from
a direct determination exploiting knowledge of zeros and poles, to the
use of the general triangulation and push-forward methods, to the
representation of the form as volume integrals over dual geometries
and contour integrals over auxiliary spaces. These methods yield
interesting representations for the canonical forms of wide classes of
positive geometries, ranging from the simplest Amplituhedra to new
expressions for the volume of arbitrary convex polytopes.}

\maketitle

\section{Introduction}

\begin{figure}
\centering

\subfloat[\Large $\frac{dx dy}{xy(1-x-y)}$]{\label{fig:intro_a}\includegraphics[width=4.5cm]{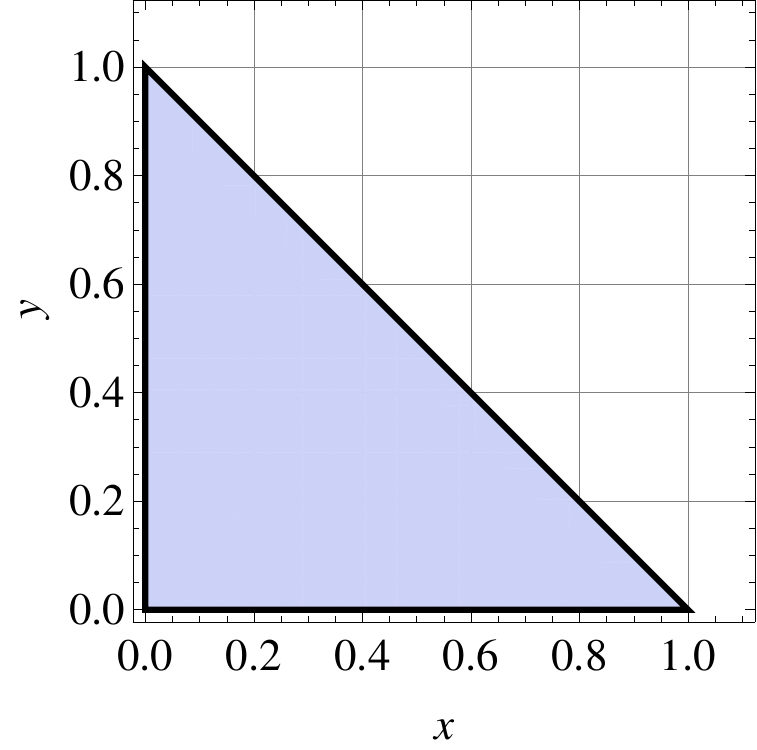}}\;\;\;\;\;\;\;\;
\subfloat[\Large $\frac{dxdy(12-x-4y)}{x(2y-x)(3-x-y)(2-y)}$]{\label{fig:intro_b}\includegraphics[width=4.5cm]{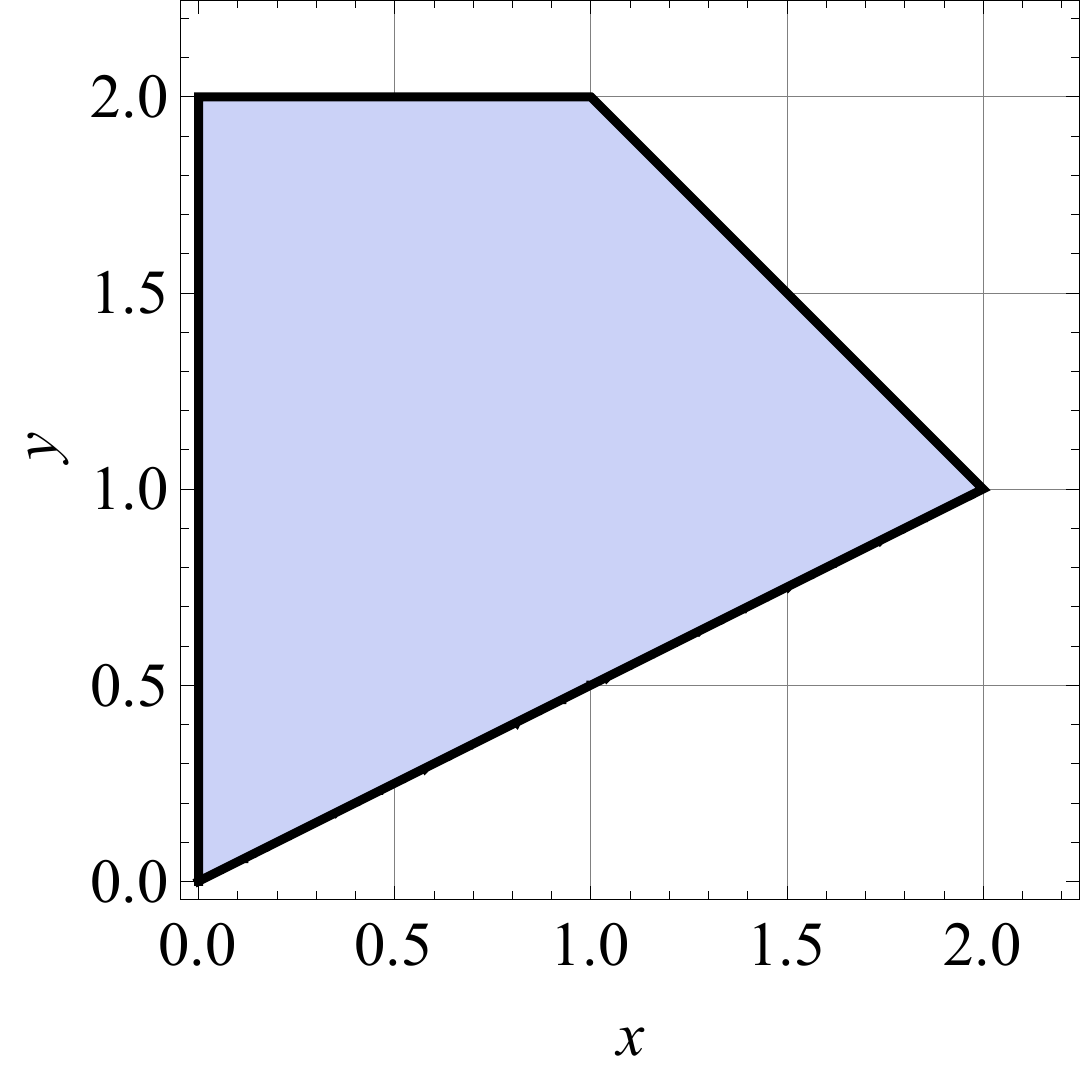}}
\\

\subfloat[\Large $\frac{(3\sqrt{11}/5)dxdy}{(1-x^2-y^2)(y-(1/10))}$]{\label{fig:intro_c}\includegraphics[width=6cm]{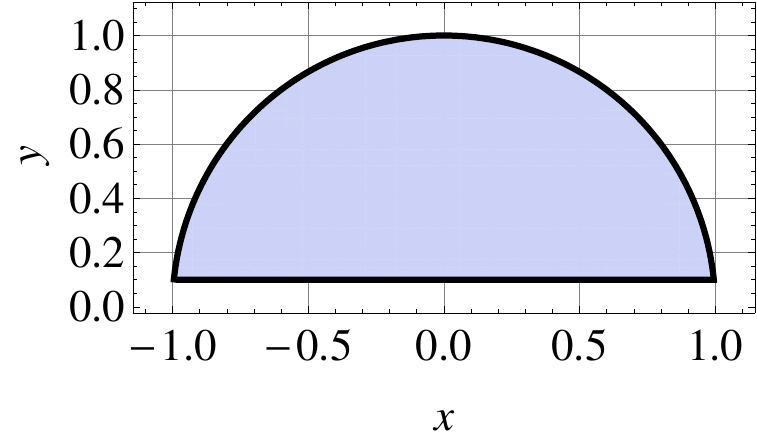}}\;\;\;\;
\subfloat[\Large $\frac{2\sqrt{3}(1+2y)dxdy}{(1-x^2-y^2)(\sqrt{3}y+x)(\sqrt{3}y-x)}$]{\label{fig:intro_d}\includegraphics[width=6cm]{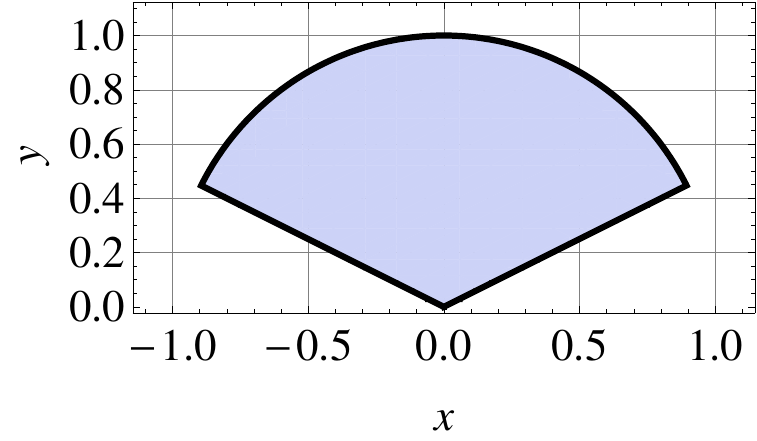}}\\
\subfloat[$0\, dxdy$]{\label{fig:intro_e}\includegraphics[width=4.5cm]{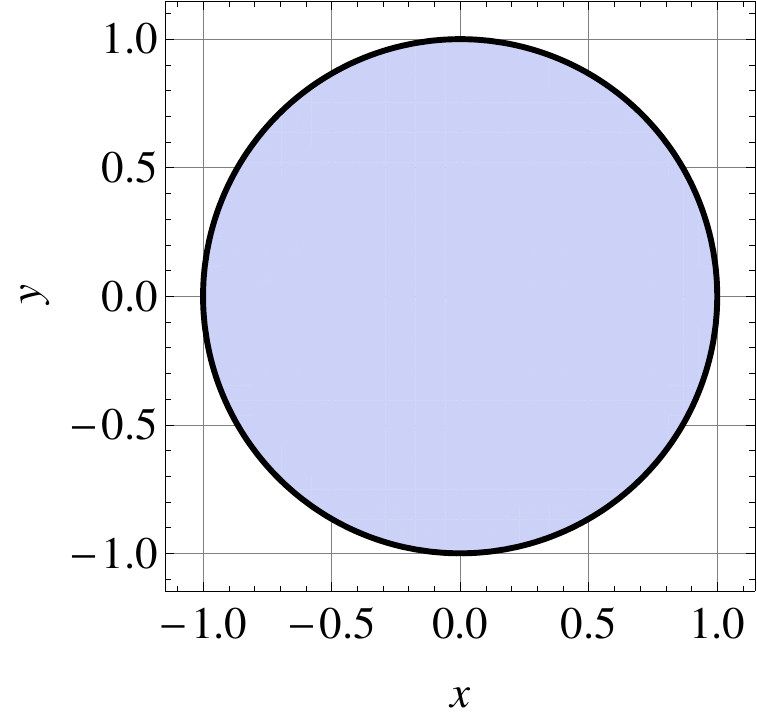}}

\caption{Canonical forms of (a) a triangle, (b) a quadrilateral, (c) a segment of the unit disk with $y\geq 1/10$, (d) a sector of the unit disk with central angle $2\pi/3$ symmetric about the $y$-axis, and (e) the unit disk. The form is identically zero for the unit disk because there are no zero-dimensional boundaries. For each of the other figures, the form has simple poles along each boundary component, all leading residues are $\pm 1$ at zero-dimensional boundaries and zero elsewhere, and the form is positively oriented on the interior.}
\label{fig:intro}
\end{figure}

Recent years have revealed an unexpected and fascinating new interplay
between physics and geometry in the study of gauge theory scattering amplitudes. In the context of planar $\mathcal{N}= 4$
super Yang-Mills theory, we now have a complete formulation of the
physics of scattering amplitudes in terms of the geometry of the
``Amplituhedron"~\cite{Grassmannian, hodges, posGrassmannian, Amplituhedron},
which is a Grassmannian generalization of polygons
and polytopes.  Neither space-time nor Hilbert space make any appearance
in this formulation -- the associated physics of locality and unitarity
arise as consequences of the geometry.

This new connection between physics and mathematics involves a number
of interesting mathematical ideas. The notion of ``positivity" plays a
crucial role. In its simplest form we consider the interior of a
simplex in projective space $\P^{n-1}$ as points with homogeneous
co-ordinates $(x_0,\ldots,x_{n-1})$ with all $x_a>0$. We can call this
the ``positive part" $\P^{n-1}_{>0}$ of projective space; thinking
of projective space as the space of 1-planes in
$n$ dimensions, we can also call this the ``positive part" of the
Grassmannian of 1-planes in $n$ dimensions, $G_{>0}(1,n)$.

This notion generalizes from $G_{>0}(1,n)$ to the ``positive part" of the
Grassmannian of $k$-planes in $n$ dimensions, $G_{>0}(k,n)$~\cite{Postnikov:2006kva, Lusztig}. The
Amplituhedron is a further extension of this idea, roughly generalizing the
positive Grassmannian in the same way that convex plane polygons
generalize triangles. These spaces have loosely been referred to as
``positive geometries" in the physics literature; like polygons and polytopes
they have an ``interior", with boundaries or facets of all
dimensionalities. Another crucial idea, which gives a dictionary for
converting these geometries into physical scattering amplitudes, is a
certain (complex) meromorphic differential form that is canonically associated
with the positive geometry. This form is fixed by the requirement of
having simple poles on (and only on) all the boundaries
of the geometry, with its residue along each being given by the canonical form of the boundary. The calculation of scattering amplitudes is
then reduced to the natural mathematical question of determining this
canonical form.

While the ideas of positive geometries and their associated canonical
forms have arisen in the Grassmannian/Amplituhedron context related to
scattering amplitudes, they seem to be deeper ideas worthy of being
understood on their own terms in their most natural and general
mathematical setting. Our aim in this paper is to take the first steps in
this direction.

To begin with, it is important to make the notion of a positive geometry precise. For instance, it is clear that the
interior of a triangle or a quadrilateral are positive geometries: we
have a two-dimensional space bounded by 1 and 0 dimensional
boundaries, and there is a unique 2-form with logarithmic
singularities on these boundaries. But clearly the interior of a
circle should {\sl not} be a positive geometry in the same sense, for the
simple reason that there are no 0-dimensional boundaries! See Figures~\ref{fig:intro_a},~\ref{fig:intro_b} \&~\ref{fig:intro_e} for an illustration.

We will formalize these intuitive notions in Section \ref{sec:positive}, and give a
precise definition of a
``positive geometry": the rough idea is to {\it define} a positive
geometry by the (recursive) {\it requirement} of the existence of a {\it unique}
form with logarithmic singularities on its boundaries. As we will see
in subsequent sections, in the plane this definition allows the
interior of polygons
but not the inside of a circle, and will also allow more general
positive geometries than polygons,
for instance bounded by segments of lines and conics.  In Figure \ref{fig:intro} we show some simple examples of positive geometries in the
plane together with their associated canonical forms.


In Sections \ref{sec:triangulations} and \ref{sec:morphism} we introduce two general methods for relating more complicated
positive geometries to simpler ones. The first method is
``triangulation". If a positive geometry $X_{\geq 0}$ can be ``tiled" by a collection of
positive geometries $X_{i,\geq 0}$ with mutually non-overlapping interiors, then the canonical form $\Omega(X_{\geq 0})$ of $X_{\geq 0}$ is given by the
sum of the forms for the pieces $\Omega(X_{\geq 0}) = \sum_i \Omega(X_{i,\geq 0})$. We say therefore that the canonical form is ``triangulation independent", a property that has played a central role in the physics literature, whose derivation we present in Section~\ref{sec:triangulations}. The
second is the ``push-forward" method. If we have a positive geometry
$X_{\geq 0}$, and if we have a \defn{morphism} (a special type of map defined in Section~\ref{sec:morphism}) that maps $X_{\geq 0}$ into another positive geometry
$Y_{\geq 0}$, then the canonical form on
$Y_{\geq 0}$ is the \defn{push-forward} of the canonical form on $X_{\geq 0}$. While both these
statements are simple and natural,
they are interesting and non-trivial. The ``triangulation" method has
been widely discussed in
the physics literature on Grassmannians and Amplituhedra. 
The
push-forward method is new, and will be applied in interesting ways in
later sections.

Sections \ref{sec:gensimplices} and \ref{sec:genpolytopes} are devoted to giving many examples of positive
geometries, which naturally divide into the simplest ``simplex-like"
geometries, and more complicated ``polytope-like" geometries.
A nice way of characterizing the distinction between the two can
already be seen by looking at the
simple examples in Figure \ref{fig:intro}. Note that the ``simplest looking" positive
geometries--the triangle and the half-disk, also have the simplest
canonical forms, with the important property of having only poles but
no zeros, while the ``quadrilateral" and ``pizza slice" have zeros
associated with non-trivial numerator factors. Generalizing these
examples, we define ``simplex-like" positive geometries to be ones for
which the canonical form has no zeros, while ``polytope-like"
positive geometries are ones for which the canonical form may have zeros.

We will provide several illustrative examples of generalized simplices (i.e. simplex-like positive geometries) in
Section \ref{sec:gensimplices}. The positive Grassmannian is an example, but we present a number of
other examples as well, including generalized simplices bounded by higher-degree
surfaces in projective spaces, as well as the positive parts of toric,
cluster and flag varieties.

In Section \ref{sec:genpolytopes} we discuss a number of examples of generalized polytopes (i.e. polytope-like positive geometries):
the familiar projective polytopes, Grassmann polytopes~\cite{Lam:2015uma}, and polytopes in partial flag varieties and 
loop Grassmannians. The tree
Amplituhedron is an important special case of a Grassmann polytope;
just as cyclic polytopes are an important special class of polytopes
in projective space. We study in detail the simplest Grassmann polytope that is
not an Amplituhedron, and determine
its canonical form by triangulation.  We also discuss loop and flag
polytopes, which
generalize the all-loop-order Amplituhedron.

In Section \ref{sec:forms} we take up the all-important question of determining the
canonical form associated with positive geometries. It is
fair to say that no ``obviously morally correct" way of finding the
canonical form for general Amplituhedra has yet emerged; instead
several interesting strategies have been found to be useful. We will
review some of these ideas and present a number of new ways of determining
the form in various special cases. We first discuss the most
direct and brute-force
construction of the form following from a detailed understanding of
its poles and zeros along the lines of~\cite{Arkani-Hamed:2014dca}. Next we illustrate the two general ideas of
``triangulation" and ``push-forward" in action. We give several
examples of triangulating more complicated positive geometries with
generalized simplices and summing the canonical form for the simplices
to determine the canonical form of the whole space. We also give
examples of morphisms from simplices into positive geometries.
Typically the morphisms involve solving coupled polynomial equation
with many solutions, and the push-forward of the form instructs us to
sum over these solutions. Even for the familiar case of polytopes in projective
space, this gives a beautiful new representation of the canonical form, which is formulated most naturally in the setting of toric varieties.  Indeed, there is a striking parallel between the polytope canonical form and the Duistermaat-Heckman measure of a toric variety.  We also
give two simple examples of the push-forward map from a simplex into
the Amplituhedron. We finally introduce a new set of ideas for
determining the canonical form associated with integral
representations. These are inspired by the Grassmannian contour
integral formula for scattering amplitudes, as well as the (still
conjectural) notion of integration
over a ``dual Amplituhedron". In addition to giving new
representations of forms for
polytopes, we will obtain new representations for classes of Amplituhedron
forms as contour integrals over the Grassmannian (for ``$\overline{NMHV}$
amplitudes" in
physics language), as well as dual-Amplituhedron-like integrals over a
``Wilson-loop"
to determine the canonical forms for all tree Amplituhedra with $m=2$.

The Amplituhedron involves a number of independent
notions of positivity. It generalizes the notion of an
``interior" in the Grassmannian, but also has a notion of
convexity, associated with demanding ``positive external data". Thus
the Amplituhedron generalizes the notion of the interior of {\it
convex} polygons, while the interior of non-convex polygons also
qualify as positive geometries by our definitions. In Section \ref{sec:convex} we
define what extra features a positive geometry should have to give a
good generalization of ``convexity", which we will call ``positive
convexity". Briefly the requirement is that the canonical form should have
no poles and no zeros inside the positive geometry. This is a rather miraculous (and largely numerically verified)
feature of the canonical form for Amplituhedra with even $m$~\cite{Arkani-Hamed:2014dca}, which is very likely ``positively
convex", while the simplest new example of a Grassmann polytope is not.

Furthermore, it is likely that our notion of
positive geometry needs to be
extended in an interesting way. Returning to the simple examples of Figure \ref{fig:intro}, it
may appear odd that
the interior of a circle is not a positive geometry while any convex
polygon is one,
given that we can approximate a circle arbitrarily well as a polygon with the
number of vertices going to infinity. The resolution of this little
puzzle is that while the canonical form
for a polygon with any fixed number of sides is a rational function
with logarithmic singularities, in the infinite
limit it is {\it not} a rational function--the poles condense to give
a function with
branch cuts instead. The notion of positive geometry we have described
in this paper is likely the special case of a ``rational" positive
geometry, which needs to be extended in some
way to cover cases where the canonical form is not rational. This is discussed in more detail in Section~\ref{sec:beyond}.

Our investigations in this paper are clearly only scratching the surface of what appears to be a deep and rich set of ideas, and in Section~\ref{sec:conclusion} we provide an outlook on immediate avenues for further exploration.
%
%
%

\section{Positive geometries}\label{sec:positive}
\subsection{Positive geometries and their canonical forms}
We let $\P^N$ denote $N$-dimensional complex projective space with the usual projection map $\C^{N{+}1}\setminus \{0\}\rightarrow \P^N$, and we let $\P^N(\R)$ denote the image of $\R^{N{+}1}\setminus\{0\}$ in $\P^N$.

Let $X$ be a \defn{complex projective algebraic variety}, which is the solution set in $\P^N$ of a finite set of homogeneous polynomial equations. We will assume that the polynomials have real coefficients.  We then denote by $X(\R)$ the \defn{real part} of $X$, which is the solution set in $\P^N(\R)$ of the same set of equations.

A \defn{semialgebraic set} in $\P^N(\R)$ is a finite union of subsets, each of which is cut out by finitely many homogeneous real polynomial equations $\{x \in \P^N(\R) \mid p(x) = 0\}$ and homogeneous real polynomial inequalities $\{x \in \P^N(\R) \mid q(x) > 0\}$.  To make sense of the inequality $q(x) > 0$ , we first find solutions in $\R^{N+1} \setminus \{0\}$, and then take the image of the solution set in $\P^N(\R)$.  

We define a $D$-dimensional \defn{positive geometry} to be a pair $(X, X_{\geq 0})$, where $X$ is an irreducible complex projective variety of complex dimension $D$ and $X_{\geq 0} \subset X(\R)$ is a nonempty oriented closed semialgebraic set of real dimension $D$ satisfying some technical assumptions discussed in Appendix~\ref{app:assumptions} where the \defn{boundary components} of $X_{\geq 0}$ are defined, together with the following recursive axioms:

\begin{itemize}
\item
For $D=0$: $X$ is a single point and we must have $X_{\geq 0} = X$. We define the $0$-form $\Omega{(X,X_{\geq 0})}$ on $X$ to be $\pm 1$ depending on the orientation of $X_{\geq 0}$. 
\item
For $D>0$: we must have
\begin{enumerate}
\item[(P1)]
Every boundary component $(C, C_{\geq 0})$ of $(X,X_{\geq 0})$ is a positive geometry of dimension $D{-}1$.
\item[(P2)]
There exists a unique nonzero rational $D$-form $\Omega{(X,X_{\geq 0})}$ on $X$ constrained by the residue relation $\Res_C \Omega{(X,X_{\geq 0})} = \Omega{(C,C_{\geq 0})}$ along every boundary component $C$, and no singularities elsewhere.
\end{enumerate}
\end{itemize}

See Appendix~\ref{app:residue} for the definition of the residue operator $\Res$. In particular, all \defn{leading residues} (i.e. $\Res$ applied $D$ times on various boundary components) of $\Omega(X,X_{\geq 0})$ must be $\pm 1$.  We refer to $X$ as the \defn{embedding space} and $D$ as the \defn{dimension} of the positive geometry.   The form $\Omega{(X,X_{\geq 0})} $ is called the \defn{canonical form} of the positive geometry $(X, X_{\geq 0})$.  As a shorthand, we will often write $X_{\geq 0}$ to denote a positive geometry $(X,X_{\geq 0})$, and write $\Omega(X_{\geq 0})$ for the canonical form.  We note however that $X$ usually contains infinitely many positive geometries, so the notation $X_{\geq 0}$ is slightly misleading.  Sometimes we distinguish the interior $X_{>0}$ of $X_{\geq 0}$ from $X_{\geq 0}$ itself, in which case we call $X_{\geq 0}$ the nonnegative part and $X_{>0}$ the positive part.  We will also refer to the {\it codimension $d$ boundary components} of a positive geometry $(X,X_{\geq 0})$, which are the positive geometries obtained by recursively taking boundary components $d$ times.

We stress that the existence of the canonical form is a highly non-trivial phenomenon.  The first four geometries in Figure \ref{fig:intro} are all positive geometries.  
\subsection{Pseudo-positive geometries}
A slightly more general variant of positive geometries will be useful for some of our arguments.  We define a $D$-dimensional \defn{pseudo-positive geometry} to be a pair $(X, X_{\geq 0})$ of the same type as a positive geometry, but now $X_{\geq 0}$ is allowed to be empty, and the recursive axioms are:
\begin{itemize}
\item
For $D=0$: $X$ is a single point. If $X_{\geq 0} = X$, we define the $0$-form $\Omega{(X,X_{\geq 0})}$ on $X$ to be $\pm 1$ depending on the orientation of $X_{\geq 0}$.  If $X_{\geq 0} = \emptyset$, we set $\Omega(X,X_{\geq 0}) = 0$.
\item
For $D>0$: if $X_{\geq 0}$ is empty, we set $\Omega(X,X_{\geq 0}) =0$.
Otherwise, we must have:
\begin{enumerate}
\item[(P1*)]
Every boundary component $(C, C_{\geq 0})$ of $(X,X_{\geq 0})$ is a pseudo-positive geometry of dimension $D{-}1$.
\item[(P2*)]
There exists a unique rational $D$-form $\Omega{(X,X_{\geq 0})}$ on $X$ constrained by the residue relation $\Res_C \Omega{(X,X_{\geq 0})} = \Omega{(C,C_{\geq 0})}$ along every boundary component $C$ and no singularities elsewhere.
\end{enumerate}
\end{itemize}
We use the same nomenclature for $X, X_{\geq 0}, \Omega$ as in the case of positive geometries.  The key differences are that we allow $X_{\geq 0} = \emptyset$, and we allow $\Omega(X,X_{\geq 0}) = 0$.  Note that there are pseudo-positive geometries with $\Omega(X,X_{\geq 0}) \neq 0$ that are not positive geometries.  When $\Omega(X,X_{\geq 0}) = 0$, we declare that $X_{\geq 0}$ is a \defn{null geometry}.  A basic example of a null geometry is the disk of Figure \ref{fig:intro_e}.

\subsection{Reversing orientation, disjoint unions and direct products}
We indicate the simplest ways that one can form new positive geometries from old ones.

First, if $(X,X_{\geq 0})$ is a positive geometry (resp. pseudo-positive geometry), then so is $(X,X_{\geq 0}^-)$, where $X_{\geq 0}^-$ denotes the same space $X_{\geq 0}$ with reversed orientation.  Naturally, its boundary components $C_i^-$ also acquire the reversed orientation, and we have $\Omega(X,X_{\geq 0}) = - \Omega(X,X_{\geq 0}^-)$.

Second, suppose $(X,X^1_{\geq 0})$ and $(X,X^2_{\geq 0})$ are pseudo-positive geometries, and suppose that they are disjoint: $X^1_{\geq 0} \cap X^2_{\geq 0} = \emptyset$.  Then the disjoint union $(X, X^1_{\geq 0} \cup X^2_{\geq 0})$ is itself a pseudo-positive geometry, and we have $\Omega(X^1_{\geq 0} \cup X^2_{\geq 0}) = \Omega(X^1_{\geq 0} ) + \Omega(X^2_{\geq 0})$.  This is easily proven by an induction on dimension.  The boundary components of $X^1_{\geq 0} \cap X^2_{\geq 0}$ are either boundary components of one of the two original geometries, or a disjoint union of such boundary components.  For example, a closed interval in $\P^1(\R)$ is a one-dimensional positive geometry (see Section \ref{sec:one}), and thus any disjoint union of intervals is again a positive geometry.  This is a special case of the notion of a \defn{triangulation} of positive geometries explored in Section \ref{sec:triangulations}.

Third, suppose $(X,X_{\geq 0})$ and $(Y,Y_{\geq 0})$ are positive geometries (resp. pseudo-positive geometries).  Then the direct product $X \times Y$ is naturally an irreducible projective variety via the Segre embedding (see \cite[I.Ex.2.14]{Hartshorne}), and $X_{\geq 0} \times Y_{\geq 0} \subset X \times Y$ acquires a natural orientation.  We have that $(Z, Z_{\geq 0}) := (X \times Y,X_{\geq 0} \times Y_{\geq 0})$ is again a positive geometry (resp. pseudo-positive geometry).  The boundary components of $(Z,Z_{\geq 0})$ are of the form $(C \times Y, C_{\geq 0} \times Y_{\geq 0})$ or $(X \times D , X_{\geq 0} \times D_{\geq 0})$, where $(C,C_{\geq 0})$ and $(D,D_{\geq 0})$ are boundary components of $(X,X_{\geq 0})$ and $(Y,Y_{\geq 0})$ respectively.  The canonical form is $\Omega(Z,Z_{\geq 0}) = \Omega(X,X_{\geq 0}) \wedge \Omega(Y,Y_{\geq 0})$.  


\subsection{One-dimensional positive geometries}\label{sec:one}
If $(X,X_{\geq 0})$ is a zero-dimensional positive geometry, then both $X$ and $X_{\geq 0}$ are points, and we have $\Omega(X,X_{\geq 0}) = \pm 1$.  If $(X,X_{\geq 0})$ is a pesudo-positive geometry instead, then in addition we are allowed to have $X_{\geq 0} = \emptyset$ and $\Omega(X,X_{\geq 0}) = 0$.

Suppose that $(X,X_{\geq 0})$ is a one-dimensional pseudo-positive geometry.  A genus $g$ projective smooth curve has $g$ independent holomorphic differentials.   Since we have assumed that $X$ is projective and normal but with no nonzero holomorphic forms (see Appendix~\ref{app:assumptions}), $X$ must have genus $0$ and is thus isomorphic to the projective line $\P^1$.  Thus $X_{\geq 0}$ is a closed subset of $\P^1(\R) \cong S^1$.  
If $X_{\geq 0} = \P^1(\R)$ or $X = \emptyset$ then $\Omega(X_{\geq 0}) = 0$ and $X_{\geq 0}$ is a pseudo-positive geometry but not a positive geometry.  Otherwise, $X_{\geq 0}$ is a union of closed intervals, and any union of closed intervals is a positive geometry. A generic closed interval is given by the following:

\begin{example}
We define the closed interval (or line segment) $[a,b]\subset\P^1(\R)$ to be the set of points $\{(1,x) \mid x \in [a,b]\} \subset \P^1(\R)$, where $a < b$.  Then the canonical form is given by
\be\label{ex:segment}
\Omega{\left([a,b]\right)} = \frac{dx}{x- a} - \frac{dx}{x-b}= \frac{(b-a)}{(b-x)(x-a)}dx.
\ee
where $x$ is the coordinate on the chart $(1,x) \in \P^1$, and the segment is oriented along the increasing direction of $x$. The canonical form of a disjoint union of line segments is the sum of the canonical forms of those line segments.
\end{example}


\section{Triangulations of positive geometries}\label{sec:triangulations}
\subsection{Triangulations of pseudo-positive geometries}
Let $X$ be an irreducible projective variety, and $X_{\geq 0} \subset X$ be a closed semialgebraic subset of the type considered in Appendix \ref{app:assumptions1}.  Let $(X,X_{i,\geq 0})$ for $i=1,\ldots,t$ be a finite collection of pseudo-positive geometries all of which live in $X$.  For brevity, we will write $X_{i,\geq 0}$ for $(X,X_{i,\geq 0})$ in this section.  We say that the collection $\{X_{i,\geq 0}\}$ \defn{triangulates} $X_{\geq 0}$ if the following properties hold:
\begin{itemize}
\item
Each $X_{i,>0}$ is contained in $X_{>0}$ and the orientations agree.
\item
The interiors $X_{i,>0}$ of $X_{i,\geq 0}$ are mutually disjoint.
\item
The union of all $X_{i,\geq 0}$ gives $X_{\geq 0}$.
\end{itemize}
Naively, a triangulation of $X_{\geq 0}$ is a collection of pseudo-positive geometries that tiles $X_{\geq 0}$.  The purpose of this section is to establish the following crucial property of the canonical form: 
\begin{align}
\label{eq:triangulationclaim}
&\mbox{If $\{X_{i,\geq 0}\}$ triangulates $X_{\geq 0}$ then $X_{\geq 0}$ is a pseudo-positive geometry and} \nonumber \\ &\mbox{$\Omega(X_{\geq 0})=\sum_{i=1}^t\Omega(X_{i,\geq 0})$.}
\end{align}
Note that even if all the $\{X_{i,\geq 0}\}$ are positive geometries, it may be the case that $X_{\geq 0}$ is not a positive geometry. A simple example is a unit disk triangulated by two half disks (see the discussion below Example~\ref{ex:circSegment}).


If all the positive geometries involved are polytopes, our notion of triangulation reduces to the usual notion of polytopal subdivision.  If furthermore $\{X_{i,\geq 0}\}$ are all simplices, then we recover the usual notion of a triangulation of a polytope.

Note that the word ``triangulation" does not imply that the  geometries $X_{i,\geq 0}$ are ``triangular" or ``simplicial" in any sense.

\subsection{Signed triangulations}\label{sec:signed}
We now define three signed variations of the notion of triangulation.  We loosely call any of these notions ``signed triangulations".

We say that a collection $\{X_{i,\geq 0}\}$ \defn{interior triangulates} the empty set if for every point $x \in \bigcup_i X_{i,\geq 0}$ that does not lie in any of the boundary components $C$ of the $X_{i,\geq 0}$ we have
\begin{align}\label{eq:signtriang}
&\#\{i \mid x \in X_{i,>0} \text{ and $X_{i,>0}$ is positively oriented at $x$}\} \nonumber\\
=& 
\#\{i \mid x \in X_{i,>0} \text{ and $X_{i,>0}$ is negatively oriented at $x$}\}
\end{align}
where we arbitrarily make a choice of orientation of $X(\R)$ near $x$.  Since all the $X_{i,>0}$ are open subsets of $X(\R)$, it suffices to check the \eqref{eq:signtriang} for a dense subset of $\bigcup_i X_{i,\geq 0} - \bigcup C$, where $\bigcup C$ denotes the (finite) union of the boundary components.  If $\{X_{1,\geq 0}, \ldots, X_{t,\geq 0}\}$ interior triangulates the empty set, we may also say that $\{X_{2,\geq 0}, \ldots, X_{t,\geq 0}\}$ interior triangulate $X^-_{1,\geq 0}$.  Thus an interior triangulation $\{X_{1,\geq 0}, \ldots, X_{t,\geq 0}\}$ of $X_{\geq 0}$ is a (genuine) triangulation exactly when a generic point $x \in X_{\geq 0}$ is contained in exactly one of the $X_{i, \geq 0}$.

We now define the notion of \defn{boundary triangulation}, which is inductive on the dimension.  Suppose $X$ has dimension $D$.  Then we say that $\{X_{i,\geq 0}\}$ is a boundary triangulation of the empty set if:
\begin{itemize}
\item
For $D=0$, we have $\sum_{i=1}^t\Omega(X_{i,\geq 0}) = 0$.
\item
For $D>0$, suppose $C$ is an irreducible subvariety of $X$ of dimension $D{-}1$.  Let $(C,C_{i,\geq 0})$ be the boundary component of $(X,X_{i,\geq 0})$ along $C$, where we set $C_{i,\geq 0} = \emptyset$ if such a boundary component does not exist.  We require that for every $C$ the collection $\{C_{i,\geq 0}\}$ form a boundary triangulation of the empty set.
\end{itemize}
As before, if $\{X_{1,\geq 0}, \ldots, X_{t,\geq 0}\}$ boundary triangulates the empty set, we may also say that $\{X_{2,\geq 0}, \ldots, X_{t,\geq 0}\}$ boundary triangulates $X^-_{1,\geq 0}$.  

We finally define the notion of \defn{canonical form triangulation}: we say that $\{X_{i,\geq 0}\}$ is a canonical form triangulation of the empty set if $\sum_{i=1}^t\Omega(X_{i,\geq 0}) =0$.  Again, we may also say that $\{X_{2,\geq 0}, \ldots, X_{t,\geq 0}\}$ canonical form triangulates $X^-_{1,\geq 0}$.

We now make the following claim, whose proof is given in Appendix~\ref{app:proofTri}:
\begin{equation}
\label{eq:signed}
\mbox{interior triangulation $\implies$ boundary triangulation $\iff$ canonical form triangulation}
\end{equation}

Note that the reverse of the first implication in \eqref{eq:signed} does not hold: $(\P^1,\P^1(\R))$ is a null geometry that boundary triangulates the empty set, but it does not interior triangulate the empty set.

We also make the observation that if $\{X_{i, \geq 0}\}$ boundary triangulates the empty set, and all the $X_{i, \geq 0}$ except $X_{1,\geq 0}$ are known to be pseudo-positive geometries, then we may conclude that $X_{1,\geq 0}$ is a pseudo-positive geometry with $\Omega(X_{1,\geq 0}) = -\sum_{i=2}^t \Omega(X_{i,\geq 0})$.  (Here, it suffices to know that $X_{1,\geq 0}$, and recursively all its boundaries, are closed semialgebraic sets of the type discussed in Appendix \ref{app:assumptions1}.)

We note that \eqref{eq:triangulationclaim} follows from \eqref{eq:signed}.  If $\{X_{i,\geq 0}\}$ triangulate $X_{\geq 0}$, we also have that $\{X_{i,\geq 0}\}$ interior triangulate $X_{\geq 0}$. Unless we explicitly refer to a {\sl signed} triangulation, the word {\sl triangulation} refers to~\eqref{eq:triangulationclaim}.  Finally, to summarize~\eqref{eq:triangulationclaim} and~\eqref{eq:signed} in words, we say that:
\be\label{eq:independence}
\mbox{The canonical form is \defn{triangulation independent}.}
\ee
We will return to this remark at multiple points in this paper.

\subsection{The Grothendieck group of pseudo-positive geometries in $X$}
\label{sec:Gro}
We define the {\sl Grothendieck group of pseudo-positive geometries in $X$}, denoted $\mathcal{P}(X)$, which is the free abelian group generated by all the pseudo-positive geometries in $X$, modded out by elements of the form
\be
\sum_{i=1}^t X_{i,\geq 0}
\ee
whenever the collection $\{X_{i,\geq 0}\}$ boundary triangulates (or equivalently by  \eqref{eq:signed}, canonical form triangulates) the empty set.  Note that in $\mathcal{P}(X)$, we have 
$X_{\geq 0}=-X_{\geq 0}^-$.  Also by \eqref{eq:signed}, $X_{\geq 0}=\sum_{i}X_{i,\geq 0}$ if $X_{i,\geq 0}$ forms an interior triangulation of $X_{\geq 0}$.

We may thus extend $\Omega$ to an additive homomorphism from $\mathcal{P}(X)$ to the space of meromorphic top forms on $X$ via:
\be
\Omega\left(\sum_{i=1}^t X_{i,\geq 0}\right)\deff \sum_{i=1}^t\Omega(X_{i,\geq 0})
\ee
Note that this homomorphism is injective, precisely because boundary triangulations and canonical form triangulations are equivalent.


\subsection{Physical versus spurious poles}
\label{sec:poles}

\begin{figure}
\centering
\includegraphics[width=6cm]{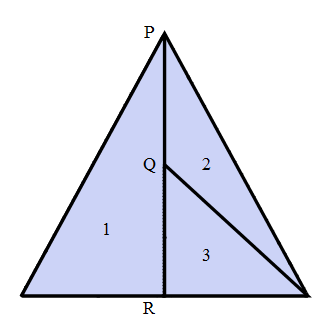}
\caption{A triangle $X_{\geq 0}$ triangulated by three smaller triangles $X_{i,\geq 0}$ for $i=1,2,3$. The vertices along the vertical mid-line are denoted $P,Q$ and $R$.}
\label{fig:triangle}
\end{figure}

In this subsection, we use ``signed triangulation" to refer to any of the three notions in Section \ref{sec:signed}.  
Let $\{X_{i,\geq 0}\}$ be a signed triangulation of $X_{\geq 0}$.  The boundary components of $X_{i,\geq 0}$ that are also a subset of boundary components of $X_{\geq 0}$ are called \defn{physical boundaries}; otherwise they are called \defn{spurious boundaries}. Poles of $\Omega(X,X_{i,\geq 0})$ at physical (resp. spurious) boundaries are called \defn{physical (resp. spurious) poles}. We sometimes refer to the triangulation independence of the canonical form \eqref{eq:triangulationclaim} as {\sl cancellation of spurious poles}, since spurious poles do not appear in the sum $\sum_i\Omega(X, X_i)$.

We now give an example which illustrates a subtle point regarding cancellation of spurious poles. It may be tempting to think that spurious poles cancel {\sl in pairs} along spurious boundaries, but this is false in general. Rather, the correct intuition is that:

\begin{quote}
Spurious poles cancel among collections of boundary components that boundary triangulate the empty set.
\end{quote}
Consider a triangle $X_{\geq 0}$ triangulated by three smaller pieces $X_{i,\geq 0}$ for $i=1,2,3$ as in Figure~\ref{fig:triangle}, but instead of adding all three terms in~\eqref{eq:triangulationclaim}, we only add the $i=1,2$ terms. Since the triangles 1 and 2 have adjacent boundaries along line $PQ$, it may be tempting to think that $\Omega(X,X_{1,\geq 0})+\Omega(X,X_{2,\geq 0})$ has no pole there. This is, however, false because the boundary components of 1 and 2 along line $PR$ forms a signed triangulation of the segment $QR$, which has a non-zero canonical form, so in particular $\Omega(X,X_{1,\geq 0})+\Omega(X,X_{2,\geq 0})$ has a non-zero residue along that line. However, had all three terms been included, then the residue would be zero, since the boundary components of the three pieces along $PR$ form a signed triangulation of the empty set. 

\section{Morphisms of positive geometries}\label{sec:morphism}

The canonical forms of different pseudo-positive geometries can be related by certain maps between them.  We begin by defining the \defn{push-forward} (often also called the \defn{trace map}) for differential forms, see \cite[II(b)]{Gri}. Consider a surjective meromorphic map $\phi:M\rightarrow N$ between complex manifolds of the same dimension. Let $\omega$ be a meromorphic top form on $M$, $b$ a point in $N$, and $V$ an open subset containing $b$. If the map $\phi$ is of degree $\deg{\phi}$, then the pre-image $\phi^{-1}(V)$ is the union of disconnected open subsets $U_i$ for $i=1,\ldots,\deg{\phi}$, with $a_i\in U_i$ and $\phi(a_i)=b$. We define the push-forward as a meromorphic top form on $N$ in the following way:
\be\label{eq:pushforward}
\phi_*(\omega)(b) \deff \sum_i \psi_i^*(\omega(a_i))
\ee
where $\psi_i\deff \phi|_{U_i}^{-1}: V \to U_i$.

Let $(X,X_{\geq 0})$ and $(Y,Y_{\geq 0})$ be two pseudo-positive geometries of dimension $D$.  A \defn{morphism} $\Phi: (X, X_{\geq 0}) \to (Y, Y_{\geq 0})$ consists of a rational (that is, meromorphic) map $\Phi: X \to Y$ with the property that the restriction $\Phi|_{X_{> 0}}: X_{> 0} \to Y_{>0}$ is an orientation-preserving diffeomorphism. A morphism where $(X, X_{\geq 0}) = (\P^D, \Delta^D)$ is also called a \defn{rational parametrization}.   If in addition, $\Phi: X \to Y$ is an isomorphism of varieties, then we call $\Phi: (X, X_{\geq 0}) \to (Y, Y_{\geq 0})$ an isomorphism of pseudo-positive geometries.  Two pseudo-positive geometries are \defn{isomorphic} if an isomorphism exists between them.

Note that if $\Phi: (X,X_{\geq 0}) \to (Y,Y_{\geq 0})$ and $\Psi: (Y,Y_{\geq 0}) \to (Z,Z_{\geq 0})$ are morphisms, then so is $\Pi = \Psi \circ \Phi$.  Pushforwards are \defn{functorial}, that is, we have the equality $\Pi_*=\Psi_*\circ \Phi_*$.  Therefore, pseudo-positive geometries with morphisms form a category. 
%


Finally, we state an important heuristic:
\begin{heuristic}\label{heuristic}
Given a morphism $\Phi: (X, X_{\geq 0}) \to (Y, Y_{\geq 0})$ of pseudo-positive geometries, we have
\be
\Phi_* (\Omega{(X,X_{\geq 0})}) = \Omega{(Y,Y_{\geq 0})}
\ee
where $\Phi_*$ is defined by \eqref{eq:pushforward}.
\end{heuristic}
We say that:
\be
\mbox{The push-forward preserves the canonical form.}
\ee

We do not prove Heuristic~\ref{heuristic} in complete generality, but we will prove it in a number of non-trivial examples (see Section \ref{sec:pushforward}).  For now we simply sketch an argument using some notation from Appendix \ref{app:assumptions}. 


The idea is to use induction on dimension and the fact that ``push-forward commutes with taking residue", formulated precisely in Proposition \ref{prop:KR}.  Let $(C,C_{\geq 0})$ be a boundary component of $(X,X_{\geq 0})$.  In general, the rational map $\Phi$ may not be defined as a rational map on $C$, but let us assume that (perhaps after a blowup of $\Phi$ \cite[I.4]{Hartshorne}, replacing $(X,X_{\geq 0})$ by another pseudo-positive geometry) this is the case, and in addition that $\Phi$ is well-defined on all of $C_{\geq 0}$.  In this case, by continuity $C_{\geq 0}$ will be mapped to the boundary $\partial Y_{\geq 0}$.  Since $C$ is irreducible, and $C_{\geq 0}$ is Zariski-dense in $C$, we see that $\Phi$ either collapses the dimension of $C$ (and thus $C$ will not contribute to the poles of $\Phi_*(\Omega(X,X_{\geq 0}))$), or it maps $C$ surjectively onto one of the boundary components $D$ of $Y$.  In the latter case, we assume that $C_{> 0}$ is mapped diffeomorphically to $D_{>0}$, so $(C,C_{\geq 0}) \to (D,D_{\geq 0})$ is again a morphism of pseudo-positive geometries.  The key calculation is then 
$$
\Res_D \Phi_*(\Omega(X,X_{\geq 0})) = \Phi_*(\Res_C \Omega(X,X_{\geq 0})) = \Phi_*(\Omega(C,C_{\geq 0})) = \Omega(D,D_{\geq 0})
$$
where the first equality is by Proposition \ref{prop:KR} and the third equality is by the inductive assumption.  Thus $\Phi_*(\Omega(X,X_{\geq 0}))$ satisfies the recursive definition and must be equal to $\Omega(Y,Y_{\geq 0})$.




\section{Generalized simplices}
\label{sec:gensimplices}

Let $(X,X_{\geq 0})$ be a positive geometry.  We say that $(X,X_{\geq 0})$ is a \defn{generalized simplex} or that it is \defn{simplex-like} if the canonical form $\Omega{(X,X_{\geq 0})}$ has no zeros on $X$.  The residues of a meromorphic top form with no zeros is again a meromorphic top form with no zeros, so all the boundary components of $(X,X_{\geq 0})$ are again simplex-like. While simplex-like positive geometries are simpler than general positive geometries, there is already a rich zoo of such objects.  

Let us first note that if $(X,X_{\geq 0})$ is simplex-like, then $\Omega(X,X_{\geq 0})$ is uniquely determined up to a scalar simply by its poles (without any condition on the residues).  Indeed, suppose $\Omega_1$ and $\Omega_2$ are two rational top forms on $X$ with the same simple poles, both of which have no zeros.  Then the ratio $\Omega_1/\Omega_2$ is a holomorphic function on $X$, and since $X$ is assumed to be projective and irreducible, this ratio must be a constant.  This makes the determination of the canonical form of a generalized simplex significantly simpler than in general.

\subsection{The standard simplex}\label{sec:standardsimplex}


The prototypical example of a generalized simplex is the positive geometry $(\P^m, \Delta^m)$, where $\Delta^m \deff \P^m_{\geq 0}$ is the set of points in $\P^m(\R)$ representable by nonnegative coordinates, which can be thought of as a projective simplex (see Section~\ref{sec:proj_simplex}) whose vertices are the standard basis vectors. We will refer to $\Delta^m$ as the \defn{standard simplex}. The canonical form is given by
\be\label{eq:dlog_Delta}
\Omega(\Delta^m) = \prod_{i=1}^m\frac{d\alpha_i}{\alpha_i}=\prod_{i=1}^m \dlog{\alpha_i}
\ee
for points $(\alpha_0,\alpha_1,\ldots,\alpha_m)\in \P^m$ with $\alpha_0=1$. Here we can identify the interior of $\Delta^m$ with $\R_{>0}^m$. Note that the pole corresponding to the facet at $\alpha_0\rightarrow 0$ has ``disappeared" due to the ``gauge choice" (i.e. choice of chart) $\alpha_0=1$, which can be cured by changing the gauge. As we will see in many examples, boundary components do not necessarily appear manifestly as poles in every chart, and different choices of chart can make manifest different sets of boundary components.

A gauge-invariant way of writing the same form is the following (see Appendix~\ref{app:projform}):
\be
\Omega(\Delta^m) = \frac{1}{m!}\frac{\left<\alpha\; d^m\alpha\right>}{\alpha_0\cdots\alpha_{m}}
\ee
There are $(m{+}1)$ codimension 1 boundary components (i.e. facets of the simplex) corresponding to the limits $\alpha_i\rightarrow 0$ for $i=0,\ldots,m$.

We say that a positive geometry $(X,X_{\geq 0})$ of dimension $m$ is \defn{$\Delta$-like} if there exists a {\sl degree one} morphism $\Phi:(\P^m,\Delta^m)\rightarrow (X,X_{\geq 0})$. The projective coordinates on $\Delta^m$ are called \defn{$\Delta-$like coordinates} of $X_{\geq 0}$. 

We point out that $\Delta$-like positive geometries are not necessarily simplex-like. Examples include BCFW cells discussed in Section~\ref{sec:BCFW}.  For now, we will content ourselves by giving an example of how a new zeros can develop under pushforwards.

\begin{example}
Consider the rational top-form on $\P^2$, given by $\omega = \dfrac{1}{(x+1)(y+1)}dx dy$ in the chart $\{(1,x,y)\} \subset \P^2$.  The form $\omega$ has three poles (along $x = 1$, $y=1$, and the line at infinity), and no zeros.  Consider the rational map $\Phi: \P^2 \to \P^2$ given by $(1,x,y) \mapsto (1,x,y/x) =: (1,u,v)$.  The map $\Phi$ has degree one, and using $dy = u dv + v du$ we compute that $\Phi_*\omega = \dfrac{u}{(u+1)(uv+1)} du dv$.  So a new zero along $u = 0$ has appeared.
\end{example}

\subsection{Projective simplices}
\label{sec:proj_simplex}

A \defn{projective $m$-simplex} $(\P^m, \Delta)$ is a positive geometry in $\P^m$ cut out by exactly ${m{+}1}$ linear inequalities. We will use $Y\in \P^m$ to denote a point in projective space with homogeneous components $Y^I$ indexed by $I=0,1,\ldots,m$. A linear inequality is of the form $Y^IW_I\ge 0$ for some vector $W\in\mathbb{R}^{m{+}1}$ with components $W_I$, and the repeated index $I$ is implicitly summed as usual. We define $Y\cdot W\deff Y^IW_I$. The vector $W$ is also called a \defn{dual vector}. The projective simplex is therefore of the form
\be
\Delta=\{Y\in\P^m(\R)\mid Y\cdot W_i \geq 0\text{ for $i=1,\ldots,m{+}1$}\}
\ee
where the inequality is evaluated for $Y$ in Euclidean space before mapping to projective space. Here the $W_{i}$'s are projective dual vectors corresponding to the \defn{facets} of the simplex. Every boundary of a projective simplex is again a projective simplex, so it is easy to see that projective simplices satisfy the axioms of a positive geometry. For notational purposes, we may sometimes write $Y^I=(1,x,y,\ldots)$ or $Y^I=(x_0,x_1,\ldots, x_m)$ or something similar.

We now give formulae for the canonical form $\Omega(\Delta)$ in terms of both the vertices and the facets of $\Delta$. Let $Z_i\in\R^{m{+}1}$ denote the vertices for $i=1,\ldots,m{+}1$, which carry upper indices like $Z_i^I$. We will allow the indices $i$ to be represented mod $m{+}1$. We have
\be \label{eq:simplexZ}
\Omega(\Delta) = \frac{s_m\langle Z_1 Z_2\cdots Z_{m+1} \rangle^m \;  \langle Y d^m Y \rangle}{m!\lb YZ_1\cdots Z_{m}\rb \lb YZ_2 \cdots Z_{m{+}1}\rb \cdots\lb YZ_{m{+}1}\cdots Z_{m-1}\rb}.
\ee
where the angle brackets $\lb\cdots\rb$ denote the determinant of column vectors $\cdots$, which is $SL(m{+}1)$-invariant, and $s_m={-}1$ for $m=1,5,9,\ldots$, and $s_m={+}1$ otherwise. See Appendix~\ref{app:projform} for the notation $\lb Yd^mY\rb$. Recall also from Appendix~\ref{app:projform} that the quantity 
\be\label{eq:simplexRationalFunc}
\aOmega(\A)\deff\Omega(\A) / \lb Yd^mY\rb
\ee
is called the \defn{canonical rational function}.

Now suppose the facet $W_i$ is adjacent to vertices $Z_{i{+}1},\ldots,Z_{i{+}m}$, then $W_i\cdot Z_j=0$ for $j=i{+}1,\ldots,i{+}m$. It follows that
\be\label{eq:W_simplex}
W_{iI}=({-}1)^{(i{-}1)(m{-}i)}\epsilon_{II_1\cdots I_m}Z_{i{+}1}^{I_1}\cdots Z_{i{+}m}^{I_m}
\ee
where the sign is chosen so that $Y\cdot W_i>0$ for $Y\in\Int(\A)$. See Section~\ref{sec:dual} for the reasoning behind the sign choice.

We can thus rewrite the canonical form in $W$ space as follows:
\be\label{eq:simplex}
\Omega(\Delta) = \frac{\langle W_1 W_2 \cdots W_{m+1} \rangle \;  \langle Y d^m Y \rangle}{m!(Y \cdot W_1) ( Y\cdot W_2) \cdots (Y\cdot W_{m+1})}.
\ee

A few comments on notation: We will often write $i$ for $Z_i$ inside an angle bracket, so for example $\lb i_0i_1\cdots i_{m}\rb \deff \lb Z_{i_0}Z_{i_1}\cdots Z_{i_{m}}\rb$ and $\lb Yi_1\cdots i_m \rb \deff \lb YZ_{i_1}\cdots Z_{i_{m}} \rb$. Furthermore, the square bracket $[1,2,\ldots,m{+}1]$ is defined to be the coefficient of $\langle Y d^m Y \rangle$ in~\eqref{eq:simplexZ}.  Thus,
\be\label{eq:RInv}
[1,2,\ldots,m{+}1] = \aOmega(\Delta) 
\ee
Note that the square bracket is antisymmetric in exchange of any pair of indices. These conventions are used only in $Z$ space.

The simplest simplices are one-dimensional line segments $\P^1(\R)$ discussed in Example~\ref{ex:segment}. We can think of a segment $[a,b]\in\P^1(\R)$ as a simplex with vertices
\be
Z_1^I=(1,a),\;\;\;Z_2^I=(1,b)
\ee
where $a<b$.

In Section~\ref{sec:polytope}, we provide an extensive discussion on convex projective polytopes as positive geometries, which can be triangulated by projective simplices.



%


\subsection{Generalized simplices on the projective plane}
\label{sec:simplexplane}
Suppose $(\P^2,\A)$ is a positive geometry embedded in the projective plane. We now argue that $\A$ can only have linear and quadratic boundary components. Let $(C,C_{\geq 0})$ be one of the boundary components of $\A$.  Then $C$ is an irreducible projective plane curve.  By our assumption (see Appendix~\ref{app:assumptions}) that $C$ is normal, we must have that $C \subset \P^2$ is a {\sl smooth} plane curve of degree $d$.  The genus of a smooth degree $d$ plane curve is equal to 
\be
g = \frac{(d-1)(d-2)}{2}.
\ee
According to the argument in Section \ref{sec:one}, for $(C,C_{\geq 0})$ to be a positive geometry (Axiom (P1)), we must have $C \cong \P^1$.  Therefore $g = 0$, which gives $d = 1$ or $d=2$.  Thus all boundaries of the positive geometry $\A$ are linear or quadratic. In Section~\ref{sec:degenerate}, we will relax the requirement of normality and give an example of a ``degenerate" positive geometry in $\P^2$. We leave detailed investigation of non-normal positive geometries for future work.

\begin{example}\label{ex:circSegment}
Consider a region $\mathcal{S}(a)\in \P^2(\R)$ bounded by one linear function $q(x,y)$ and one quadratic function $f(x,y)$, where $q=y-a \geq 0$ for some constant ${-}1<a<1$, and $f=1-x^2-y^2 \geq 0$. This is a ``segment" of the unit disk. A picture for $a=1/10$ is given in Figure~\ref{fig:intro_c}.

We claim that $\mathcal{S}(a)$ is a positive geometry with the following canonical form
\be
\Omega(\mathcal{S}(a))=\frac{2\sqrt{1-a^2}dx dy}{(1-x^2-y^2)(y-a)}
\ee

Note that for the special case of $a=0$, we get the canonical form for the ``northern half disk".
\be
\Omega(\mathcal{S}(0))=\frac{2dx dy}{(1-x^2-y^2)y}
\ee

We now check that the form for general $a$ has the correct residues on both boundaries. On the flat boundary we have
\be
\Res_{y{=}a}\Omega(\mathcal{S}(a))=\frac{2\sqrt{1-a^2}dx}{1-a^2-x^2}=\frac{2\sqrt{1-a^2}dx}{(\sqrt{1-a^2}-x)(x+\sqrt{1-a^2})}
\ee
Recall from Section~\ref{sec:one} that this is simply the canonical form on the line segment $x\in [-\sqrt{1-a^2},\sqrt{1-a^2}]$, with positive orientation since the boundary component inherits the counter-clockwise orientation from the interior.

The residue on the arc is more subtle. We first rewrite our form as follows:
\be
\Omega(\mathcal{S}(a))=\left(\frac{\sqrt{1-a^2} dy }{ x(y-a)}\right)\frac{df}{f}
\ee
which is shown by applying $df={-}2(xdx+ydy)$. The residue along the arc is therefore
\be
\Res_{f=0}\Omega(\mathcal{S}(a))=\frac{\sqrt{1-a^2} dy }{ x(y-a)}
\ee
Substituting $x=\sqrt{1-y^2}$ for the right-half of the arc gives residue $+1$ at the boundary $y=a$, and substituting $x=-\sqrt{1-y^2}$ for the left-half of the arc gives residue $-1$.

We give an alternative calculation of these residues. Let us parametrize $(x,y)$ by a parameter $t$ as follows.
\be
(x,y)=\left(\frac{(t+t^{-1})}{2},\frac{(t-t^{-1})}{2i}\right)
\ee
which of course satisfies the arc constraint $f(x,y)=0$ for all $t$.

Rewriting the form on the arc in terms of $t$ gives us
\be
\Res_{f{=}0}\Omega(\mathcal{S}(a))=\frac{2\sqrt{1-a^2}dt}{t^2-2iat-1}=\frac{(t_+-t_-)dt}{(t-t_+)(t-t_-)}
\ee
where $t_\pm =ia\pm\sqrt{1-a^2}$ are the two roots of the quadratic expression in the denominator satisfying
\be
t_++t_- = 2ia,\;\;\; t_+t_- = -1
\ee

The corresponding roots $(x_\pm,y_\pm)$ are
\be
(x_\pm,y_\pm)=(\pm \sqrt{1-a^2},a)
\ee
which of course correspond to the boundary points of the arc.

The residues at $t_\pm$ and hence $(x_\pm,y_\pm)$ are $\pm 1$, as expected.

\end{example}

By substituting $a=-1$ in Example \ref{ex:circSegment} we find that the unit disk $\mathcal{D}^2\deff\mathcal{S}(-1)$ has vanishing canonical form, or equivalently, is a null geometry. Alternatively, one can derive this by triangulating (see Section \ref{sec:triangulations}) the unit disk into the northern half disk and the southern half disk, whose canonical forms must add up to $\Omega(\mathcal{D}^2)$. A quick computation shows that the canonical forms of the two half disks are negatives of each other, so they sum to zero. A third argument goes as follows: the only pole of $\Omega(\mathcal{D}^2)$, if any, appears along the unit circle, which has a vanishing canonical form since it has no boundary components. So in fact $\Omega(\mathcal{D}^2)$ has no poles, and must therefore vanish by the non-existence of nonzero holomorphic top forms on the projective plane. More generally, a pseudo-positive geometry is a null geometry if and only if all its boundary components are null geometries.

However, not all conic sections are null geometries. Hyperbolas are notable exceptions. From our point of view, the distinction between hyperbolas and circles as positive geometries is that the former intersects a line at infinity. So a hyperbola has two boundary components, while a circle only has one. We show this as a special case of the next example.

\begin{example}\label{ex:quadratic}
Let us consider a generic region in $\P^2(\R)$ bounded by one quadratic and one linear polynomial. Let us denote the linear polynomial by $q=Y\cdot W\geq 0$ with $Y^I=(1,x,y)\in\P^2(\R)$ and the quadratic polynomial by $f=YY\cdot Q\deff Y^IY^JQ_{IJ}$ for some real symmetric bilinear form $Q_{IJ}$. We denote our region as $\mathcal{U}(Q,W)$. 

The canonical form is given by
\be
\Omega(\mathcal{U}(Q,W)) = \frac{\sqrt{QQWW} \ip{Y dY dY}}{(YY\cdot Q) (Y \cdot W)}
\ee
where $QQWW\deff -\frac{1}{2}\epsilon^{IJK}\epsilon^{I'J'K'}Q_{II'}Q_{JJ'}W_KW_{K'}$ and $\epsilon^{IJK}$ is the Levi-Civita symbol with $\epsilon^{012}=1$, and $\lb\cdots \rb$ denotes the determinant. The appearance of $\sqrt{QQWW}$ ensures that the result is invariant under rescaling $Q_{IJ}$ and $W_I$ independently, which is necessary. It also ensures the correct overall normalization as we will show in examples.

It will prove useful to look at this example by putting the line $W$ at infinity $W_I = (1,0,0)$ and setting $Y^I = (1,x,y)$, with $YY\cdot Q = y^2 - (x-a)(x-b)$ for $a\neq b$, which describes a hyperbola. The canonical form becomes
\be
\Omega(\mathcal{U}(Q,W))=\frac{2dx dy}{y^2-(x-a)(x-b)}
\ee
Note that taking the residue on the quadric gives us the 1-form on $y^2-(x-a)(x-b)=0$:
\be
\Res_{Q}\Omega(\mathcal{U}(Q,W))=dx/y=2dy/((x-a)+(x-b))
\ee
Suppose $a\neq b$, then this form is smooth as $y \to 0$ where $x \to a$ or $x \to b$, which is evident in the second expression above. The only singularities of this 1-form are on the line $W$, which can be seen by reparametrizing the projective space as $(z,w,1)\sim (1,x,y)$ so that $z=1/y,w=x/y$, which gives the 1-form on $1-(w-az)(w-bz)=0$:
\be\label{eq:hyperbola}
\Res_{Q}\Omega(\mathcal{U}(Q,W))=dw-\frac{dz}{z}=\frac{[(w-az)(-1+bz)+(w-bz)(-1+az)]dz}{z((w-az)+(w-bz))}
\ee
Evidently, there are only two poles $(z,w)=(0,\pm 1)$, which of course are the intersection points of the quadric $Q$ with the line $W$. The other ``pole" in~\eqref{eq:hyperbola} is not a real singularity since the residue vanishes.

Note however that as the two roots collide $a \to b$, the quadric degenerates to the product of two lines $(y + x - a)(y - x + a)$ and we get a third singularity at the intersection of the two lines $(x,y)=(a,0)$.

Note also another degenerate limit here, where the line $W$ is taken to be tangent to the quadric $Q$. We can take the form in this case to be $(dx dy)/(y^2 - x)$. Taking the residue on the parabola gives us the 1-form $dy$, that has a double-pole at infinity, which violates our assumptions. This corresponds to the two intersection points of the line $W$ with $Q$ colliding to make $W$ tangent to $Q$. In fact, we can get rid of the line $W$ all together and find that the parabolic boundary is completely smooth and hence only a null geometry.

Moreover, we can consider the form
$(dx dy)/(x^2 + y^2 - 1)$, i.e. associated with the interior of a circle. But for the same reason as for the parabola, the circle is actually a null geometry. Despite this, it is of course possible by analytic continuation of the coefficients of a general quadric to go from a circle to a hyperbola which {\it is} a positive geometry. 

Now let us return to the simpler example of the segment $\mathcal{U}(Q,W)\deff\mathcal{S}(a)$, where
\be
Q_{IJ}=\begin{pmatrix}
1 & 0 & 0\\
0 & -1 & 0\\
0 & 0 & -1
\end{pmatrix},\;\;\;\;
W_I=(-a,0,1)
\ee

Substituting these into the canonical form we find
\be
QQWW=(1-a^2)>0,\;\;\;
YY\cdot Q=1-x^2-y^2,\;\;\;
Y\cdot W=a-y,\;\;\;
\ee
and therefore
\be
\Omega(\mathcal{U}(Q,W))=\Omega(\mathcal{S}(a))
\ee
as expected.

\end{example}


\subsubsection{An example of a non-normal geometry}\label{sec:degenerate}
We would now like to give an example of a {\sl non-normal} positive geometry. Consider the geometry $\mathcal{U}(C)\subset \P^2(\R)$ defined by a cubic polynomial $YYY\cdot C\geq 0$, where $YYY\cdot C \deff C_{IJK} Y^I Y^J Y^K$ for some real symmetric tensor $C_{IJK}$. The canonical form must have the following form:
\begin{equation}
\Omega (\mathcal{U}(C))= \frac{C_0\langle Y dY dY \rangle}{YYY\cdot C }
\end{equation}
where $C_0$ is a constant needed to ensure that all leading residues are $\pm 1$. Of course, $C_0$ must scale linearly as $C_{IJK}$ and must be dependent on the Aronhold invariants for cubic forms. For our purposes we will work out $C_0$ only in specific examples.

\begin{figure}
\centering

\subfloat[Curve: $y^2-x(x-1/2)(x+1)=0$]{\includegraphics[width=6cm]{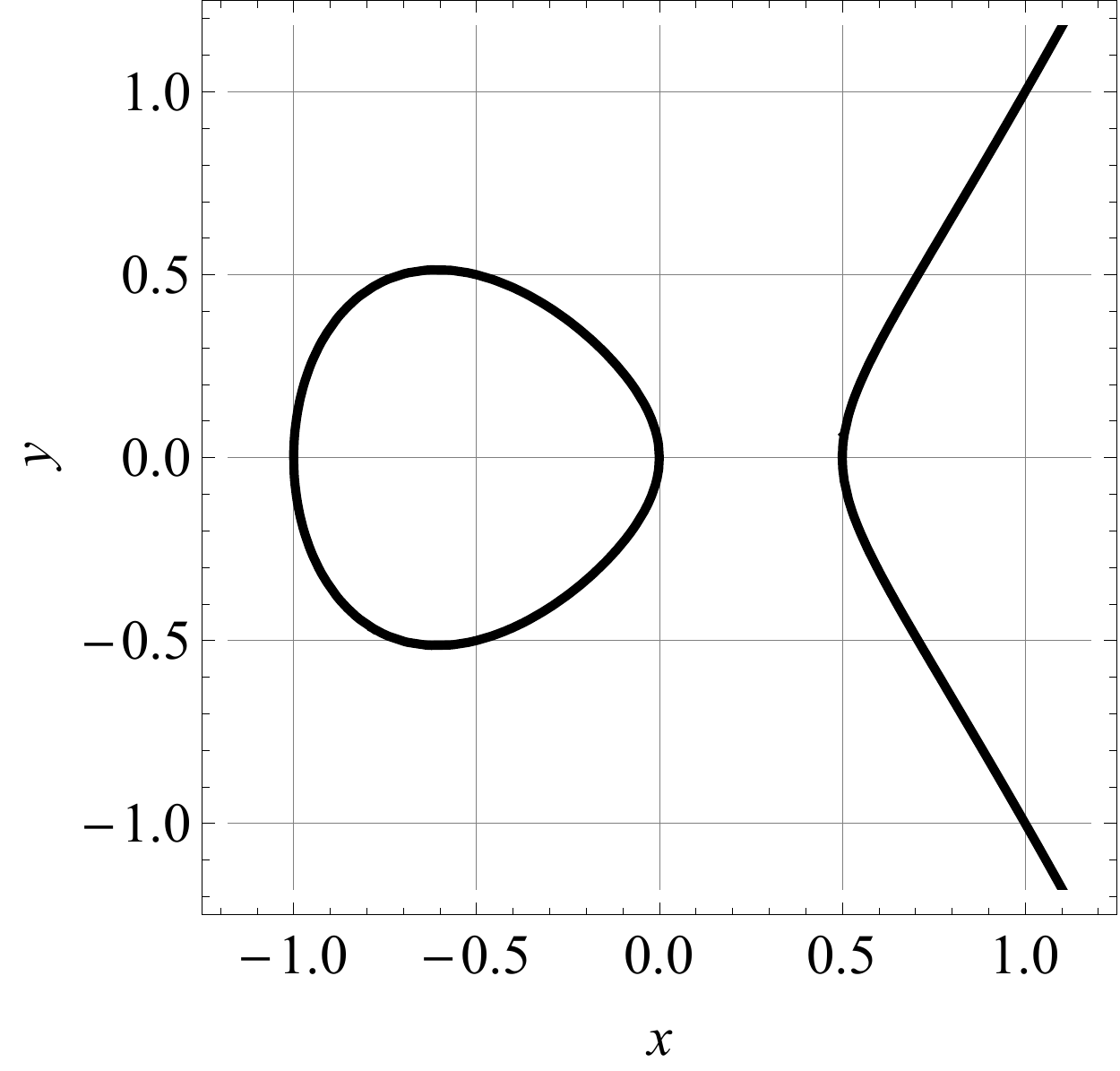}}\;\;\;\;\;\;\;\;
\subfloat[Boundary: $y^2-x^2(x+1)=0$]{\includegraphics[width=6cm]{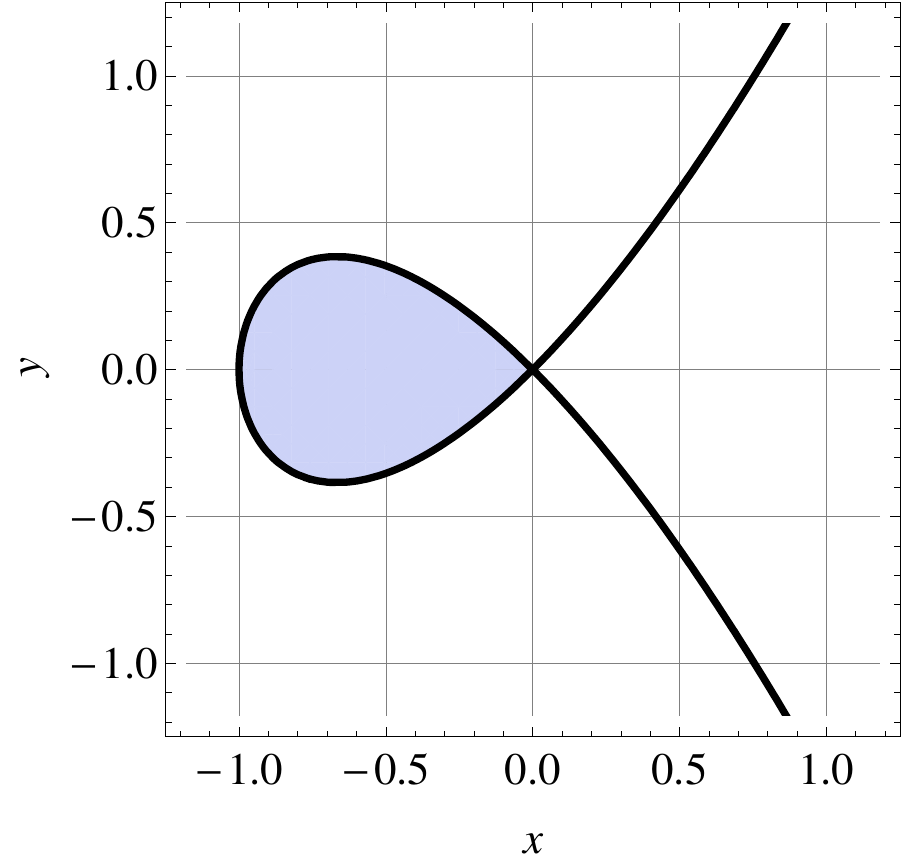}}

\caption{(a) A non-degenerate vs. (b) a degenerate elliptic curve. The former does not provide a valid embedding space for a positive geometry, while the shaded ``tear-drop" is a valid (non-normal) positive geometry.}
\label{fig:elliptic}
\end{figure}

Let us consider a completely generic cubic, which by an appropriate change of variables can always be written as $YYY\cdot C = y^2 - (x - a)(x-b)(x-c)$ for constants $a,b,c$. If the three constants are distinct, then there is no positive geometry associated with this case because the 1-form obtained by taking a residue on the cubic is $dx/y$ which is the standard holomorphic one-form associated with a {\sl non-degenerate} elliptic curve.  We can also observe directly that there are no singularities as $x \to a,b,c$, and that as we go to infinity, we can set $y \to 1/t^3, x \to 1/t^2$ with $t \to 0$, and $dx/y \rightarrow -2dt$ is smooth. The existence of such a form makes the canonical form non-unique, and hence ill-defined. By extension, no positive geometry can have the non-denegerate cubic as a boundary component either.

However, if the cubic degenerates by having two of the roots of the cubic polynomial in $x$ collide, then we {\it do} get a beautiful (non-normal) positive geometry, one which has only one zero-dimensional boundary. Without loss of generality let us put the double-root at the origin and consider the cubic $y^2 - x^2(x + a^2)$. Taking the residue on the cubic, we can parametrize $y^2 - x^2(x+a^2)=0$ as $y=t(t^2 - a^2), x = (t^2 - a^2)$, then $dx/y = dt/(t^2 - a^2)$ has logarithmic singularities at $t = \pm a$. Note that these two points correspond to the same point $y = x = 0$ on the cubic! But the boundary is oriented, so we encounter the same logarithmic singularity point from one side and then the other as we go around. We can cover the whole interior of the ``teardrop" shape for this singular cubic by taking
\begin{equation}
x = u(t^2 - a^2),\;\;\;y = u t (t^2 - a^2)
\end{equation}
which, for $u \in (0,1)$ and $t \in (-a,a)$ maps 1-1 to the teardrop interior, dutifully reflected in the form 
\begin{equation}
\frac{dx dy}{y^2 - x^2(x + a^2)} = \frac{dt}{(a-t)(t+a)} \frac{du}{u(1-u)}
\end{equation}
Note that if we {\it further} take $a \to 0$, we lose the positive geometry as we get a form with a double-pole, much as our example with the parabola in Example~\ref{ex:quadratic}. 

\subsection{Generalized simplices in higher-dimensional projective spaces}
Let us now consider generalized simplices $(\P^m,\A)$ for higher-dimensional projective spaces.  Let $(C,C_{\geq 0})$ be a boundary component of $\A$, which is an irreducible {\sl normal} hypersurface in $\P^m$.  For $(C,C_{\geq 0})$ to be a positive geometry, $C$ must have no nonzero holomorphic forms.  Equivalently, the \defn{geometric genus} of $C$ must be 0.  This is the case if and only if $C$ has degree less than or equal to $m$.  Thus in $\P^3$, the boundaries of a positive geometry are linear, quadratic, or cubic hypersurfaces.

It is easy to generalize Example \ref{ex:circSegment} to simplex-like positive geometries in $\P^m(\R)$: take a
positive geometry bounded by $(m-1)$ hyperplanes $W_i$ and a quadric
$Q$, which has  canonical form
\begin{equation}
\Omega(\A) = \frac{C_0\langle Y d^m Y\rangle}{(Y \cdot W_1) \cdots (Y \cdot W_{m-1}) (YY\cdot Q)}.
\end{equation}
for some constant $C_0$. Note that the $(m-1)$ planes intersect generically on a line, that in turn intersects
the quadric at two points, so as in our two-dimensional example this
positive geometry has two zero-dimensional boundaries.

Let us consider another generalized simplex, this time in $\P^3(\R)$.  We take a three-dimensional region $\A\subset \P^3(\R)$ bounded by a cubic surface and a plane.  If we take a generic cubic surface $C$ and generic plane $W$, then their intersection would be a generic cubic curve in $W$, which as discussed in Section \ref{sec:simplexplane} cannot contain a (normal) positive geometry.

On the other hand, we can make a special choice of cubic surface $\A$ that gives a positive geometry.  A
pretty example is provided by the ``Cayley cubic" (see Figure~\ref{fig:cayley}). If $Y^I =
(x_0,x_1,x_2,x_3)$ are coordinates on $\P^3$, let the cubic $C$ be defined by
\begin{equation}
C\cdot YYY \deff x_0 x_1 x_2 + x_1x_2 x_3 + x_2 x_3 x_0 + x_3 x_0 x_1 = 0
\end{equation}

\begin{figure}
\centering
\includegraphics[width=7cm]{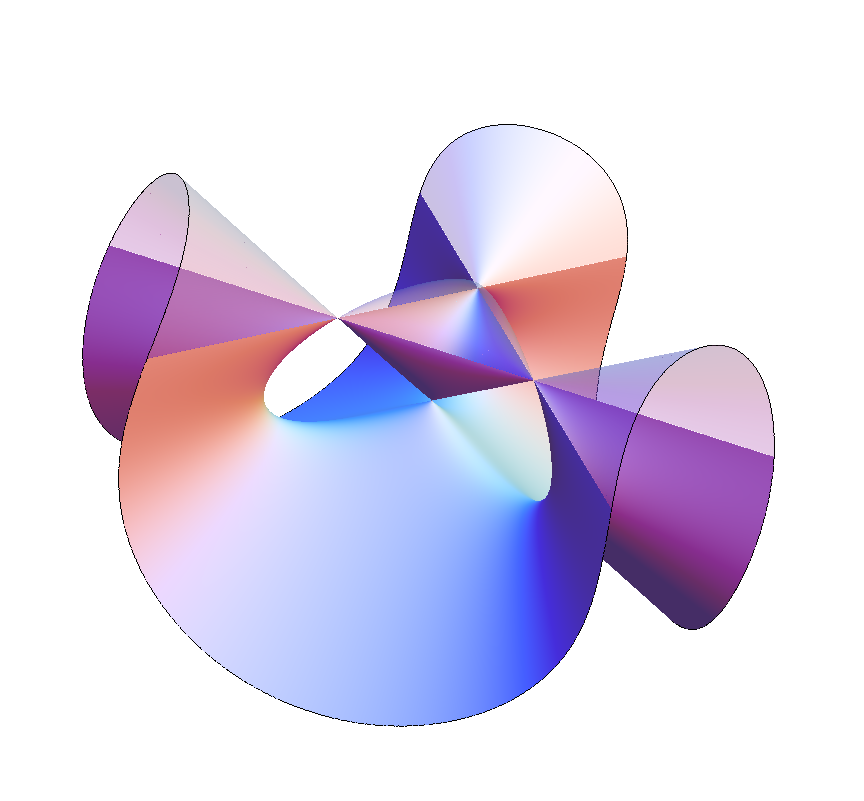}\caption{The Cayley cubic curve. The plane separating the translucent and solid parts of the surface is given by $x_0=0$.}
\label{fig:cayley}
\end{figure}

This cubic has four singular points at $X_0 = (1,0,0,0),\ldots,
X_3=(0,0,0,1)$.  Note that $C$ gives a singular surface, but it still satisfies the normality criterion of a positive geometry.  Let us choose three of the singular points, say $X_1,X_2,X_3$,
and let $W$ be the hyperplane passing through
these three points; we consider the form
\begin{equation}
\Omega(\A) = C_0\frac{\langle Y d^3 Y \rangle}{(YYY\cdot C) \langle Y X_1X_2X_3\rangle}
\end{equation}
where $C_0$ is a constant. A natural choice of variables turns this into a ``dlog" form. Consider
\begin{equation}
x_i = s y_i, \mbox{ for $i=1,2,3$}; \qquad x_0 = -\frac{y_1 y_2 y_3}{y_1 y_2 +
y_2 y_3 + y_3 y_1}
\end{equation}
Then if we group the three $y$'s as coordinates of ${\P}^2 = \{y=(y_1,
y_2, y_3)\}$, we have
\begin{equation}
\Omega(\A) = \frac{\langle y d^2 y \rangle}{2y_1 y_2 y_3} \frac{ds}{(s-1)}
\end{equation}
This is the canonical form of the positive geometry given by the bounded component of the region cut out by $YYY\cdot C \geq 0$ and
$\langle X_1X_2X_3Y \rangle \geq 0$.

We can generalize this construction to ${\P}^m(\R)$, with $Y=(x_0, \cdots, x_{m})$ and a degree $m$ hypersurface
$$
Q_m\cdot Y^m = \sum_{i=0}^{m} x_0 \cdots \hat{x}_i \cdots x_m,
$$
where $Q_m\cdot Y^m\deff Q_{mI_1\dots I_m}Y^{I_1}\cdots Y^{I_m}$
and the 
singular points are $X_0 = (1,0,\cdots, 0), \ldots,$ $X_{m}=(0,\cdots,
0,1)$. Then if we choose $m$ of these points and a linear factor
corresponding to the hyperplane going through them,
\begin{equation}
\Omega(\A) = C_0\frac{\langle Y d^m Y \rangle}{(Y^m\cdot Q_m) \langle Y X_1 \cdots X_m \rangle}
\end{equation}
is the canonical form associated with the bounded component of the positive geometry cut out by $Y^m\cdot Q_m\geq 0,
\langle X_1 \cdots X_m Y \rangle \geq 0$, for some constant $C_0$.

\subsection{Grassmannians}

In this section we briefly review the positroid stratification of the positive Grassmannian, and argue that each cell of the stratification is a simplex-like positive geometry.
\subsubsection{Grassmannians and positroid varieties}\label{sec:posGrassmannian}
Let $G(k,n)$ denote the Grassmannian of $k$-dimensional linear subspaces of $\C^n$.  We recall the \defn{positroid stratification} of the Grassmannian.  Each point in $G(k,n)$ is represented by a $k \times n$ complex matrix $C = (C_1,C_2,\ldots,C_n)$ of full rank, where $C_i \in \C^k$ denote column vectors. Given $C \in G(k,n)$ we define a function $f: \ZZ \to \ZZ$ by the condition that
\be \label{eq:permutation}
C_i \in \sp(C_{i+1},C_{i+2},\ldots,C_{f(i)})
\ee
and $f(i)$ is the minimal index satisfying this property. In particular, if $C_i=0$, then $f(i)=i$. Here, the indices are taken mod $n$. The function $f$ is called an \defn{affine permutation}, or ``decorated permutation", or sometimes just ``permutation" \cite{Postnikov:2006kva,KLS,posGrassmannian}.  Classifying points of $G(k,n)$ according to the affine permutation $f$ gives the positroid stratification 
\be \label{eq:positroidstrata}
G(k,n) = \bigsqcup_f \oPi_f.
\ee
where, for every affine permutation $f$, the set $\oPi_f$ consists of those $C$ matrices satisfying~\eqref{eq:permutation} for every integer $i$.
We let the \defn{positroid variety} $\Pi_f \subset G(k,n)$ be the closure of $\oPi_f$.  Then $\Pi_f$ is an irreducible, normal, complex projective variety \cite{KLS}.
If $k = 1$ then $G(k,n) \cong \P^{n-1}$ and the stratification \eqref{eq:positroidstrata} decomposes $\P^{n-1}$ into coordinate hyperspaces.

\subsubsection{Positive Grassmannians and positroid cells}
Let $G(k,n)(\R)$ denote the real Grassmannian.  Each point in $G(k,n)$ is represented by a $k \times n$ complex matrix of full rank. The \defn{(totally) nonnegative Grassmannian} $G_{\geq 0}(k,n)$ (resp. \defn{(totally) positive Grassmannian} $G_{>0}(k,n)$) consists of those points $C \in G(k,n)(\R)$ all of whose $k \times k$ minors, called Pl\"ucker coordinates, are nonnegative (resp. positive) \cite{Postnikov:2006kva}.  The intersections
\be
\Pi_{f,> 0} \deff G_{\geq 0} \cap \oPi_f, \qquad \Pi_{f, \geq 0} \deff G_{\geq 0} \cap \Pi_f
\ee
are loosely called (open and closed) \defn{positroid cells}.

For any permutation $f$, we have
\be
 (\Pi_f, \Pi_{f,\geq 0}) \mbox{ is a positive geometry.}
\ee
The boundary components of $(\Pi_f, \Pi_{f,\geq 0})$ are certain other positroid cells $(\Pi_g,\Pi_{g,\geq 0})$ of one lower dimension. The canonical form $\Omega(f) \deff \Omega{(\Pi_f, \Pi_{f,\geq 0})}$ was studied in \cite{KLS,posGrassmannian}.  We remark that $\Omega(f)$ has no zeros, so $(\Pi_f,\Pi_{f,\geq 0})$ is simplex-like. 

The canonical form $\Omega(G_{\geq 0}(k,n)) \deff \Omega{(G(k,n),G_{\geq 0}(k,n))}$ of the positive Grassmannian was worked out and discussed in~\cite{posGrassmannian}:
\be
\Omega(G_{\geq 0}(k,n))\deff \frac{\prod_{s=1}^k\lb C d^{n{-}k}C_s\rb}{((n{-}k)!)^k\prod_{i=1}^n(i,i{+}1,\ldots,i{+}k{-}1)}
\ee
where $C\deff(C_1,\ldots,C_k)^T$ is a $k\times n$ matrix representing a point in $G(k,n)$, and the parentheses $(i_1,i_2,\ldots,i_k)$ denotes the $k\times k$ minor of $C$ corresponding to columns $i_1,i_2,\ldots,i_k$ in that order. We also divide by the ``gauge group" $\GL(k)$ since the matrix representation of the Grassmannian is redundant.  The canonical forms $\Omega(f)$ on $\Pi_f$ are obtained by iteratively taking residues of $\Omega(G_{\geq 0}(k,n))$.

The Grassmannian $G(k,n)$ has the structure of a cluster variety \cite{Scott}, as discussed in Section \ref{sec:cluster}.  The cluster coordinates of $G(k,n)$ can be constructed using plabic graphs or on-shell diagrams.  Given a sequence of \defn{cluster coordinates} $(c_0,c_1,\ldots,c_{k(n{-}k)})\in \P^{k(n{-}k)}$ for the Grassmannian $G(k,n)$, the positive Grassmannian is precisely the subset of points representable by positive coordinates. It follows that $G_{\geq}(k,n)$ is $\Delta$-like with the degree-one cluster coordinate morphism $\Phi:(\P^{k(n{-}k)},\Delta^{k(n{-}k)})\rightarrow (G(k,n),G_{\geq 0}(k,n))$.  Note of course that a different degree-one morphism exists for each choice of cluster.

According to Heuristic~\ref{heuristic}, we expect that the canonical form on the positive Grassmannian is simply the push-forward of $\Omega(\Delta^{k(n{-}k)})$. That is,
\be \label{eq:pushGrass}
\Omega(G_{\geq 0}(k,n))=\pm\Phi_*\left(\frac{\lb c\;d^{k(n{-}k)}c\rb}{(k(n{-}k))!\prod_{I=0}^{k(n{-}k)}c_I}\right)
\ee
where the overall sign depends on the ordering of the cluster coordinates.  Equation \eqref{eq:pushGrass} is worked out in \cite{Lam:2015uma}.  It follows in particular that the right hand side of \eqref{eq:pushGrass} is independent of the choice of cluster. 

\subsection{Toric varieties and their positive parts}
\label{sec:toric}
In this section we show that positive parts of projectively normal toric varieties are examples of positive geometries.

\subsubsection{Projective toric varieties}
Let $z = (z_1, z_2,\ldots, z_n) \in (\ZZ^{m+1})^n$ be a collection of integer vectors in $\ZZ^{m+1}$.  We assume that the set is {\it graded}, so:
\be \label{eq:graded}
\mbox{There exists $a \in \Q^{m+1}$ so that $a \cdot z_i = 1$ for all $i$.}  
\ee

We define a (possibly not normal) projective toric variety $X(z) \subset \P^{n-1}$ as 
\be
X(z) &= \overline{\{(X^{z_1}, X^{z_2}, \ldots , X^{z_n}) \mid X \in (\C^\ast)^{m+1}\}} \subset \P^{n-1}.
\ee
where
\be
X^{z_i} \deff X_0^{z_{0i}} X_1 ^{z_{1i}} \cdots X_{m}^{z_{m,i}}
\ee
Equivalently, $X(z)$ is the closure of the image of the monomial map $\theta^z$ 
\be \label{eq:thetaz}
\theta^z: X = (X_0, \ldots,X_{m}) \longmapsto (X^{z_1}, \ldots , X^{z_n}) \in \P^{n-1}.
\ee 

We shall assume that  $z$ spans $\ZZ^{m+1}$ so that $\dim X(z) = m$.  The intersection of $X(z)$ with $\{(C_1,C_2,\ldots,C_n) \mid C_i \in \C^*\} \subset \P^{n-1}$ is a dense complex torus $T \cong (\C^*)^m$ in $X(z)$.  Define the nonnegative part $X(z)_{ \geq 0}$ of $X(z)$ to be the intersection of $X(z)$ with $\Delta^{n-1} = \{(C_1,C_2,\ldots,C_n) \mid C_i \in \R_{\geq 0}\} \subset \P^{n-1}(\R)$.  Similarly define $X(z)_{>0}$.  Equivalently, $X(z)_{> 0}$ is simply the image of $\R_{>0}^{m+1}$ under the monomial map $(\C^\ast)^{m+1} \to \P^{n-1}$.  It is known that $X(z)_{\geq 0}$ is diffeomorphic to the polytope $\A(z) \deff \Conv(z)$, see \cite{Fulton,Sottile}.  We establish a variant of this result in Appendix \ref{app:diffeo}. Note that we do not need to assume that the $z_i$ are vertices of $\A(z)$; some of the points $z_i$ may lie in the interior.

Our main claim is that
\be \label{eq:toricpositive}
(X(z) ,X(z)_{\geq 0}) \mbox{ is a positive geometry}
\ee
whenever $X(z)$ is {\it projectively normal} (which implies normality).  
It holds if and only if we have the following equality of lattice points in $\ZZ^{m+1}$
\be \label{eq:normal}
\Cone(z) \cap \sp_{\ZZ}(z) = \sp_{\ZZ_{\geq 0}}(z).
\ee
If the equality \eqref{eq:normal} does not hold, we can enlarge $z$ by including additional lattice points in $\Cone(z) \cap \sp_{\ZZ}(z)$ until it does.

The torus $T$ acts on $X(z)$ and the torus orbits are in bijection with the faces $F$ of the polytope $\A(z)$.  For each such face $F$, we denote by $X_F$ the corresponding torus orbit closure; then $X_F$ is again a projective toric variety, given by using the points $z_i$ that belong to the face $F$.  If $X(z)$ is projectively normal, then all the $X_F$ are as well.

\subsubsection{The canonical form of a toric variety}


The variety $X(z)$ has a distinguished rational top form $\Omega_{X(z)}$ of top degree.  The rational form $\Omega_{X(z)}$ is uniquely defined by specifying its restriction to the torus
$$
\Omega_{X(z)}|_T \deff \Omega_T \deff \prod_{i=1}^m \frac{dx_i}{x_i},
$$
where $x_i$ are the natural coordinates on $T$, and $\Omega_T$ is the natural holomorphic non-vanishing top form on $T$.  
In Appendix \ref{app:toricform}, we show that when $X(z)$ is projectively normal, the canonical form $\Omega_{X(z)}$ has a simple pole along each facet toric variety $X_F$ and no other poles, and furthermore, for each facet $F$ the residue $\Res_{X_F} \Omega_{X(z)}$ is equal to the canonical form $\Omega_{X_F}$ of the facet.  

This property of $\Omega_{X(z)}$ establishes \eqref{eq:toricpositive}, apart from the uniqueness part of Axiom (P2) which is equivalent to the statement that $X(z)$ has no nonzero holomorphic forms.  This is well known: when $X(z)$ is a smooth toric variety, this follows from the fact that smooth projective rational varieties $V$ have no nonzero holomorphic forms, or equivalently, have geometric genus $\dim H^0(V, \omega_V)$ equal to $0$.  Normal toric varieties have rational singularities so inherit this property from a smooth toric resolution.  We have thus established \eqref{eq:toricpositive}, and  the canonical form $\Omega(X(z),X(z)_{\geq 0})$ is $\Omega_{X(z)}$.

We remark that $\Omega_{X(z)}$ has no zeros, and thus $(X(z),X(z)_{\geq 0})$ is a simplex-like positive geometry.

\begin{example}
Take $n = 4$ and $m = 2$, with
$$
z_1 = (1,0,0), \qquad z_2 = (1,1,0), \qquad z_3 = (1,1,1), \qquad z_4 = (1,0,1).
$$
The polytope $\A(z)$ is a square.  The toric variety $X(z)$ is the closure in $\P^3$ of the set of points $\{(x,xy,xyz,xz) \mid x,y,z \in \C^* \}$, or equivalently of $\{(1,y,yz,z) \mid y,z \in \C^* \}$.  This closure is the quadric surface $C_1 C_3 - C_2 C_4 = 0$.  In fact, $X(z)$ is isomorphic to the $\P^1 \times \P^1$, embedded inside $\P^3$ via the Segre embedding.

The nonnegative part $X(z)_{\geq 0}$ is the closure of the set of points $\{(1,y,yz,z) \mid y,z \in \R_{>0} \}$, and is diffeomorphic to a square.  There are four boundaries, given by $C_i = 0$ for $i = 1,2,3,4$, corresponding to $y \to 0,\infty$ and $z \to 0,\infty$.  For example, when $z \to 0$ we have the boundary component $D = \overline{\{(1,y,0,0) \mid y \in \C^*\}} \subset \P^3$, which is isomorphic to $\P^1$.  In these coordinates, the canonical form is given by
$$
\Omega(X(z),X(z)_{\geq 0}) = \Omega_{X(z)} = \frac{dy}{y} \frac{dz}{z}.
$$
The residue $\Res_{z = 0}\Omega_{X(z)}$ is equal to $dy/y$, which is the canonical form of the boundary component $D \cong \P^1$ above.

\end{example}

The condition \eqref{eq:graded} that $z$ is graded implies that $\theta^z: (\C^*)^{m+1} \to T \subset X(z)$ factors as
\be
(\C^*)^{m+1} \longrightarrow (\C^*)^{m+1}/\C^* \stackrel{\ttheta^z}{\longrightarrow} T,
\ee where the quotient $S = (\C^*)^{m+1}/\C^*$ arises from the action of $t \in \C^*$ given by
\be
t \cdot (X_1,\ldots,X_{m+1}) = (t^{\ta_1} X_1,\ldots, t^{\ta_{m+1}}X_{m+1})
\ee
where $\ta \in \ZZ^{m+1}$ is a scalar multiple of $a \in \Q^{m+1}$ that is integral.  For example, if $z_i = (1,z'_i)$ for $z'_i \in \ZZ^m$ as in Section \ref{sec:push}, then $a = (1,0,\ldots,0)$ and $S$ can be identified with the subtorus $\{(1,X_1,X_2,\ldots,X_m)\} \subset (\C^*)^{m+1}$.  The map $\ttheta^z: S \to T$ is surjective, but may not be injective.  By Example \ref{ex:tori}, we have
\be \label{eq:pushtheta}
\ttheta^z_*(\Omega_S) = \Omega_T.
\ee

%
%
%

\subsection{Cluster varieties and their positive parts}\label{sec:cluster}
We speculate that reasonable cluster algebras give examples of positive geometries.  Let $A$ be a cluster algebra over $\C$ (of geometric type) and let $\oX = \Spec(A)$ be the corresponding affine cluster variety \cite{FZ}; thus the ring of regular functions on $\oX$ is equal to $A$.  We will assume that $\oX$ is a smooth complex manifold, see e.g. \cite{Mul,LS} for some discussion of this.  

The generators of $A$ as a ring are grouped into \defn{clusters} $(x_1,x_2,\ldots,x_n)$, where $n = \dim \oX$.  Each cluster corresponds to a subtorus $T \cong (\C^*)^n$ with an embedding $\iota_T: T \hookrightarrow \oX$.  Different clusters are related by \defn{mutation}:
\be \label{eq:mutation}
x_i x'_i = M + M',
\ee
swapping the coordinate $x_i$ for $x'_i$, where $M, M'$ are monomials in the $x_j, j \neq i$.  It is clear from \eqref{eq:mutation} that if $(x_1,\ldots,x_i, \ldots, x_n)$ are all positive real numbers, then so are $(x_1,\ldots,x'_i, \ldots, x_n)$.  We thus define the \defn{positive part} of $\oX$ to be $\oX_{>0} \deff \iota_T(\R_{>0}^n)$.

Furthermore, we define the canonical form $\Omega_{\oX} \deff \prod_{i=1}^n dx_i/x_i$.  By \eqref{eq:mutation}, we have
\be
x_i dx'_i + x'_i dx_i = dM + dM'
\ee
so wedging both sides with $\prod_{j \neq i} dx_j/x_j$, we deduce that the canonical form $\Omega_{\oX}$ does not depend on the choice of cluster.  In fact, under some mild assumptions, $\Omega_{\oX}$ extends to a holomorphic top-form on $\oX$.

We speculate that there is a compactification $X$ of $\oX$ such that 
\be \label{eq:cluster}
(X, X_{\geq 0} \deff \overline{\oX_{>0}}) \mbox{ is a positive geometry with canonical form } \Omega(X,X_{\geq 0}) = \Omega_{\oX}.
\ee

Furthermore, we expect that all boundary components are again compactifications of cluster varieties. We also expect that the compactification can be chosen so that the canonical form has no zeros.


\subsection{Flag varieties and total positivity}\label{sec:flagTP}
Let $G$ be a reductive complex algebraic group, and let $P \subset G$ be a parabolic subgroup.  The quotient $G/P$ is known as a generalized flag variety.  If $G = \GL(n)$, and $P = B \subset \GL(n)$ is the subgroup of upper triangular matrices, then $G/B$ is the usual flag manifold.  If $G = \GL(n)$ and 
\be
P = \left\{\left(\begin{array}{cc} A&B \\ 0 &C \end{array}\right)\right\} \subset \GL(n)
\ee
with block form where $A, B, C$ are respectively $k \times k$, $k \times (n-k)$, and $(n-k) \times (n-k)$, then $G/P \cong G(k,n)$ is the Grassmannian of $k$-planes in $n$-space.

In \cite{Lusztig}, the totally nonnegative part $(G/P)_{\geq 0} \subset G/P(\R)$ of $G/P$ was defined, assuming that $G(\R)$ is split over the real numbers.   We sketch the definition in the case that $G = \GL(n)$.  An element $g \in \GL(n)$ is called totally positive if all of its minors (of any size) are positive.  Denoting the totally positive part of $\GL(n)$ by $\GL(n)_{>0}$, we then define
$$
(\GL(n)/P)_{\geq 0} \deff \overline{\GL(n)_{>0} \cdot e},
$$
where $\GL(n)_{>0} \cdot e$ denotes the orbit of $\GL(n)_{>0}$ acting on a basepoint $e \in G/P$, which is the point in $G/P$ represented by the identity matrix.  In the case that $\GL(n)/P$ is the Grassmannian $G(k,n)$, we have $(G/P)_{\geq 0} = G(k,n)_{\geq 0}$, though this is not entirely obvious!

In \cite{Rietsch} it was shown that there is a stratification $G/P = \bigcup \oPi_u^w$ such that each of the intersections $\oPi_u^w \cap (G/P)_{\geq 0}$ is homeomorphic to $\R_{>0}^d$ for some $d$.  The closures $\Pi_u^w = \overline{\oPi_u^w}$ are known as \defn{projected Richardson varieties}, and in the special case $G/P \cong G(k,n)$, they reduce to the positroid varieties of Section \ref{sec:posGrassmannian} \cite{KLS,KLS2}.

The statement
\be \label{eq:GP}
(G/P, (G/P)_{\geq 0}) \mbox{ is a positive geometry}
\ee
was essentially established in \cite{KLS2}, but in somewhat different language.  Namely, it was proved in \cite{KLS2} that for each stratum $\Pi_u^w$ (with $G/P$ itself being one such stratum), there is a meromorphic form $\Omega_u^w$ with simple poles along the boundary strata $\{\Pi_{u'}^{w'}\}$ such that $\Res_{\Pi_{u'}^{w'}}\Omega_u^w =\alpha \cdot  \Omega_{u'}^{w'}$ for some scalar $\alpha$.  By identifying $\Omega_u^w$ with the push-forward of the dlog-form under the identification $\R_{>0}^d \cong \oPi_u^w \cap (G/P)_{\geq 0}$, we expect all the scalars $\alpha$ can be computed to be equal to 1.  Note that this also shows that $(\Pi_u^w,\Pi_u^w \cap (G/P)_{\geq 0})$ is itself a positive geometry.

We remark that it is strongly expected that $\oPi_u^w$ (and in particular the open stratum inside $G/P$) is a cluster variety \cite{Lec}.  Thus \eqref{eq:GP} is a special case of \eqref{eq:cluster}.  For example, the cluster structure of $G(k,n)$ is established in \cite{Scott}.

\section{Generalized polytopes}
\label{sec:genpolytopes}

In this section we investigate the much richer class of  \defn{generalized polytopes}, or \defn{polytope-like} geometries, which are positive geometries whose canonical form may have zeros.

\subsection{Projective polytopes}\label{sec:polytope}
The fundamental example is a convex polytope embedded in projective space.
Most of our notation was already established back in Section~\ref{sec:proj_simplex}. In Appendix \ref{app:cones}, we recall basic terminology for polytopes and explain the relation between projective polytopes and cones in a real vector space.

%

\subsubsection{Projective and Euclidean polytopes}
Let $Z_1,Z_2,\ldots,Z_n \in \R^{m+1}$, and denote by $Z$ the $n \times (m+1)$ matrix whose rows are given by the $Z_i$.  Define $\A\deff \A(Z)\deff\A(Z_1,Z_2,\ldots ,Z_n) \subset \P^m(\R)$ to be the convex hull
\be \label{eq:poly}
\A = \Conv(Z) = \Conv(Z_1,\ldots,Z_n) \deff \left\{\sum_{i=1}^n C_i Z_i \in \P^m(\R)  \mid C_i \geq 0, i=1,\ldots ,n\right\}.
\ee
We make the assumption that $Z_1,\ldots,Z_n$ are all vertices of $\A$.
In \eqref{eq:poly}, the vector $\sum_{i=1}^n C_i Z_i \in \R^{m+1}$ is thought of as a point in the projective space $\P^m(\R)$.  The polytope $\A$ is well-defined if and only if $\sum_{i=1}^n C_i Z_i$ is never equal to 0 unless $C_i = 0$ for all $i$.  A basic result, known as ``Gordan's theorem"~\cite{Ziegler}, states that this is equivalent to the condition:
\be \label{eq:Gordan}
\mbox{There exists a (dual) vector $X\in \R^{m+1}$ such that $Z_i\cdot X >0$ for $i = 1,2,\ldots,n$.}
\ee
The polytope $\A$ is called a convex \defn{projective polytope}.

Every projective polytope $(\P^m,\A)$ is a positive geometry.  This follows from the fact that every polytope $\A$ can be triangulated (see Section~\ref{sec:triangulations}) by projective simplices.  By Section \ref{sec:proj_simplex}, we know that every simplex is a positive geometry, so by the arguments in Section \ref{sec:triangulations} we conclude that $(\P^m,\A)$ is a positive geometry.  The canonical form $\Omega(\A)$ of a projective polytope will be discussed in further detail from multiple points of view  in Section \ref{sec:forms}.

It is clear that the polytope $\A$ is unchanged if each $Z_i$ is replaced by a positive multiple of itself.  This gives an action of the little group $\R_{>0}^n$ on $Z$ that fixes $\A$.  To visualize a polytope, it is often convenient to work with \defn{Euclidean polytopes} instead of projective polytopes.  To do so, we use the little group to ``gauge fix" the first component of $Z$ to be equal to 1 (if possible), so that $Z = (1,Z')$ where $Z' \in \R^m$.  The polytope $\A \subset \P^m$ can then be identified with the set
\be
\left\{\sum_{i=1}^n C_i Z'_i
 \subset \mathbb{R}^m \mid C_i \geq 0, i=1,\ldots ,n \text{ and } C_1 +C_2 + \cdots + C_n = 1\right\}
\ee
inside Euclidean space $\R^m$. The $C_i$ variables in this instance can be thought of as center-of-mass weights. Points in projective space for which the first component is zero lie on the $(m{-}1)$-plane at infinity.

The points $Z_1,\ldots,Z_n$ can be collected into a $n \times (m+1)$ matrix $Z$, which can be thought of as a linear map $Z: \R^n \to \R^{m+1}$ or a rational map $Z: \P^{n-1} \to \P^m$.  The polytope $\A$ is then the image $Z(\Delta^{n-1})$ of the standard $(n-1)$-dimensional simplex in $\P^{n-1}(\R)$.

\subsubsection{Cyclic polytopes}\label{sec:cyclic}
We call the point configuration $Z_1,Z_2,\ldots,Z_n$ \defn{positive} if $n \geq m+1$, and all the $(m+1) \times (m+1)$ {\sl ordered} minors of the matrix $Z$ are strictly positive.  Positive $Z$ always satisfy condition \eqref{eq:Gordan}.  In this case, the polytope $\A$ is known as a \defn{cyclic polytope}. For notational convenience, we identify $Z_{i+n} \deff Z_i$, so the vertex index is represented mod $n$.

For even $m$, the facets of the cyclic polytope are
\be
\Conv(Z_{i_1-1},Z_{i_1},\ldots,Z_{i_{m/2}-1},Z_{i_{m/2}})
\ee
for $1\le i_1{-}1<i_1<i_2{-}1<i_2 <\cdots <i_{m/2}{-}1<i_{m/2}\le n{+}1$.

For odd $m$, the facets are 
\be
\Conv(Z_1,Z_{i_1-1},Z_{i_1},\ldots,Z_{i_{(m{-}1)/2}-1},Z_{i_{(m{-}1)/2}})
\ee
for $2 \le i_1{-}1< i_1<i_2{-}1<i_2<\cdots<i_{(m{-}1)/2}{-}1<i_{(m{-}1)/2}\le n$ and
\be
\Conv(Z_{i_1-1},Z_{i_1},\ldots,Z_{i_{(m{-}1)/2}-1},Z_{i_{(m{-}1)/2}},Z_n)
\ee
for $1 \le i_1{-}1< i_1<i_2{-}1<i_2< \cdots<i_{(m{-}1)/2}{-}1<i_{(m{-}1)/2}\le n{-}1$.

This description of the facets is commonly known as \defn{Gale's evenness criterion}~\cite{Ziegler}.

An important example for the physics of scattering amplitudes in planar $\mathcal{N}=4$ super Yang-Mills theory is the $m=4$ cyclic polytope which has boundaries:
\be
\Conv(Z_{i-1},Z_i,Z_{j-1},Z_j)
\ee
for $1\le i{-}1<i<j{-}1<j\le n{+}1$. The physical applications are explained in Section~\ref{sec:scattering}.

\subsubsection{Dual polytopes}
\label{sec:dual}
Let $(\P^m,\A)$ be a convex polytope and let $Y\in\P^m(\R)$ be a point away from any boundary component. We now define the \defn{dual of $\A$ at $Y$}, denoted $\A_Y^*$, which is a convex polytope in the linear dual of $\P^m$ (also denoted $\P^m$). For the moment let us ``de-projectivize" $Y$ so that $Y\in\R^{m{+}1}$. 

Recall that each facet of $\A$ is given by the zero-set (along $\partial\A$) of some dual vector $W\in\R^{m{+}1}$. Before going to projective space, we pick the overall sign of $W$ so that $W\cdot Y>0$ for our $Y$. We say that the facets are \defn{oriented relative to $Y$}. Now assume that the facets of $\A$ are given by $W_1,\ldots, W_r$. Then we define:
\be \label{eq:dualY}
\A_Y^*\deff \Conv(W_1,\ldots, W_r) \deff\left\{ \sum_{j=1}^r C_jW_j\in\P^m\mid C_j\geq 0,j=1,\ldots, r \right\}
\ee
Not that the (relative) signs of the $W_j$'s are crucial, hence so is the position of $Y$ relative to the facets.

In the special case where $Y\in\Int(\A)$, we let $\A^*\deff \A_Y^*$ and refer to this simply as the \defn{dual of $\A$}. An equivalent definition of the dual of $\A$ is:
\be
\A^* = \{W \in \mathbb{P}^m\mid W\cdot Y \geq 0\text{ for all } Y\in \A\}.
\ee
A priori, the inequality $W \cdot Y \geq 0$ may not make sense when $W$ and $Y$ are projective.  As in \eqref{eq:dualY}, we give it precise meaning by working first in $\R^{m+1}$, and then taking the images in $\P^m$ (see also Appendix \ref{app:cones}).

For generic $Y$, we also wish to assign an orientation to $\A_Y^*$ as follows. Suppose we orient the facets $W_1,\ldots, W_r$ relative to the {\sl interior} of $\A$. Let $s$ denote the number of terms $W_j\cdot Y$ that are negative. Then we orient $\A_Y^*$ based on the parity of $s$. In particular, the dual for $Y\in\Int(\A)$ is positively oriented.

It should be obvious that $(\P^m,\A_Y^*)$ is a positive geometry for each $Y$.

An important observation about the dual polytope $\A_Y^*$ is that it has ``opposite" combinatorics to $\A$. In other words, vertices of $\A$ correspond to the facets of $\A_Y^*$ and vice versa.  Since we assumed that $\A$ has $n$ vertices $Z_1,Z_2,\ldots,Z_n$, the dual polytope $\A_Y^*$ has $n$ facets corresponding to $\{W\mid W\cdot Z_i = 0\}$ for $i = 1,2,\ldots,n$. Suppose a facet of $\A$ is adjacent to vertices $Z_{i_1},\ldots,Z_{i_{m}}$, then the dual vertex $W$ satisfies $W\cdot Z_{i_1}=\ldots=W\cdot Z_{i_m}=0$ so that
\be
W_I=\epsilon_{II_1\ldots I_m}Z_{i_1}^{I_1}\ldots Z_{i_m}^{I_m}
\ee 
where we have ordered the vertices so that $Y\cdot W>0$. This is a straightforward generalization of~\eqref{eq:W_simplex}. Sometimes we will also write
\be\label{eq:W_notation}
W=(i_1\ldots i_m)
\ee
to denote the same quantity. Applying this to Section~\ref{sec:cyclic} gives us all the vertices of the polytope dual to a cyclic polytope.



We now state an important fact about dual polytopes. Given a {\sl signed} triangulation $\A=\sum_i\A_i$ of a convex polytope $\A$ by other convex polytopes $\A_i$ (see Sections~\ref{sec:triangulations} and~\ref{sec:Gro}), and a point $Y$ not along any boundary component, we have
\be
{\A=\sum_i\A_i}\;\;\;\;\Rightarrow\;\;\;\;{\A_Y^* =\sum_i\A_{iY}^*}
\ee
In words, we say that
\be\label{eq:commute}
\mbox{Dualization of polytopes ``commutes" with triangulation.}
\ee
This is a crucial geometric phenomenon to which we will return. While we do not provide a direct geometric proof, we will argue its equivalence to the triangulation independence of the canonical form and the existence of a volume interpretion for the form in Section~\ref{sec:dualpolytopeform}.

\subsection{Generalized polytopes on the projective plane}
\label{sec:polytopeplane}
Let us now discuss a class of positive geometries in $\P^2$ which includes Examples \ref{ex:circSegment} and \ref{ex:quadratic}.  Let $C \subset \R^2$ be a closed curve that is piecewise linear or quadratic.  Thus $C$ is the union of curves $C_1,C_2,\ldots,C_r$ where each $C_i$ is either a line segment, or a closed piece of a conic.  We assume that $C$ has no self-intersections, and let $\mathcal{U} \subset \R^2$ be the closed region enclosed by $C$.  We will further assume that $\mathcal{U}$ is a convex set.  Define the degree $d(\mathcal{U})$ of $\mathcal{U}$ to be the sum of the degrees of the $C_i$.  We now argue that if $d \geq 3$, 
\be
(\P^2,\mathcal{U}) \mbox{ is a positive geometry.}
\ee
We will proceed by induction on $d = d(\mathcal{U})$.  For the base case $d = 3$, there are two possibilities: (a) $\mathcal{U}$ is a triangle, or (b) $\mathcal{U}$ is a convex region enclosed by a line and a conic.  For case (a), $\Omega(\mathcal{U})$ was discussed in Section \ref{sec:proj_simplex}.  For case (b), $\Omega(\mathcal{U})$ was
studied in Example~\ref{ex:quadratic}. In both cases, $(\P^2,\mathcal{U})$ is a positive geometry.

Now suppose that $d(\mathcal{U}) = d \geq 4$ and that one of the $C_i$ is a conic.  Let $L$ be the line segment joining the endpoints of $C_i$.  By convexity, $L$ lies completely within $\mathcal{U}$, and thus decomposes $\mathcal{U}$ into the union $\mathcal{U} = \mathcal{U}_1 \cup \mathcal{U}_2$ of two regions where $\mathcal{U}_1$ has $C_i$ and $L$ as its ``sides", while $\mathcal{U}_2$ satisfies the same conditions as $\mathcal{U}$, but has more linear sides. In other words, $\mathcal{U}$ is triangulated by $\mathcal{U}_{1,2}$, and by the discussion in Section~\ref{sec:triangulations}, $\mathcal{U}$ is a pseudo-positive geometry with $\Omega(\mathcal{U})=\Omega(\mathcal{U}_1) + \Omega(\mathcal{U}_2)$. In addition, we argue that $\mathcal{U}$ must be a positive geometry, since all its boundary components (line segments on the projective line) are positive geometries.

If none of the $C_i$ is a conic, then $\mathcal{U}$ is a convex polygon in $\R^2$.  We can slice off a triangle and repeat the same argument.

Let us explicitly work out a simple example defined by two linear boundaries and one quadratic: a ``pizza" slice.

\begin{example}\label{ex:pizza}
Consider a ``pizza" shaped geometry; that is, a sector $\mathcal{\mathcal{T}}(\theta_1,\theta_2)$ of the unit circle between polar angles $\theta_1,\theta_2$, which is bounded by two linear equations $q_1=-x\sin\theta_1+y\cos\theta_1\geq 0$ and $q_2=x\sin\theta_2-y\cos\theta_2\geq 0$, and an arc $f=1-x^2-y^2\geq 0$. Let us assume for simplicity of visualization that $0\le\theta_1<\theta_2\le\pi$. See Figure~\ref{fig:intro_d} for a picture for the case $(\theta_1,\theta_2)=(\pi/6,5\pi/6)$.

The canonical form is given by 
\be\label{eq:pizza}
\Omega(\mathcal{\mathcal{T}}(\theta_1,\theta_2))=\frac{\left[\sin(\theta_2-\theta_1)+(-x\sin\theta_1+y\cos\theta_1)+(x\sin\theta_2-y\cos\theta_2)\right]}{(1-x^2-y^2)(-x\sin\theta_1+y\cos\theta_1)(x\sin\theta_2-y\cos\theta_2)}dxdy
\ee
Let us take the residue along the linear boundary $q_1=0$ via the limit $x\rightarrow y \cot\theta_1$, and confirm that the result is the canonical form on the corresponding boundary component. We get
\be
\Res_{q_1{=}0}\Omega(\mathcal{T}(\theta_1,\theta_2))=\frac{(\sin\theta_1) }{(\sin\theta_1-y)y}dy
\ee
which is the canonical form on the line segment $y\in[0,\sin\theta_1]$ with positive orientation. The upper bound $y<\sin\theta_1$ is simply the vertical height of the boundary $q_1=0$.

Similarly, the residue at $q_2=0$ is given by
\be
\Res_{q_2{=}0}\Omega(\mathcal{T}(\theta_1,\theta_2))=-\frac{(\sin\theta_2)}{(\sin\theta_2-y)y}dy
\ee
which is the canonical form on the line segment $y\in [0,\sin\theta_2]$, again with negative orientation.

The residue along the arc can be computed in a similar manner as for the segment of the disk in Example~\ref{ex:circSegment}, so we leave this as an exercise for the reader.

For $\theta_1=0,\theta_2=\pi$, we are reduced to the northern half disk from Example~\ref{ex:circSegment}, so $\mathcal{T}(0,\pi)=\mathcal{S}(0)$. A quick substitution shows that the canonical forms match as well.
\end{example}

\subsection{A naive positive part of partial flag varieties}\label{sec:flag}
In this section, let us consider a \defn{partial flag variety} $\Fl(n;\k)$ with $k\deff(k_1<k_2<\cdots < k_r)$ where
\be
\Fl(n;\k)\deff \{V = (V_1 \subset V_2 \subset \cdots \subset V_r \subset \C^n) \mid \dim V_i = k_i\}
\ee
is the space of flags, whose components are linear subspaces of specified dimensions $k_1 < k_2 < \cdots < k_r $.  Compared to Section \ref{sec:flagTP}, we have $\Fl(n;\k) = \GL(n)/P$ for an appropriate choice of parabolic subgroup $P$.

In Section \ref{sec:flagTP}, we defined the totally nonnegative part $\Fl(n;\k)_{\geq 0}$.
We now define the \defn{naive nonnegative part} $\Fl(n;\k)_{\tgeq 0}$ of $\Fl(n;\k)$ to be the locus of $V$ where $V_i \in G_{\geq 0}(k_i,n)$ for all $i$.  The $\tgeq 0$ is to remind us that $\Fl(n;\k)_{\tgeq 0}$ may differ from the totally nonnegative part $\Fl(n;\k)_{\geq 0}$. We speculate that
\be
(\Fl(n;\k), \Fl(n;\k)_{\tgeq 0}) \mbox{ is a positive geometry.}
\ee
The naive nonnegative part $\Fl(n;\k)_{\tgeq 0}$ has the following symmetry property that the totally nonnegative part $\Fl(n;\k)_{\geq 0}$ lacks: if all the $k_i$ have the same parity, then the cyclic group $\ZZ/n\ZZ$ acts on $\Fl(n;\k)_{\tgeq 0}$.  To see this, we represent points of $\Fl(n;\k)_{\tgeq 0}$ by $k_r \times n$ full-rank matrices $C$ so that $V_i$ is the span of the first $k_i$ rows.  If $C_1,C_2,\ldots,C_n$ denote the $n$ columns of $C$, then the cyclic group acts by sending $C$ to the $C'$ with columns $\pm C_n, C_1,C_2,\ldots,C_{n-1}$, where the sign is $+$ if the $k_i$ are odd, and $-$ otherwise.

Unlike $\Fl(n;\k)_{\geq 0}$, the naive nonnegative part $\Fl(n;\k)_{\tgeq 0}$ is in general a polytope-like positive geometry: the canonical forms will have zeros.
The case $\Fl(n;1,3)_{\tgeq 0}$ is studied in some detail in \cite{Bai:2015qoa}, and we will discuss its canonical form in Section \ref{sec:1loop}.  

\subsection{$L$-loop Grassmannians}\label{sec:loopgrass}

We now define the \defn{$L$-loop Grassmannian} $G(k,n; \k)$, where $\k\deff (k_1,\ldots ,k_L)$ is a sequence of positive integers. 
A point $V$ in the $L$-loop Grassmannian $G(k,n;\k)$ is a collection of linear subspaces $V_S\subset \C^n$ indexed by $S$, where $S\deff\{s_1,\ldots ,s_l\}$ is any subset of $\{1,2,\ldots,L\}$ for which $k_S\deff k_{s_1}{+}\ldots{+}k_{s_L} \le n{-}k$. Moreover, we require that $\dim V_S = k + k_S$, and $V_S \subset V_{S'}$ whenever $S \subset S'$.  For simplicity we sometimes write 
$V_i=V_{\{i\}}, V_{ij}=V_{\{i,j\}}$, and so on and so forth. In particular, $V_\emptyset\subset V_s$ for any one-element set $S = \{s\}$.

We say that a point $V \in G(k,n;\k)$ is ``generic" if $V_i\cap V_j=V_\emptyset$ for any $i\neq j$. For such points, we have $V_S = \sp(V_{s_1},V_{s_2}, \ldots, V_{s_\ell})$ for any $S = \{s_1,s_2,\ldots,s_\ell\}$, so that $V$ is determined completely by $V_1,V_2,\ldots ,V_L$. The space $G(k,n;\k)$ is naturally a subvariety of a product of Grassmannians, and in particular it is a projective variety.

If $L = 0$, then the $0$-loop Grassmannian reduces to the usual Grassmannian $G(k,n)$.  If $L = 1$, the $1$-loop Grassmannian reduces to the partial flag variety $\Fl(n;k,k+k_1)$.  We may refer to the $0$-loop Grassmannian as the \defn{tree Grassmannian}. The distinction between ``trees" and ``loops" comes from the terminology of scattering amplitudes.  We caution that ``loop Grassmannian" commonly refers to another infinite dimensional space in the mathematical literature, which is also called the affine Grassmannian.

We now define the positive part $G_{> 0}(k,n;\k)$ of $G(k,n;\k)$.   Consider the set of $(k{+}K)\times n$ matrices $M(k+K,n)$, where $K\deff k_1+\ldots+k_L$. We will denote each matrix as follows:
\be\label{eq:loopMatrix}
P\deff 
\begin{pmatrix}
C\\ D_1\\ \ldots \\ D_L
\end{pmatrix}
\ee
where $C$ has $k$ rows, and $D_i$ has $k_i$ rows for $i=1,\ldots,L$.  We say that $P$ is \defn{positive} if for each $S = \{s_1,s_2,\ldots,s_\ell\}$ the matrix formed by taking the rows of $C, D_{s_1}, D_{s_2},\ldots,D_{s_\ell}$ is positive, that is, has positive $(k+k_S) \times (k+k_S)$ minors.
Each positive matrix $P$ gives a point $V \in G(k,n;\k)$, where $V_S$ is the span of the rows of $C, D_{s_1}, D_{s_2},\ldots,D_{s_\ell}$. Two points $P_{1,2}$ are equivalent if they map to the same point $V$.   Each equivalence class defines a point on the (strictly) positive part $G_{>0}(k,n;\k)$.  The nonnegative part $G_{\geq 0}(k,n;\k)$ is the closure of $G_{>0}(k,n;\k)$ in $G(k,n;\k)(\R)$. 


There exists a subgroup $\G(k;\k)$ of $\GL(k+K)$ acting on the left whose orbits in $M(k+K,n)$ are equivalence classes of matrices $P$ defining the same point $V$.  Elements of $\G(k;\k)$ allow row operations within each of $C$, or the $D_i$, and also allows adding rows of $C$ to each $D_i$.  

We caution that not every point $V \in G(k,n;\k)$ is representable by a matrix $P$.  For example, for $G(0,3;1,1)$, the $P$ matrix is a $2 \times 3$ matrix.  Consider a point on the boundary where the two rows of $P$ are scalar multiples of each other, then $V_1 = V_2$ and additional information is required to specify the 2-plane $V_{1,2}$. So there are even points in the boundary of $G_{\geq 0}(k,n;\k)$ that cannot be represented by $P$.

Let us focus our attention on the $L$-loop Grassmannian $G(k,n;\ell^L)$ where $l^L \deff (l,\ldots,l)$ with $l$ appearing $L$ times.  The case $l=2$ is the case of primary physical interest.  The $L$-loop Grassmannian $G(k,n;l^L)$ has an action of the symmetric group $S_L$ on the set $\{1,2,\ldots,L\}$ with $L$ elements.  For a permutation $\sigma \in S_L$,  we have $\sigma(V)_S = V_{\sigma(S)}$.  In addition, when $l$ is even, the action of $S_L$ on $G(k,n;l^L)$ preserves the positive part: permuting the blocks $D_i$ of the matrix $P$ preserves the positivity conditions.
%

The $L$-loop Grassmannian $G(k,n;\k)$ is still very poorly understood in full generality. For instance, a complete stratification extending the positroid stratification of the positive Grassmannian is still unknown. The existence of the canonical form is also unknown. Nonetheless, we have identified the canonical form for some $L=1$ cases, which we discuss in Section~\ref{sec:1loop}, and some $L=2$ cases.

The proven existence of some $L=1$ canonical forms is highly non-trivial. We therefore speculate that 
\begin{equation}
\mbox{$(G(k,n;\k),G_{\geq 0}(k,n;\k))$ is a positive geometry}.
\end{equation}
Note that for even $\ell$, the speculative positive geometry $G(k,n;\ell^L)$ should have an action of $S_L$: the symmetric group acts on the boundary components, and furthermore the canonical form will be invariant under $S_L$.

We may also denote a generic point in the loop Grassmannian as
\be\label{eq:ampspace}
\mathcal{Y}=(Y,Y_1,\ldots,Y_L)
\ee
where $V_S=\sp\{Y,Y_{s_1},\ldots,Y_{s_l}\}$ for any $S$. Or, for $\k=2^L$ we may use notation like $(Y,(AB)_1,\ldots,(AB)_L)$ or $(Y,AB,CD,EF,\ldots)$, which are common in the physics literature.

\begin{example}\label{twoloop}
Consider the space $G(0,4;2^2)$ which is isomorphic to the ``double Grassmannian" $(G(2,4)\times G(2,4))$, and has an action of the symmetric group $S_2$.
Let $(C_1,C_2)\in G(2,4)^2$ denote a point in the space, with both $C_1$ and $C_2$ thought of as $2\times 4$ matrices modded out by $\GL(2)$ from the left. The interior of the positive geometry is given by the following points in matrix form:
\be
(G(2,4)^2)_{> 0}\deff \{(C_1,C_2)\;:\: C_1,C_2\in G_{> 0}(2,4)\text{ and }\det\begin{pmatrix}C_1\\C_2\end{pmatrix}> 0\}
\ee
In other words, both $C_1,C_2$ are in the positive Grassmannian, and their combined $4\times 4$ matrix (i.e. the two rows of $C_1$ stacked on top of the two rows of $C_2$) has positive determinant. 

We can parametrize the interior with eight variables as follows:
\be
\begin{pmatrix}C_1\\C_2\end{pmatrix}=
\begin{pmatrix}1&x_1&0&-w_1\\
0&y_1&1&z_1\\
1&x_2&0&-w_2\\
0&y_2&1&z_2
\end{pmatrix}
\ee
The conditions $C_1,C_2\in G_{> 0}(2,4)$ impose that all eight variables be positive, while the final condition requires
\be
\det\begin{pmatrix}C_1\\C_2\end{pmatrix}=-(x_1-x_2)(z_1-z_2)-(y_1-y_2)(w_1-w_2)>0
\ee
The canonical form can be computed by a triangulation argument given in~\cite{Arkani-Hamed:2013kca}:
\be\label{eq:2loops}
\Omega((G(2,4)^2)_{\geq 0}) = \frac{dx_1dy_1dz_1dw_1dx_2dy_2dz_2dw_2(x_1z_2{+}x_2z_1{+}y_1w_2{+}y_2w_1)}{x_1x_2y_1y_2z_1z_2w_1w_2[(x_1{-}x_2)(z_1{-}z_2){+}(y_1{-}y_2)(w_1{-}w_2)]}
\ee
The 9 poles appearing in the denominator account for the boundaries defined by the 9 inequalities. Note that the form is symmetric under exchanging $1\leftrightarrow 2$, as expected from the action of $S_2$.

Let us illustrate the simple method by which this result is obtained
by looking at a smaller example: a 4-dimensional
boundary. Let us go to the
boundary where $y_{1,2}=w_{1,2}=0$. The geometry is then simply given by
\begin{equation}
x_{1,2}>0, z_{1,2}>0, \; (x_1 - x_2)(z_1 - z_2)<0
\end{equation}
But this can clearly be triangulated (see Section~\ref{sec:triangulations}) in two pieces. We either have
$z_1>z_2>0$ and $x_1<x_2$, or vice-versa. We can trivially triangulate
$z_1>z_2>0$ by saying $z_2 = a, z_1 = a + b$ with $a>0,b>0$ and so the
form is $da\;db/ab  = dz_2dz_1/z_2(z_1 - z_2)$. Thus the full form is
\be
dx_1 dx_2 dz_2 dz_2 &\times& \left(\frac{1}{z_2 (z_1 - z_2)}
\frac{1}{x_1 (x_2 - x_1)} + \frac{1}{z_1 (z_2 - z_1)} \frac{1}{x_2
(x_1 - x_2)} \right) \\
&=& \frac{dx_1 dx_2 dz_1 dz_2 (x_2 z_1 + x_1
z_2)}{x_1 x_2 z_1 z_2 (x_2 - x_1)(z_1 - z_2)}.
\ee
which of course can be obtained from~\eqref{eq:2loops} by taking residues at the corresponding boundaries.
\end{example}

\subsection{Grassmann, loop and flag polytopes}\label{sec:grassmann}
Let us begin by reviewing the construction of \defn{Grassmann polytopes} \cite{Lam:2015uma}, which will motivate the construction of \defn{flag polytopes} and \defn{loop polytopes}.

Let $Z_1,Z_2,\ldots,Z_n \in \R^{k+m}$ be a collection of vertices.  The linear map $Z: \R^n \to \R^{k+m}$ induces a rational map $Z: G(k,n) \to G(k,k+m)$ in the obvious way.  If the map $Z$ is well-defined on $G_{\geq 0}(k,n)$, we define the image 
\be
Z(G_{\geq 0}(k,n)) = \{C \cdot Z \mid C \in G_{\geq 0}(k,n)\}, \qquad \mbox{or more generally } Z(\Pi_{f,\geq 0})
\ee
to be a Grassmann polytope. Note that we allow $Z(\Pi_{f,\geq 0})$ in the definition because we want faces of Grassmann polytopes to also be Grassmann polytopes.

In \cite{Lam:2015uma}, it is shown that the following condition
\be \label{eq:GrassGordan}
\mbox{There exists a $(k{+}m) {\times} k$ real matrix $M$ so that $Z {\cdot}M$ has positive $k {\times} k$ minors.}
\ee
implies that $Z(G_{\geq 0}(k,n))$ is well-defined.  When $k=1$, \eqref{eq:GrassGordan} reduces to \eqref{eq:Gordan}.   In \cite{Lam:2015uma}, \eqref{eq:Gordan} was used to define Grassmann polytopes but recently Galashin has announced that there exist images $Z(G_{\geq 0}(k,n))$ that do not satisfy \eqref{eq:GrassGordan}.  See~\cite{Lam:2015uma,Karp} for a discussion of\eqref{eq:GrassGordan} and other related criteria. 

We now define \defn{flag polytopes} and \defn{loop polytopes} generalizing Grassmann polytopes, and give a condition similar to \eqref{eq:GrassGordan}. We may sometimes refer to Grassmann polytopes also as \defn{tree Grassmann polytopes} and loop polytopes as \defn{loop Grassmann polytopes}.

Let $Z$ be a full rank $n \times (k{+}m)$ real matrix thought of as a linear map $Z: \R^n \to \R^{k+m}$.  Set
$$
X \deff \Fl(n;\k) \qquad \mbox{and} \qquad d \deff k_r = \max_i (k_i)
$$ 
for the former and assume that $k{+}m \geq d$, and 
$$
X\deff G(k,n;\k) \qquad \mbox{and} \qquad d\deff\max\{\dim V_S|\dim V_S \leq k{+}m\}
$$ 
for the latter. 
In the first case, we have a rational map $Z: \Fl(\k;n) \to \Fl(\k;k+m)$ sending $V_i$ to $Z(V_i)$ for $i=1,\ldots,r$.  In the second case we have a rational map $Z: G(k,n;\k) \to G(k,k+m;\k)$  sending $V_S$ to $Z(V_S)$.  We then define the flag polytope, or loop polytope to be the image $Z(X_{\geq 0})$, whenever the map $Z$ is well-defined on $X_{\geq 0}$ (here $X_{\geq 0}$ refers to the naive nonnegative part $\Fl(n;\k)_{\tgeq 0}$ in the case $X = \Fl(n;\k)$).  In the case of $X = G(k,n;0)$, this reduces to the definition of a Grassmann polytope.  We speculate that
\be \label{eq:Grassmannpositive}
Z(X_{\geq 0}) \mbox{ is a positive geometry,}
\ee
where the ambient complex variety is taken to be the Zariski closure of $Z(X_{\geq 0})$.

We now introduce the condition
\be \label{eq:loopGordan}
\mbox{There exists a $(k{+}m){\times}d$ real matrix $M$ such that $Z{\cdot}M$ has positive $d{\times}d$ minors.}\;\;
\ee
generalizing \eqref{eq:GrassGordan}.  We claim that
\be\label{eq:Zimage}
\mbox{Under \eqref{eq:loopGordan}, the image $Z(X_{\geq 0})$ is well-defined in both cases.}  
\ee
Note that any positive $Z$ satisfies \eqref{eq:loopGordan} \cite[Lemma 15.6]{Lam:2015uma}.

Let us now prove \eqref{eq:Zimage}.  In \cite{Lam:2015uma}, it is shown that $Z(G_{\geq 0}(d,n))$ is well-defined if $Z$ satisfies \eqref{eq:loopGordan}.  We shall show that \eqref{eq:loopGordan} implies the same statement for all $1 \leq d' \leq d$.  This will show that $Z(V_s)$ in the $\Fl(\k)$ case (resp. $Z(V_S)$ in the $G(k,n;l^L)$ case) is well-defined for $V \in X_{\geq 0}$, proving \eqref{eq:Zimage}.

We think of $Z$ as a point in $G(k+m,n)$.  The condition \eqref{eq:loopGordan} is equivalent to the condition that there exists $U \in G_{>0}(d,n)$ such that $U \subset Z$, i.e. $Z$ contains a totally positive subspace of dimension $d$.  In \cite[Lemma 15.6]{Lam:2015uma}, it is shown that if $U \in G_{>0}(d,n)$ then it contains $U' \in G_{>0}(d',n)$ for all $d' \leq d$.  It follows that \eqref{eq:loopGordan} implies the same condition for all values $d' \leq d$, completing the proof.

\subsection{Amplituhedra and scattering amplitudes}
\label{sec:scattering}
Following the discussion in Section~\ref{sec:grassmann}, let us further suppose that $Z$ is positive: all the ordered $(k+m) \times (k+m)$ minors are positive. In other words, the rows of $Z$ form the vertices of a cyclic polytope. Then the Grassmann polytope
\be
\A(k,n,m) \deff Z(G_{\geq 0}(k,n))
\ee
is known as the \defn{tree Amplituhedron} \cite{Amplituhedron,Arkani-Hamed:2013kca}. For instance, the tree Amplituhedron for $k = 1$ is a cyclic polytope in $\P^m(\R)$.

Now consider the loop Grassmannian $G(k,n;l^L)$ . The corresponding $L$-loop Grassmann polytope is called the \defn{L-loop Amplituhedron}~\cite{Amplituhedron,Arkani-Hamed:2013kca}:
\be \label{eq:defampli}
\A(k,n,m;l^L)\deff Z(G_{\geq 0}(k,n;l^L))
\ee
We also refer to this space as the \defn{Amplituhedron at $L$-loops} or simply the \defn{Amplituhedron} whenever the \defn{loop level} $L$ is understood. In particular, the $0$-loop Amplituhedron is the tree Amplituhedron.  The following special case of \eqref{eq:Grassmannpositive} is at the heart of the connection between scattering amplitudes and our work.

\begin{conjecture}\label{conj:Amplituhedron}
Grassmann polytopes (both trees and loops), particularly Amplituhedra, are positive geometries. 
\end{conjecture}

It follows from the Tarski-Seidenberg theorem (see Appendix~\ref{app:tarski}) that tree and loop Grassmann polytopes are all semialgebraic sets.  In other words, they are given as a finite union of sets that can be ``cut out" by polynomial equations in the Pl\"ucker coordinates and polynomial inequalities in the Pl\"ucker coordinates.  While the Tarski-Seidenberg theorem guarantees that our geometries are semi-algebraic, it does not provide us with the homogeneous inequalities needed to ``cut out" the geometries. Identifying such inequalities is an outstanding problem, progress on which will be reported in~\cite{winding}. However, this still does not prove the existence of the canonical form.



The $m=4,l=2$ Amplituhedron is the most interesting case for physics, because it appears to provide a complete geometric formulation of all the planar scattering amplitudes in $\mathcal{N}=4$ super Yang-Mills. More precisely,
\be
\text{the $n$-particle N$^k$MHV tree amplitude } &=& \Omega(\A(k,n,4))\\ \text{the integrand of the $n$-particle N$^k$MHV $L$-loop amplitude }&=&\Omega(\A(k,n,4;2^L))
\ee
We will often denote the \defn{physical Amplituhedron} simply as $\A(k,n;L)$, and the physical tree Amplituhedron more simply as $\A(k,n)$. Historically, the scattering amplitudes were first computed using techniques from quantum field theory, and were subsequently recognized as top-degree meromorphic forms with simple poles on the boundary of the Amplituhedron and unit leading residues, thus providing the original motivation for the study of canonical forms. The existence of the scattering amplitudes provides strong evidence for the existence of canonical forms on Amplituhedra.

\section{Canonical forms}
\label{sec:forms}

The main purpose of this section is to establish a list of methods and strategies for computing the canonical form of positive geometries. A summary of the main methods is given as follows:
\begin{itemize}
\item
{\bf Direct construction from poles and zeros}: We propose an ansatz for the canonical form as a rational function and impose appropriate constraints from poles and zeros.

\item
{\bf Triangulations}: We triangulate a generalized polytope by generalized simplices and sum the canonical form of each piece.

\item
{\bf Push-forwards}: We find morphisms from simpler positive geometries to more complicated ones, and apply the push-forward via Heuristic~\ref{heuristic}.

\item
{\bf Integral representations}: We find expressions for the canonical form as a volume integral over a ``dual" geometry, or as a contour integral over a related geometry.

\end{itemize}

\subsection{Direct construction from poles and zeros}\label{sec:matching}
Consider a positive geometry $(X,X_{\geq 0})$ of dimension $m$ for which there exists a degree-one morphism $\Phi:(\P^m,\A)\rightarrow (X,X_{\geq 0})$ for some positive geometry $\A$ in projective space. By Heuristic \ref{heuristic}, it suffices to compute the canonical form on $\A$ before pushing forward the result onto $X_{\geq 0}$ via $\Phi$. Since the map is of degree one, the push-forward is usually trivial.

Suppose $\A$ is defined by homogeneous polynomial inequalities $q_i(Y)\geq 0$ indexed by $i$ for $Y\in\P^m$. Then an ansatz for the canonical form is the following:
\be
\Omega(\A)=\frac{q(Y)\lb Yd^mY\rb}{\prod_i q_i(Y)}
\ee
for some homogeneous polynomial $q(Y)$ in the numerator which must satisfy:
\be\label{eq:degree}
\deg q = \sum_i \deg q_i-m-1
\ee
so that the form is invariant under \defn{local} $\GL(1)$ action $Y\rightarrow \alpha(Y)Y$. The method of \defn{undetermined numerator} is the idea that the numerator can be solved by imposing residue constraints. Note that this method operates under the \defn{assumption} that a solution to the numerator exists, which in most cases is a non-trivial fact.

We illustrate the idea with a few simple examples below.

\begin{example}\label{ex:square}
Consider the quadrilateral $\A\deff \A(Z_1,Z_2,Z_3,Z_4)$ in $\P^2(\R)$ with facets given by the four inequalities $q_1=x\geq 0,\;q_2=2y-x\geq 0,\;q_3=3-x-y\geq 0,\;q_4=2-y\geq 0$. The picture is given in Figure~\ref{fig:intro_b}, and the vertices are 
\be
Z_1^I=(1,0,0),\;\;
Z_2^I=(1,2,1),\;\;
Z_3^I=(1,1,2),\;\;
Z_4^I=(1,0,2)
\ee
with $(1,x,y)\in \P^2(\R)$ as usual.

We will derive the canonical form with the following ansatz.
\be
\Omega(\A)=\frac{(A+Bx+Cy)dxdy}{x(2y-x)(3-x-y)(2-y)}
\ee
for undetermined coefficients $A,B,C$. Note that the numerator must be linear by~\eqref{eq:degree}.

There are six (4 choose 2) double residues. Those corresponding to vertices of the quadrilateral must have residue $\pm 1$ (the sign is chosen based on orientation), while those corresponding to two opposite edges must have residue zero. We list these requirements as follows, where we denote $\Res_{ji}\deff\Res_{q_j=0}\Res_{q_i=0}$:

\be
\Res_{12}&=&\frac{A}{12}=+1\\
\Res_{23}&=&\frac{A+2B+C}{6}=+1\\
\Res_{34}&=&\frac{A+B+2C}{3}=+1\\
\Res_{41}&=&\frac{A+2C}{4}=+1\\
\Res_{13}&=&\frac{A+3C}{6}=0\\
\Res_{24}&=&-\frac{A+4B+2C}{12}=0
\ee
By inspection, the only solution is $(A,B,C)=(12,{-}1,{-}4)$. It follows that
\be\label{eq:square}
\Omega(\A)=\frac{(12-x-4y)dxdy}{x(2y-x)(3-x-y)(2-y)}
\ee
We observe that since there are many more equations than undetermined coefficients, the existence of a solution is non-obvious.

\end{example}

This method becomes complicated and intractable pretty fast. However, in the case of cyclic polytopes, a general solution was identified in~\cite{Arkani-Hamed:2014dca} which we review below.

\subsubsection{Cyclic polytopes}\label{sec:numerator}

We now apply the numerator method to the
cyclic polytope geometry, described in \cite{Arkani-Hamed:2014dca}. Let us illustrate how it works
for the case of a quadrilateral $\A \deff\A(Z_1,Z_2,Z_3,Z_4)$. We know that the form must have poles
when $\ip{Y12}, \ip{Y23}, \ip{Y34}, \ip{Y41} \to 0$; together with weights this tells us that
\begin{equation}
\aOmega(\A) = \frac{L_I Y^I}{2!\ip{Y12}\ip{Y23}\ip{Y34}\ip{Y41}}
\end{equation}
for some $L_I$. We must also require that the codimension 2
singularities of this form only occur at the corners of the
quadrilateral. But for generic $L$, this will not be the case; writing $(ij)$ for the line $\ip{Yij}=0$, we will
also have singularities 
at the
intersection of the lines $(12)$ and $(34)$, and also at the intersection of $(23),(14)$. The
numerator must put a zero on these configurations, and thus we have
that $L$ must be the line that passes through $(12) \cap (34)$ as well
as $(23) \cap (41)$:
\begin{equation}
L_I = \epsilon_{I J K} ((12)\cap(34))^J ((23)\cap(14))^K
\end{equation}
If we expand out $(a b) \cap (c d) \deff Z_a \ip{b c d} - Z_b \ip{a c d} =- Z_c \ip{a b d} + Z_d \ip{a b c}$, we can reproduce the expressions for this area in terms of triangulations. 

Note that we can interpret the form as the area of the dual quadrilateral bounded by the edges $Z_1,Z_2,Z_3,Z_4$ and hence the vertices
$W_1=(12),W_2=(23),W_3=(34),W_4=(41)$. See~\eqref{eq:W_notation} for the notation. By going to the affine space
with Y at infinity as $Y=(1,0,0)$, $W_i = (1,W'_i)$, we find the
familiar expression for the area of a quadrilateral with the vertices
$W_1',W_2',W_3',W_4'$,
\begin{equation}
(W'_3 -
W'_1) \times (W'_4 - W'_2).
\end{equation}
where the $\times$ denotes the Euclidean cross product.

We can continue in this way to determine the form for any polygon. For
instance for a pentagon $\A$, we have the general form
\begin{equation}
\aOmega(\A) = \frac{L_{IJ} Y^I Y^J}{2!\ip{Y12}\ip{Y23}\ip{Y34}\ip{Y45}\ip{Y51}}
\end{equation}
but now the numerator must put a zero on all the 5 bad singularities
where $(12),(34)$ intersect, $(23),(45)$ intersect and so on. Thus
$L_{IJ}$ must be the unique conic that passes through all these five
points. If we let $B^I_i = [(i,i{+}1) \cap (i{+}2,i{+}3)]^I$ be the bad
points, then
\begin{equation}
L_{IJ} = \epsilon_{(I_1 J_1) \cdots (I_5 J_5)(IJ)} (B_1 B_1)^{(I_1
J_1)} \cdots (B_5 B_5)^{(I_5 J_5)}
\end{equation}
where $\epsilon_{(I_1 J_1)(I_2 J_2) \cdots (I_6 J_6)}$ is the unique
tensor that is symmetric in each of the $(IJ)$'s but antisymmetric
when swapping $(I_i J_i) \leftrightarrow (I_j J_j)$.

This construction for the numerator generalizes for all $n$-gons.
Just from the poles $\Omega(\A)$ takes the form

\begin{equation}
\Omega({\cal A}) = \frac{N_{I_1 \cdots I_{n-3}} Y^{I_1} \cdots
Y^{I_{n-3}}}{\langle Y 1 2 \rangle \cdots \langle Y n 1 \rangle}
\end{equation}

Now there are $N=n(n-1)/2 - n = (n^2 - 3n)/2$ ``bad" intersections
$B_{a,b}$ of non-adjacent lines, $B^I_{a,b} = [(a,a+1) \cap
(b,b+1)]^I$. But there is a unique (up to scale) numerator that puts a
zero on these bad intersections:
\begin{equation}
L_{I_1 \cdots I_{n-3}} = \epsilon_{(I^{(1)}_1 \cdots
I^{(1)}_{n-3})(I^{(2)}_{1} \cdots I^{(2)}_{n-3}) \cdots (I^{(N)}_1
\cdots I^{(N)}_{n-3}) (I_1 \cdots I_{n-3})}
\prod_{S=1}^N B_S^{I_1^{(S)}}\cdots B_S^{I_{n{-}3}^{(S)}}
\end{equation}
where we have re-labeled the intersections as $B_S$ for $S=1,\ldots,N$. Note that in order for the $\epsilon$ tensor to exist, it is crucial that $N{+}1$ is the dimension of rank $n{-}3$ symmetrc tensors on $\R^3$.

It is interesting that the polygon lies entirely inside the zero-set
defined by the numerator: the ``bad" singularities are ``outside" the
good ones.

For higher-dimensional polytopes the story is much more interesting as
described in \cite{Arkani-Hamed:2014dca}. Here we content ourselves with presenting one of
the examples which illustrates the novelties that
arise.

Consider the $m=3$ cyclic polytope for $n=5$, with vertices
$Z_1,\ldots,Z_5$. The boundaries of the cyclic polytope in $m=3$
dimensional projective space are of the form $(1,i,i+1)$ and
$(n,j,j+1)$, which here are simply
$(123),(134),(145),(125),(235),(345)$.

\begin{equation}
\frac{N_{IJ} Y^I Y^J}{\langle Y 123 \rangle \langle Y 134 \rangle
\langle Y145 \rangle \langle Y 125 \rangle \langle Y 235 \rangle
\langle Y 345 \rangle}
\end{equation}

The numerator corresponds to a quadric in $\P^3$ which has $4
\times 5/2 = 10$ degrees of freedom, and so can be specified in
principle by vanishing on 9 points.

Naively, however, much more is required of the numerator than vanishing
on points. The only edges of this polytope
correspond to the lines $(i,j)$, but there are six pairs of the above
faces that do not intersect on lines $(i,j)$; we find spurious residues
at $L_1 =(123) \cap (145),L_2 = (123) \cap (345), L_3 = (134) \cap
(125),L_4=(134) \cap (235),L_5=(145) \cap(235),L_6=(125) \cap (345)$.
So the numerator must vanish on these lines; the quadric must contain
the six lines $L_i$. Also the zero-dimensional boundaries must simply
correspond to the $Z_i$, while there are six possible intersections of
the denominator planes that are not of this form, so the numerator
must clearly vanish on these six points $X_{1,\ldots, 6}$ as well. Of
course these 6 ``bad points" all lie on the ``bad lines", so if the
numerator kills the bad lines the bad points are also killed.

But there is a further constraint that was absent for the case of the
polygon. In the polygon story, the form $\Omega$ was guaranteed to
have sensible logarithmic singularities and we only had to kill the
ones in the wrong spots, but even this is not guaranteed for cyclic
polytopes. Upon taking two residues, for generic numerators we can
encounter double (and higher) poles. Suppose we approach the point $Y
\to Z_1$, by first moving $Y$ to the line $(13)$ which is the
intersection of the planes $(123),(134)$. If we put $Y = Z_1 + y Z_3$,
then two of the remaining poles are $\langle Y 1 4 5 \rangle = - y
\langle 1 3 4 5 \rangle$ and $\langle Y 125 \rangle = y \langle 1 2 3 5
\rangle$, and so we get a double-pole $y^2$. In order to avoid this
and have sensible logarithmic singularities, the numerator must vanish
linearly as $Y \to Z_1$; the same is needed as $Y \to Z_3$ and $Y \to
Z_5$. Thus in addition to vanishing on the six lines $L_i$, the
numerator must also vanish at $Y \to Z_1,Z_3,Z_5$.

It is not a-priori obvious that this can be done; however quite
beautifully the geometry is such that the 6 lines $L_i$ break up into
two sets, which mutually intersect on the 9 points $X_i$ and
$Z_{1,3,5}$, with three intersection points lying on each line. The
numerator can thus be specified by vanishing on these points, which
guarantees that it vanishes as needed on the lines.

More intricate versions of the same phenomenon happens for more
general cyclic polytopes: unlike for polygons, apart from simply
killing ``bad points" the zero-set of the numerator must ``kiss" the
positive geometry at codimension 2 and lower surfaces, to guarantee
getting logarithmic singularities. The form for $m=4$ cyclic polytopes
were constructed in this way. We have thus seen what this most
brute-force, direct approach to determining the canonical form by
understanding zeros and poles entails. The method can be powerful in
cases where the geometry is completely understood, though as we have seen
this can be somewhat involved even for polytopes.

\subsubsection{Generalized polytopes on the projective plane} \label{sec:genpoly}
Let us return to the study of positive geometries $\A$ in $\P^2$, in view of the current discussion of canonical forms.
In Sections \ref{sec:simplexplane} and \ref{sec:polytopeplane} we explained that under the assumptions of Appendix \ref{app:assumptions}, the boundary components are smooth curves and are thus either linear or quadratic.


Let us suppose that we allow the boundary components to be singular curves $p(Y) = 0$, as in Section \ref{sec:degenerate}, and furthermore we now allow $p(Y)$ to be of arbitrary degree $d$, and have $r$ double-points.  For a degree $d$ curve with $r$ double-points, the genus is given by $(d-1)(d-2)/2 -
r$. Now recall that curves of non-zero genus admit non-trivial holomorphic top forms, which leads to non-unique canonical forms, thus violating our assumptions in Section~\ref{sec:positive}. We therefore require the curve to have $((d-1)(d-2)/2)$
double-points and hence genus zero.  This means that $p(Y) = 0$ can
be rationally parametrized as $Y^I = Y^I(t)$, or equivalently, that the normalization of the curve $P(Y) = 0$ is isomorphic to $\P^1$. In practice, it is easy
to reverse-engineer the polynomial defining the curve of interest from
a rational parametrization. Working with co-ordinate $Y=(1,x,y)$, a
rational parametrization is of the form $x(t) = P_x(t)/Q(t), y(t) =
P_y(t)/Q(t)$. Then, the resultant $p(x,y) = R(P_x(t) - x Q(t); P_y(t)
- y Q(t))$ gives us the polynomial defining the parametrized curve.
For instance taking $x(t) = (t(t^2 + 1))/(1 + t^4), y(t) = (t (t^2 -
1))/(1 + t^4)$ yields the quartic ``lemniscate" curve $p(x,y) = 4((x^2 +
y^2)^2 - (x^2 - y^2))$, which has three double-points; one at $x=y=0$
and two at infinity.

As for the canonical form, the numerator (see Section \ref{sec:matching}) must be chosen to kill all the undesired residues. Recall that for a $d$-gon, the
numerator has to put zeros on $d(d-1)/2 - d = (d^2 - 3d)/2$ points,
and that there is a unique degree $(d - 3)$ polynomial that passes
through those points, which determines the numerator uniquely up to overall scale.  It is
interesting to consider an example which is the opposite extreme of a
polygon. Consider an irreducible degree $d$ polynomial, with
$(d-1)(d-2)/2 = (d^2 - 3d)/2 + 1$ singular points. To get a positive
geometry, we can kill the residues on all but one of these singular
points, leaving just a single zero-dimensional boundary just as in our
teardrop cubic example of Section \ref{sec:degenerate}. These are the same number $(d^2 - 3d)/2$ of
points we want to kill as in the polygon example, and once again there
is a unique degree $(d-3)$ curve that passes through those points.

\begin{figure}
\centering
\includegraphics[width=8cm]{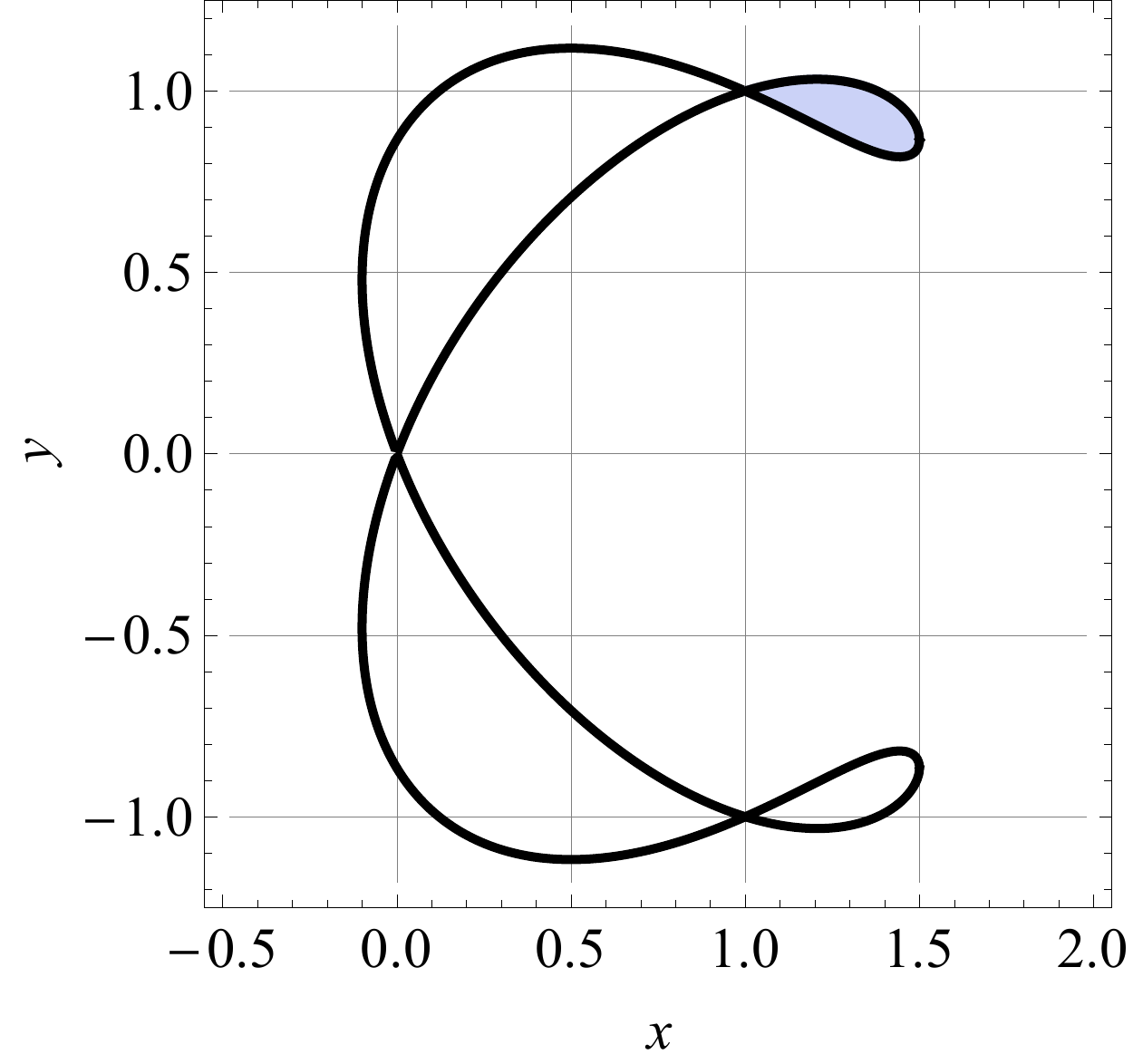}\caption{An ``ampersand" curve with boundary given by a quartic polynomial. The shaded ``teardrop" is a positive geometry.}
\label{fig:ampersand}
\end{figure}

An example is provided by the ``ampersand" geometry (see Figure~\ref{fig:ampersand}) associated with
the quartic curve $P(x,y) = (y^2 - x^2)(x-1)(2x - 3) - 4(x^2 + y^2 -
2x)^2$, which has three singular points at $(0,0),(1,1),(1,-1)$.
If we choose the numerator to be the line that kills e.g. the points
$(0,0)$ and $(1,-1)$, then we get a positive geometry corresponding to
the ``teardrop" in the upper quadrant.

Two more examples: consider a
region bounded by two quadrics $Q_1$ and $Q_2$. The numerator of the
form is $Y\cdot L$ for some line $L$. Now two generic quadrics
intersect at four points $P_1,P_2,P_3,P_4$.  Mirroring the determination of
the form for the case of the quadrilateral, we can choose the line $L$
appearing in the numerator to kill two of $P_{i}$'s, and this will give
us the canonical form associated with the geometry $(Q_{1,2}\cdot YY)\geq 0$.
Similarly, consider a positive geometry defined by a singular cubic
$C$ and a line $W$.  Again we have a numerator of the form  $Y \cdot
L$. The line $W$ intersects the cubic in three points $P_{1},P_2,P_3$. If we pick
the the line $L$ to pass through one of the $P_i$ as well as the
singular
point of the cubic, we get the canonical form associated with the
geometry. These constructions can be extended to higher dimensions,
where (as with the cyclic polytope example) we will generically
encounter numerators whose zeros touch the positive geometry on
co-dimension two (and lower dimensional) boundaries.

\subsection{Triangulations}

Recall from Section~\ref{sec:triangulations} that if a positive geometry is triangulated by a collection of other positive geometries, its canonical form is given by the sum of the canonical forms of the collection. We now apply this method to compute the canonical form of various generalized polytopes.

\subsubsection{Projective polytopes}\label{sec:polytopetriangulation}

Let $\A\deff\A(Z_1,\ldots ,Z_n)$ be a convex projective polytope. The canonical form $\Omega(\A)$ can be obtained from a triangulation of $\A$ (see Section \ref{sec:triangulations} and Appendix \ref{app:cones}).  Let $\Delta_1,\Delta_2,\ldots,\Delta_r$ be a triangulation of $\A$ into simplices. For simplicity let us assume that the simplex interiors are mutually non-overlapping.

The canonical form $\Omega(\A)$ is given by
\be \label{eq:triangulation}
\Omega(\mathcal{A}) = \sum_{i} \Omega(\Delta_i)
\ee
The fact that the simplicial canonical forms {\sl add} is dependent on the assumption that the orientation of the interior of $\Delta_i$ agrees with that of $\A$ for each $i$.  
More generally, for any polytopal subdivision of $\A$ into polytopes $\A_i$ (i.e. a ``triangulation" by polytopes), we have
\be
\Omega(\A)=\sum_i \Omega(\A_i).
\ee

We begin with the simplest case: line segments in $\P^1(\R)$.

\begin{example}
Consider a triangulation of the segment $[a,b]$ from Example~\ref{ex:segment} by a sequence of successively connected segments:
\be
[a,b]=\bigcup_{i=1}^r[c_{i{-}1},c_i]
\ee
where $a=c_0<c_1<\ldots <c_r=b$. It is straightforward to check that
\be
\Omega([a,b])=\frac{(b-a)dx}{(b-x)(x-a)}=\sum_{i=1}^r\frac{(c_i-c_{i{-}1})dx}{(c_i-x)(x-c_{i{-}1})}=\sum_{i=1}^r\Omega([c_{i{-}1},c_i])
\ee
More generally, for the positive geometry $\A=\bigcup_i[a_i,b_i]\subset \P^1$ which is triangulated by finitely many line segments with mutually disjoint interiors, the canonical form is:
\be
\Omega\left(\bigcup_i[a_i,b_i]\right)=\sum_i \Omega([a_i,b_i])
\ee
\end{example}

\begin{example}
Suppose $\A$ is a convex projective polytope, and $Z_*$ is a point in its interior, then $\A$ is triangulated by
\be
\A=\bigcup_\text{facets}\Conv(Z_*,Z_{i_1},Z_{i_2},\ldots,Z_{i_m})
\ee
where we take the union over all choice of indices $i_1,\ldots ,i_m$ for which $\Conv(Z_{i_1},Z_{i_2}\ldots ,Z_{i_m})$ is a facet of the polytope, and we avoid repeated permutations of the same set of indices. For each facet, we order the indices so that $Z_*,Z_{i_1},\ldots ,Z_{i_m}$ is positively oriented. It follows that
\be
\Omega(\A)=\sum_{\text{facets}}[*,i_1,\ldots ,i_m]
\ee

\end{example}

Recalling the facets of cyclic polytopes from Section~\ref{sec:cyclic}, we have the following corollaries.

\begin{example}\label{cyclic1}
The canonical rational function of a cyclic polytope $\A$ for even $m$ can be obtained as follows.
\be
\aOmega(\A)=\sum_{1\le i_1{-}1<i_1< \cdots <i_{m/2}{-}1<i_{m/2}\le n{+}1}[*,i_1{-}1,i_1,\ldots,i_{m/2}{-}1,i_{m/2}]
\ee
For arbitrary $Z_*$, this is called a \defn{CSW triangulation}. For $Z_*=Z_i$ for some $i$, this is called a \defn{BCFW triangulation}.
\end{example}

\begin{example}\label{cyclic2}
The canonical rational function of a cyclic polytope $\A$ for odd $m$ can be obtained as follows.
\be
\aOmega(\A)=\sum_{2 \le i_1{-}1<i_1< \cdots<i_{(m{-}1)/2}{-}1<i_{(m{-}1)/2}\le n}-[*,1,i_1{-}1,i_1,\ldots,i_{(m{-}1)/2}{-}1,i_{(m{-}1)/2}]\;\;\\
+\sum_{1 \le i_1{-}1<i_1< \cdots<i_{(m{-}1)/2}{-}1<i_{(m{-}1)/2}\le n{-}1}[*,i_1{-}1,i_1,\ldots,i_{(m{-}1)/2}{-}1,i_{(m{-}1)/2},n]\;\;
\ee
for any $Z_*$. If we set $Z_*=Z_1$ or $Z_n$, then we get
\be\label{eq:odd_m_1}
\aOmega(\A)=\sum_{2 \le i_1{-}1<i_1< \cdots<i_{(m{-}1)/2}{-}1<i_{(m{-}1)/2}\le n{-}1}[1,i_1{-}1,i_1,\ldots,i_{(m{-}1)/2}{-}1,i_{(m{-}1)/2},n]\;\;
\ee

\end{example}

\subsubsection{Generalized polytopes on the projective plane}

In this section we verify that the canonical form for the ``pizza slice" geometry from Example~\ref{ex:pizza} can be obtained by triangulation.
\begin{example}
Recall the ``pizza slice" geometry $\mathcal{T}(\theta_1,\theta_2)$ from Example~\ref{ex:pizza}. For simplicity, we will assume reflection symmetry about the $y$-axis and let $\mathcal{T}(\theta)\deff\mathcal{T}(\theta,\pi-\theta)$ for some $0\le \theta<\pi/2$. Denote the vertices of the geometry by $Z_i^I\in\P^1(\R)$ for $i=1,2,3$, where
\be
Z_1^I=(1,0,0),\;\; Z_2^I=(1,\cos\theta,\sin\theta),\;\; Z_3^I=(1,{-}\cos\theta,\sin\theta)
\ee

The pizza slice is clearly the union of a segment of the disk (see Example~\ref{ex:circSegment}) and a triangle (see Section \ref{sec:proj_simplex}).
\be
\mathcal{T}(\theta)=\mathcal{S}(\sin\theta)\cup\A(Z_1,Z_2,Z_3)
\ee
It follows that
\be
\Omega(\mathcal{T}(\theta))&=&\Omega(\mathcal{S}(\sin\theta))+\Omega(\A(Z_1,Z_2,Z_3))\\
&=&\frac{(2\cos\theta) dxdy}{(1{-}x^2{-}y^2)(y{-}\sin\theta)}+\frac{(2\sin^2\theta\cos\theta)dxdy}{(\sin\theta{-}y)({-}x\sin\theta{+}y\cos\theta)(x\sin\theta{+}y\cos\theta)}\\
&=&\frac{2\cos\theta(y+\sin\theta)dxdy}{(1-x^2-y^2)({-}x\sin\theta{+}y\cos\theta)(x\sin\theta{+}y\cos\theta)}
\ee
which is equivalent to ~\eqref{eq:pizza} for $\theta_1=\theta$ and $\theta_2=\pi-\theta$, with $y-\sin\theta=0$ as a spurious pole.

\end{example}

\subsubsection{Amplituhedra and BCFW recursion}
\label{sec:BCFW}

Motivated by physical principles of quantum field theory, a recursion relation was discovered for the canonical form of the physical Amplituhedron called \defn{BCFW recursion}~\cite{Britto:2004ap,all loop}. This is a rich subject on its own. While a full explanation of the recursion relation is beyond the scope of this paper, we present a sketch of the idea here.

Let $\A$ be the amplituhedron defined in \eqref{eq:defampli}.  We begin by introducing an extra parameter $z$ by making the shift $Z_n\rightarrow Z_n+zZ_{n{-}1}$, which gives $\Omega(\A)\rightarrow \Omega(\A(z))$. The principle of \defn{locality} suggests that the canonical form can only develop {\sl simple poles} in $z$ (including possibly a simple pole at infinity), which can be seen by studying the structure of Feynman propagators. It follows that
\be
\Omega(\A)=\oint_C \frac{dz}{z}\Omega(\A(z))
\ee
where the contour $C$ is a small counter-clockwise loop around the origin. Applying Cauchy's theorem by expanding the loop to infinity gives
\be\label{eq:BCFWpoles}
\Omega(\A)=-\sum_{i}\Res_{z\rightarrow z_i}\frac{\Omega(\A(z))}{z}+\Omega(\A(\infty))
\ee
where $z_i$ denotes all the poles. The residue at infinity is simply the Amplituhedron with $Z_n$ removed.

The residues at $z_i$, however, are more involved. Based on extensive sample computations, we make the following observations assuming $D=\dim(\A)$:
\begin{itemize}
\item
There exists a $\Delta$-like positive geometry $\mathcal{C}_i$ in the loop Grassmannian $G(k,n;2^L)$ of dimension $D$.
\item
There exists a subset $\A_i$ of $\A$ also of dimension $D$, called a \defn{BCFW cell}.
\item
The map under $Z:\mathcal{C}_i\rightarrow \A_i$ is a degree-one morphism. Since $\mathcal{C}_i$ is $\Delta$-like,  hence so is $\A_i$.
\item
Given $\Delta-$like coordinates $(1,\alpha_{i1},\ldots,\alpha_{iD})\in \P^D(\R)$ on $\mathcal{C}_i$, the residue at $z_i$ is given by the push-forward:
\be\label{eq:bcfw_res}
-\Res_{z\rightarrow z_i}\frac{\Omega(\A(z))}{z}=Z_*\left(\prod_{j=1}^D\frac{d\alpha_{ij}}{\alpha_{ij}}\right)=\Omega(\A_i)
\ee
\end{itemize}

At $L=0$, each set $\mathcal{C}_i$ is a positroid cell of the positive Grassmannian and $\A_i$ is the image under $Z$. For $L>0$ some generalization of this statement is expected to hold.

The precise construction of $\mathcal{C}_i$ is explained in~\cite{Bai:2014cna}. While the map is never explicitly mentioned in the reference, its geometric structure is explained in terms of \defn{momentum twistor diagrams}, which are loop extensions of Postnikov's plabic graphs~\cite{Postnikov:2006kva}. In particular, the $\Delta$-like coordinates can be read off from labels appearing on the graph, while the diagrams at any $k,n,L$ can be constructed from diagrams of lower $k,n$ or $L$, hence the {\sl recursive} nature of BCFW.

We point out that while BCFW cells are $\Delta$-like, they are not necessarily simplex-like. Namely, their canonical forms may have zeros.



From~\eqref{eq:bcfw_res} and the discussion in Section~\ref{sec:triangulations}, it follows that the BCFW cells form a boundary triangulation of the Amplituhedron:
\be
\Omega(\A)=\sum_{i}\Omega(\A_i)
\ee

Furthermore, it appears based on extensive numerical checks that the BCFW cells have mutually disjoint interiors. So they triangulate the Amplituhedron in the first sense defined in Section~\ref{sec:triangulations}. We point out that if our assumptions on BCFW cells hold, then the Amplituhedron must be a positive geometry.

Historically, BCFW recursion was first discovered in the context of quantum field theory as an application of Cauchy's theorem on the deformation parameter $z$. See~\cite{Britto:2004ap}. The poles at $z_i$ correspond to Feynman propagators going on shell (i.e. locality), while the residues at $z_i$ were constructed via the principle of \defn{unitarity} in terms of amplitudes (or loop integrands) of lower $k,n$ or $L$ with modifications on the particle momenta.

Since most positive geometries are not directly connected to field theory scattering amplitudes, we do not expect BCFW recursion to extend to all cases. Nonetheless, it is conceivable that the canonical form can be reconstructed by an application of Cauchy's theorem to a clever shift in the boundary components. A solution to this problem would allow us, in principle, to construct the canonical form of arbitrarily complicated positive geometries from simpler ones.

\subsubsection{The tree Amplituhedron for $m=1,2$}
\label{sec:Amplituhedron12}
There is by now a rather complete understanding of the tree Amplituhedron with $m=1,2$, for any $k$ and $n$. The following results will be presented in detail elsewhere \cite{winding}.  (For $m=1$, see also \cite{KarpWilliams}.)  Here we will simply present (without proof) some simple triangulations of these Amplituhedra, and give their associated canonical forms. We let $\A\deff\A(k,n,m)$ whenever $k,n,m$ are understood.

In the $m=1$ Amplituhedron,  $Y_s^I$ is a $k$-plane in $(k{+}1)$ dimensions with $s=1,\ldots,k$ indexing a basis for the plane and $I=0,\ldots,k$ indexing the vector components in $(k{+}1)$ dimensions. We will triangulate by images of $k-$dimensional cells of $G_{>0}(k,n)$; in other words for each cell we will look at the image $Y_s^I = \sum_{i=1}^n C_{s i}(\alpha_1,\ldots,\alpha_k) Z_i^I$, where $\alpha_1,\ldots, \alpha_k$ are positive $\Delta-$like co-ordinates for that cell. For every collection of $k$ integers $\{i_1, \ldots, i_k\}$ with $1 \leq i_1 < i_2 < \cdots < i_k \leq n-1$, there is a cell where  
\begin{equation}
C^{\{i_1,\ldots,i_k\}}_{s a} = \left\{\begin{array}{cc} 1 &  a = i_{s} \\  \alpha_s &  a=i_{s} + 1 \\ 0 & {\rm otherwise} \end{array} \right\}
\end{equation}
with the positive variables $\alpha_s \geq 0$. 
In other words in this cell we have $Y_s^I = Z^I_{i_s} + \alpha_s Z^I_{i_{s}+1}$. 

Geometrically, in this cell the $k$-plane $Y$ intersects the cyclic polytope of external data in the $k$ 1-dimensional edges $(Z_{i_1},Z_{i_1+1}), \ldots, (Z_{i_k} Z_{i_k+1})$. The claim that these cells triangulate the $m=1$ Amplituhedron is then equivalent to the statement that this Amplituhedron is the set of all $k$-planes which intersect the cyclic polytope in precisely $k$ of its 1-dimensional {\sl consecutive} edges (i.e. an edge between two consecutive vertices).

The motivation for this triangulation will be given (together with associated new characterizations  of the Amplituhedron itself) in \cite{winding}. For now we can at least show that these cells are non-overlapping in $Y$ space, by noting that in this cell, the following sequence of minors
\begin{equation}
\{\langle Y 1 \rangle, \langle Y 2 \rangle, \cdots, \langle Y n \rangle\}
\end{equation}
have precisely $k$ sign flips, with the flips occurring at the locations $(i_s,i_s+1)$ for $s=1,\ldots, k$. Since the sign patterns are different in different cells, we can see that the cells are non-overlapping; the fact that they triangulate the Amplituhedron is more interesting and will be explained at greater length in \cite{winding}. 

The $k$-form associated with this cell is 
\begin{equation}
\Omega^{\{i_1, \ldots, i_k\}} = Z_*\left(\prod_{s=1}^k d\log  \alpha_s\right)= \prod_{s=1}^k d {\rm log}\left(\frac{\langle Y,i_s{+}1 \rangle}{\langle Y i_s \rangle}\right) 
\end{equation}
and the full form is 
\begin{equation}
\Omega(\A) = \sum_{1 \leq i_1 < \cdots < i_k \leq n-1} \Omega^{\{i_1, \ldots, i_k\}}
\end{equation}
It is easy to further simplify this expression since the sums collapse telescopically. For instance for $k=1$ we have 
\begin{equation}
\sum_{1 \leq i_1 \leq n-1} d{\rm log}\left(\frac{\langle Y,i_1{+}1\rangle}{\langle Y i_1 \rangle}\right) = d {\rm log} \left(\frac{\langle Y n \rangle}{\langle Y 1 \rangle}\right)
\end{equation}
Note the cancellation of spurious poles in the sum leading nicely to the final result. The same telescopic cancellation occurs for general $k$, and for even $k$ we are left with the final form 
\begin{eqnarray}
\Omega(\A) = 
\frac{{\rm d}^{k \times (k+1)} Y}{\rm Vol\;GL(k)} \sum_{\substack{2 \leq j_1{-}1<j_1<\cdots \\< j_{k/2}{-}1<j_{k/2}\leq n}} [1,j_1{-}1,j_1,\ldots,j_{k/2}{-}1,j_{k/2}]
\end{eqnarray}
while for odd $k$ we are left with
\be
\Omega(\A) = \frac{{\rm d}^{k \times (k+1)} Y}{{\rm Vol\; GL(}k)}
\sum_{\substack{2\leq j_1{-}1<j_1<\cdots\\<j_{(k{-}1)/2}{-}1<j_{(k{-}1)/2}\leq n{-}1}}
[1,j_1{-}1,j_1,\ldots,j_{(k{-}1)/2}{-}1,j_{(k{-}1)/2},n]\;\;
\ee
where the brackets denote
\be
[j_0,j_1,\ldots, j_{k}]\deff \frac{\lb j_0\ldots j_k\rb}{\lb Yj_0\rb\cdots\lb Yj_k\rb}
\ee
for any indices $j_0,\ldots, j_k$. The brackets satisfy
\be
\prod_{s=1}^k d\log\left(\frac{\lb Y j_{s}\rb}{\lb Y j_{s{-}1}\rb}\right)=
\frac{{\rm d}^{k \times (k+1)} Y}{{\rm Vol\; GL(}k)}
[j_0,j_1,\ldots, j_{k}]
\ee

Note that each bracket can be interpreted as a simplex volume in $\P^k(\R)$ whose vertices are the $Z_{j_0},\ldots,Z_{j_k}$, with $Y$ the hyperplane at infinity.
Note also that the combinatorial structure of the triangulation is identical to triangulation of cyclic polytopes discussed in Section~\ref{sec:polytopetriangulation}. Indeed, when $\lb Y i\rb >0$ for each $i$, the canonical rational function can be interpreted as the volume of the convex cyclic polytope with vertices $Z$. However, these conditions do not hold for $Y$ on the interior of the Amplituhedron, since $Y$ must pass through the interior of the cyclic polytope as it intersects $k$ 1-dimensional consecutive edges.

Starting with $k=2$, this is not a {\it positively convex} geometry (see Section~\ref{sec:convex}): the form has zeros (and poles) on the interior of the Amplituhedron. We can see this easily for e.g. $k=2,n=4$. 
We have a single term with a pole at $\langle Y 2 \rangle \to 0$, so it is indeed a boundary component, but it is trivial to see that $\lb Y2\rb$ can take either sign in the Amplituhedron, so the form has poles and zeros on the interior. 

We now consider the case $m=2$, and triangulate the Amplituhedron with the image of a collection of $2 \times k$ dimensional cells of $G_{>0}(k,n)$. That is, for each cell we look at the image 
$Y_s^I = \sum_{i=1}^nC_{s i}(\alpha_1,\beta_1,\ldots,\alpha_k,\beta_k) Z_i^I$. Similar to $m=1$, the cells are indexed by $\{i_1, \ldots, i_k\}$ with $2 \leq i_1 < i_2 < \cdots < i_k \leq (n-1)$. Now the $C$ matrices are given by 
\begin{equation}
C^{\{i_1,\ldots,i_k\}}_{s i} = \left\{\begin{array}{cc} (-1)^{s -1 }  & i=1 \\ \alpha_s & i=i_{s} \\ \beta_{s} & i=i_s + 1 \\ 0  & {\rm otherwise} \end{array} \right\}
\end{equation}
In other words, in this cell we have $Y_s^I = (-1)^{s - 1} Z_1^I + \alpha_s Z^I_{i_s} + \beta_{s} Z^I_{i_s + 1}$. As for $m=1$, it is easy to see that images of these cells are non-overlapping in $Y$ space for essentially the identical reason; in this cell it is easy to check that the sequence of minors
\begin{equation} 
\{\langle Y 1 2 \rangle, \langle Y 1 3 \rangle, \ldots, \langle Y 1 n \rangle\}
\end{equation}
again has precisely $k$ sign flips, that occur at the locations $(i_s, i_s+1)$ for $s=1,\ldots, k$.

The $2k$-form associated with this cell is 
\begin{eqnarray}
\Omega^{\{i_1, \cdots, i_k\}} &=& Z_*\left(\prod_{s=1}^k d\log\alpha_s \;d\log \beta_s\right) \nonumber \\ 
&=& \prod_{s=1}^k d {\rm log} \left(\frac{\langle Y 1 {i_s} \rangle}{\langle Y{i_s} ,{{i_s}+1} \rangle}\right) d {\rm log} \left(\frac{\langle Y 1,{i_s+1} \rangle}{\langle Y {i_s}, {{i_s}+1} \rangle} \right) \nonumber \\
&=&\frac{{\rm d}^{k(k+2)}Y}{{\rm Vol\;GL}(k)}
[1,i_1,i_1{+}1;\ldots;1,i_k,i_k{+}1]
\end{eqnarray}
where
\be
[p_1,q_1,r_1;\ldots;p_k,q_k,r_k]\deff
\frac{\left[
\lb(Y^{k{-}1})^{s_1}p_1q_1r_1\rb
\cdots
\lb (Y^{k{-}1})^{s_k}p_kq_kr_k\rb
\epsilon_{s_1\cdots s_k}\right]^k}{
2^k\lb Yp_1q_1\rb\lb Yq_1r_1\rb\lb Yp_1r_1\rb\cdots
\lb Yp_kq_k\rb\lb Yq_kr_k\rb\lb Yp_kr_k\rb
}\;\;
\ee
for any indices $p_s,q_s,r_s$ with $s=1,\ldots, k$ and
\be
(Y^{k{-}1})^s \deff  Y_{s_1}\wedge\cdots \wedge Y_{s_{k{-}1}}\epsilon^{ss_1\cdots s_{k{-}1}}
\ee

As usual the full form arises from summing over the form for each piece of the triangulation 
\begin{equation}
\Omega(\A) = \sum_{2 \leq i_1 < \cdots i_k \leq n-1} \Omega^{\{i_1, \cdots, i_k\}}
\end{equation}

In particular, for $k=2$, we have
\be\label{eq:kermit} 
\Omega(\A(2,2,n)) &=&\lb Yd^2Y_1\rb\lb Yd^2Y_2\rb \times\\
&&\sum_{2\le i<j\le n{-}1}\frac{\det
\begin{pmatrix}
\lb Y_1,i{-}1,i,i{+}1 \rb & \lb Y_1,j{-}1,j,j{+}1\rb \\
\lb Y_2,i{-}1,i,i{+}1 \rb & \lb Y_2,j{-}1,j,j{+}1\rb
\end{pmatrix}^2}{2^2\lb Y1i\rb\lb Y1,i{+}1\rb\lb Yi,i{+}1\rb\lb Y1j\rb\lb Y1,j{+}1\rb\lb Yj,j{+}1\rb}\nonumber
\ee

This is called the \defn{Kermit representation}, and the summands $[1,i,i{+}1;1,j,j{+}1]$ are called \defn{Kermit terms}. These are important for 1-loop MHV  scattering amplitudes whose physical Amplituhedron $\A(0,n;L{=}1)$ is isomorphic to the Amplituhedron $\A(2,n,2)$.

Returning to general $k$, of course the form also has spurious poles that cancel between the terms, though unlike the case of $m=1$ it cannot be trivially summed into a simple expression with only physical poles. However there is an entirely different representation of the form, not obviously related to the triangulation of the Amplituhedron, which is (almost) free of all spurious poles. This takes the form 
\be\label{eq:localForm}
& &\Omega(\A) = \frac{{\rm d}^{k(k+2)}Y}{{\rm Vol\;GL}(k)}\times \\
& & \sum_{1\le i_1<i_2<\cdots < i_k\le n} \frac{\langle \left(Y^{k-1}\right)^{s_1} {i_1{-}1},{i_1},{i_1{+}1}\rangle  \cdots \langle \left(Y^{k-1}\right)^{s_k} {i_k{-}1},{i_k},{i_k{+}1}\rangle  \epsilon_{s_1 \cdots s_k} \langle X{i_1}\cdots{i_k} \rangle}{\langle Y X \rangle \prod_{s=1}^k \langle Y{i_s},{{i_s}{+}1} \rangle}\nonumber
\ee
Note the presence of a reference $X^{IJ}$ in this expression, playing an analagous role to a ``triangulation point" in a triangulation of a polygon into triangles $[X,i,i{+}1]$. The final expression is however $X$-independent. Note also that apart from the $\langle Y X \rangle$ pole, all the poles in this expression are physical.

This second ``local" representation of the form $\Omega(\A)$ allows us to exhibit something that looks miraculous from the trianguation expression: $\Omega(\A)$ is positive when $Y$ is inside the Amplituhedron, and so the $m=2$ Amplituhedron is indeed a positively convex geometry (see Section \ref{sec:convex})! Indeed, if we choose $X$ judiciously to be e.g. $X^{IJ} = (Z_{l} Z_{l+1})^{IJ}$ for some $l$, then trivially all the factors in the denominator are positive. Also, $\langle X Z_{i_1} \cdots Z_{i_k} \rangle > 0$  trivially due to the positivity of the $Z$ data. The positivity of the first factor  in the numerator is not obvious; however, it follows immediately from somewhat magical positivity properties of the following ``determinants of minors". For instance for $k=2$ the claim is that as long as the $Z$ data is positive, 
\begin{equation}
{\rm det} \left( \begin{array}{cc} \langle a,{i-1},i,{i+1} \rangle & \langle a,{j-1},j,{j+1} \rangle \\ \langle b,{i-1},i,{i+1} \rangle & \langle b,{j-1},j,{j+1} \rangle \end{array} \right)>0
\end{equation}
for any $a<b$ and $i<j$. Similarly for $k=3$, 
\begin{equation}
{\rm det} \left(\begin{array}{ccc} \langle a,b,{i-1},i,{i+1} \rangle & \langle a,b,{j-1},j,{j+1} \rangle & \langle a,b,{k-1},k,{k+1} \rangle \\ \langle a,c,{i-1},i,{i+1} \rangle & \langle a,c,{j-1},j,{j+1} \rangle & \langle a,c,{k-1},k,{k+1} \rangle \\ 
\langle b,c,{i-1},i,{i+1} \rangle & \langle b,c  ,{j-1} ,j,{j+1} \rangle & \langle b,c,{k-1},k,{k+1} \rangle \end{array} \right)>0 
\end{equation}
for any $a<b<c; i<j<k$, 
with the obvious generalization holding for higher $k$. These identities hold quite non-trivially as a consequence of the positivity of the $Z$ data. 

The existence of this second representation of the canonical form, and especially the way it makes the positivity of the form manifest, is quite striking. The same phenomenon occurs for $k=1$ and any $m$--the canonical forms are always positive inside the polytope, even though the determination of the form obtained by triangulating the polytope does not make this manifest. For polytopes, this property is made manifest by the much more satisfying representation of the form as the volume integral over the dual polytope. The fact that the same properties hold for the Amplituhedron (at least for even $m$)  suggests that we should think of the ``local" expression \eqref{eq:localForm} for the form we have seen for $m=2$ as associated with the ``triangulation" of a ``dual Amplituhedron". We will have more to say about dual positive Grassmannians and Amplituhedra in \cite{dual}.


\subsubsection{A $1$-loop Grassmannian}
\label{sec:1loop}
We now consider the special case of the 1-loop positive Grassmannian $G_{> 0}(1,n;2)$ with $k = 1$ and $L=1$, which is directly relevant for the {\sl physical} 1-loop Amplituhedron $\A(1,n;L{=}1)$.  The space $G(1,n;2)$ is the $(3n-7)$-dimensional space consisting of $3 \times n$ matrices 
\begin{equation}
P = \left( \begin{array}{c} C \\ D \end{array} \right)
\end{equation}
where the last two rows form the $D$ matrix and the first row forms the $C$ matrix. The $P$ matrix has positive $3\times 3$ minors while $C$ has positive components. Alternatively, $G(1,n;2)$ is the space of partial flags $\{(V,W) \mid 0 \subset V \subset W \subset \C^n\}$ of subspaces where $\dim V = 1$ and $\dim W = 3$. The nonnegative part $G_{\geq 0}(1,n;2)$ is given by those matrices $P$ where the row spaces of $C$ and $P$ are in the totally nonnegative Grassmannian.

There is an action of $T_+ = \R_{>0}^n$ on $G_{\geq 0}(1,n;2)$.  The quotient space $G_{\geq 0}(1,n;2)/T_+$ can be identified with the space of $n$ points $(P_1,\ldots,P_n)$ in $\P^2(\R)$ in convex position (with a fixed distinguished line at infinity $L_\infty$ with dual coordinates $Q_\infty=(1,0,0)$), considered modulo affine transformations. The $P_i$ are of course just the columns of $P$. 
For the precise definition we refer the reader to \cite{Bai:2015qoa}, although we warn the reader that the positive geometry is denoted $G_+(1,n;1)$ in the reference. Using the geometry of such convex $n$-gons, a collection of simplex-like cells $\tPi_{f, \geq 0} \subset G_{\geq 0}(1,n;2)$ are constructed in analogy to the positroid cells $\Pi_{f,\geq 0}$ of Section~\ref{sec:grassmann}, where now $f$ varies over a different set of affine permutations.  Note that some of the cells $\tPi_{f,\geq 0}$ are special cases of the totally positive cells $(\Pi_u^w)_{\geq 0}$ of Rietsch and Lusztig described in Section \ref{sec:flagTP}, but some are not.

We sketch how $\tPi_{f,\geq 0}$ can be shown to be a positive geometry.  First, associated to each $\tPi_{f,\geq 0}$ is a class of momentum-twistor diagrams $D_f$ \cite{Bai:2014cna,Bai:2015qoa}.  The edges of any such diagram $D_f$ can be labeled to give a (degree-one) rational parametrization $\R_{>0}^d \cong \tPi_{f, >0}$, and the canonical form $\Omega(\tPi_{f,\geq 0})$ is given by $\prod_i d\alpha_i/\alpha_i$ where $\alpha_i$ are the edge labels. Different diagrams for the same cell provide different parametrizations. Each boundary component $C_{\geq 0}$ of $\tPi_{f, \geq 0}$ can be obtained by removing one of the edges of some momentum-twistor diagram $D_f$ for $\tPi_{f,\geq 0}$, although some boundaries may be visible to only a proper subset of diagrams.  This shows that the residues of $\Omega(\tPi_{f,\geq 0})$ give the forms $\Omega(C_{\geq 0})$.

Unlike the nonnegative Grassmannian $G_{\geq 0}(k,n)$, there are $\binom{n}{3}$ (instead of just one!) simplex-like top cells in $G(1,n;2)$, denoted $\Pi_{\{a,b,c\}}$ where $\{a,b,c\} \subset \{1,2,\ldots,n\}$ is a 3-element subset.  As an example, with $n = 5$, the cells $\Pi_{\{1,2,5\}}, \Pi_{\{2,3,5\}}, \Pi_{\{3,4,5\}}$ correspond to the diagrams shown in Figure~\ref{fig:1loop}.

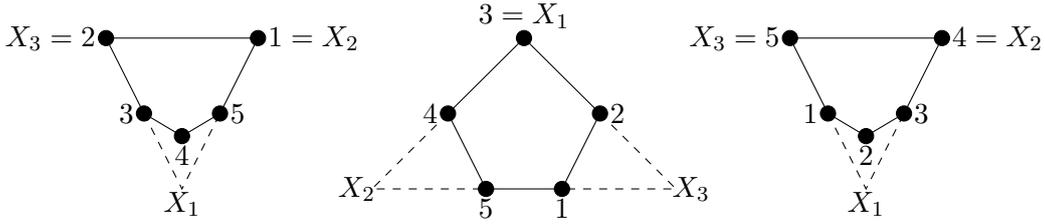
\begin{figure}

\begin{center}

\begin{tikzpicture}

\draw (0,0) -- (-2,0) -- (-1.5,-1)--(-1,-1.3)-- (-0.5,-1)--(0,0);

\draw[fill] (0,0) circle (0.1cm) node[right] {$1 = X_2$};

\draw[fill] (-2,0) circle (0.1cm) node[left] {$X_3 = 2$};

\draw[fill] (-1.5,-1) circle (0.1cm) node[left] {$3$};

\draw[fill] (-1,-1.3) circle (0.1cm) node[below] {$4$};

\draw[fill](-0.5,-1) circle (0.1cm) node[right] {$5$};

\draw[dashed] (-1.5,-1) -- (-1,-2) -- (-0.5,-1);

\node at (-1,-2.2) {$X_1$};



\begin{scope}[shift={(3,-2)}]

\draw (0,0) -- (1,0) -- (1.5,1) -- (0.5,2) -- (-0.5,1) -- (0,0);

\draw[fill] (0,0) circle (0.1cm) node[below] {$5$};

\draw[fill] (1,0) circle (0.1cm) node[below] {$1$};

\draw[fill] (1.5,1) circle (0.1cm) node[right] {$2$};

\draw[fill] (0.5,2) circle (0.1cm) node[above] {$3 = X_1$};

\draw[fill](-0.5,1) circle (0.1cm) node[left] {$4$};

\draw[dashed] (1.5,1) -- (2.5,0) -- (1,0);

\draw[dashed] (-0.5,1) -- (-1.5,0) -- (0,0);

\node at (-1.7,0) {$X_2$};

\node at (2.7,0) {$X_3$};

\end{scope}

\begin{scope}[shift={(9,0)}]

\draw (0,0) -- (-2,0) -- (-1.5,-1)--(-1,-1.3)-- (-0.5,-1)--(0,0);

\draw[fill] (0,0) circle (0.1cm) node[right] {$4 = X_2$};

\draw[fill] (-2,0) circle (0.1cm) node[left] {$X_3 =5$};

\draw[fill] (-1.5,-1) circle (0.1cm) node[left] {$1$};

\draw[fill] (-1,-1.3) circle (0.1cm) node[below] {$2$};

\draw[fill](-0.5,-1) circle (0.1cm) node[right] {$3$};

\draw[dashed] (-1.5,-1) -- (-1,-2) -- (-0.5,-1);

\node at (-1,-2.2) {$X_1$};



\end{scope}

\end{tikzpicture}

\end{center}

\caption{Spaces of pentagons (with vertices labeled $1,2,3,4,5$ for
$P_1,P_2,P_3,P_4,P_5$) that correspond to the cells $\Pi_{\{1,2,5\}},
\Pi_{\{2,3,5\}}, \Pi_{\{3,4,5\}}$ respectively. In the leftmost
picture, the dashed line indicates that the edges $15$ and $23$
intersect on the side of $12$ containing the pentagon interior as shown. Each pentagon is inscribed in a ``big" triangle with vertices denoted $X_1,X_2,X_3$. If we label the solid edge between $i,i{+}1$ as $i$, then the cell $\Pi_{\{a,b,c\}}$ corresponds to the big triangle with (extended) edges $a,b,c$. Finally, note that a generic
pentagon in the plane belongs to exactly one of these three classes.}
\label{fig:1loop}

\end{figure}

A parametrization-independent formula for $\Omega(\Pi_{\{a,b,c\}})$ is given in \cite{Bai:2015qoa}:
\begin{equation}
\Omega(\Pi_{\{a,b,c\}}) = \frac{d^{3n-7}P}{\prod_{i=1}^n (i,i+1,i+2)} \{X_1,X_2,X_3\}
\end{equation}
where the parentheses are defined by $(i,j,k)\deff\det(P_i,P_j,P_k)$ for any (cyclically extended) indices $i,j,k$, and $\{X_1,X_2,X_3\}$ is the area of the triangle with vertices $X_1,X_2,X_3$ (see dotted triangle in the polygon picture above),
\be
\{X_1,X_2,X_3\}=\frac{(X_1,X_2,X_3)}{[X_1][X_2][X_3]}
\ee
where $[X]$ denotes the component of $X$ along $C$ (i.e. $Q_\infty\cdot X$),
and the measure is defined by
\be
d^{3n{-}7}P=\lb Cd^{n{-}1}C\rb\lb CDd^{n{-}3}D_1\rb\lb CDd^{n{-}3}D_2\rb
\ee
with $D_1,D_2$ denoting the two rows of $D$.

Again, unlike the usual nonnegative Grassmannian, the 1-loop nonnegative Grassmannian $G_{\geq 0}(1,n;2)$ is itself a {\it polytope-like} positive geometry.  Indeed, it is shown in \cite{Bai:2015qoa} that $G_{\geq 0}(1,n;2)$ can be triangulated by a collection of the top cells $\Pi_{\{a,b,c\}}$.  Namely, given a usual triangulation $(\{a,b,c\},\{d,e,f\},\ldots)$ of the $n$-gon (necessarily with $n-2$ pieces), the collection $\Pi_{\{a,b,c\}}, \Pi_{\{d,e,f\}},\ldots$ gives a triangulation of $G_{\geq 0}(1,n;2)$, and we have
\begin{equation}
\Omega(G_{\geq 0}(1,n;2)) = \Omega(\Pi_{\{a,b,c\}}) + \Omega(\Pi_{\{d,e,f\}}) + \cdots.
\end{equation}
Thus for example, we have
\be
\Omega(G_{\geq 0}(1,5;2)) = \Omega(\Pi_{\{1,2,5\}})+ \Omega(\Pi_{\{2,3,5\}}) + \Omega( \Pi_{\{3,4,5\}})
\ee
Remarkably, the right hand side sums to the area of the pentagon with vertices $P_1,\ldots,P_5$, which be be seen pictorially.

More generally, a triangulation independent formula for $\Omega(G_{\geq 0}(1,n;2))$ is given in \cite{Bai:2015qoa}:
\begin{equation}
\Omega(G_{\geq 0}(1,n;2)) = \frac{d^{3n{-}7}P}{\prod_{i=1}^n (i,i+1,i+2)}\mathcal{M}(P).
\end{equation}
where $\mathcal{M}(P)$ is the area of the $n$-gon with vertices $P_1,\ldots,P_n$. A simple formula for this area is
\be
\mathcal{M}(P)=\sum_{i=2}^{n{-}1}\frac{(1,i,i{+}1)}{[1][i][i{+}1]}
\ee
Note that $\mathcal{M}(P)$ is positive in $G_{> 0}(1,n;2)$, so that $\Omega(G_{\geq 0}(1,n;2))$ is positively oriented on $G_{> 0}(1,n;2)$.  Thus $G_{\geq 0}(1,n;2)$ is a positively convex geometry in the sense of Section \ref{sec:convex}.

\subsubsection{An example of a Grassmann polytope} \label{sec:Grassmann}


We now consider an example of a Grassmann polytope for $k > 1$ that is combinatorially different to the Amplituhedron.  We set $k=2$, $n =5$, and $m = 2$.  We take $Z$ to be a $5 \times 4$ matrix such that $\ip{1234} < 0$ but $\ip{1235}, \ip{1245}, \ip{1345}, \ip{2345} >0$.  Explicitly, we may gauge fix 
$$
Z=\left(
\begin{array}{cccc}
 a & -b & c & d \\
 1 & 0 & 0 & 0 \\
 0 & 1 & 0 & 0 \\
 0 & 0 & 1 & 0 \\
 0 & 0 & 0 & 1 \\
\end{array}
\right)
$$
where $a,b,c,d>0$.   The matrix 
$$
M=\left(
\begin{array}{cc}
 1 & 0 \\
 0 & 1 \\
 -\alpha  & \beta  \\
 -\delta  & \gamma  \\
\end{array}
\right)
$$
with $0 < \alpha,\beta,\gamma,\delta\ll 1$, chosen to satisfy $\alpha/\beta < \delta/\gamma < a/b$ shows that $Z$ satisfies the condition \eqref{eq:GrassGordan}.
Let $\A = Z(G_{\geq 0}(2,5))$ denote the Grassmann polytope.  Then for $Y \in \A$, the global inequalities $\ip{Y14} \leq 0$, $\ip{Y15} \geq 0$, and $\ip{Y45} \geq 0$ hold: for example, expanding $Y = C \cdot Z$, we get 
$$
\ip{Y14} = (23)\ip{2314} + (25)\ip{2514} + (35)\ip{3514} = (23) \ip{1234} -(25)\ip{1245} -(35)\ip{1345} \leq 0
$$
where $(ab)$ denotes the minor of $C \in G_{\geq 0}(2,5)$ at columns $a,b$.  In contrast the tree Amplituhedron $\A(2,5,2)$ where $\ip{1234} > 0$ satisfies $\ip{Yi,i{+}1} \geq 0$ for all $i$.

Our methods are not able to rigorously prove it, but a triangulation of this Grassmann polytope appears to be given by the images under $Z$ of the four $4$-dimensional positroid cells
\begin{align}\label{eq:GPtriang}
[2][3][4][5]: (*1) = 0, \qquad 
&[1][2][3][5]: (*4) = 0,  \nonumber \\
[1][2][3][4]: (*5) = 0, \qquad
&[12][34][5]: (12) = (34) =0 
\end{align}
or of the two $4$-dimensional positroid cells
\begin{equation}\label{eq:GPtriang2}
[1][3][4][5]: (2*) = 0 , \qquad [1][23][45]: (23)=(45)=0.
\end{equation}
Here, the notation $[12][34][5]$ denotes the $C$ matrices where columns $C_1$ and $C_2$ (resp. $C_3$ and $C_4$) are parallel, but $C_5$ is linearly independent; this cell is cut out by the two equations $(12)=(34)=0$.  Similarly, $[2][3][4][5]$ denotes the $C$ matrices where $C_1 = 0$ and the other columns have no relations; this cell is cut out by $(12)=(13)=(14)=(15)=0$.  

Thus the canonical rational function is given by 
\be
\aOmega(\A)=
 [2,3,4;2,4,5]+[1,2,3;1,3,5] \\
 - [1,2,3;1,3,4]+[5,1,2;5,3,4] 
\ee
where the terms correspond respectively to the cells in the triangulation~\eqref{eq:GPtriang}, and the brackets are Kermit terms (see below~\eqref{eq:kermit}).

Alternatively, the canonical rational function can be given by two terms based on the second triangulation~\eqref{eq:GPtriang2}:
\be
\Omega(\A)= [1,3,4;1,4,5]+[1,2,3;1,4,5]
\ee
This gives a remarkable algebraic equality between the two Kermit representations above, which of course results from the phenomenon of triangulation independence of the form.

Note that the Kermit triangulation of the Amplituhedron $\A(2,5,2)$ differs from \eqref{eq:GPtriang} and \eqref{eq:GPtriang2} by the cell $C = [1][2][3][4]$, whose canonical form is given by the Kermit term $[1,2,3;1,3,4]$. The intuition for this is as follows.  Starting with a positive $Z$ matrix, we may continuously vary the entries so that $\ip{1234}$ changes from positive to negative, but all other maximal minors remain positive.  At the moment when $\ip{1234} = 0$, the image of the cell $C$ in $G(2,4)$ is no longer four-dimensional, but is three-dimensional, i.e., it collapses.  Changing $\ip{1234}$ from positive to negative thus changes whether $C$ is used in a triangulation.  An analogous situation for a quadrilateral is illustrated in Figure \ref{fig:interiorpt}.

Finally, despite not having proven the triangulations~\eqref{eq:GPtriang} and~\eqref{eq:GPtriang2}, our confidence in our claim comes from the fact that the resulting canonical form agrees numerically with the $i\epsilon$ contour representation discussed in Section~\ref{sec:contour1}.

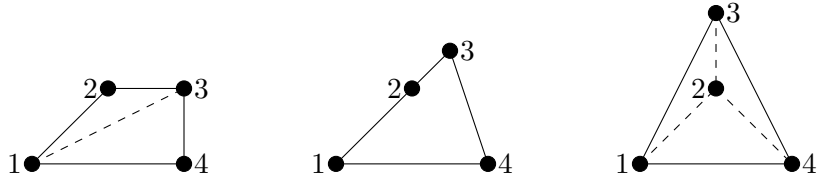
\begin{figure}
\begin{center}
\begin{tikzpicture}
\draw (0,0) -- (1,1) -- (2,1)--(2,0)-- (0,0);
\draw[fill] (0,0) circle (0.1cm) node[left] {$1$};
\draw[fill] (1,1) circle (0.1cm) node[left] {$2$};
\draw[fill] (2,1) circle (0.1cm) node[right] {$3$};
\draw[fill] (2,0) circle (0.1cm) node[right] {$4$};
\draw[dashed] (0,0) -- (2,1);
\begin{scope}[shift={(4,0)}]
\draw[fill] (0,0) circle (0.1cm) node[left] {$1$};
\draw[fill] (1,1) circle (0.1cm) node[left] {$2$};
\draw[fill] (1.5,1.5) circle (0.1cm) node[right] {$3$};
\draw[fill] (2,0) circle (0.1cm) node[right] {$4$};
\draw (0,0) -- (1,1) -- (1.5,1.5)--(2,0)-- (0,0);
\end{scope}
\begin{scope}[shift={(8,0)}]
\draw[fill] (0,0) circle (0.1cm) node[left] {$1$};
\draw[fill] (1,1) circle (0.1cm) node[left] {$2$};
\draw[fill] (1,2) circle (0.1cm) node[right] {$3$};
\draw[fill] (2,0) circle (0.1cm) node[right] {$4$};
\draw (0,0) --  (1,2)--(2,0)-- (0,0);
\draw[dashed] (0,0) -- (1,1) -- (1,2);
\draw[dashed] (1,1) -- (2,0);
\end{scope}
\end{tikzpicture}
\end{center}
\caption{Changing $\lb 123\rb$ from $>0$, to $0$, then to $<0$.  The triangulation $[124] + [234]$ is changed to $[213]+[124] + [234]$.}
\label{fig:interiorpt}
\end{figure}

\subsection{Push-forwards} \label{sec:pushforward}
Recall from Heuristic~\ref{heuristic} that for any morphism $\Phi:(X,X_{\geq 0})\rightarrow (Y,Y_{\geq 0})$, we expect the push-forward to preserve the canonical form:
\be
\Phi_*(\Omega(X,X_{\geq 0}))=\Omega(Y,Y_{\geq 0})
\ee
This procedure is useful for computing the canonical form of the image when the canonical form of the domain is already known. However, finding morphisms of degree $>1$ is a difficult challenge. We demonstrate a few non-trivial examples in this section.

\subsubsection{Projective simplices}
\label{sec:simplexPush}
In this section we consider morphisms $\Phi:(\P^m,\Delta)\rightarrow (X,X_{\geq 0})$ from projective simplices to positive geometries $X_{\geq 0}$. In most cases we will assume that $\Delta =\Delta^m$ is the standard simplex, since they are isomorphic.

We begin with morphisms $\Phi:(\P^m,\Delta^m)\rightarrow (X,X_{\geq 0})$ of degree one, in which case $X_{\geq 0}$ is $\Delta-$like (as defined in Section \ref{sec:standardsimplex}), and its canonical form is given by:
\be
\Omega{(X,X_{\geq 0})}=\Phi_* \left(\prod_{i=1}^m \frac{d\alpha_i}{\alpha_i}\right)
\ee
where $(1,\alpha_1,\ldots ,\alpha_m)\in\P^m$. The simplest $\Delta$-like positive geometry is a projective simplex:

\begin{example}
A projective simplex $\Delta\subset \P^m(\R)$ is isomorphic to the standard simplex $\Delta^m$ by the following map $\Phi:(\P^m,\Delta^m)\rightarrow (\P^m,\Delta)$:
\be
\Phi(\alpha)=\sum_{i=0}^m \alpha_i Z_{i{+}1}
\ee
where $Z_i\in\P^m(\R)$ are the vertices of $\Delta$ for $i=1,\ldots, m{+}1$. As a matter of convention, the projective variables and the vertices are indexed slightly differently. Note that the positive part $\Delta^m$ (i.e. $\alpha_i>0$ for each $i$) is mapped diffeomorphically onto the interior of $\Delta$.

The canonical form on $\Delta$ is therefore
\be
\Omega(\Delta)=\Phi_*\left(\prod_{i=1}^m\frac{d\alpha_i}{\alpha_i}\right)
\ee
where we have made the ``gauge choice" $\alpha_0=1$ as usual. Alternatively, pulling back the form~\eqref{eq:simplexZ} onto $\Delta^m$ gives the form on $\alpha$-space.

 The image of the hyperplane $\{\alpha_i = 0\} \subset \P^m$ intersects $\Delta$ along the facet opposite the vertex $Z_{i+1}$.  Taking the residue of $\Omega(\Delta^m)$ along $\alpha_i = 0$ before pushing forward gives the canonical form of that facet.  We note that the pole for localizing on the facet opposite $Z_1$ is hidden, as explained in Section~\ref{sec:standardsimplex}. 

\end{example}

We now consider {\sl higher degree} morphisms $\Phi:(\P^m,\Delta^m)\rightarrow (X,X_{\geq 0})$. Let us assume for the moment that $X=\P^m$. We now provide a general analytic argument for why the push-forward should have no poles on $X_{>0}$. The behavior of the push-forward near the boundary of $X_{\geq 0}$ is more subtle and will be discussed subsequently on a case-by-case basis.

Suppose the map is given by $\alpha\mapsto\Phi(\alpha)$ and 
let $\beta_0$ be a point in $X_{>0}$. Furthermore, assume if possible that the push-forward has a singularity at $\beta_0$. It follows that the Jacobian $J(\alpha)$ of $\Phi(\alpha)$ must vanish at some point $\alpha_0$ for which $\Phi(\alpha_0)=\beta_0$. Let $\beta=\Phi(\alpha)$ which we expand near the critical point.
\be\label{eq:noPoles}
\beta=\Phi(\alpha_0)+\lambda \sum_i\epsilon_i \frac{\partial \Phi(\alpha_0)}{\partial \alpha_i}+\frac{1}{2}\lambda^2\sum_{i,j}\epsilon_i\epsilon_j\frac{\partial^2\Phi(\alpha_0)}{\partial \alpha_i\partial \alpha_j}+O(\lambda^3)
\ee
where we have set $\alpha=\alpha_{0}+\lambda \epsilon$ with $\lambda$ a small parameter and $\epsilon$ a constant vector. Since the Jacobian vanishes at $\alpha_0$, a generic point $\beta$ in a small neighborhood of $\beta_0$ cannot be approximated by the linear term. However, unless the quadratic term degenerates, $\beta$ can be approximated quadratically by choosing $\epsilon$ so that the first variation vanishes. Namely,
\be
\sum_a \epsilon_i\frac{\partial \Phi(\alpha_0)}{\partial \alpha_i}=0
\ee
It follows that the variation is even in $\lambda$, so there are two roots $\lambda_\pm$ (corresponding to points $\alpha_\pm$, respectively) that approximate $\beta$, with $\lambda_+=-\lambda_-$. Since the Jacobian is clearly linear in $\lambda$ for small variations near $\alpha_0$, therefore $J(\alpha_+)=-J(\alpha_-)+O(\lambda_-^2)$. Since the push-forward is a sum of $1/J(\alpha)\sim 1/\lambda$ over all the roots, the roots corresponding to $\alpha_\pm$ therefore cancel in the limit $\beta\rightarrow \beta_0$, and there is no pole. 

We now show a few examples of higher degree push-forwards, beginning with self-morphisms of the standard simplex.
\begin{example}\label{ex:tori}

Let $\Phi:(\P^m,\Delta^m)\rightarrow(\P^m,\Delta^m)$ be a morphism of the standard simplex with itself, defined by
\be
\Phi(1,\alpha_1,\ldots, \alpha_m)&=&(1,\beta_1,...,\beta_m)\\
\beta_j&=&\prod_{i=1}^m\alpha_i^{a_{ij}}
\ee
where $a_{ij}$ is an invertible integer matrix. We assume the determinant is positive so that the map is orientation preserving. While this map is a self-diffeomorphism of $\Int(\Delta^m)$, it is not necessarily one-to-one on $\P^m$. The push-forward gives
\be \label{eq:monomialpush}
\Phi_*\left(\prod_{i=1}^m\frac{d\alpha_i}{\alpha_i}\right)=\sum_{\text{roots}}\frac{d^m\beta}{\frac{\partial(\beta_1\ldots\beta_m)}{\partial(\alpha_1\ldots\alpha_m)}\prod_{i=1}^m\alpha_i}=\frac{\deg(\Phi)}{\det(a_{ij})}\prod_{j=1}^m\frac{d\beta_j}{\beta_j}
\ee
where $\deg(\Phi)$ is the number of roots, and we have substituted the Jacobian:
\be
\frac{\partial(\beta_1\ldots\beta_m)}{\partial(\alpha_1\ldots\alpha_m)}=\det(a_{ij})\frac{\prod_{j=1}^m\beta_j}{\prod_{i=1}^m\alpha_i}
\ee
It is easy to see that the degree of $\Phi$ must be $|\det(a_{ij})|$: after an integral change of basis, the matrix $(a_{ij})$ can be put into Smith normal form, that is, made diagonal.  For a diagonal matrix $(a_{ij})$ it is clear that the degree of $\Phi$ is simply the product of diagonal entries.  Thus \eqref{eq:monomialpush} verifies Heuristic~\ref{heuristic} in this case.

By contrast, we note that the pull-back along $\Phi$ gives $\Phi^*(\Omega(\Delta^m))=\det(a_{ij})\Omega(\Delta^m)$, which does not preserve leading residues.
\end{example}

The next few examples explore the push-forward in one dimension. They are all applications of Cauchy's theorem in disguise.

\begin{example}

A simple non-trivial example is a quadratic push-forward $\Phi:\Delta^1\rightarrow \Delta^1$ given by $\Phi(1,\alpha)=(1,a \alpha^2+2b\alpha)$ for some real constants $a>0;b\geq 0$. The assumptions suffice to make $\Phi$ a self-morphism of $\Delta^1$. Setting $(1,\beta)=\Phi(1,\alpha)$ we get two roots $\alpha_\pm$ from solving a quadratic equation. The push-forward is therefore
\be
\Phi_*(\dlog \alpha)=\sum_\pm \dlog\left(\alpha_\pm\right)=\sum_\pm \dlog\left(\frac{-b\pm\sqrt{b^2+a\beta}}{a}\right)
\ee
Since we are summing over roots, a standard Galois theory argument implies that the result should be rational.  Indeed, the square-root disappears, and direct computation gives $\Phi_*(\dlog \alpha) = \dlog \beta$.

We can also do the sum without directly solving the quadratic equation. The result should only depend on the sum and product of the roots $x_\pm$, since the result must be a rational function.
\be
\Phi_*(\dlog \alpha)=\sum_\pm \frac{d\beta}{\alpha_\pm \left(\frac{d\beta}{d\alpha}\right)_\pm}=\sum_\pm \frac{1}{2\alpha_\pm(a\alpha_\pm+b)}
\ee
Substituting $\alpha_\pm(a\alpha_\pm+b)=\beta-b\alpha_\pm$, which comes from the original equation, we get
\be
f_*(\dlog \alpha)=\sum_\pm \frac{1}{2(\beta-b\alpha_\pm)}=
\frac{2\beta-b(\alpha_++\alpha_-)}{2(\beta^2-b\beta(\alpha_++\alpha_-)+b^2\alpha_+\alpha_-)}
\ee
We now use the identities $\alpha_++\alpha_-=-2b/a$ and $\alpha_+\alpha_-=-\beta/a$, which give the desired result $\Phi_*(\dlog \alpha) = \dlog \beta$.
\end{example}

\begin{example}\label{ex:1Dpush}
We now go ahead and tackle the same example for a polynomial of arbitrary degree. Suppose $\Phi$ is a self-morphism of $\Delta^1$ given by $\beta=f(\alpha)=\alpha^n+a_{n{-}1}\alpha^{n{-}1}+\ldots +a_1\alpha$.

We first define the holomorphic function
\be
g(\alpha)=\frac{1}{\alpha(f(\alpha)-\beta)}
\ee
which has no pole at infinity since $f(\alpha)$ is at least of degree one. The sum over all the residues of the function is therefore zero by Cauchy's theorem. It follows that
\be
-\frac{1}{\beta}+\sum_i\frac{1}{\alpha_if'(\alpha_i)}=0
\ee
where we sum over all the roots $\alpha_i$ of $f(\alpha)=\beta$. Therefore,
\be
\Phi_*(\dlog \alpha)=\sum_i \frac{d\beta}{\alpha_i f'(\alpha_i)}=d\log\beta
\ee
\end{example}

Finally, we consider a simple but instructive push-forward of {\sl infinite} degree.

\begin{example}\label{ex:pushsegment}
Consider the map $\Phi:(\P^1,[-\pi/2,\pi/2])\rightarrow(\P^1,[-1,1])$ between two closed line segments given by:
\be\label{eq:trigpush}
\Phi(1,\theta)=(1,\sin\theta)
\ee
While this map is not rational, we can nevertheless verify Heuristic~\ref{heuristic} for $\Phi$ by explicit computation. For any point $(1,\sin\theta)$ in the image, there are infinitely many roots of~\eqref{eq:trigpush} given by $\theta_n\deff \theta+2\pi n$ and $\theta_n'\deff -\theta+\pi(2n+1)$ for $n\in\mathbb{Z}$. It is easy to show that both sets of roots contribute the same amount to the push-forward, so we will just sum over $\theta_n$ twice:
\be\label{eq:trigpushseries}
\Phi_*\left(\frac{\pi d\theta}{(\pi/2-\theta)(\theta+\pi/2)}\right)&=&2\sum_{n\in\mathbb{Z}}\frac{\pi dx}{(\pi/2-\theta_n)(\theta_n+\pi/2)\cos\theta_n}=2\frac{dx}{\cos^2\theta}\\
&=&\frac{2dx}{(1-x)(x+1)}
\ee
which of course is the canonical form of $[-1,1]$. The infinite sum can be computed by an application of Cauchy's theorem.

We note here that our definition~\eqref{eq:pushforward} of the push-forward only allows finite degree maps while $\Phi$ is of infinite degree. Nonetheless, it appears that Heuristic~\ref{heuristic} still holds when the push-forward is an {\sl absolutely convergent} series like~\eqref{eq:trigpushseries}. We stress that some push-forwards give conditionally convergent series, such as the morphism $(\P^1,\Delta^1)\rightarrow(\P^1,[0,1])$ given by $(1,x)\rightarrow(1,e^{-x})$, in which case the push-forward is ill-defined since there is no canonical order in which to sum the roots.
\end{example}

\subsubsection{Algebraic moment map and an algebraic analogue of the \\ Duistermaat-Heckman measure}\label{sec:momentmap}
Recall from Section \ref{sec:toric} that toric varieties $X(z)$ are positive geometries.  We now show how the canonical form of a polytope $\A = \Conv(Z)$ can be obtained as the push-forward of the canonical form of a toric variety, establishing an instance of Heuristic \ref{heuristic}.

Associated to the torus action of $T$ on $X(z)$ is a {\it moment map} $\mu:X(z) \to \P^{m}(\R)$
\be \label{eq:moment}
X(z) \ni (C_1: \cdots : C_n) \longmapsto \sum_{i=1}^n |C_i|^2 z_i
\ee
which is an important object in symplectic geometry.  The image $\mu(X(z))$ is the polytope $\A(z) \in \P^m$ with vertices $z_1,\ldots,z_n$.

Now, let us suppose we have a polytope $\A = \Conv(Z)$ with vertices $Z_1,\ldots,Z_n$ such that $Z$ and $z$ have the same ``shape".  Namely, we insist that the determinants
\be \label{eq:samesign}
\langle Z_{i_0} \cdots Z_{i_{m}} \rangle \text{ and } \langle z_{i_0} \cdots z_{i_{m}} \rangle \text{ have the same sign}
\ee
for all $1\le i_0, i_1, \ldots , i_{m}\le n$.  Here, two real numbers have the same sign if they are both positive, both negative or both zero.  In other words, we ask that the vector configurations $Z$ and $z$ have the same \defn{oriented matroid} (see Appendix \ref{sec:matroids}).  We caution that the integer matrix $z$ may not exist if $Z$ is not ``realizable over the rationals".

Some basic terminology concerning oriented matroids is recalled in Appendix \ref{sec:matroids}.
We then have the (rational) linear map $Z: X(z) \to \P^m$
\be \label{eq:algmoment}
X(z) \ni (C_1: \cdots : C_n) \longmapsto \sum_{i=1}^n C_i Z_i.
\ee
When $Z = z$, this map is called the {\it algebraic moment map} \cite{Fulton,Sottile}.  Note that the moment map has image in a real projective space but the algebraic moment map has image in a complex projective space.  We shall show in Section \ref{app:diffeo} that the image $Z(X(z)_{\geq 0})$ of the nonnegative part is the polytope $\A$.

Suppose $ Y \in \P^m$.  Then the inverse image of $Y$ under the map $Z$ is a linear slice $L_Y \subset \P^{n-1}$ of dimension $n-1-m$.  For a typical $Y$, the slice $L_Y$ intersects $X(z)$ in finitely many points and we have the elegant equality
\be \label{eq:solutions}
\#|L_Y \cap X(z)| = m! \cdot \mbox{volume of $\A(z)$}.
\ee
Here the volume of $\A(z)$ is taken with respect to the lattice generated by the vectors $z_1,\ldots,z_n$, so that the unit cube in this lattice has volume 1.
In geometric language, \eqref{eq:solutions} states that the {\it degree} of $X(z)$ is equal to $m!$ times the volume of $\A(z)$.    
When $X(z)$ is a smooth complex projective variety, it is also a symplectic manifold.  In this case it has a real $2m$-form $\omega$ (not meromorphic!), called its {\it symplectic volume}.  The {\it Duistermaat-Heckman measure} is the push-forward $\mu_* \omega$ on $\P^{m}(\R)$, where $\omega$ is thought of as a measure on $X(z)$.  Identifying $\A(z)$ with a polytope inside $\R^m$, a basic result states that $\mu_* \omega$ is equal to the standard Lebesgue measure inside the polytope $\A(z)$, and is zero outside.

We may replace the moment map $\mu$ by the linear map $Z$, and the symplectic volume $\omega$ by our canonical holomorphic form $\Omega({X(z)_{\geq 0}})$ from Section \ref{sec:toric}.  The crucial result is the following.

\begin{theorem}\label{thm:push-forward}
Assume that $z$ and $Z$ have the same oriented matroids and that $z$ is graded.  Then $Z_*(\Omega({X(z)_{\geq 0}})) = \Omega(\A)$ is the canonical form of the polytope $\A = \Conv(Z)$.
\end{theorem} 


We remark that the graded condition \eqref{eq:graded} is a mild condition: since $Z$ satisfies the condition \eqref{eq:Gordan}, it follows from Proposition \ref{prop:OM} that $z$ does as well.  Thus the vectors $z_i$ can be scaled by positive integers until the set $\{z_1,\ldots,z_n\}$ is graded~\eqref{eq:graded}.

Note that in Theorem \ref{thm:push-forward} we do not need to assume that $X(z)$ is projectively normal.

Theorem \ref{thm:push-forward} is proved in Appendix \ref{app:push-forwardproof}.  We sketch the main idea. Both $(X(z),X(z)_{\geq 0})$ and $(\P^m,\A)$ are positive geometries.  The condition that $z$ and $Z$ have the same oriented matroid implies that the polytopes $\A(z)$ and $\A$ have the same combinatorial structure.  Thus the axioms (P1) and (P2) for the two positive geometries have the same recursive structure.  We prove the equality $Z_*(\Omega({X(z)_{\geq 0}})) = \Omega(\A)$ by induction, assuming that the equality is already known for all the facets of $\A$.  At the heart of the inductive step is the fact that taking push-forwards and taking residues of meromorphic forms are commuting operations.

\subsubsection{Projective polytopes from Newton polytopes}
\label{sec:push}
In this section we discuss morphisms from $(\P^m,\Delta^m)$ to convex polytopes $(\P^m,\A)$ and their push-forwards. Our main results here overlap with results from our discussion on toric varieties in Section~\ref{sec:momentmap}. 
However, our intention here is to provide a self-contained discussion. Our focus here is also more geometric in nature, emphasizing the fact that any such morphism restricts to a diffeomorphism $\Int(\Delta^m)\rightarrow \Int(\A)$.

Now let $\A \subset \P^m$ be a convex polytope in projective space with vertices $Z_i$ for $i=1,\ldots,n$.  
Let $z_1, z_2,\ldots, z_n \in \ZZ^{m+1}$ be an {\em integer} matrix with the same oriented matroid as $Z_1,\ldots,Z_n$, that is, satisfying \eqref{eq:samesign}.

For simplicity we now assume that $z_i = (1,z'_i)=(1,z'_{1i},z'_{2i},\ldots,z'_{mi})$ (see Section \ref{sec:toric} for how to relax this condition).

Let us define a rational map $\Phi:(\P^m,\Delta^m)\rightarrow (\P^m,\A)$ given by:
\be\label{eq:Phi}
\Phi(X)&=&\sum_i C_i(X) Z_i\\
C_i(X)&\deff & X^{z_i} \deff X_1^{z'_{1i}} X_2 ^{z'_{2i}} \cdots X_{m}^{z'_{m,i}}
\ee
where $(1,X)\in \P^m$. This is called the \defn{Newton polytope map} and the polytope with integer vertices $z_i$ is called the \defn{Newton polytope}.

We make two major claims in this section. The first is the following:
\be\label{eq:claim1}
\mbox{{\bf Claim 1:} The map $\Phi$ is a \defn{morphism} provided that~\eqref{eq:samesign} holds.}
\ee
That is, it restricts to a \defn{diffeomorphism} on $\Int(\Delta^m)\rightarrow \Int(\A)$.

This is a non-trivial fact for which we provide two proofs in Appendix~\ref{app:diffeo}. At the heart of Claim~\eqref{eq:claim1} is that the Jacobian of $\Phi$ is uniformly positive on $\Int(\Delta^m)$:
\be
J(\Phi)=\sum_{1\le i_0< \cdots <i_{m}\le n}C_{i_0} \cdots C_{i_{m}}\lb z_{i_0} \cdots z_{i_{m}}\rb\lb Z_{i_0} \cdots Z_{i_{m}}\rb
\ee
where $J(\Phi)$ denotes the Jacobian of $\Phi$ with respect to $u_i\deff\log(X_i)$, which is clearly positive provided that $z_i$ and $Z_i$ have the same oriented matroid (see \eqref{eq:samesign}). While this is only a {\sl necessary} condition for Claim~\eqref{eq:claim1}, it is nevertheless the key to proving it.
We provide a graphical example of the diffeomorphism for the pentagon in Figure~\ref{fig:pentagon}.

This establishes the way to our second claim, which is an instance of Heuristic~\ref{heuristic}:
\be\label{eq:claim2}
\mbox{{\bf Claim 2:} The canonical form of the polytope is given by the push-forward:}\nonumber\\
\Omega(\A) = \Phi_*\left(\frac{d^m X}{\prod_{a=1}^m X_a}\right)\hskip 1.7in
\ee

Equation \eqref{eq:claim2} follows from Theorem~\ref{thm:push-forward}.  To see this, note that in \eqref{eq:pushtheta}, we have $\Omega_S = \prod_{i=1}^m dX_i/X_i$ and $\Omega_T = \Omega_{{X(z)}}$ as meromorphic forms.  Thus \eqref{eq:claim2} follows from \eqref{eq:pushtheta} and Theorem~\ref{thm:push-forward}.   We provide an alternative analytic argument avoiding toric varieties in Section~\ref{sec:sumOverRoots} by directing manipulating the sum-over-roots procedure in the push-forward computation.


\begin{figure}
\centering
\includegraphics[width=6cm]{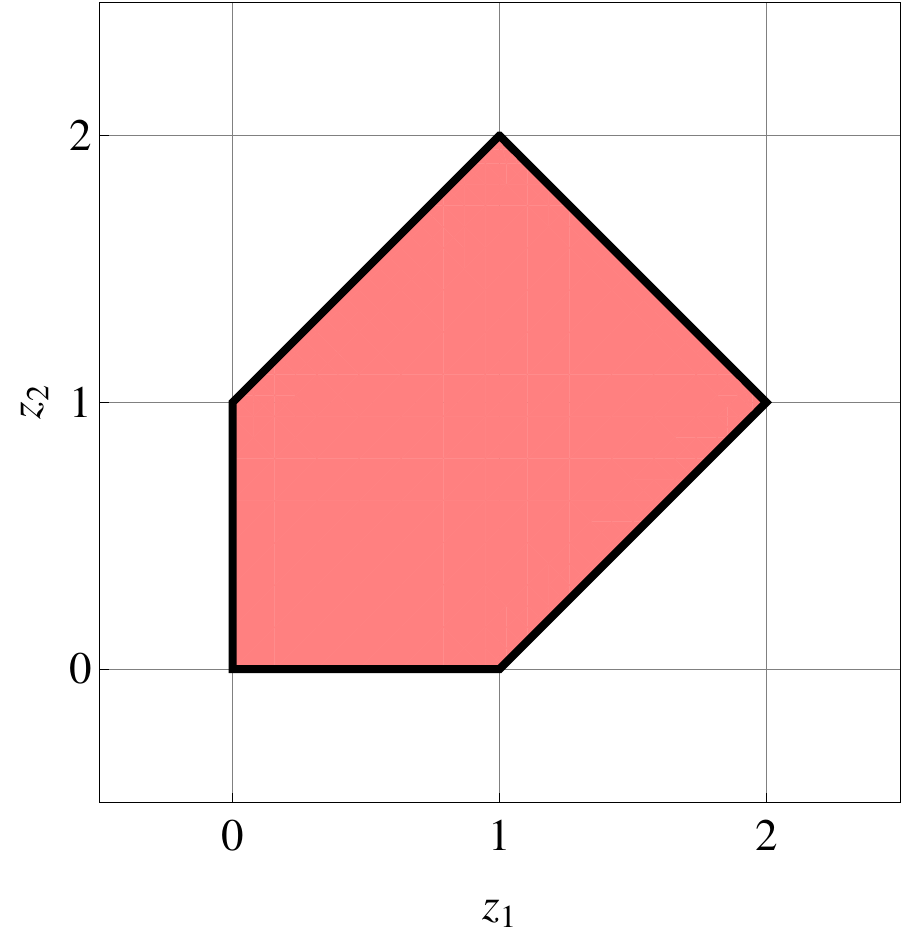}\;\;\;\;\;
\includegraphics[width=6.75cm]{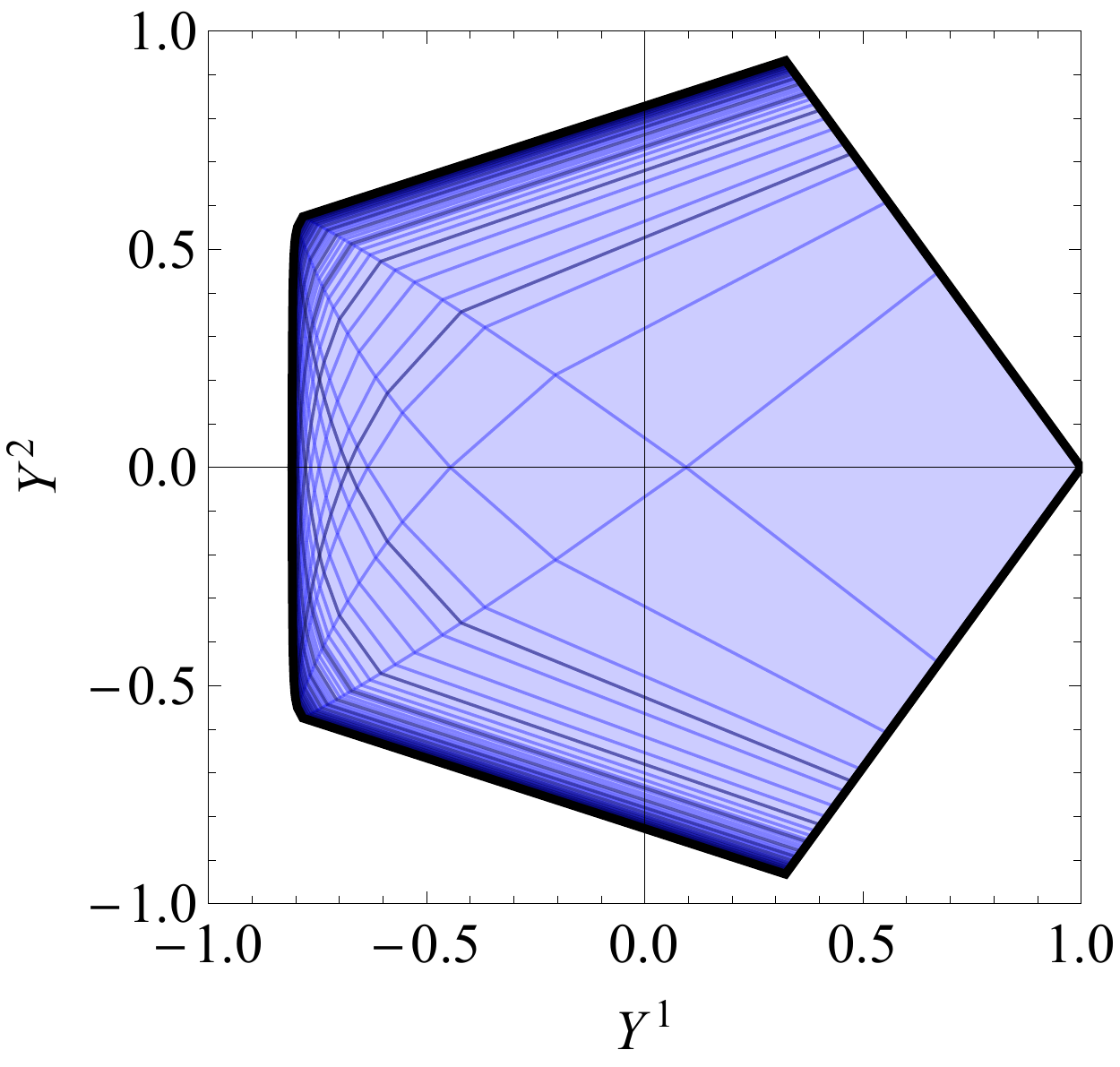}
\caption{A Newton pentagon (left) with vertices $z'=((0,0),(1,0),(2,1),(1,2),(0,1))$ and the image of the corresponding diffeomorphism $\Phi(X)$ (right). The dark lines in the interior of the pentagon (right) denote lines of constant $X_1$ or $X_2$. While the two pentagons are not identical as sets, they do have the same oriented matroid.}
\label{fig:pentagon}
\end{figure}

For computational purposes, a more convenient way to express the (canonical rational function of the) push-forward is the following integral formula:
\be\label{eq:deltapush}
\aOmega(\A)(Y)=\int \frac{d^m X}{\prod_{a=1}^m X_a} \delta^{m}(Y;\Phi(X))
\ee
where for any two points $Y,Y'$ in projective space $\mathbb{P}^m$, the integral
\be\label{eq:delta4}
\delta^m(Y,Y') = \frac{1}{m!}\int\frac{d\rho}{\rho}\delta^{m+1}(Y-\rho Y')
\ee
is the delta function with weights $(-m{-}1,0)$. Here the sum over pre-images appearing in the push-forward is equivalent to the sum over solutions to the constraints imposed by the delta functions. The relevant Jacobian factors are taken care of by the delta functions as well. We treat the delta functions formally as an ``analytic" object without taking absolute values in the Jacobian factor, and we sum over all {\sl complex} roots.

In most cases numerical methods are needed to compute the push-forward, which we have done extensively as verification of Claim~\eqref{eq:claim2}. However, here we show a simple case where the push-forward can be done by hand.

\begin{example}
Let us consider a quadrilateral $\A$ on the projective plane. By a convenient gauge-fixing, we will take the external data and $Y$ to be of the form
\be
Z=\begin{pmatrix}
1 & 0 & 0 \\
0 & 1 & 0 \\
0 & 0 & 1 \\
a & -1 & b
\end{pmatrix},\;\;\;\;
Y = (y,1,x)
\ee

We can trivially compute the canonical rational function by e.g. triangulation as $[123] +
[134]$, and finding
\begin{equation}
\Omega(\A)=\frac{(ab+a x + b y) }{x y (a + y)(b+x)} d x d y
\end{equation}
and we will now reproduce this result from the Newton polytope map \eqref{eq:Phi}, with the Newton polytope being the square with vertices
$(1,0),(1,1),(0,1),(0,0)$, respectively. Thus our map is
\be
Y =\Phi(X)= X_1 Z_1 +
X_1 X_2 Z_2 + X_2 Z_3+Z_4 = \left(\frac{X_1+a}{X_1 X_2-1},1,\frac{X_2+b}{X_1X_2-1}\right)
\ee
which are projective equalities. We must solve the equations
\begin{equation}
(X_1X_2 - 1) y = X_1 + a, \;\;\; (X_1X_2 - 1) x = X_2+ b
\end{equation}
We then have
\begin{equation}
\frac{d^2X}{X_1X_2} = dx dy \times
\frac{(X_1X_2 - 1)^2}{X_1X_2 (x X_1+y X_2-1)}
\end{equation}
It therefore suffices to show that
\begin{equation}
\sum_{{\rm roots}} \frac{(X_1 X_2 - 1)^2}{X_1 X_2 (x X_1+y X_2-1)} = \frac{(ab+a x + b y) }{x y (a + y)(b+x)}
\end{equation}
If we define $\rho = (X_1X_2 - 1)$, then we have $X_1 = \rho
y - a, X_2= \rho x - b$ and so $(\rho + 1) = (\rho y - a)(\rho x -
b)$. So $\rho$ satisfies the quadratic equation
\begin{equation}
\rho^2 - \frac{(1 + a x+by)}{x y}\rho + \frac{a b - 1}{x y} = 0
\end{equation}
which has two roots $\rho_\pm$. Note that 
\begin{equation}
xX_1+yX_2-1 = x (\rho y -
a)+ y (\rho x - b)-1 = 2 x y \left(\rho - \frac{1}{2}\frac{(1 + ax + by)}{x y}\right)
\end{equation}
Comparing with the quadratic equation for $\rho$, we can identify the term corresponding to the sum of the roots $(\rho_++\rho_-)$, giving:
\be
(xX_1+yX_2-1)_\pm =\pm xy(\rho_+ -\rho_-)
\ee

Thus, the sum over roots becomes
\begin{eqnarray} \nonumber
\frac{1}{x y (\rho_+ - \rho_-)} \times \left(\frac{\rho_+^2}{(\rho_+ y -
a)(\rho_+ x - b)} - \frac{\rho_-^2}{(\rho_- y - a)(\rho_- x - b)}\right) \\ \label{eq:pushsquare} =
\frac{1}{x y} \frac{a b (\rho_+ +
\rho_-)-\rho_+ \rho_-(b y + a x)}{(\rho_+ \rho_- y^2 - a(\rho_+ +\rho_-) y + a^2)(\rho_+ \rho_- x^2 -  b (\rho_+ +\rho_-)x + b^2)}
\end{eqnarray}
And reading off
\be
\rho_+ \rho_- = \frac{a b - 1}{x y}, \qquad
\rho_+ +\rho_- =\frac{1 +  a x+by}{x y}
\ee
from the quadratic equation we finally find that \eqref{eq:pushsquare} equals
\begin{equation}
\frac{a b + ax+by}{x y (a + y)(b + x)}
\end{equation}
as expected.
\end{example}

\subsubsection{Recursive properties of the Newton polytope map}
\label{sec:sumOverRoots}

We now derive the push-forward identity~\eqref{eq:claim2} for the Newton polytope by explicitly manipulating the push-forward computation and applying induction on dimension. In fact, we will derive a more general version for lower dimensional polytopes in $\P^m(\R)$. The statement is the following:

Consider the variety $H\subset\P^m$ defined by linear equations $Y\cdot H_1,\ldots,Y\cdot H_s=0$ for some dual vectors $H_1,\ldots, H_s\in\P^m(\R)$ with $Y\in\P^m$. Let $\A$ denote a convex polytope in $H(\R)$ of dimension $D=m{-}s$ with vertices $Z_1,\ldots,Z_n$. Let $\Phi:(\P^D,\Delta^D)\rightarrow(H,\A)$ be a morphism defined by a Newton polytope with vertices $z_i=(1,z_i')\in\R^{D{+}1}$. Furthermore, we assume that $\A$ and the Newton polytope have the same oriented matroid. Namely, there exist constant vectors $K_1,\ldots,K_s\in\P^m(\R)$ such that
\be
\mbox{$\lb K_1\cdots K_s Z_{i_0}\cdots Z_{i_D}\rb$ and $\lb z_{i_0}\cdots z_{i_D}\rb$ have the same sign}
\ee
for any set of indices $1\le i_0,\ldots,i_D\le n$.

The $D=0$ case is trivial. Now assume $D>0$. Let $\B$ denote a facet of $\A$ with vertices $Z_{j_1},\ldots,Z_{j_p}$, and let $B$ denote the corresponding facet of the Newton polytope. We now argue that there exists a change of variables $(X_1,\ldots,X_D)\rightarrow (U_1,\ldots, U_D)$ given by $X_b=\prod_{a=1}^D{U_a}^{\alpha_{ab}}$ for some invertible integer matrix $\alpha_{ab}$ that has the following properties: (Note that we define $\Phi'(U)\deff\Phi(\psi(U))$ with $X\deff\psi(U)$)
\begin{itemize}
\item
The map $U\rightarrow X$ is a self-morphism of $\Delta^D$ which may be of degree $>1$. In particular, it is a self-diffeomorphism of $\Int(\Delta^D)$.
\item
The restriction of $\Phi'$ to $U_D=0$ is a Newton polytope map $\Delta^{D{-}1}\rightarrow\B$ whose Newton polytope has the same shape as $B$.
\end{itemize}

We construct the change of variables as follows. Let us begin by considering the matrix formed by column vectors $(z_{j_1},\ldots,z_{j_D},z_{j})$ where the first $D$ columns are vertices of $B$, while $z_{j}$ is any vertex of the Newton polytope away from the facet. Let $\alpha\deff(\alpha_{ab})$ be the inverse of the matrix, which has only rational components. Finally, let us rescale the rows of $\alpha_{ab}$ (indexed by $a$) by positive integers so that all its components are integral. This provides the change of variables $X\rightarrow U$.

Let us redefine $z_i$ for every $i$ to be the vertex of the Newton polytope with respect to $U$, which again has the same oriented matroid as $\A$. In particular, the matrix of column vectors $(z_{j_1},\ldots, z_{j_D},z_{j})$ is identity. Since every vertex of $B$ is a linear combination of $z_{j_1},\ldots, z_{j_D}$, it must therefore have zero as its last component. Furthermore, since every other vertex of the Newton polytope is a linear combination of $z_{j_1},\ldots, z_{j_D},z_j$ with positive coefficient in $z_j$, its last component must be positive. Geometrically, this means that the $U_D=0$ limit is mapped to the facet $\B$ by $\Phi'$.

Now, let $\Psi:\Delta^{D{-}1}\rightarrow\B$ denote the map $\Phi'$ restricted to $U_D=0$, and let $w_{j_1},\ldots, w_{j_p}\in\R^{D}$ denote the vertices of the Newton polytope of $\Psi$, which are equivalently the vectors $z_{j_1},\ldots, z_{j_p}$, respectively, with the last component removed. It is straightforward to check that the Newton polytope of $\Psi$ has the same oriented matroid as $\B$, so our induction hypothesis can be applied to $\Psi$.

We now argue that the push-forward by $\Phi'$ has a pole at the boundary $\B$ with the expected residue $\Omega(\B)$. Let $Y=\Phi'(U)$. The main strategy is to observe that the roots $(U_1,\ldots,U_{D{-}1})$ are independent of $U_D$ near $\B$.

Indeed, at $Y\in \mathcal{B}$ (i.e. $U_D=0$), the map becomes
\be
Y = C_{j_1}Z_{j_1}+\cdots+C_{j_p}Z_{j_p}
\ee
which gives us a collection of roots $(U_1,\ldots, U_{D{-}1})$ independent of $U_D$. It follows therefore from the induction hypothesis that
\be\label{eq:newton_induction}
\sum_{\text{roots }(U_1,\ldots,U_{D{-}1})}\frac{d^{D{-}1}U}{U_1\cdots U_{D{-}1}}
=\Psi_*\left(\frac{d^{D{-}1}U}{U_1\cdots U_{D{-}1}}\right)=\Omega(\B)
\ee

Moreover, the roots $U_D$ are given by the equation
\be
\lb K_1\ldots K_s Y j_1\ldots j_D\rb \sim U_D^q F(U_1,\ldots,U_{D{-}1})
\ee
for some function $F$ independent of $U_D$ and some exponent $q$. Hence, 
\be\label{eq:newton_single}
\sum_{\text{roots }U_D}\frac{dU_D}{U_D}=d\log\lb K_1\ldots K_s Y j_1\ldots j_D\rb+\cdots
\ee
where the $\cdots$ denotes term proportional to $dU_1,\ldots,dU_{D{-}1}$.

It follows that in the limit $Y\rightarrow \B$, we have
\be
\sum_{\text{roots}}\frac{d^DU}{U_1\cdots U_D}&=&\left(\sum_{\text{roots }(U_1,\ldots,U_{D{-}1})}\frac{d^{D{-}1}U}{U_1\cdots U_{D{-}1}}\right)\left(\sum_{\text{roots }U_D}\frac{dU_D}{U_D}\right)+\cdots\\
&=&\Psi_*\left(\frac{d^{D{-}1}U}{U_1\cdots U_{D{-}1}}\right)d\log\lb K_1\ldots K_s Y j_1\ldots j_D\rb+\cdots\\
&=&\Omega(\B)\;d\log\lb K_1\ldots K_s Y j_1\ldots j_D\rb+\cdots
\ee
where $\cdots$ denotes terms smooth in the limit, and we have applied~\eqref{eq:newton_single} and~\eqref{eq:newton_induction}. This of course is the expected pole and residue.

Furthermore, by the discussion in Section~\ref{sec:simplexPush}, we find that there are no poles on the interior of $\A$.

Finally, we can re-express our result as a push-forward from $X$-space by the following:
\be
\Phi_*\left(\frac{d^DX}{X_1\cdots X_D}\right)=\Phi_*'\left(\frac{d^DU}{U_1\cdots U_D}\right)
\ee
since $\Phi_*' = \Phi_*\circ \psi_*$, and $\psi_*$ pushes the standard simplex form on $U$ to the same form on $X$ (see~\eqref{ex:tori}).

\subsubsection{Newton polytopes from constraints}

We now provide an alternative way of thinking about equation~\eqref{eq:deltapush} as \defn{constraints} on $C$ space. Most of the notation in this section is borrowed from Section~\ref{sec:contour1}.

The map $X\rightarrow C_i(X)$ parametrizes a subset of the projective space $C\in \P^{n{-}1}$. This subset can be equivalently ``cut out" using constraints on the $C_i$ variables.
\be\label{eq:orthogonal}
\prod_{i=1}^nC_i^{w_{pi}}=1
\ee
with one constraint for each value of $p=1,...,n{-}m{-}1$, where $w_{pi}$ is a constant \defn{integer} matrix to be determined. For the constraint to be well-defined projectively, we must have $\sum_i w_{pi}=0$ for each $p$.

Substituting the parametrization $C_i(X)=\prod_{a=0}^m X_a^{z_{ai}}$, we find that
\be
\sum_{i=1}^nz_{ai}w_{pi}=0
\ee
for each $a,p$. 

We can now replace the push-forward formula~\eqref{eq:deltapush} by an integral over all $C_i$ with imposed constraints.
\be\label{eq:constraints}
\aOmega(\A)=\frac{1}{m!}\int \;\frac{d^nC}{\prod_{j=1}^nC_j}\;\left[\prod_{p=1}^{n{-}m{-}1}\delta\left(1-\prod_{i=1}^n C_i^{w_{pi}}\right)\right]\;\delta^{m{+}1}\left(Y-\sum_{i=1}^nC_iZ_i\right)
\ee

Note that each delta function in the square brackets imposes one constraint, and in the absence of constraints we arrive at the familiar cyclic measure on $C$ space. 

It is evident that explicitly summing over the roots of these polynomial equations will become prohibitively more 
difficult for large $n$. But of course there is a standard approach to summing rational functions of roots of polynomial equations, making use of  
the {\sl global residue theorem}~\cite{griffithsharris} which we review in Appendix~\ref{app:GRT}. This method has been applied ubiquitously in the literature on scattering amplitudes, from the earliest understanding of the relations between leading singularities and the connection between the twistor-string and Grassmannian formalisms~\cite{Grassmannian,ArkaniHamed:2009sx,ArkaniHamed:2009dg}, to the recent application to scattering equations \cite{CHY1,CHY2,CHY3,CHY4}.
Appropriately taking care of some minor subtleties, the global residue theorem also works for our problem, naturally connecting the Newton polytope formula to triangulations of the polytope.

We now provide explicit computations for the quadrilateral and the pentagon with constraints.

\begin{example}
For $m=2$, $n=4$, there is only one constraint given by
\be
w_{p}\deff w_{1p}=(Q_1, -Q_2, Q_3,-Q_4)
\ee
for some integers $Q_i>0$, $i=1,2,3,4$. The positivity of the $Q_i$ variables is imposed by the positivity of the Newton polytope.

Now recall that delta functions can be identified with residues in the following manner:
\be
\int dz \;\delta(z)f(z)=\Res_{z\rightarrow 0} \left(\frac{1}{z}f(z)\right)=f(0)
\ee
We will therefore think of the delta function constraints appearing in the square brackets Eq.~\eqref{eq:constraints} as residues, and apply the global residue theorem~\ref{thm:GRT}.

We have
\be\label{eq:constraints4}
\aOmega(\A)=\frac{1}{2}\int\frac{d^4C}{C_1C_2C_3C_4}\left[\frac{C_2^{Q_2}C_4^{Q_4}}{\underline{C_2^{Q_2}C_4^{Q_4}-C_1^{Q_1}C_3^{Q_3}}}\right]\delta^3\left(Y-\sum_{i=1}^4C_iZ_i\right)
\ee
Here we have underlined the pole corresponding to the constraint. The reader should imagine that the three remaining constraints (i.e. $Y=C\cdot Z$) also appear as poles, but for notational convenience we will simply write them as delta functions. For higher $n$, we will have additional constraints, and hence additional underlined polynomials.

Now ~\eqref{eq:constraints4} instructs us to sum over all global residues where the four constraints vanish, which in general is very hard to compute by explicit summation. However, the global residue theorem vastly simplifies the problem. Note that there are eight poles in our integral (i.e. four constraints plus the four cyclic factors $C_i$). In order to apply the theorem, we will need to group the poles into four groups. Let us group the cyclic factors with the underlined constraint, and each of the delta function constraints forms its own group. It follows therefore that
\be
\aOmega(\A)=-\frac{1}{2}\sum_{i=1}^4\int \frac{d^4C}{C_1\ldots\underline{C_i}\ldots C_4}\left[\frac{C_2^{Q_2}C_4^{Q_4}}{C_2^{Q_2}C_4^{Q_4}-C_1^{Q_1}C_3^{Q_3}}\right]\delta^3\left(Y-\sum_{i=1}^4C_iZ_i\right)
\ee
This gives us four residues corresponding to $C_i\rightarrow 0$ for each $i$. Now the positivitiy of $Q_i>0$ comes into play. Clearly, the $C_2,C_4\rightarrow 0$ residues vanish due to the appearance of $C_2^{Q_2}C_4^{Q_4}$ in the numerator. However, for the $C_1,C_3\rightarrow 0$ residues, the square bracket becomes unity, and the result is equivalent to $\Res(1)$, $\Res(3)$, respectively, as defined in Eq.~\ref{eq:ResJ}. It follows that
\be
\aOmega(\A)=\Res(1)+\Res(3)=[234]+[124]
\ee
Alternatively, we could have assumed $Q_i<0$ for each $i$. 
This would have given us another triangulation:
\be
\aOmega(\A)=\Res(2)+\Res(4)=[134]+[123]
\ee

We note that our technique can be applied to convex cyclic Newton polytopes for any $m$, provided that $n=m{+}2$. The constraint matrix is given by a sequence with alternating signs.
\be
w_p=(Q_1,-Q_2,Q_3,-Q_4,...,({-}1)^{m{+}1}Q_{m{+}2})
\ee
where the constants $Q_i$ are either all positive or all negative. If $Q_i>0$, then the canonical rational function becomes
\be
\aOmega(\A)=\sum_{\text{$i$ odd}}\Res(i)
\ee
And for $Q_i<0$, we get
\be
\aOmega(\A)=\sum_{\text{$i$ even}}\Res(i)
\ee
\end{example}

\begin{example}
Let us move on to the $n=5,m=2$ pentagon $\A$ with the Newton polytope having vertices $(0,0),(1,0),(2,1),(1,2),(0,1)$. Then the Newton polytope formula becomes
\begin{equation}
\aOmega(\A) = \int dC_1 dC_2 dC_3 dC_4 dC_5 \frac{1}{C_1 C_2 C_5} \underline{\frac{C_1^2}{C_1^2 C_3 - C_2^2 C_5}}\; \underline{\frac{C_1^2}{C_1^2 C_4 - C_2 C_5^2}} \delta^3(Y - C\cdot Z) 
\end{equation}
Now on the support of the delta function we have a two-dimensional integral, and the underline tells us to sum over all the roots of the two polynomials with $C_1 \neq 0$.  


In order to use the global residue theorem for our pentagon problem, we group the denominator factors into five groups, given by the three delta function constraints and $f_1,f_2$:
\begin{equation}
f_1 = C_2 C_5 (C_1^2 C_3 - C_2^2 C_5), \;\; f_2 = (C_1^2 C_4 - C_2 C_5^2), \;\; g = C_1^3
\end{equation}
where $g$ is the numerator.

The roots of $f_1=0,f_2=0$ certainly include the ``complicated" solutions of the Newton polytope problem, where $C_1 \neq 0$, and at these the residues are well-defined. The only subtlety is that we have other roots, where $(C_2,C_1)=(0,0)$, and where
$(C_5,C_1)=(0,0)$. These zeros are quite degenerate; the Jacobian factor vanishes and the residue is not directly well-defined. 

We can deal with this by slightly deforming the functions; we will take instead 
\begin{equation}
f_1 = C_2 C_5 ((C_1^2 - \epsilon_3^2) C_3 - C_2^2 C_5), \;\; f_2 = ((C_1^2 - \epsilon_4^2) C_4 - C_2 C_5^2), \;\; g = C_1^3
\end{equation}
where we will imagine that both $\epsilon_3,\epsilon_4$ are eventually sent to zero. 
The singular zeros we previously found, with $(C_2,C_1)=(0,0)$ and with $(C_5,C_1)=(0,0)$ will now be split into a number of non-degenerate solutions, and we will be able to use the global residue theorem. 

Let us see where these deformed zeros are located. For $f_1=0$, we can set either $C_2=0,C_5=0$ or $(C_1^2 - \epsilon_3^2) C_3 - C_2^2 C_5=0$. The first two cases are of course trivial; the roots and corresponding residues, as we take $\epsilon_3,\epsilon_4 \to 0$, are 
\begin{eqnarray}
C_2= 0, & C_4 = 0, & {\rm Res} = [135] \\
C_5=0, & C_4 = 0, & {\rm Res} = [123] \\
C_2 = 0, & C_1 = \pm \epsilon_4, & \sum_\pm {\rm Res} = \frac{\epsilon_4^2}{\epsilon_4^2 - \epsilon_3^2} [345] \\
C_5 = 0, & C_1 = \pm \epsilon_4, & \sum_\pm {\rm Res} = \frac{\epsilon_4^2}{\epsilon_4^2 - \epsilon_3^2} [234]
\end{eqnarray}
Note interestingly that the residues for the last two cases depend on the ratio $\epsilon_3/\epsilon_4$ even as $\epsilon_{3,4} \to 0$. 
We next have to look at the solutions to 
\begin{equation}
(C_1^2 - \epsilon_3^2) C_3 - C_2^2 C_5 = 0, \, (C_1^2 - \epsilon_4^2) C_4 - C_2 C_5^2=0
\end{equation}
We are interested in the solutions where $C_1$ is close to zero; which here means that either $C_1$ is close to $\epsilon_3$ or $C_1$ is close to $\epsilon_4$. More formally, we can set $\epsilon_{3,4} = \epsilon E_{3,4}$ and ask what the solutions look like as $\epsilon \to 0$ with $E_{3,4}$ fixed.  It is then easy to see that if $C_1^2 \to \epsilon_3^2$, we must have that $C_5$ is non-zero and $C_2 \to (\epsilon_3^2 - \epsilon_4^2) C_4/C_5^2$, while if $C_1^2 \to \epsilon_4^2$, we must have that $C_2$ is non-zero and $C_5 \to (\epsilon_4^2 - \epsilon_3^2) C_3/C_2^2$. It is trivial to compute the residues in these cases, and we find
\begin{eqnarray}
C_1 \to \pm \epsilon_3, & C_2 \to  (\epsilon_3^2 - \epsilon_4^2) C_4/C_5^2, & \sum_\pm {\rm Res} = \frac{\epsilon_3^2}{\epsilon_4^2 - \epsilon_3^2} [345] \\ 
C_1 \to \pm \epsilon_4, & C_5 \to (\epsilon_4^2 - \epsilon_3^2) C_3/C_2^2, & \sum_\pm {\rm Res} = \frac{\epsilon_4^2}{\epsilon_3^2 - \epsilon_4^2} [234]
\end{eqnarray} 
By the global residue theorem, the sum over all these residues gives us $\aOmega_(\A)$, so we find 
\begin{eqnarray}
\aOmega(\A) &=& [123] + [135] + [345] \left(\frac{\epsilon_4^2}{\epsilon_4^2 - \epsilon_3^2} - \frac{\epsilon_3^2}{\epsilon_4^2 - \epsilon_3^2} \right) + [234] \left(\frac{\epsilon_4^2}{\epsilon_4^2 - \epsilon_3^2} - \frac{\epsilon_4^2}{\epsilon_4^2 - \epsilon_3^2} \right) \\ & = & [123] + [135] + [345] 
\end{eqnarray}
which is a standard triangulation of the pentagon. 
 
It would be interesting to carry out the analog of this analysis for the general Newton polytope expression associated with the canonical form of any polytope. The same resolution of the the singular roots will clearly be needed, and it will be interesting to see how and which natural class of triangulations of the polytope emerges in this way.  

\end{example}

\subsubsection{Generalized polytopes on the projective plane}

We now consider push-forwards onto generalized polytopes on the projective plane, particularly the pizza slice.

\begin{example} 
We construct the canonical form of the pizza slice $\mathcal{T}(\theta_1,\theta_2)$ from Example \ref{ex:pizza} via push-forward.
Let $z_1 = \tan(\theta_1/2)$ and $z_2 = \tan(\theta_2/2)$.
Consider the morphism $\Phi:(\P^1 \times \P^1,[z_1,z_2]\times[0,\infty])\rightarrow \mathcal{T}(\theta_1,\theta_2)$ given by:
\be
(1,x,y)=\Phi(z,t) &\deff& \left (1,\frac{1-z^2}{(1+z^2)},\frac{2z}{(1+z^2)} \right)+t Z_* \\ &=& \left(1,\frac{1-z^2}{(1+z^2)(1+t)},\frac{2z}{(1+z^2)(1+t)}\right)
\ee
where $z,t$ are coordinates on the two $\P^1$'s, and $Z_*^I \deff (1,0,0)$ is the ``tip" of the pizza.  Note that the equations are projective.  As $z$ varies in $[z_1,z_2]$, the point $(\frac{1-z^2}{(1+z^2)},\frac{2z}{(1+z^2)})$ sweeps out the circular arc $(\cos(\theta),\sin(\theta))$ where $\theta$ varies from $\theta_1$ to $\theta_2$.  The variable $t$ acts like a radial coordinate that goes from $0$ (the unit arc) to $\infty$ (the tip).

For a generic point $(x,y)$ there are two roots in $\Phi^{-1}(1,x,y)$, say $(z,t)$ and $(-1/z,-t-2)$.  We compute
\be
&& \Phi_*\left(\frac{(z_2-z_1)}{(z_2-z)(z-z_1)} dz \frac{1}{t} dt\right) \\
&=& \frac{(1+t)^3}{2} \left(\frac{(1+z^2)(z_2-z_1)}{(z_2-z)(z-z_1)}\frac{1}{t} - \frac{(1+z^2)(z_2-z_1)}{(zz_2+1)(1+zz_1)}\frac{1}{(t+2)}\right) dx dy \label{eq:pushz}
\ee 
We have the two identities
\be \label{eq:id1}
\frac{1}{2}\frac{(1+z^2)(z_2-z_1)}{(z_2-z)(z-z_1)} &=& \frac{\sin(\theta_2-\theta_1)+\sin(\theta-\theta_1)+\sin(\theta_2-\theta)}{2\sin(\theta-\theta_1)\sin(\theta_2-\theta)} \\ \label{eq:id2}
-\frac{1}{2}\frac{(1+z^2)(z_2-z_1)}{(zz_2+1)(1+zz_1)}&=& \frac{\sin(\theta_2-\theta_1)-\sin(\theta-\theta_1)-\sin(\theta_2-\theta)}{2\sin(\theta-\theta_1)\sin(\theta_2-\theta)}
\ee
which are related by a shift $ \theta\rightarrow \theta+\pi$.
Substituting into \eqref{eq:pushz} and using $1-x^2-y^2=t(2+t)/(1+t)^2$ gives
\be \label{eq:pizzafinal}
\frac{\left[\sin(\theta_2-\theta_1)+(-x\sin\theta_1+y\cos\theta_1)+(x\sin\theta_2-y\cos\theta_2)\right]}{(1-x^2-y^2)(-x\sin\theta_1+y\cos\theta_1)(x\sin\theta_2-y\cos\theta_2)}dxdy
\ee
which of course is the canonical form $\Omega(\mathcal{T}(\theta_1,\theta_2))$. Note that sine summation formulas were used to get the denominator factors.

The push-forward calculation can also be done with more intuitive angular coordinates, using an infinite degree map similar to Example~\ref{ex:pushsegment}.  We take the map $\Psi:(\P^1 \times \P^1,[\theta_1,\theta_2]\times[0,\infty])\rightarrow \mathcal{T}(\theta_1,\theta_2)$ given by:
\be
(1,x,y)=\Psi(1,\theta,t)\deff (1,\cos\theta,\sin\theta)+t Z_*=\left(1,\frac{\cos\theta}{1+t},\frac{\sin\theta}{1+t}\right)
\ee The variable $\theta$ acts as an angular coordinate between $\theta_1$ and $\theta_2$.  For any point $(x,y)=(\cos\theta,\sin\theta)/(1{+}t)$, there are two sets of roots given by:
\be
(\theta_n,t_n)&= &(\theta+2\pi n, t)\\
(\theta_n',t_n')&=&(\theta+\pi (2n+1),-t{-}2)
\ee
where $n\in\mathbb{Z}$. Summing over all roots in the push-forward gives
\be\label{eq:pushpizza}
\Psi_*\left(\frac{(\theta_2-\theta_1)d\theta}{(\theta_2-\theta)(\theta-\theta_1)}\frac{dt}{t}\right)=\frac{(1+t)^3}{t}\sum_{n\in\mathbb{Z}}\frac{(\theta_2-\theta_1)}{(\theta_2-\theta_n)(\theta_n-\theta_1)}dx dy \\
+\frac{(1+t)^3}{t+2}\sum_{n\in\mathbb{Z}}\frac{(\theta_2-\theta_1)}{(\theta_2-\theta_n')(\theta_n'-\theta_1)}dx dy
\ee
Ignoring the rational $t$ factors, the summations give:
\be \label{eq:sum1}
\sum_{n\in\mathbb{Z}}\frac{(\theta_2-\theta_1)}{(\theta_2-\theta_n)(\theta_n-\theta_1)}&=&\frac{\sin(\theta_2-\theta_1)+\sin(\theta-\theta_1)+\sin(\theta_2-\theta)}{2\sin(\theta-\theta_1)\sin(\theta_2-\theta)}\\ \label{eq:sum2}
\sum_{n\in\mathbb{Z}}\frac{(\theta_2-\theta_1)}{(\theta_2-\theta_n')(\theta_n'-\theta_1)}&=&\frac{\sin(\theta_2-\theta_1)-\sin(\theta-\theta_1)-\sin(\theta_2-\theta)}{2\sin(\theta-\theta_1)\sin(\theta_2-\theta)}
\ee
which are exactly \eqref{eq:id1} and \eqref{eq:id2}.  Thus we again obtain \eqref{eq:pizzafinal} as the canonical form of the pizza slice.  The summation identities \eqref{eq:sum1} and \eqref{eq:sum2} can also be interpreted as the push-forward summation for the infinite degree map $\phi: \theta \mapsto z = \tan(\theta/2)$, and we have $\Psi = \Phi \circ (\phi \times \id)$, where $\id: t \mapsto t$ is the identity map on the $t$ coordinate.

%
\end{example}

\subsubsection{Amplituhedra}\label{sec:Amplituhedron}

In this section, we conjecture that the \defn{Amplituhedron form} (i.e. the canonical form of the Amplituhedron) can be formulated as a push-forward from the standard simplex. While this is nothing more than a direct application of our favorite Heuristic~\ref{heuristic}, we wish to emphasize the important open problem of constructing morphisms needed for the push-forward. In this section, we let $\A\deff\A(k,n,m; l^L)$.

\begin{conjecture}
Given a morphism $\Phi:\Delta^D\rightarrow \A(k,n,m;l^L)$ from positive coordinates to the Amplituhedron, the Amplituhedron form is given by the push-forward
\be
\Omega(\A(k,n,m;l^L))=\Phi_*\left(\prod_{a=1}^D\frac{dX_a}{X_a}\right)
\ee
where $D$ is the dimension of the Amplituhedron.
\end{conjecture}

Before providing explicit examples, we reformulate our conjecture in terms of the {\sl canonical rational function} $\aOmega(\A)(\mathcal{Y})$ of the Amplituhedron (see~\eqref{eq:ampspace} for the notation $\mathcal{Y}$), which was first mentioned in~\ref{eq:simplexRationalFunc} for the $k=1$ Amplituhedron, and extended to the loop Amplituhedron in Appendix~\ref{app:projform}. We will also refer to $\aOmega(\A)$ as the \defn{amplitude}, motivated by the physical interpretation discussed in Sec~\ref{sec:Amplituhedron}. Keep in mind, of course, that the physical interpretation only exists for the {\sl physical} Amplituhedron, and for $L>0$ the canonical rational function is strictly speaking the \defn{$L$-loop integrand} of the amplitude.

\begin{conjecture}
Given a morphism $\Phi:\Delta^D\rightarrow \A(k,n,m;l^L)$ from positive coordinates to the Amplituhedron, the amplitude is given by:
\be\label{pushamp}
\aOmega(\A)(\mathcal{Y}) = \int \frac{d^DX}{\prod_{a=1}^{D}X_a}\delta^D(\mathcal{Y}; \Phi(X))
\ee
for any point $\mathcal{Y}\in\mathcal{A}$. For computational purposes, $\mathcal{Y}$ and $\Phi(X)$ are represented as matrices modded out by a left group action, as discussed below~\eqref{eq:loopMatrix}.
\end{conjecture}

The delta function $\delta^D(\mathcal{Y},\mathcal{Y}')$ is the unique (up to an overall normalization) delta function on $G(k,k{+}m;l^L)$ that is invariant under the group action $\mathcal{G}(k;\k)$ (defined below~\ref{eq:loopMatrix}) on $\mathcal{Y}'$, and scales inversely as the measure~\eqref{eq:measure} under the same group action on $\mathcal{Y}$. The inverted scaling ensures that the canonical form is locally invariant under the action.


\begin{example}
For $k=1$ and $m=4$, the {\sl physical} tree Amplituhedron $\A$ is a convex cyclic polytope with vertices $Z_i\in\P^4(\R)$. Morphisms $\Phi: (\P^4,\Delta^4) \rightarrow (\P^4,\A)$ are given by convex cyclic Newton polytopes from \eqref{eq:Phi}, and the amplitude is given by:
\be
\aOmega(\A)(Y)=\int \frac{d^4X}{X_1X_2X_3X_4}\delta^4(Y;\Phi(X))
\ee
where $Y\in \P^4$ and the delta function is given in~\eqref{eq:delta4} with $m=4$.
\end{example}

We now provide some examples for $k=2,m=2,l=2,L=0$. While we do not have a complete construction of such morphisms for all $n$, we did find a few non-trivial examples for small $n$. We stress that the existence of such morphisms is an absolutely remarkable fact. We also speculate that a complete understanding of morphisms to the tree Amplituhedron may provide insight on extending toric varieties to the Grassmannian, as morphisms to the polytope were given by projective toric varieties.

Let us begin by clarifying the push-forward computation, which for the present case is given by:
\be\label{eq:pushMHV}
\aOmega(\A)=\int\frac{d^4X}{X_1X_2X_3X_4}\delta^{4}(Y;\Phi(X))
\ee
where
\be
\delta^4(Y;Y')=\frac{1}{4}\int \frac{d^{2\times 2}\rho}{(\det{\rho})^2}\delta^{2\times 4}(Y-\rho\cdot Y')
\ee
is the delta function that imposes the constraint $Y=Y'$ for any pair of $2\times 4$ matrices $Y,Y'$, where the tilde indicates that the two sides are only equivalent up to an overall $\GL(2)$ transformation. The $\rho$ is a $2\times 2$ matrix that acts as a $\GL(2)$ transformation on $Y'$ from the left, and the $d^{2\times 2}\rho/(\det \rho)^2$ measure is $\GL(2)$ invariant in $\rho$ both from the left and the right. We see therefore that the delta function has $\GL(2)$ weights $(-4, 0)$ in $(Y,Y')$, as expected.

\begin{example}\label{ex:24}
Consider the Amplituhedron for $(k,m,n,l,L)=(2,2,4,2,0)$, which is simply
\be
\A = \{C \cdot Z \in G(2,4)\mid C\in G_{\geq 0}(2,4)\}
\ee
where $Z$ is a $4\times 4$ matrix with positive determinant. Since $Z$ is non-singular, $\A$ is isomorphic to $G_{\geq 0}(2,4)$. A (degree-one) morphism from $(\P^4,\Delta^4)$ to the Amplituhedron is therefore:
\be
C(X) &=& \begin{pmatrix}
1 & X_1 & 0 & -X_4\\
0 & X_2 & 1 & X_3
\end{pmatrix}\\
\Phi(X)&=&C(X) \cdot Z
\ee
sending the interior $\Int(\Delta^4)$ (where $X_i>0$ for $i=1,2,3,4$) to the interior of the Amplituhedron.

Substituting into~\eqref{eq:pushMHV} gives:
\be
\aOmega(\A)=\frac{1}{4}\frac{\lb 1234\rb^2}{\lb Y 12\rb\lb Y 23\rb\lb Y 34\rb\lb Y 14\rb}
\ee
\end{example}

\begin{example}\label{ex:25}
We now move on to the first non-trivial example beyond the polytope. Consider the same case as Example~\ref{ex:24} but now with $n=5$. The Amplituhedron is
\be
\A = \{C \cdot Z\in G(2,4)\mid C\in G_{\geq 0}(2,5)\}
\ee
with all the maximal ordered minors of $Z$ positive. The interior of this Amplituhedron has no obvious diffeomorphisms with the interior $\Int(\Delta^4)$ of the $4$-simplex. But we have managed to stumble across a few lucky guesses, such as the following:
\be\label{eq:pushMHV5}
C(X) &=& \begin{pmatrix}
1 & X_1 & X_1 X_2 & X_2 & 0\\
0 & X_1 X_3 & X_1 X_2(X_3{+}X_3X_4) & X_2(X_3{+}X_3X_4{+}X_4) & 1
\end{pmatrix}\\
\Phi(X)&=&C(X)\cdot Z
\ee
for $X_1,X_2,X_3,X_4>0$. We have verified numerically that the push-forward given by~\eqref{eq:pushMHV5} gives the correct 5-point 1-loop integrand. A careful analytic argument can be provided to prove that the restriction $\Int(\Delta^4)\isom \Int(\A)$ is indeed a diffeomorphism.
\end{example}

\begin{example}
We also provide an example for $n=6$. The setup is the same as Example \ref{ex:25}, with the $C(X)$ matrix given by
\be
C(X) &=& \begin{pmatrix}
X_1^3X_2^2 & X_3{+}X_1^2X_2^2X_3 & X_1^5 X_2^2 X_3 X_4^2 & X_1^3 X_3 X_4^2 & X_1 & 0\\
-X_1^4 X_2^2 & -X_1 X_3 & 0 & X_1^2X_3X_4^2 & 1{+}X_1^2X_4^2 & X_1^5X_2^2X_4^2
\end{pmatrix}\;\;
\ee
The push-forward was verified numerically. We challenge the reader to prove that this map restricts to a diffeomorphism $\Int(\Delta^4)\isom\Int(\A)$.
\end{example}

\subsection{Integral representations}
We now present integral representations (e.g. volume integrals, contour integrals) of various canonical forms.

\subsubsection{Dual polytopes}
\label{sec:dualpolytopeform}

Consider a convex polytope $(\P^m, \A)$ and a point $Y$ not along any boundary component. We argue that the canonical rational function $\aOmega(\A)$ at $Y$ is given by the volume of the dual polytope $\A_Y^*$ (defined in Section \ref{sec:dual}) under a $Y$-dependent measure. More precisely,
\be \label{eq:volume}
\aOmega(\A)(Y)=\Vol(\A_Y^*) \deff \frac{1}{m!}\int_{W\in \A_Y^*} \frac{\lb W d^m W \rb}{(Y\cdot W)^{m+1}}
\ee
In order for this integral to be well-defined on projective space, the integrand must be invariant under \defn{local}  $\GL(1)$ transformations $W\rightarrow W'=\alpha(W)W$, which we proved in~\eqref{eq:local}. Moreover, observe that by construction of $\A_Y^*$, we have $Y\cdot W>0$ for every point $W\in\Int(\A_Y^*)$, which is important for the integral to converge. However, the overall sign of the integral is dependent on the orientation of the dual. 
We say therefore that the volume is \defn{signed}.

Now let us prove this claim for simplices by explicit computation:

\begin{example}
Let $Y\in\Int(\Delta)$ for some simplex. The volume of the dual simplex $\Delta_Y^*$ with vertices $W_1,\ldots,W_{m{+}1}$ (so that $Y\cdot W_i>0$ for each $i$) can be computed using a Feynman parameter technique. Since the form is locally $\GL(1)$ invariant, we can gauge fix the integration to $W=W_1+\alpha_1W_2+\cdots+\alpha_m W_{m{+}1}$ and integrate over $\alpha_i>0$. This gives
\be\label{eq:simplexVol}
\Vol(\Delta^*_Y)&=&\int \frac{d^m\alpha\lb W_1 \cdots W_{m{+}1}\rb}{((Y \cdot W_1)+\alpha_1(Y\cdot W_2)+ \cdots +\alpha_{m}(Y \cdot W_{m{+}1}))^{m+1}}\\
&=& \frac{\lb W_1 \cdots W_{m{+}1}\rb}{m!(Y\cdot W_1) \cdots (Y \cdot W_{m{+}1})}
\ee
agreeing with \eqref{eq:simplex}.

Note that the result is independent of the sign of the vertices $W_i$. If, for instance, we move $Y$ outside $\Delta$ so that $Y\cdot W_1<0$ but $Y\cdot W_i> 0$ for every $i\neq 1$, then we would have needed to use $-W_1$ in the integration, but the result would still have the same form.
\end{example}

We now argue that the formula holds for an arbitrary convex polytope based on three observations:
\begin{itemize}
\item

``Dualization of polytopes commutes with triangulation"~\eqref{eq:commute}. This means that
\be \label{eq:filliman}
\A=\sum_i\A_{i}\;\;\;\;\Rightarrow\;\;\;\; \A_Y^*=\sum_i\A_{iY}^*
\ee
For triangulation by simplices, this formula is known as {\it Filliman duality} \cite{Fil}.

\item

The signed nature of the volume is crucial, because it implies \defn{triangulation independence}. This means that
\be
{\A_Y^*=\sum_i\A_{iY}^*}\;\;\;\;\Rightarrow\;\;\;\; {\Vol(\A_Y^*)=\sum_i\Vol(\A_{iY}^*)}
\ee

\item
The canonical rational function is triangulation independent (Section \ref{sec:triangulations}):
\be
\A=\sum_i\A_i\;\;\;\;\Rightarrow\;\;\;\; \aOmega(\A)=\sum_i\aOmega(\A_i)
\ee
\end{itemize}

Combining these three statements for a {\sl signed} triangulation by simplices $\A_i =\Delta_i$, it follows that
\be\label{eq:dualArgument}
\Vol(\A_Y^*)=\sum_i\Vol(\Delta_{iY}^*)=\sum_i\aOmega(\Delta_i)=\aOmega(\A)
\ee
which is the desired claim.

We have not given a proof of the first observation \eqref{eq:filliman}.  In Section~\ref{sec:laplace} we will give an independent proof of the volume formula~\eqref{eq:volume}.  Then by the second and third assumptions, we find that $\Vol(\A_Y^*)=\sum_i\Vol(\A_{iY}^*)$ for every $Y$, which implies that $\A_Y^*=\sum_i\A_{i}^*$, thus deriving \eqref{eq:filliman}.  Alternatively, we can also say that the the volume formula combined with the first two assumptions implies the triangulation independence of the canonical rational function.

We now wish to comment on the distinction between ``physical" poles and ``spurious" poles (see Section~\ref{sec:poles}) in the context of the volume formula, first pointed out by Andrew Hodges as an interpretation of spurious poles appearing in BCFW recursion of NMHV tree amplitudes~\cite{hodges}. 

Given a triangulation of $\A$ by polytopes $\A_i$ with mutually non-overlapping interiors, spurious poles appear along boundary components of the triangulating pieces that do not belong to the boundary components of $\A$. From the dual point of view, the duals $\A_{iY}^*$ form an {\sl overlapping} triangulation of $\A_Y^*$, and a spurious pole corresponds to a subset of volume terms $\Vol(\A_{iY}^*)$ going to infinity individually, but whose sum remains finite. Geometrically, the signed volumes overlap and cancel. More details are given in~\cite{hodges}.

However, suppose instead that $Y\in\Int(\A)$ and the duals $\A_{iY}^*$ have non-overlapping interiors, then $\Vol(\A_{iY}^*)\le \Vol(\A_Y^*)$. It follows that all the volume terms are finite on $\Int(\A)$, and there are no spurious poles. Of course, the sum is independent of $X$ since the integral is surface-independent. This is called a \defn{local triangulation}. An example of a local triangulation of cyclic polytopes in $\P^4(\R)$ is given in~\cite{ArkaniHamed:2010gg}.

On a final note, we argue that these simplicial volumes are identical to Feynman parameters~\cite{peskin} appearing in the computation of loop scattering amplitudes. 
A general Feynman parameter formula takes the following form.
\be
\frac{1}{\prod_{i=0}^{m}A_i}=m!\int_{x\in I^{m{+}1}} d^{m{+}1}x\frac{\delta(1-\sum_{i=0}^{m}x_i)}{\left(\sum_{i=0}^{m}x_i A_i\right)^{m{+}1}}
\ee
where $I^{m{+}1}$ is the unit cube given by $0<x_i<1$ for all $i$. Integrating over $0<x_{0}<1$ yields
\be
\frac{1}{\prod_{i=0}^{m}A_i}=m!\int_{\substack{x\in I^m\\0<\sum_{i=1}^{m}x_i<1}}d^mx\frac{1}{\left(A_{0}+\sum_{i=1}^mx_i(A_i-A_{0})\right)^{m{+}1}}
\ee

Now change variables $x_i\rightarrow \alpha_i$ so that
\be
x_i=\frac{\alpha_i}{1+\sum_{j=1}^m\alpha_i}
\ee
and $\alpha_i>0$ for all $i$ is the equivalent region of integration. The Jacobian for the change of measure is
\be
d^mx = \frac{d^m\alpha}{(1+\sum_{i=1}^m\alpha_i)^{m{+}1}}
\ee
Putting everything together, we get
\be
\frac{1}{m!}\frac{1}{\prod_{i=0}^{m}A_i}=\int_{\alpha_i>0}d^m\alpha\frac{1}{(A_{0}+\sum_{i=1}^m\alpha_iA_i)^{m{+}1}}
\ee
We see that the right hand side is very reminiscent of the volume formula~\eqref{eq:volume} for a (dual) simplex. 
The lesson here is that loop integrals can be reinterpreted as polytope volumes, and Feynman parameters are coordinates on the interior of the polytope over which we integrate. 

\subsubsection{Laplace transforms}
\label{sec:laplace}
For a convex projective polytope $\A \subset \P^m(\R)$, we let $\bar \A \subset \R^{m+1}$ be the cone over $\A$, so that $\bar \A \deff \{Y\in \mathbb{R}^{m+1}\mid Y \in \A\}$.  Similarly, one has the cone of the dual $\bar \A_Y^* \subset \R^{m+1}$ in the dual vector space. For simplicity, we will only consider the dual for which $Y\in \Int(\A)$, in which case $\A^*\deff\A_Y^*$.

We argue that the canonical rational function at any point $Y\in\Int(\A)$ is given by a Laplace transform over the cone of the dual:
\be\label{eq:laplace}
 \aOmega(\bar\A)(Y) = \frac{1}{m!}\left(\int_{W \in \bar \A_Y^*} e^{-W \cdot Y} d^{m+1}W \right).
\ee
Such integral formulae have been considered in~\cite{Brion-Vergne}. Extensions to the Grassmannian are discussed in~\cite{ferroVolume}

We now show the equivalence of~\eqref{eq:laplace} with the dual volume~\eqref{eq:volume}. We begin with the Laplace transform, and change variables by writing $W=\rho \widehat{W}$ so that $\widehat{W}$ is a unit vector and $\rho>0$ is the Euclidean norm of $W$. It follows that
\be
d^{m{+}1}W=\frac{1}{m!}\rho^m d\rho\lb \widehat{W}d^m\widehat{W}\rb
\ee
The Laplace transform rational function becomes
\be
\aOmega(\A)=\frac{1}{m!}\int_{\rho>0,\widehat{W}\in\A^*} e^{-\rho\widehat{W}\cdot Y}\frac{\rho^m d\rho}{m!} \lb\widehat{W}d^m\widehat{W}\rb =\frac{1}{m!}\int_{\widehat{W}\in \A^*}\frac{\lb \widehat{W}d^m\widehat{W}\rb}{(\widehat{W}\cdot Y)^{m{+}1}}.
\ee
This result is identical to~\eqref{eq:volume} but ``gauge-fixed" so that $W=\widehat{W}$ has unit norm. Removing the gauge choice and rewriting $\widehat{W}$ as $W$ recovers~\eqref{eq:volume}.

In principle, this completes the argument, but let us do a few important examples.
\begin{example}
Suppose $\Delta$ is a simplex, so $\Delta^*$ is also a simplex, generated by facets $W_1,\ldots,W_{m+1} \in \R^{m+1}$.  Then setting $W = \alpha_0 W_1 + \cdots \alpha_m W_{m+1}$, we compute
\begin{align}
 \aOmega(\bar\Delta)(Y) & =\frac{1}{m!} \left(\int_{W \in \bar \A^*} e^{-W \cdot Y} d^{m+1}W \right)  \\
& = \frac{1}{m!}\langle W_1 W_2 \cdots W_{m+1}\rangle\left(\int_{\R_{>0}^{m+1}} e^{-(\alpha_0 (W_1 \cdot Y) + \alpha_1 (W_2 \cdot Y)  + \cdots + \alpha_{m} (W_{m+1} \cdot Y) )} d^{m{+}1}\alpha \right)  \\
&= \frac{\langle W_1 W_2 \cdots W_{m+1}\rangle}{ m!(Y \cdot W_1)  (Y \cdot W_2) \cdots  (Y\cdot W_{m+1})},
\end{align}
in agreement with \eqref{eq:simplex}.
\end{example}

We recognize the result simply as the volume $\Vol(\Delta^*)$ (note the implicit dependence on $Y$). For a general convex polytope $\A$, we may triangulate its dual $\A^*=\sum_i\Delta_i^*$ by (dual) simplices. Then the cone of the dual is also triangulated by the cones $\overline{\Delta_i}^*$ of the dual simplices. For simplicity let us assume that these dual cones are non-overlapping except possibly at mutual boundaries. We can therefore triangulate the Laplace transform integration and obtain:
\be
\aOmega(\A)(Y)=\sum_i\Vol(\Delta_{i}^*)
\ee
which of course is equivalent to $\Vol(\A^*)$ as expected.

Let us now try a sample computation for a non-simplex polytope.
\begin{example}\label{ex:squarelaplace}
Let $\A \subset \P^2$ be the polytope with vertices:
$$\{(1,0,0),(0,1,0),(0,0,1),(1,1,-1)\}$$
Then the dual cone is
\be
\bar \A^* = \{W = (W_0,W_1,W_2) \in \R_{\geq 0}^3 \mid W_0+ W_1 \geq W_2\}.
\ee
We compute 
\begin{align}
\aOmega(\A) &= \frac{1}{2!}\int_{W \in \bar \A^*} e^{-W \cdot Y} d^{m+1}W  \\
&=\frac{1}{2!} \int_{0}^\infty e^{-W_0Y^0} \int_0^\infty e^{-W_1Y^1} \int_{0}^{W_0 + W_1} e^{-W_2Y^2} dW_2 dW_1 dW_0 \\
&=\frac{1}{2!}\frac{Y^0 +Y^1 + Y^2}{(Y^0+Y^2) (Y^1+Y^2) Y^0 Y^1}.
\end{align}
To verify our calculation, we observe that
\be \label{eq:squaretriangulation}
\frac{1}{2!}\frac{Y^0 +Y^1 + Y^2}{(Y^0+Y^2) (Y^1+Y^2) Y^0 Y^1} = \frac{1}{2!}\left(\frac{1}{Y^0Y^1Y^2} - \frac{1}{(Y^0+Y^2)(Y^1+Y^2)Y^2}\right)
\ee
corresponding to the triangulation $\A = \Delta_1 \cup \Delta_2$ where $\Delta_1$ has three vertices given by $\{(1,0,0),(0,1,0),(0,0,1)\}$ and $\Delta_2$ has vertices $\{(1,0,0),(1,1,-1),(0,1,0)\}$.
\end{example}
%

Let us now argue that the Laplace transform makes manifest all the poles and residues of the canonical form, and thus satisfies the recursive properties of the form. Suppose $\{Y\in\partial\A \mid W_0 \cdot Y = 0\}$ defines one of the facets of the polytope $\A$, which we denote by $\B$, and hence $W_0$ is one of the vertices of $\A^*$. We now show that \eqref{eq:laplace} has a first order pole along the hyperplane $W_0 \cdot Y= 0$ and identify its residue.  Consider the expansion
\be\label{eq:W0}
W=\alpha W_0+\bar{W}
\ee
where $\alpha$ is a scalar and $\bar{W}\in \gamma(W_0)$, where $\gamma(W_0)$ denotes the union of all the facets of the cone $\bar{\A}^*$ \defn{not} adjacent to $W_0$. It is 
straightforward to prove that every point $W$ in the interior of the integration region $\bar{\A^*}$ has a unique expression in the form \eqref{eq:W0} with $(\alpha,\bar{W})\in \mathbb{R}_{>0}\times \gamma(W_0)$. 
Furthermore, some simple algebra shows that $d^{m{+}1}W=d\alpha \lb W_0 d^m\bar{W}\rb/m!$. It follows that
\be
\aOmega(\A)=\frac{1}{m!}\int_0^\infty d\alpha\;e^{-\alpha(W_0\cdot Y)}\int_{\bar{W}\in\gamma(W_0)}\;e^{-\bar{W}\cdot Y}\frac{\lb W_0 d^m\bar{W}\rb}{m!}
\ee
The $\alpha>0$ integral can be trivially computed to reveal the first order pole.
\be
\aOmega(\A)=\frac{1}{m!}\frac{1}{(W_0\cdot Y)}\int_{\bar{W}\in\gamma(W_0)}\;e^{-\bar{W}\cdot Y}\frac{\lb W_0 d^m\bar{W}\rb}{m!}
\ee

If we now take a residue at $(W_0 \cdot Y)\rightarrow 0$, we see that the remaining integral is invariant under \defn{local} shifts $\bar{W}\rightarrow\bar{W}+\beta(\bar{W}) W_0$ along the direction of $W_0$, where $\beta(\bar{W})$ is a scalar dependent on $\bar{W}$. So the integral is performed within the quotient space $\P^m/W_0$, which of course is just the dual of the co-dimension one variety $\{Y \mid Y\cdot W_0=0\}$ containing $\B$. We say that the integral has been \defn{projected} through $W_0$.

We can change the region of integration to $\bar{\B}^*$ which we define to be the cone over the union of all the facets of $\A^*$ adjacent to $W_0$. The change is possible because the surface $\gamma(W_0)\cup\bar{\B}^*$ is closed under the projection. The choice of notation $\bar{\B}^*$ suggests that the region can be interpreted as the cone of the dual of $\B$. Indeed, the facets of $\A^*$ adjacent to $W_0$ correspond to the vertices of $\B$.

It follows that
\be
\lim_{(W_0 \cdot Y) \to 0}\aOmega(\A)(W_0\cdot Y) = \frac{1}{m!}\left(\int_{\bar{W} \in \B^*}\;e^{\bar{W}\cdot Y} d^m\bar{W}\right)
\ee
where $d^m\bar{W}$ is defined to be the measure $\lb W_0 d^m\bar{W}\rb/m!$ on $\P^m/W_0$. We can interpret the residue as a Laplace transform for $\B$, as expected from the recursive nature of the residue.

\subsubsection{Dual Amplituhedra}
\label{sec:dualAmp}
As argued in Section \ref{sec:triangulations}, a signed triangulation of a positive geometry $\A=\sum_i\A_i$ implies $\Omega(\A) = \sum_i \Omega(\A_i)$. This was argued based on the observation that spurious poles cancel among the triangulating terms, thus leaving only the physical poles.

In our discussion of the dual polytope in Section~\ref{sec:dualpolytopeform}, we argued that the volume formulation of the canonical rational function (under certain assumptions) provides an alternative understanding of triangulation independence. We therefore wish to extend these assumptions to the Amplituhedron and conjecture a volume interpretation of the canonical rational function. We restrict our conjecture to even $m$, since the volume formula implies positive convexity (see Section~\ref{sec:convex}), which in most cases does not hold for odd $m$.

Let $\A$ denote an Amplituhedron, and let $X$ denote the irreducible variety in which it is embedded. We conjecture the following:

\begin{conjecture}
\begin{enumerate} 
\item
For each $\mathcal{Y}\in X$ not on a boundary component of $\A$, there exists an irreducible variety $X_{\mathcal{Y}}^*$, called the \defn{dual variety at $\mathcal{Y}$}, with a bijection $(X,\mathcal{B})\isom(X_{\mathcal{Y}}^*,\mathcal{B}_{\mathcal{Y}}^*)$ that maps positive geometries in $X$ to positive geometries in $X_{\mathcal{Y}}^*$. In particular, there exists a \defn{dual Amplituhedron} $\A_{\mathcal{Y}}^*$ for each $\mathcal{Y}$.

\item
Let $\mathcal{B}_{1,2}$ be positive geometries in $X$. Then $\mathcal{B}_1\subset\mathcal{B}_2$ if and only if $\mathcal{B}_{2\mathcal{Y}}^*\subset \mathcal{B}^*_{1\mathcal{Y}}$.

\item
Given a triangulation of $\B$ in $X$, we have
\be
\mathcal{B}=\sum_i\mathcal{B}_i\;\;\;\;
\Rightarrow\;\;\;\;
\mathcal{B}_{\mathcal{Y}}^*=\sum_i\mathcal{B}_{i\mathcal{Y}}^*.
\ee
\item
There exists a $\mathcal{Y}$-dependent measure $d\Vol$ on $X_{\mathcal{Y}}^*(\R)$ so that:
\be
\aOmega(\mathcal{B})(\mathcal{Y}) = \Vol(\mathcal{B}_{\mathcal{Y}}^*)\deff \int_{\mathcal{B}_{\mathcal{Y}}^*}d\Vol
\ee
\end{enumerate}
\end{conjecture}

Following the reasoning outlined around~\eqref{eq:dualArgument}, we arrive at the desired conclusion:
\be
\A=\sum_i\A_i\;\;\;\;&&\Rightarrow \;\;\;\;
\A_{\mathcal{Y}}^*=\sum_i\A_{i\mathcal{Y}}^*\;\;\;\;\Rightarrow\\
\Vol(\A_{\mathcal{Y}}^*)=\sum_i\Vol(\A_{i\mathcal{Y}}^*)\;\;\;\;&&\Rightarrow
\;\;\;\;\aOmega(\A)=\sum_i\aOmega(\A_i)
\ee
In words, we say that the triangulation independence of the dual volume explains the triangulation independence of the form. This provides an alternative argument to cancellation of spurious poles.



\subsubsection{Dual Grassmannian volumes}

We now discuss a first attempt at constructing a dual of the (tree) Amplituhedron which extends the dual polytope construction in Section~\ref{sec:dualpolytopeform}, but does not necessarily follow the format of Section~\ref{sec:dualAmp}. 

Let us begin by thinking about projective space ($k=1$ and any $m$) and the most obvious form associated with the space of points $(Y,W)$ where $Y\in\P^m(\R)$ and $W\in\P^m(\R)$ is in the dual:
\begin{equation}
d\Vol_{k=1,m}(Y,W) \deff \frac{d^{m+1} Y}{{\rm Vol\;GL(1)}} \frac{d^{m+1} W}{{\rm Vol\;GL(1)}} \frac{1}{(Y\cdot W)^{m+1}}
\end{equation}
Completely equivalently but slightly more naturally, we can think of the $W_I$ as being $m$-planes $W^{I_1\cdots I_m}$ in $(m{+}1)$ dimensions (Note that our dimension counting is Euclidean.), with $d\Vol_{k=1,m}$ taking essentially the same form:
\begin{equation}
d\Vol_{k=1,m}(Y,W) = \frac{d^{1\times (m+1)} Y}{{\rm Vol\;GL(1)}} \frac{d^{m \times (m+1)} W}{{\rm Vol\;GL(m)}} \frac{1}{\langle Y W \rangle^{m+1}}
\end{equation}
Thus we are taking the full $(m{+}1)$-dimensional space and looking at the space of 1-planes $Y$ and $m$-planes $W$ in it, with the only natural form possible on $(Y,W)$ space. 

Now let $\A$ be a convex polytope in $Y$ space as usual. We would like to integrate the form $d\Vol_{k=1,m}$ over some natural region in $W$ space. An obvious way of motivating the integration region for $W$ is to ensure the form never has singularities when $Y$ is in the polytope's interior. Since we can put $Y = \sum_{i=1}^n C_i Z_i$ with $C_i>0$,  we have $\langle Y W \rangle = \sum_{i=1}^n C_i \langle i W \rangle$, so it is natural to say that the region in $W$ space is defined by $\langle i W \rangle>0$. Indeed, since we know the vertices of the polytope are $Z_i$, and these are the facets of the dual polytope, this gives us the definiton of the dual polytope in the space of hyper-planes. And as we have indeed seen, this leads correctly to the canonical form associated with the polytope:
\begin{equation}
\Omega({\cal A})(Y) = \int_{\langle iW \rangle \geq 0} d\Vol_{k=1,m}(Y,W)
\end{equation}
where for simplicity we have absorbed any combinatorial factor into the measure.

All of this has an obvious extension to the Amplituhedron for general $k$, and even $m$ where we expect the geometry to be positively convex (see Section~\ref{sec:convex}). In the full $(k+m)$-dimensional space where $Z_i$ lives, we consider $k$-planes $Y$ and $m-$planes $W$, and the form on $(Y,W)$ space:
\begin{equation}
d\Vol_{k,m}(Y,W) = \frac{d^{k \times (m+k)} Y}{{\rm \Vol GL(k)}} \frac{d^{m \times (m+k)} W}{{\rm\Vol GL(m)}} \frac{1}{\langle Y W \rangle^{m+k}}
\end{equation}

We conjecture that for $Y\in\Int(\A)$, the canonical form $\Omega(\A)$ is given by an integral over a ``dual Amplituhedron" in exactly the same way. As with the polytope example, we want to ensure this form has no singularities when $Y$ is in the Amplituhedron, and the most obvious way of ensuring this is to define the dual Amplituhedron region by $\langle {i_1} \cdots {i_k} W\rangle \geq 0$,  for $1\leq 1_1< \cdots, < i_k\leq n$. Thus it is natural to conjecture
\begin{equation}
\Omega({\cal A})(Y) = \int_{\langle {i_1} \cdots {i_k} W\rangle\geq 0} d\Vol_{k,m}(Y,W) \;\;\;\;\;\; \text{(incorrect for $k>1$)}
\end{equation}
Where, again, the measure is appropriately normalized. This conjecture turns out to be incorrect for $k>1$, but let us do a simple example so we can discuss the subtleties.

Consider the first non-trivial case where $m=2,k=2,n=4$, so $Y$ and $W$ each corresponds to a 2-plane in 4 dimensions. We can put the $Z$ matrix to be identity. Then the inequalities defining the dual Amplituhedron are $\langle W 1 2 \rangle \geq 0, \langle W 23 \rangle \geq 0,\langle W 34 \rangle \geq 0, \langle W 14 \rangle \geq 0, \langle W 1 3 \rangle \geq 0, \langle W 2 4 \rangle \geq 0$. 
Note crucially that $\langle W 1 3 \rangle, \langle W 2 4 \rangle$ are both {\it positive} (in the interior), whereas for $Y$ in the Amplituhedron interior, $\langle Y 1 3 \rangle, \langle Y 2 4 \rangle $ are both {\it negative}. 
This region in $W$ space can be parametrized by $W^{IJ}=(W_1 W_2)^{IJ}$ with 
\begin{equation}
W =  \left(\begin{array}{cccc} 1 &-\alpha_1& 0& \alpha_4 \\ 0 &\alpha_2& -1& \alpha_3 \end{array}\right)
\end{equation}
where $\alpha_{1,2,3,4}>0$. We can also conveniently put $Y^{IJ} = (Y_1 Y_2)^{IJ}$ in the Amplituhedron by setting 
\begin{equation}
Y = \left(\begin{array}{cccc} 1 & \beta_1 & 0 & -\beta_4 \\0 & \beta_2 & 1 & \beta_3 \end{array} \right)
\end{equation}
for $\beta_{1,2,3,4}>0$. Then $\langle Y W \rangle = (\alpha_2+\beta_2)(\alpha_4+\beta_4) + (\alpha_1+\beta_1)(\alpha_3+\beta_3)$ and we must integrate 
\begin{equation}
\int_{\alpha\in\R_{>0}^4} d^4\alpha \frac{3}{((\alpha_2+\beta_2)(\alpha_4+\beta_4) + (\alpha_1+\beta_1)(\alpha_3+\beta_3))^4}
\end{equation}
The integrals can be performed easily, leading to 
\be
\frac{1}{4}
\left[- \frac{1}{\beta_1\beta_2\beta_3\beta_4} + \frac{1}{\beta_1^2 \beta_3^2} {\rm log} \left(1 + \frac{\beta_1 \beta_3}{\beta_2 \beta_4}\right) + \frac{1}{\beta_4^2 \beta_2^2} {\rm log} \left(1 + \frac{\beta_4 \beta_2}{\beta_1 \beta_3} \right)
\right]\ee
which can be written more invariantly as
\begin{eqnarray}
 &&- \frac{\langle 1 2 3 4 \rangle^2}{4\langle Y 12 \rangle \langle Y 23 \rangle \langle Y 34 \rangle \langle Y 14 \rangle}+ \\ 
 \frac{\langle 1 2 3 4\rangle^2}
{4\langle Y 2 3 \rangle^2 \langle Y 1 4 \rangle^2} &&{\rm log}\left(1 + \frac{\langle Y 2 3\rangle \langle Y 1 4 \rangle}{\langle Y 1 2 \rangle \langle Y 3 4 \rangle}\right)+
 \frac{\langle 1 2 3 4 \rangle^2}
{4\langle Y 1 2 \rangle^2 \langle Y 3 4 \rangle^2} {\rm log}\left(1 + \frac{\langle Y 1 2\rangle \langle Y 3 4 \rangle}{\langle Y 2 3 \rangle \langle Y 1 4 \rangle}\right)\nonumber
\end{eqnarray}

Here we encounter a surprise.  The first term looks like the expected canonical form associated with the $k=2,m=2,n=4$ Amplituhedron, but with the wrong sign.  And we have extra logarithmic corrections, so the form is not rational! In fact it is easy to check that this form has the interesting property of having {\it none} of the usual poles expected of the canonical form! For instance we expect poles when $\langle Y 12 \rangle \to 0$; but in this limit, the pole from the first rational term precisely cancels what we get from expanding the logs. However it is certainly amusing that the expected canonical form can be identified (even if with the wrong sign!) as the ``rational part" of this expression. We plan to return to these questions in~\cite{dual}.

\subsubsection{Wilson loops and surfaces}
\label{sec:wilson}

Continuing our search for the ``dual Amplituhedron", let us now consider another approach.  We wish to construct a simple ``dual integral" formula (not necessarily a volume) yielding the canonical rational function for the $m=2$ tree Amplituhedron.

\begin{figure}
\centering

\subfloat[
The standard measure on the Grassmannian $G(k,k{+}m)$ is formed by connecting the arrows on the left with those on the right in the most natural way. Namely, for each blue node, the $k$ outgoing arrows should be connected with each of the $k$ orange nodes.\newline]
{\label{fig:tensor_a}\includegraphics[width=7cm]{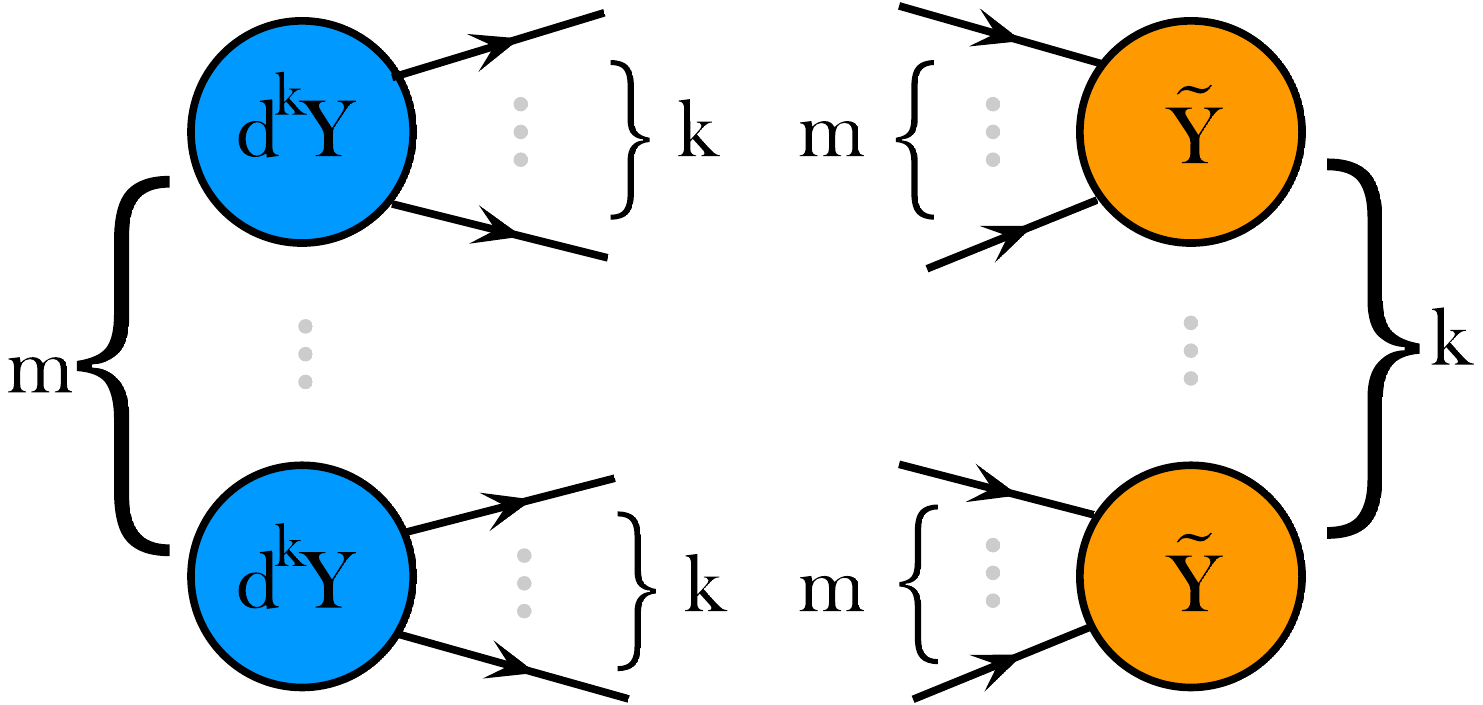}}\\

\subfloat[The measure for $(d\tilde{W}Wd\tilde{W})_{I_1\cdots I_{2k{-}2}}$. The index contractions are defined graphically. In particular, we note that $W^{IJ}$ should have 2 upstairs indices and hence 2 outgoing arrows, while $d\tilde{W}_{I_1\ldots I_{k}}$ should have $k$ incoming arrows.\newline]{\label{fig:tensor_b}\includegraphics[width=9.5cm]{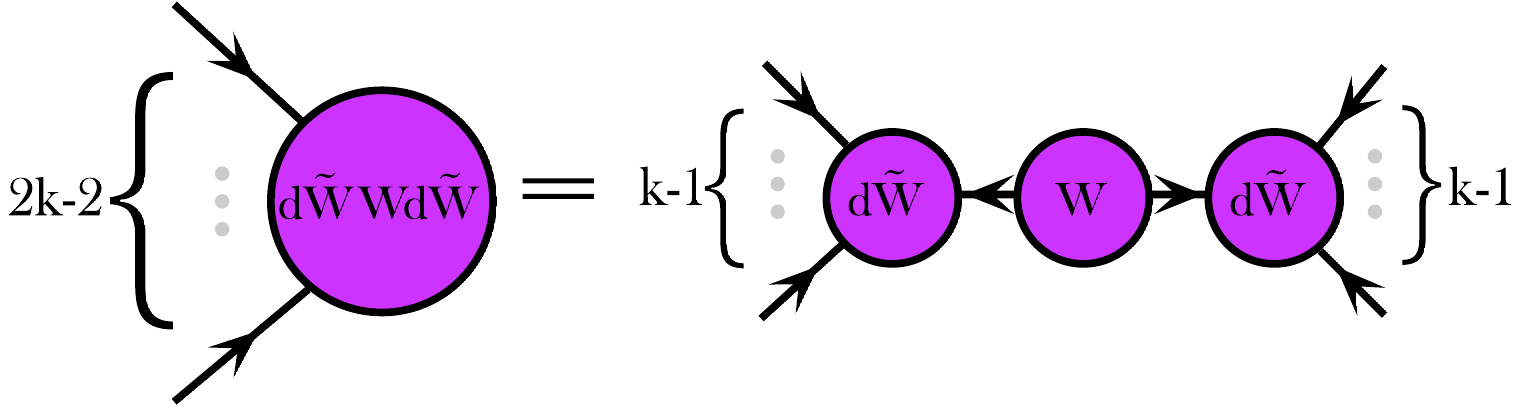}}
\\

\subfloat[Index contraction for the numerator of $\omega_k(W_1,\ldots,W_k;Y)$. For each orange node, the $k$ outgoing arrows should be connected with each of the $k$ purple nodes.]{\label{fig:tensor_c}\includegraphics[width=7.5cm]{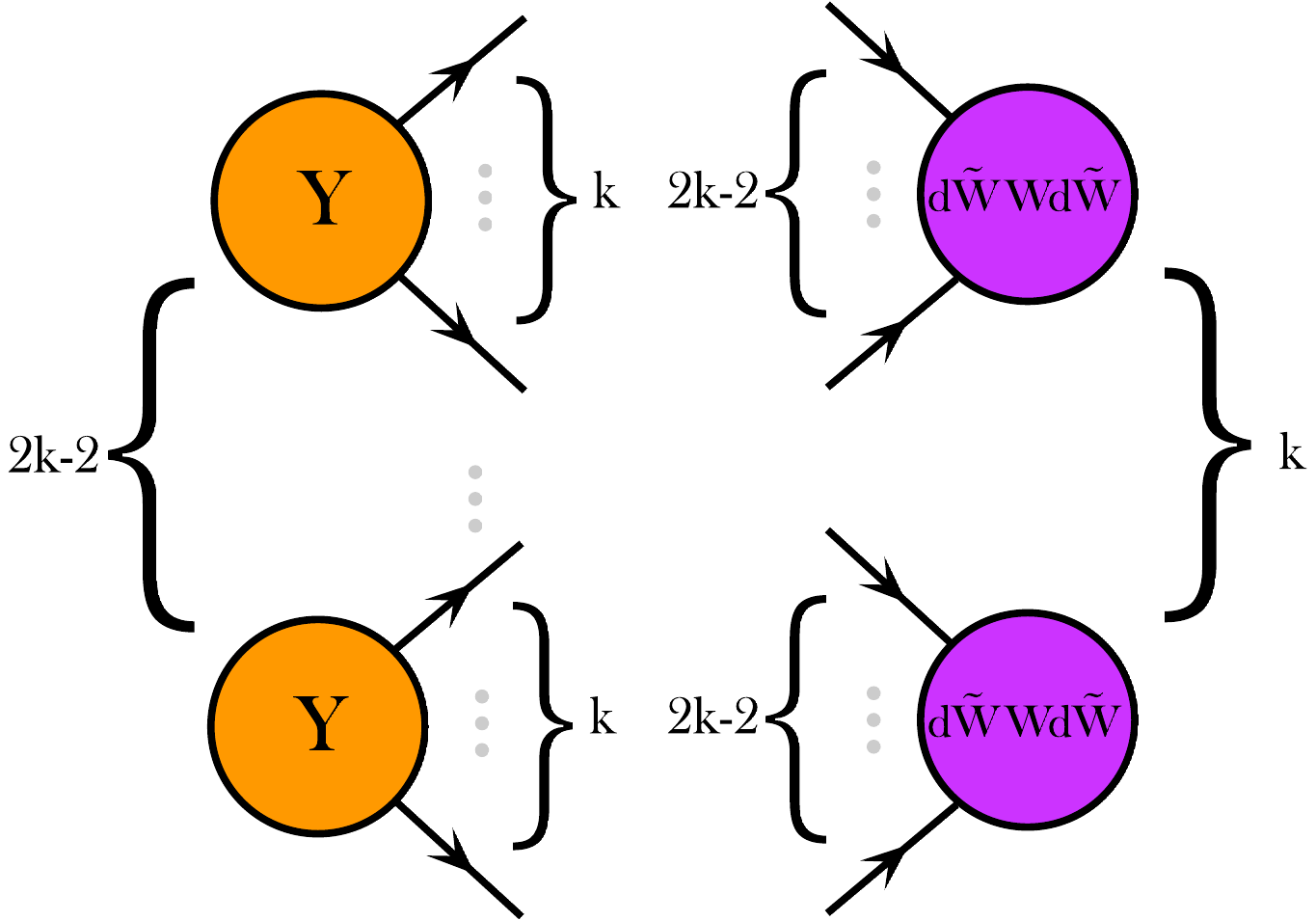}}\;\;\;\;

\caption{Diagrams representing tensor contractions. Each node denotes a tensor, with each outgoing arrow denoting an upstairs index of the tensor, and each incoming arrow denoting a downstairs index.\newline}
\label{fig:tensor}
\end{figure}

Let us begin by describing why the problem is challenging, returning to the
$k=1$ polygon example for $m=2$. Here we have a polygon whose
vertices are $Z_i^I$, and the Amplituhedron (a polygon) is just the
convex hull of these vertices. Now for any $k$, the co-dimension one boundaries
of the $m=2$ Amplituhedron are $W_{i}^{IJ} \deff (Z_i Z_{i+1})^{IJ}$, which are
2-planes in $(2+k)$ dimensions. Here the product $Z_iZ_{i{+}1}$ is a wedge product, so $W_{i}^{IJ}$ is alternating, and projectively can be thought of as points in $\wedge^2\R^4\cong \R^6$ or $\P^5$ when projectivized. Loosely speaking, we would
like the dual Amplituhedron to be the ``convex hull of the $W_i^{IJ}$".
But there is a basic difficulty: while we can add vectors (i.e. 1-planes)
together to get
other vectors, we cannot in general add $k-$planes to get other
$k-$planes. In the special case where $k=1$, the $W_{i}^{IJ} =
\epsilon^{IJK} W_{i K}$ are
dual to points $W_{i K}$, and so can be added. Note of course that the $W_{iK}$ are just vertices of the dual polytope. But for general $k$, if
a natural notion of the
``dual" Amplituhedron is to exist along these lines we must learn how to deal with
adding $k$-planes together.

We will defer a discussion of general strategies for doing this to~\cite{dual}, here we will discuss one approach that can be carried to completion for the special case of $m=2$. We
begin by noting that while
we cannot add arbitrary 2-planes $W_{a}^{IJ}$ to get other 2-planes,
adding two consecutive $W$'s does yield
a 2-plane; to wit:
\begin{equation}
\alpha W_{i-1}^{IJ} + \beta W_{i}^{IJ} = [Z_{i} (-\alpha Z_{i-1} + \beta Z_{i+1})]^{IJ}
\end{equation}
For $\alpha,\beta>0$, this gives us a line from $W_{i{-}1}$ to $W_i$ within the Grassmannian. Thus there is a natural``polygon" $P$ in $G(2,2+k)$
with the vertices $W_i^{IJ}$
joined consecutively by line segments as above. Here we have not specified an interior for the polygon, only its edges. In fact, since the embedding space $G(2,2{+}k)$ has dimension greater than 2 for $k>1$, the polygon does not necessarily have a unique interior. We return to this important point shortly.

Since our canonical form has rank $2 \times k$, it is
natural to expect the ``dual integral" representation to be an integral over a $2 \times k$ dimensional space. Given this
canonical one-dimensional boundary
in $G(2,2+k)$, a simple possibility presents itself.
Consider {\sl any} 2-dimensional surface whose boundary is
the one-dimensional polygon $P$. Now consider any $k$ of these surfaces
$\Sigma_{s}$ for $s=1,\cdots,k$, which may be distinct.
Then we would like to consider a $k$-fold integral over the space $\Sigma\deff \Sigma_1\times\cdots\times\Sigma_k$. This gives us a $2 \times k$ dimensional
integral as desired but appears to
depend on the choice of $\Sigma_s$ for each $s$; our only hope is that the forms are closed (in each of the $k$ components independently)
and thus the integral depends only on the (canonical) polygon $P$ and
not on the particular surface spanning it. As
we will see, the structure of this form is essentially fixed by
demanding that it is consistently
defined on the Grassmannian, and it will indeed turn out to be closed in each component independently.

Let us first recall what fixes the structure of forms on the Grassmannian (see also the discussion in Appendix~\ref{app:projform}). Consider first the Grassmannian $G(k,k{+}m)$ associated with a matrix $Y_s^I$ with $s = 1,\ldots, k$ and $I=1,\cdots, k{+}m$. We can think of $Y$ in a more $\GL(k)$ invariant way as a $k-$fold antisymmetric tensor $Y^{I_1 \cdots I_k} = \epsilon^{s_1 \cdots s_k} Y^{I_1}_{s_1} \cdots Y^{I_k}_{s_k}$. It is also natural to consider the $m$-plane $\tilde{Y}_{J_1 \cdots J_{m}} = \epsilon_{I_1 \cdots I_k J_1 \cdots J_{m}} Y^{I_1 \cdots I_k}$. The Pl\"ucker relations satisfied by $Y$ are then contained in the simple statement $Y_s^K Y_{K J_2 \cdots J_{m}} = 0$.

It will be convenient to introduce a graphical notation for the $\GL(k{+}m)$ indices here. Each node represents a tensor; and for each tensor, an upstairs index is denoted by an arrow outgoing from the node and a downstairs index by an incoming one. Then $Y^{I_1\cdots I_k}$ is a node with $k$ outgoing arrows and $\tilde Y_{I_1\ldots I_{m}}$ is a node with $m$ incoming ones, as shown by the orange nodes in Figure~\ref{fig:tensor}.


Now as discussed in Appendix~\ref{app:projform}, in order for a differential form to be well-defined on the Grassmannian, it must be invariant under {\it local} $\GL(k)$ transformations, $Y_s^I \to L_s^t(Y) Y_t^I$. We repeat the argument here from a graphical point of view. Since $dY \to  L ((L^{-1} d L) Y + dY)$ under this transformation, we must have that the measure is unchanged if we replace any single factor of $dY_s^I$ with any $Y_t^I$. This fixes  the standard measure factor $\langle Y d^m Y \rangle$ in projective space up to scale. Generalizing to the Grassmannian, we are looking for a $k \times m$ form. It is natural to consider the $\GL(k)$ invariant $k$-form $(d^k Y)^{I_1 \cdots I_k}$, defined as minors of the matrix of $dY$'s, or $(d^k Y)^{I_1 \cdots I_k} \deff  \epsilon^{s_1 \cdots s_k} dY^{I_1}_{s_1} \cdots dY^{I_k}_{s_k}$. For local $\GL(k)$ invariance, every leg of $(d^k Y)$ must be contracted with some $\tilde{Y}$, so that the replacement $dY_s^I \to Y_t^I$ vanishes by Pl\"ucker. Thus there is a natural $k \times m$ form on the Grassmannian, whose diagram is a complete graph connecting $m$ factors of $(d^k Y)$ on one side, and $k$ factors of  $\tilde{Y}$ on the other side, as in Figure~\ref{fig:tensor_a}. It is easy to see that in the standard gauge-fixing by $\GL(k)$ where a $k \times k$ block of the matrix representation of $G(k,k{+}m)$ is set to the identity, this form is simply the wedge product of the remaining variables. 
Any top-form on $G(k,k{+}m)$ is expressible as this universal factor multiplied by a $GL(k)$ co-variant function of the $Y$'s with weight $-(k{+}m)$. While the contraction described here is different from the one given in~\eqref{eq:grassStandardMeasure}, they are equivalent, which can be shown by gauge fixing.

We now use the same ideas to determine the structure of the $2k$-form on $\Sigma$. Let us start with the case $k=2$. We are looking for a 4-form that is the product of two 2-forms, on the space of $W_1^{IJ}$ and $W_2^{IJ}$. Here the subscripts $1,2$ index the integration variables, not the vertices $W_i$; the distinction should be clear form context. By the same logic as above, we will build the form out of the building blocks $(d\tilde{W} W  d\tilde{W})_{IJ}$ (see Figure~\ref{fig:tensor_b}), which are invariant under local $\GL(2)$. 
We then find a form with appropriate weights under both the $W_s$ and $Y$ rescaling, given by 
\begin{equation}\label{eq:wilsonForm2}
\omega_{k=2}(W_1,W_2;Y)=\frac{ \left(d \tilde{W}_1 W_1 d \tilde{W}_1 \right)_{IJ} \left(d \tilde{W}_2 W_2 d \tilde{W}_2 \right)_{KL}  Y^{IK} Y^{JL}}{(\tilde{W}_1 \cdot Y)^3 (\tilde{W}_2 \cdot Y)^3}
\end{equation}
We then claim that the canonical rational function for the $m=2,k=2$ Amplituhedron can be expressed as 
\begin{equation}
\aOmega(\A(2,2,n))(Y) = \int_{\Sigma_1 \times \Sigma_2} \omega_{k=2}(W_1,W_2;Y)
\end{equation}

For general $k$ we will have $k$ factors of $(W \cdot Y)^3$ in the denominator, thus to have the correct weight $-(2+k)$ in minors of $Y$, we have to have $2k-2$ factors of $Y$ upstairs. Now the objects $(d \tilde{W}_s W_s d \tilde{W}_i)_{I_1 \cdots I_{2k-2}}$ each have $2k-2$ incoming arrows (as in Figure~\ref{fig:tensor_b}), while every $Y^{I_1 \cdots I_k}$ has $k$ outgoing arrows. We can thus express 
$\omega_{k}(W_1,\cdots,W_k;Y)$ graphically as the complete graph linking the $k$ factors of $(d \tilde{W} W d \tilde{W})$ on one side and the $(2 k -2)$ $Y$'s on the other, as shown in Figure~\ref{fig:tensor_c}. We then claim that 
\begin{equation}
\label{eq:wilsonForm}
\aOmega(\A(k,2,n))(Y) =  \int_{\Sigma} \omega(W_1, \ldots, W_k;Y)
\end{equation}

These expressions indeed reproduce the correct canonical rational function for the $m=2$ Amplituhedron for all $k$. Let us illustrate how this works for the case of $k=2$. A straightforward computation shows that the form $\omega_{k=2}(W_1,W_2;Y)$ is closed in $W_1,W_2$ independently; indeed it is closed even if the $W^{IJ}$ are not constrained by the Pl\"ucker relations, and can be thought of as being general points in $\P^5$. Thus the result of the integral is independent of the surface $\Sigma$, provided that $\partial\Sigma_s=P$ for each $s$. We will thus construct each surface $\Sigma_s$  by triangulating it like the interior of a polygon. We begin by picking an arbitrary reference point $X^{IJ}$, and taking the triangle in $\P^5$ with vertices $X,W_{i{-}1},W_{i}$, which we denote by $[X,W_{i{-}1},W_i]$. It follows that the union $\cup_{i=1}^n[X,W_{i{-}1},W_i]$ forms a surface with boundary $P$, so we will take it to be our definition of $\Sigma_s$ for each $s$.

Suppose we integrate over the triangle pair $(W_1,W_2)\in [X,W_{i{-}1},W_i]\times [X,W_{j{-}1},W_j]$. We can parametrize the triangle pair by
\be
W_1 &=& X+\alpha_1 Z_{i{-}1}Z_i+\beta_1 Z_iZ_{i{+}1}\\
W_2 &=& X+\alpha_2 Z_{j{-}1}Z_j+\beta_2 Z_jZ_{j{+}1}
\ee
where $\alpha_{1,2}>0; \beta_{1,2}>0$.

The only non-trivial part of the computation is working out the numerator index contraction in~\eqref{eq:wilsonForm2}. After a series of index gymnastics, we get for each $i,j$:
\be
\omega_{k=2}=
\frac{1}{2}\frac{
(\mathcal{N}_1(i,j)+\mathcal{N}_2(i,j))d^2\alpha d^2\beta
}{
\left(\lb YX\rb + \alpha_1\lb Yi{-}1,i\rb+\beta_1\lb Yi,i{+}1\rb\right)^3
\left(\lb YX\rb + \alpha_1\lb Yj{-}j,j\rb+\beta_1\lb Yj,j{+}1\rb\right)^3
}\nonumber\\
\ee
where
\be
\mathcal{N}_1(i,j)&\deff &\lb YX \rb\lb X ij\rb\lb Y(i{-}1,i,i{+}1)\cap(j{-}1,j,j{+}1)\rb\\
\mathcal{N}_2(i,j)&\deff &-\lb Y(i{-}1,i,i{+}1)\cap(Xj)\rb\lb Y(j{-}1,j,j{+}1)\cap(Xi)\rb
\ee
where $\lb Y(abc)\cap(def)\rb\deff \lb Y_1 abc\rb\lb Y_2 def\rb-(Y_1\leftrightarrow Y_2)$ for any vectors $a,b,c,d,e,f$.

Integrating over $\alpha_{1,2}>0; \beta_{1,2}>0$ and summing over all $i,j$ gives us
\be
\int_\Sigma \omega_{k=2} = \frac{1}{8}\sum_{i,j=1}^n
\frac{\mathcal{N}_1(i,j)+\mathcal{N}_2(i,j)}{\lb YX\rb^2\lb Yi{-}1,i\rb\lb Yi,i{+}1\rb\lb Yj{-}1,j\rb\lb Yj,j{+}1\rb}
\ee
This is one of several \defn{local} expressions for the canonical rational function $\aOmega(\A(2,2,n))$. That is, there are no spurious singularities, except at $\lb YX\rb\rightarrow 0$. Interestingly, if we keep only one of the numerator terms $\mathcal{N}_{1},\mathcal{N}_{2}$, then the result would still sum to (half) the correct answer. Perhaps some clever manipulation of the integration measure would make this manifest. If we only keep the first numerator, then we recover the local form given in~\eqref{eq:localForm}.

It is possible, through a clever choice of the surface $\Sigma$, to recover the Kermit representation~\eqref{eq:kermit} of the canonical rational function. While the equivalence between the Kermit representation and the local form appears non-trivial as an algebraic statement, it follows easily from the surface-independence of the integral. 

This computation can be extended easily to higher $k$. Furthermore, since the surfaces $\Sigma_1,\ldots, \Sigma_k$ are independent, it is possible to have picked a different $X$ for each surface, giving a local form with arbitrary reference points $X_1,\ldots, X_k$.

\subsubsection{Projective space contours part I}
\label{sec:contour1}
We now turn to a new topic. We show that the rational canonical form of convex projective polytopes (and possibly more general positive geometries) can be represented as a contour integral over a related projective space.

Recall for convex projective polytopes $\A$ that the $n \times (m+1)$-matrix $Z$ can be considered a linear map $\P^{n-1} \to \P^m$, sending the standard simplex $\Delta^{n{-}1}$ to the polytope $\A$. Letting $(C_j)_{j=1}^n$ denote the coordinates on $\P^{n-1}$, we have the equation $Y = Z(C) = \sum_{j=1}^nC_jZ_j$ describing the image $Y \in \P^m$ of a point $C \in \P^{n-1}$ under the map. 

We begin with the simplex canonical form on $(C_j)_{j=1}^n\in \C^{n}$ constrained on the support of the expression $Y=\sum_{j=1}^n C_jZ_j$:
\be
\int\frac{d^nC}{\prod_{j=1}^n C_j}\delta^{m{+}1}\left(Y-\sum_{j=1}^n C_jZ_j\right).
\ee
Typically, the delta function on projective space $\delta^m(Y,Z(C))$ is reduced from rank $m{+}1$ to $m$ by integrating over a $\GL(1)$ factor like $\rho$ in~\eqref{eq:delta4}. Instead, we have absorbed the $d\rho/\rho$ measure into the canonical form on $C$-space, thus giving a rank-$n$ measure on $\C^n$.

We now describe a contour in the $C$-space such that the above integral gives the canonical rational function $\aOmega(\A)$. A naive integral over all $C \in \R^n$ is obviously ill-defined due to the $1/C_j$ singularities. However, with some inspiration from Feynman, we can integrate slightly away from the real line $C_j\rightarrow c_j = C_j + i \epsilon_j$ for some small constants $\epsilon_j>0$, with $c_j\in \mathbb{R}$ being the contour. We will assume that $\sum_{j=1}^n\epsilon_jZ_j=\epsilon Y$ for some $\epsilon>0$, and let $\bar{Y} \deff (1+i\epsilon)Y$. After completing the contour integral, we can take the limit $\epsilon\rightarrow 0$ to recover $\bar{Y}\rightarrow Y$. This is reminiscent of the epsilons appearing in the loop integration of amplitudes. Finally, we assume that the $Z_j$ vertices form a real convex polytope 
and $\bar{Y}$ is a positive linear combination of the $Z_j$'s. Note that $\bar{Y}$ is real and $Y$ is now slightly imaginary due to the $i\epsilon$ shift.

We claim, in the $\epsilon_j \to 0$ limit,
\be \label{eq:contour}
\aOmega(\A)(Y)=\frac{1}{(2\pi i)^{n{-}m{-}1}m!}\int\frac{d^nc}{\prod_{j=1}^n (c_j-i \epsilon_j)}\delta^{m+1}\left(\bar{Y}-\sum_{j=1}^n c_jZ_j\right)
\ee
with integration over $c_j\in\mathbb{R}$ for each $j$. The overall constants have been inserted to achieve the correct normalization.

Before proving this identity in full generality, we give a few examples.

\begin{example}
The simplest example occurs for $n=m{+}1$ where $\A$ is just a simplex in $m$ dimensions. In that case, there is no contour, and we can immediately set $\epsilon_j\rightarrow 0$ for each $j$ and thus $\epsilon\rightarrow 0$. We get
\be
\aOmega(\A)(Y)=\frac{1}{m!}\int \frac{d^{m{+}1}c}{\prod_{j=1}^{m{+}1}c_j}\delta^{m{+}1}\left(Y-\sum_{j=1}^{m{+}1}c_jZ_j\right)
\ee
We can uniquely solve for each $c_j$ on the support of the delta function, which gives
\be
c_j = (-1)^{j{-}1}\frac{\lb Y12 \cdots \hat{j}\cdots (m{+}1)\rb}{\lb 12 \cdots (m{+}1)\rb}
\ee
where the ``hat" denotes omission, and the Jacobian of the delta function is $\lb12\cdots (m{+}1)\rb$. It follows that (see~\eqref{eq:RInv})
\be
\aOmega(\A)(Y)=[1,2,\ldots,m{+}1]
\ee
which is the familiar canonical rational function for a simplex. 

More generally, suppose we integrate over a contour in the original $C$ space that picks up a set of poles at $C_j\rightarrow 0$ for all $j$ except $j=j_0,\ldots,j_{m}$, which we assume to be arranged in ascending order. Then the residue is given by
\be
\frac{1}{m!}\int \left(\prod_{j\in J}\Res_{C_j = 0}\right)\frac{d^nC}{\prod_{j=1}^nC_j}\delta\left(Y-\sum_{j=1}^n C_j Z_j\right)=[j_0,j_1,j_2,\ldots ,j_m]
\ee
where $J=\{1,2,\ldots,n\}-\{j_0,j_1,j_2,\ldots,j_{m}\}$. Suppose the indices contained in $J$ are $k_1,\ldots,k_{n{-}m{-}1}$ in increasing order, then we will denote the residue collectively as $\Res(J)$ or $\Res(k_1, k_2,\ldots,k_{n{-}m{-}1})$. So
\be\label{eq:ResJ}
\Res(J)\deff \Res(k_1,\ldots,k_{n{-}m{-}1}) \deff [j_1,j_2,\ldots,j_m]
\ee
We note that the result may come with a negative sign if the contour is negatively oriented. The $\Res$ operator, however, assumes a positive orientation. This formula will be very convenient for the subsequent examples.
\end{example}

We now move on to higher $n$ examples for the polygon (i.e. $m=2$).

\begin{example}

We now perform the contour integral explicitly at $n=4$ points for the $m=2$ quadrilateral. There are four integrals over $c_i$ constrained by three delta functions. We will integrate out $c_{1,2,3}$ to get rid of the delta functions with a Jacobian factor $\lb 123\rb^3$, which gives us
\be
\frac{1}{2!(2\pi i)}\int_{c_4\in \mathbb{R}}\frac{dc_4/\lb 123\rb^3}{(c_1-i\epsilon_1)(c_2-i\epsilon_2)(c_3-i\epsilon_3)(c_4-i\epsilon_4)}
\ee
where $c_1,c_2$, and $c_3$ depend on $c_4$ through the three equations $\bar{Y}=Z(c)$.

In the large $c_4$ limit, all three dependent variables scale like $O(c_4)$, so the integrand scales like $O(c_4^{-4})$ which has no pole at infinity. We can now integrate over $c_4$ by applying Cauchy's theorem. The key is to first work out the location of the four poles relative to the real line.

We begin with the constraints
\be
\bar{Y}=c_1Z_1+c_2 Z_2+c_3 Z_3+c_4Z_4
\ee
which can be re-expressed in terms of three scalar equations by contracting with $Z_2Z_3$, $Z_3Z_1$, and $Z_1Z_2$, respectively.
\be
c_4\lb 234\rb &=& \lb \bar{Y}23\rb-c_1 \lb 123\rb\\
c_4\lb 134\rb &=& -\lb \bar{Y}31\rb+c_2 \lb 123\rb \\
c_4 \lb 124\rb &=& \lb \bar{Y}12\rb - c_3 \lb 123\rb
\ee
We have written the equation so that each bracket is positive, provided that $Z_i$ is cyclically convex and $Y$ is inside the polygon. Both positivity conditions are crucial, because they prescribe the location of the poles.

\begin{figure}
\centering
\includegraphics[width=8cm]{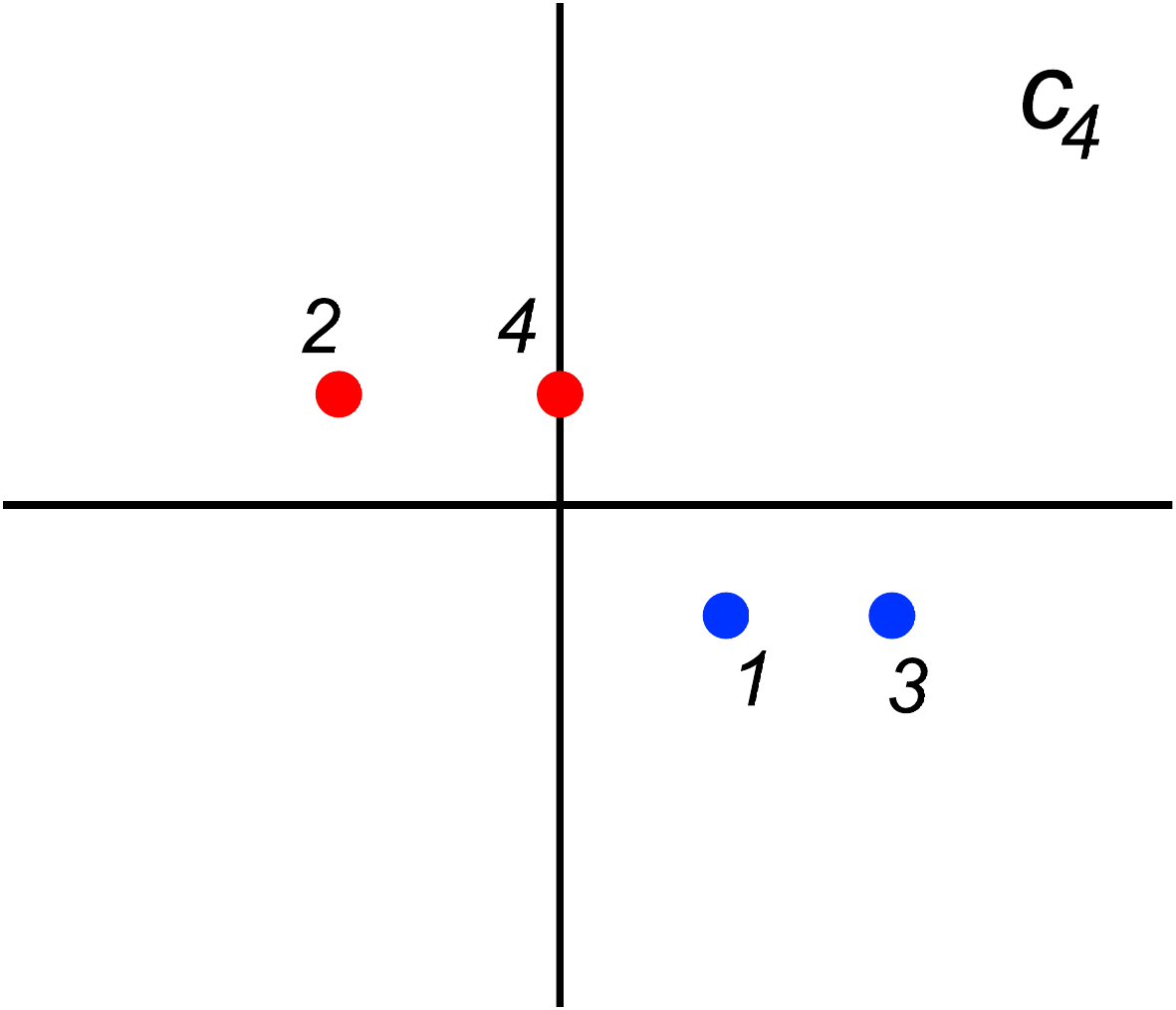}
\caption{The locations of the four poles for the quadrilateral $i\epsilon$ contour, with $c_1,c_2$ and $c_3$ dependent on $c_4$ through $\bar{Y}=c\cdot Z$. A pole is colored red if a counterclockwise contour picks up the residue with a plus sign (e.g. $+\Res(2), +\Res(4)$), while a pole is colored blue if a counterclockwise contour picks up the residue with a minus sign (e.g. $-\Res(1), -\Res(3)$). Of course, the signs are reversed if the contour is clockwise.}
\label{fig:ieps4}
\end{figure}

Now, for any index $i=1,2,3,4$, we let the ``pole at $i$" refer to the pole at $c_i\rightarrow i\epsilon_i$. From the constraints, we find that poles 1 and 3 both provide a negative imaginary part to $c_4$, so they are below the real line, while poles at $2$ and $4$ are above the real line (see Figure~\ref{fig:ieps4}). Since there is no pole at infinity, the integral can be performed by closing the contour below to pick up poles $1,3$, or closing the contour above to pick up poles $2,4$. The former gives us
\be
\aOmega(\mathcal{A}) = \Res(1)+\Res(3) = [2,3,4]+[1,2,4]
\ee
while the latter gives
\be
\aOmega(\mathcal{A}) = \Res(2)+\Res(4) = [1,3,4]+[1,2,3]
\ee
Of course, these are two equivalent triangulations of the quadrilateral. We see therefore that the $i\epsilon$ contour beautifully explains the triangulation independence of the canonical rational function as a consequence of Cauchy's theorem.

\end{example}

\begin{example}
We now compute the contour for a pentagon, where new subtleties emerge. We begin as before by integrating over $c_1,c_2$ and $c_3$ to get rid of the delta functions. We then re-express the delta function constraints as three scalar equations.
\be
c_4\lb 234\rb+c_5\lb 235\rb &=& \lb \bar{Y}23\rb-c_1\lb 123\rb\\
c_4 \lb 134\rb+c_5\lb 135\rb&=&-\lb \bar{Y}31\rb+c_2\lb 123\rb\\
c_4\lb 124\rb+c_5\lb 125\rb&=&\lb \bar{Y}12\rb-c_3\lb 123\rb
\ee
We now integrate over $c_4\in \mathbb{R}$. The locations of the poles are the same as before (see Figure~\ref{fig:ieps4}), and we are free to choose how we close the contour. For sake of example, let us close the contour above so we pick up poles 2 and 4. We now analyze both poles individually.

The pole at 4 induces $c_4\rightarrow i\epsilon_4$. Our constraints therefore become:
\be
c_5\lb 235\rb &=& \lb \bar{Y}23\rb-c_1\lb 123\rb-i\epsilon_4\lb 234\rb\\
c_5\lb 135\rb&=&-\lb \bar{Y}31\rb+c_2\lb 123\rb-i\epsilon_4\lb 134\rb\\
c_5\lb 125\rb&=&\lb \bar{Y}12\rb-c_3\lb 123\rb-i\epsilon_4 \lb 124\rb
\ee
There are now four poles left for $c_5$, corresponding to 1,2,3 and 5. The poles at 1 and 3 are clearly below the real line, as evident in the equations above, while the the pole at 5 is obviously above (see Figure~\ref{fig:ieps5}). The pole at 2, however, is more subtle and deserves closer attention. In the limit $c_2\rightarrow i\epsilon_2$, we find from the second equation above that
\be
c_5\lb 135\rb = -\lb \bar{Y}31\rb+iq
\ee
where $q=\epsilon_2\lb 123\rb-\epsilon_4\lb 134\rb$. If $q>0$, then the pole at 2 is above the real line, and if $q<0$, then the pole is below. So the relative size of the $\epsilon_j$ parameters makes a difference to the computation. However, as we now show, the final result is unaffected by the sign of $q$. Suppose $q>0$, then we can close the $c_5$ contour above and pick up the following poles
\be
\Res(45)+\Res(24)=[123]+[135]
\ee
Alternatively, we can close below and pick up
\be
\Res(34)+\Res(14)=[125]+[235]
\ee
The two results, of course, are identical since there is no pole at infinity. So the pole at 4 produces the following result regardless of how the $c_5$ contour is closed, provided that $q>0$.
\be
c_4\rightarrow i\epsilon_4\;;\; q>0 \;\;\sim\; \aOmega(Z_1,Z_2,Z_3,Z_5)
\ee
If $q<0$, then $\Res(24)$ migrates below the real line, in which case closing the contour above would give
\be
\Res(45) = [123]
\ee
and closing below would give
\be
\Res(34)+\Res(14)-\Res(24)=[125]+[235]-[135]
\ee
which of course are equivalent. The minus sign in front of $\Res(24)$ appears due to the orientation of the contour. 
In summary,
\be
c_4\rightarrow i\epsilon_4\;;\; q<0 \;\;\sim\; \aOmega(Z_1,Z_2,Z_3)
\ee

\begin{figure}
\centering
\includegraphics[width=5cm]{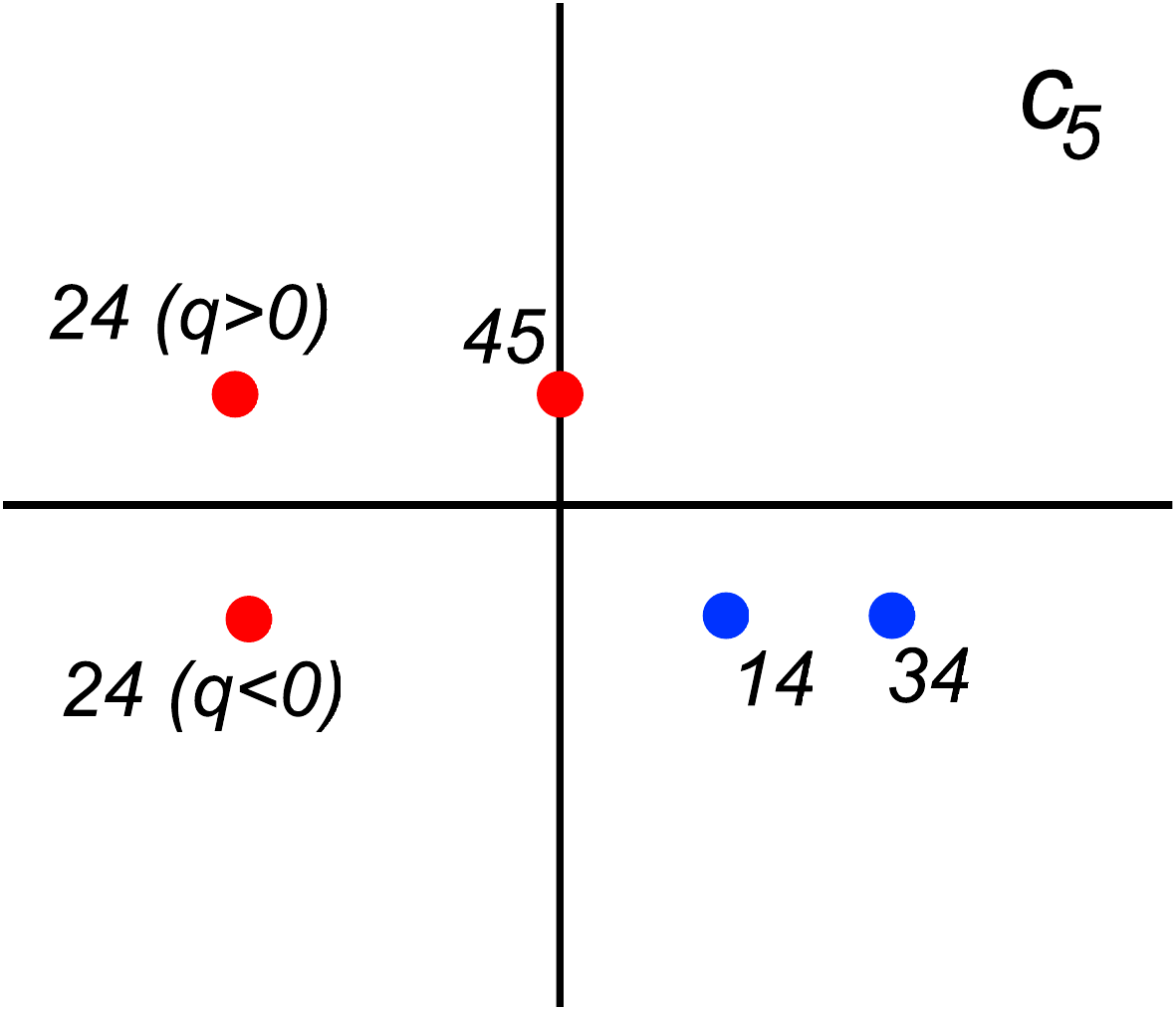}\;\;\;\;\;\;
\includegraphics[width=5cm]{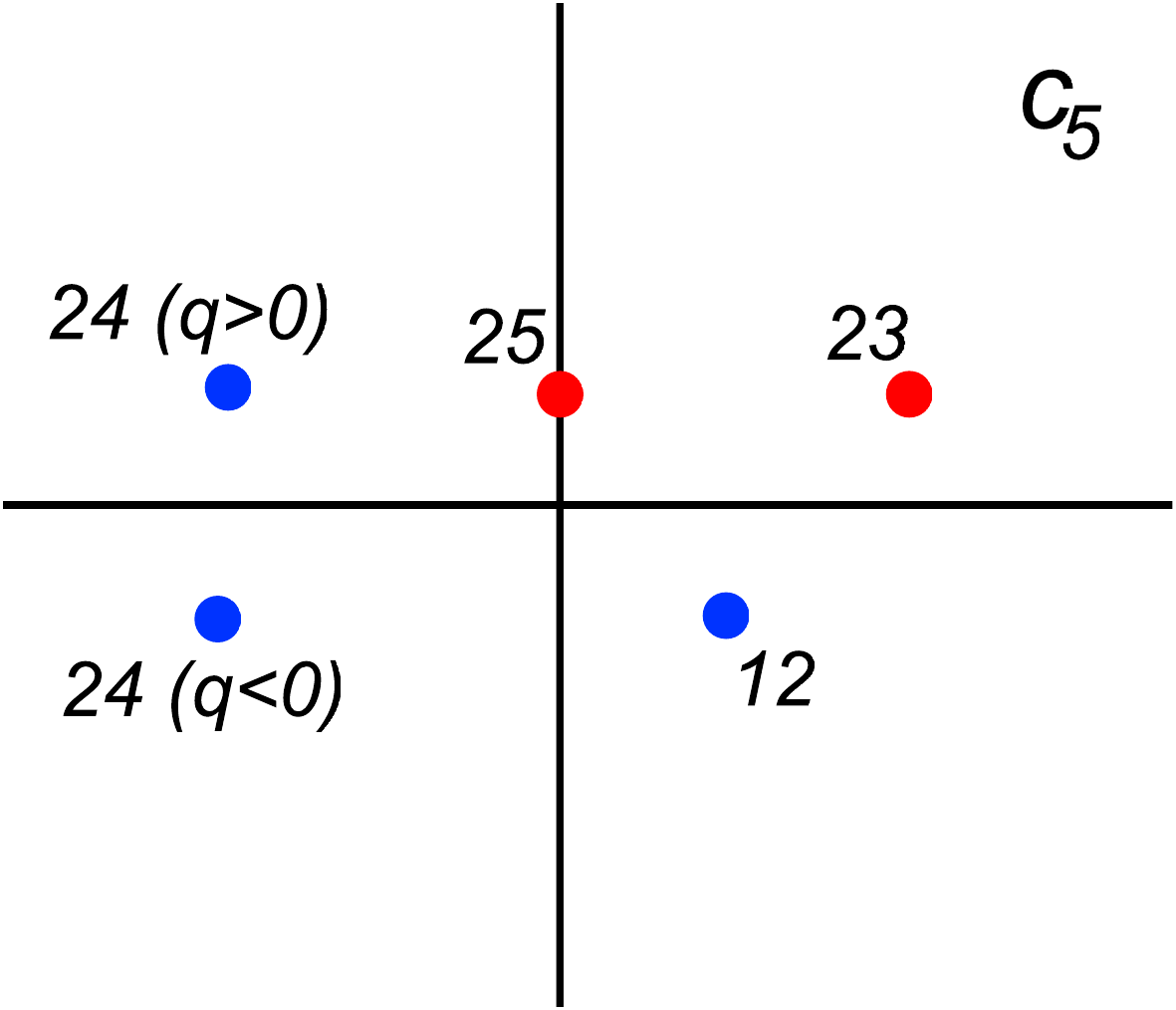}
\caption{The pole locations of $c_5$ under $c_4\rightarrow i\epsilon_4$ (left) and $c_2\rightarrow i\epsilon_2$ (right). Note that in both graphs the position of the 24 pole relative to the real line differs depending on the sign of $q$.  See the caption under Figure~\ref{fig:ieps4} for explanations of the pole coloring.}
\label{fig:ieps5}
\end{figure}

Now let us move on to the pole of $c_4$ at 2 which induces $c_2\rightarrow i\epsilon_2$. Again, there are four poles left for $c_5$, this time corresponding to 1,3,4,5. After re-arranging the constraints using Schouten identities, we find
\be
c_5\lb 123\rb\lb 345\rb &=&  \lb\bar{Y}34\rb\lb 123\rb-\lb 123\rb(c_1\lb 134\rb+i\epsilon\lb 234\rb)\\
c_4\lb 145\rb\lb 123\rb &=&-\lb \bar{Y} 41\rb\lb 123\rb+\lb 123\rb(i\epsilon_2\lb 124\rb+c_3\lb 134\rb)\\
c_5\lb 135\rb&=&\lb\bar{Y}13\rb+i\epsilon_2\lb 123\rb-c_4\lb 134\rb
\ee
Evidently, the poles at 3 and 5 are above the real line, while the pole at 1 is below (see Figure~\ref{fig:ieps5}). The pole at 4, however, is above if $q>0$ and below if $q<0$. 

For $q>0$, closing the $c_5$ contour above gives
\be
\Res(23)-\Res(24)+\Res(25)=[145]-[135]+[134]
\ee
while closing below gives
\be
\Res(12)=[345]
\ee
which are equivalent, so
\be
c_2\rightarrow i\epsilon_2\;;\; q>0 \;\;\sim\; \aOmega(Z_3,Z_4,Z_5)
\ee

For $q<0$, closing the contour above gives
\be
\Res(23)+\Res(25)=[145]+[134]
\ee
while closing below gives
\be
\Res(12)+\Res(24)=[345]+[135]
\ee
which are again equivalent, so
\be
c_2\rightarrow i\epsilon_2\;;\; q>0 \;\;\sim\; \aOmega(Z_1,Z_3,Z_4,Z_5)
\ee

We now sum over the $c_4\rightarrow i\epsilon_4$ and $c_2\rightarrow i\epsilon_2$ contributions. For $q>0$, we get
\be
\aOmega(\A)=\aOmega(Z_1,Z_2,Z_3,Z_5)+\aOmega(Z_3,Z_4,Z_5)
\ee
and for $q<0$, we get
\be
\aOmega(\A)=\aOmega(Z_1,Z_2,Z_3)+\aOmega(Z_1,Z_3,Z_4,Z_5)
\ee
Both results are of course correct. Again, we see triangulation independence as a result of Cauchy's theorem, but with additional subtleties involving the sign of $q$. Furthermore, we could have closed the $c_4$ contour below instead and achieved a different set of triangulations.

\end{example}

We now argue for all $m>0$ and all $n\ge m{+}1$ that the $i\epsilon$ contour integral \eqref{eq:contour} is equivalent to the dual volume formula \eqref{eq:volume}.  Let $S$ denote the right hand side of \eqref{eq:contour}.  We begin by writing the delta function in terms of its Fourier transform with dual variable $W\in\mathbb{R}^{m{+}1}$.
\be
\delta^{m{+}1}\left(\bar{Y}-\sum_{j=1}^nc_jZ_j\right)=\frac{1}{(2\pi)^{m{+}1}}\int d^{m{+}1}W\;e^{i\left(-W \cdot \bar{Y}+\sum_{j=1}^nc_jW \cdot Z_j\right)}
\ee


Then we integrate over each $c_j$ using the following identity:
\be
\int_\R \frac{dc_j}{c_j-i\epsilon_j}e^{i c_j W \cdot Z_j} = (2\pi i)\; \theta(W \cdot Z_j)
\ee
where $\theta(x)$ is the Heaviside step function. It follows that
\be
S = \frac{i^{m{+}1}}{m!}\int d^{m{+}1}W  e^{-iW \cdot \bar{Y}}\prod_{j=1}^n\theta(W \cdot Z_j)
\ee

We also change variables to radial coordinates $\rho \deff |W|$ and $\hat{W}\deff W/|W|$ with $|W|$ denoting the Euclidean norm of $W$. We have $W=\rho \hat{W}$ so that the measure now becomes $d^{m{+}1}W=\rho^md\rho \;\lb\hat{W}d^m\hat{W}\rb/m!$, where $\lb\hat{W}d^m\hat{W}\rb$ is the pull back of $\lb Wd^mW\rb$ to the sphere $\hat{W}\in S^m$. 

Now recall the Fourier transform identity:
\be
\frac{1}{m!}\int_{\R} dx\; x^m\theta(x)e^{-i x y} = \frac{(-i)^{m{+}1}}{(y-i \epsilon')^{m{+}1}} 
\ee
for some small $\epsilon'>0$ which we take to zero in the end. The integral over $\rho>0$ therefore becomes:
\be
\frac{1}{m!}\int_0^\infty d\rho\;\rho^m e^{-i\rho \hat{W}\cdot \bar{Y}} = \frac{(-i)^{m{+}1}}{(\hat{W} \cdot \bar{Y}-i\epsilon')^{m+1}} = 
\frac{(-i)^{m{+}1}}{(\hat{W} \cdot Y)^{m+1}} 
\ee
where we set $\epsilon'\rightarrow 0$ on the right. 

It follows that
\be
 S&=&\frac{1}{m!}\int_{\hat{W}\in S^m} \lb \hat{W}d^m\hat{W}\rb\frac{1}{(\hat{W}\cdot \bar{Y})^{m{+}1}}\prod_{j=1}^n\theta(\hat{W} \cdot Z_j)\\
&=&\frac{1}{m!}\int_{\A^*} \lb Wd^mW\rb\frac{1}{(W \cdot \bar{Y})^{m{+}1}} 
\ee
which is the volume formula \eqref{eq:volume} for $\aOmega(\A)$ in the limit $\bar{Y}\rightarrow Y$.
In the last step, we used that fact that the inequalities $W \cdot Z_j>0$ imposed by the step functions carve out the interior of the dual polytope $\A^* \subset \P^m$. Furthermore, on the last line we removed the ``gauge" choice $\hat{W}\in S^m$, since the resulting integral is gauge invariant under $GL(1)$ action. 

We stress that it is absolutely crucial for $\bar{Y}$ to be on the polytope's interior. Otherwise, the denominator factor $W\cdot \bar{Y}$ can vanish and cause divergent behavior. From the contour point of view, as shown in examples above, the position of $\bar{Y}$ affects the location of poles relative to the contour.

\subsubsection{Projective space contours part II}
We now consider an alternative contour integral that is in essence identical to the previous, but represented in a very different way. The result is the following
\be \label{eq:kernelcontour}
\aOmega(\A)(Y) &=& \frac{1}{(2\pi)^{n{-}m{-}1}m!}\int \frac{d^nB}{\prod_{j=1}^n (C_{j}^0-i B_j)}\delta^{m+1}\left(\sum_{j=1}^nB_jZ_j\right) \\
&=& \frac{1}{(2\pi)^{n{-}m{-}1}m!}\int_{B\in K} \frac{d^{n-m-1} B}{\prod_{j=1}^n (C_{j}^0-i B_j)}
\ee
for any point $C^0 = (C_1^0,C_2^0,\ldots,C_n^0) \in \R^n$ such that $Y = \sum_{j=1}^n C_{j}^0Z_j$ and $C_j^0>0$. In the second equation, $K \subset \R^n$ denotes the $(n{-}m{-}1)$ dimensional kernel of the map $Z: \R^n \to \R^{m+1}$. The measure on $K$ is defined as
\be
d^{n{-}m{-}1}B\deff \int d^nB\;\delta^{m{+}1}\left(\sum_{j=1}^n B_j Z_j\right)
=\frac{dB_{k_1}\cdots dB_{k_{n{-}m{-}1}}}
{\lb j_0\cdots j_m\rb}
\ee
for any partition $\{j_0,\ldots, j_m\}\cup\{k_1,\ldots,k_{n{-}m{-}1}\}$ of the index set $\{1,\ldots,n\}$. On the right, it is understood that $B\cdot Z = 0$.

We now argue that \eqref{eq:kernelcontour} is equivalent to \eqref{eq:contour}. An immediate consequence is that this contour integral is independent of the choice of $C^0$.

We begin with the $i\epsilon$ contour~\eqref{eq:contour}, fix a choice $c^0$ satisfying $\bar{Y}=c^0\cdot Z$, and change integration variables $c_j\rightarrow b_j=c_j^0-c_j$.
\be
\aOmega(\A)(Y)=\frac{1}{(2\pi i)^{n{-}m{-}1}m!}\int\;\frac{d^nb}{\prod_{j=1}^n(c_j^0-b_j-i\epsilon_j)}\delta^{m{+}1}\left(\sum_{j=1}^nb_jZ_j\right)
\ee

For each $b_j$, there is a pole at $b_j=c_j^0-i\epsilon_j$ which lives in the fourth quadrant on the complex plane of $b_j$. Since we are currently integrating over the real line, it is possible to do a clockwise Wick rotation $b_j\rightarrow B_j=-ib_j$ to integrate over the imaginary line (i.e. $B_j\in \mathbb{R}$) without picking up any of the poles. It follows that

\be
\aOmega(\A)(Y)=\frac{1}{(2\pi)^{n{-}m{-}1}m!}\int\;\frac{d^nB}{\prod_{j=1}^n(c_j^0-iB_j-i\epsilon_j)}\delta^{m{+}1}\left(\sum_{j=1}^nB_jZ_j\right)
\ee

We can now set $\epsilon_j\rightarrow 0$ as it no longer affects the contour of integration. In this limit, we recover $c_i^0\rightarrow C_i^0$ and $\bar{Y}\rightarrow Y$, giving $Y=C^0\cdot Z$ and hence the desired result. The formula \eqref{eq:kernelcontour} was first established in \cite[Theorem 5.5]{BT}.

Let us work out one example.

\begin{example}
Let $\A \subset \P^2$ have vertices $\{(1,0,0),(0,1,0),(0,0,1),(1,1,-1)\}$, as in Example \ref{ex:squarelaplace}.  The kernel is $K = \{(x,x,-x,-x):x\in\R\} \subset \R^4$.  For $Y = (Y^0,Y^1,Y^2)$ we pick $C^0 = (Y^0-a,Y^1-a,Y^2+a,a)$ for some $a\in \R$, with the assumption that all entries are positive. Thus
\begin{align}
\aOmega(\A)(Y) & = 
\frac{1}{(2\pi)2!}\int_{-\infty}^\infty \frac{dx}{(Y^0-a-ix)(Y^1-a-ix)(Y^2+a+ix)(a+ix)} \\
&= \frac{1}{2!}\left( \frac{1}{Y^0Y^1Y^2} - \frac{1}{(Y^0+Y^2)(Y^1+Y^2)Y^2}\right)
\end{align}
where we have chosen to close the contour by picking the poles $x = ia$ and $x =i(Y^2+a)$ with positive imaginary part.  This is independent of $a$, as expected, and agrees with \eqref{eq:squaretriangulation}.
\end{example}

It is interesting to contrast these contour representations with the Newton polytope push-forward formula discussed in Section~\ref{sec:push}. For the former, the combinatorial structure of the convex polytope is automatically ``discovered" by the contour integration without prior knowledge. For the latter, the combinatorial structure must be reflected in the choice of the Newton polytope, which may not be possible for polytopes whose combinatorial structure cannot be realized with integral coordinates (e.g. non-rational polytopes~\cite{nonrational}).

\subsubsection{Grassmannian contours}

The generalization of the $i\epsilon$ contour formulation to the Amplituhedron for $k>1$ is an outstanding problem. We now give a sketch of what such a generalization may look like at tree level, if it exists.

We conjecture that the canonical rational function of the tree Amplituhedron, denoted $\A(k,n,m;L{=}0)$, is given by a contour integral of the following form:
\be
\aOmega(\A)=\frac{1}{(2\pi i)^{k(n{-}m{-}k)}(m!)^k}\int_\Gamma\;\frac{d^{k\times n }C}{\prod_{i=1}^n C_{i,i{+}1,...,i{+}k{-}1}}
\delta^{k\times(k{+}m)}\left(Y-C\cdot Z\right)
\ee
where the delta functions impose the usual $Y=C\cdot Z$ constraint while the measure over $C$ is the usual cyclic Grassmannian measure on $G_{\geq 0}(k,n)$. The integral is performed over some contour $\Gamma$ which we have left undefined.

For $n=m{+}k$, there is no contour and we trivially get the expected result. For $n=m{+}k{+}1$, we have naively guessed a contour that appears to work according to numerical computations. The idea is to simply change variables $C\rightarrow c$ so that for every Pl\"ucker coordinate we have $c_{i_1,\ldots,i_k}=C_{i_1,\ldots,i_k}+i\epsilon_{i_1,\ldots i_k}$ for a small constant $\epsilon_{i_1,\ldots,i_k}>0$. We will assume that $\sum_{1\leq i_1<\cdots <i_k\leq n}\epsilon_{i_1\ldots i_k}Z_{i_1}\wedge\ldots\wedge Z_{i_k}=\epsilon Y$ for some $\epsilon>0$ and define $\bar{Y}\deff (1+i\epsilon)Y$ (We are re-scaling Pl\"ucker coordinates here, not components of the matrix representation.) so we can re-express $Y=C\cdot Z$ as $\bar{Y}=c\cdot Z$ as we did for the polytope. Again, we assume that $\bar{Y}$ is real and $Y$ is slightly complex.

We then integrate over the real part of $c$. 
\be
\aOmega(\A)(Y)&=&\frac{1}{(2\pi i)^{k}(m!)^k}\int_\Gamma\;\frac{d^{k\times (m{+}k{+}1) }c}{\prod_{i=1}^{m{+}k{+}1}\left[c_{i,i{+}1,\ldots,i{+}k{-}1}{-}i\epsilon_{i,i{+}1,\ldots,i{+}k{-}1}\right]}\times\nonumber\\
&&\delta^{k\times(k{+}m)}\left(\bar{Y}-c\cdot Z\right)
\ee
After integrating over the delta function constraints, there are effectively only $k$ integrals left to do. Conveniently, the poles appearing in the cyclic measure are linear in each integration variable, and can therefore be integrated by applying Cauchy's theorem.


One obvious attempt for generalizing beyond $n>m{+}k{+}1$ is to impose the exact same $i\epsilon$ deformation. However, after having integrated over the delta functions, the cyclic minors are at least of quadratic order in the integration variables, thus making the integral very challenging to perform (and possibly ill-defined in the small $\epsilon$ limit). It is therefore still a challenge to extend the contour integral picture to the Amplituhedron.

Despite the difficulty of extending the $i\epsilon$ contour, our optimism stems from the well-known observation that the canonical rational function of the {\sl physical} Amplituhedron is given by a sum of {\sl global residues} (see Appendix~\ref{app:GRT}) of the Grassmannian measure constrained on $Y=C\cdot Z$. This was seen in the study of scattering amplitudes. See for instance~\cite{Grassmannian} and references therein. It is therefore reasonable to speculate that a choice of contour would pick up the correct collection of poles whose global residues sum to the expected result, and that different collections of residues that sum to the same result appear as different deformations of the same contour.

We expect these constructions to continue for the $L$-loop Amplituhedron. Namely, we expect that the canonical rational function of the $L$-loop Amplituhedron $\A(k,n,m;l^L)$ to be given by a sum over global residues of the canonical form of the $L$-loop Grassmannian $G(k,k{+}m;l^L)$ constrained by $\mathcal{Y}=Z(\mathcal{C})$ for $\mathcal{C}\in G(k,k{+}m;l^L)$ and $\mathcal{Y}\in\A(k,n,m;l^L)$. A non-trivial example for the 1-loop physical Amplituhedron $\A(1,6;L{=}1)$ is given in~\cite{Bai:2015qoa}. The complete solution to this proposal is still unknown since the canonical form of the loop Grassmannians are mostly unknown.

\section{Integration of canonical forms}
\subsection{Canonical integrals}
In this section we provide a brief outline of an important topic: integration of canonical forms. 

Consider a positive geometry $(X,X_{1,\geq 0})$ with canonical form $\Omega(X,X_{1,\geq 0})$. Naively, integration of the form over $X(\R)$ or $X_{1,\geq 0}$ is divergent due to the logarithmic singularities. However, consider another positive geometry $(X,X_{2,\geq 0})$ which does not intersect the boundary components of $X_{1,\geq 0}$ except possibly at isolated points. We can therefore define the integral of the canonical form of $X_{1,\geq 0}$ over $X_{2,\geq 0}$:
\be
\omega_{X_{2,\geq 0}}(X_{1,\geq 0})=\int_{X_{2,\geq 0}}\Omega(X,X_{1,\geq 0})
\ee
We will refer to such integrals as \defn{canonical integrals}.

In particular, if $X_{1,\geq 0}$ is positively convex, and $X_{2,\geq 0}\subset X_{1,\geq 0}$, and $X_{2,> 0}$ is oriented in the same way as $X_{1,> 0}$ where they intersect, then the integral is positive. 

We observe that $\omega$ is triangulation independent in both arguments. Namely, given a triangulation $X_{2,\geq 0}=\sum_i Y_{i,\geq 0}$, we have
\be
\omega_{X_{2,\geq 0}}(X_{1,\geq 0})=\sum_{i}\omega_{Y_{i,\geq 0}}(X_{1,\geq 0})
\ee
A similar summation formula holds for a triangulation of $X_{1,\geq 0}$.

The simplest example is the familiar polylogarithm functions, which can be reproduced as canonical integrals.
\begin{example}
Let $0<z<1$. Consider the standard simplex $(\P^m, \Delta^m)$ and the simplex $(\P^m,\Delta^m(z))$ with vertices
\be
(1,0,0,0,\ldots,0,1),(1,z,0,0,\ldots, 0,1),(1,z,z,0,&\ldots&,0,1),\ldots, (1,z,z,z,\ldots,z,1),\\
(1,z,z,z,&\ldots&, z,1{-}z)
\ee
where the first row consists of $m$ vectors containing 0,1,2,\ldots, $m{-}1$ components of $z$, respectively, and the second row consists of an extra point. We will refer to $\Delta^m(z)$ as the \defn{polylogarithmic simplex of order $m$}.

Alternatively, the simplex $\Delta^m(z)$ can be constructed recursively in $m$ by beginning with $\Delta^1(z)$ which is simply the line segment $x\in[1-z,1]$ with $(1,x)\in \P^1(\R)$, and observe that
\be\label{eq:polylog_recursive}
\Delta^m(z)=\bigcup_{t\in [0,z]}\{t\}\times\Delta^{m{-}1}(t)
\ee
where $\{t\}\times \Delta^{m{-}1}(t)$ denotes all the points $(1,t,x)\in \P^m(\R)$ with $(1,x)\in\Delta^{m{-}1}(t)$.

We argue by induction on $m\geq 1$ that the integral of the canonical form of $\Delta^m$ over $\Delta^m(z)$ is precisely the polylogarithm $\text{Li}_m(z)$. That is,
\be
\text{Li}_m(z)=\int_{\Delta^m(z)}\Omega(\Delta^m)
\ee
Indeed, for $m=1$, we have
\be
\int_{\Delta^1(z)}\Omega(\Delta^1)\deff \int_{1{-}z}^1 \frac{dt}{t}=-\log(1-z)=\text{Li}_1(z)
\ee
Now suppose $m>1$. Given the recursive construction~\eqref{eq:polylog_recursive}, it is straightforward to show that
\be
\int_{\Delta^m(z)}\Omega(\Delta^m)=\int_0^z \frac{dt}{t}\int_{\Delta^{m{-}1}(t)}\Omega(\Delta^{m{-}1})
\ee
This is the same recursion satisfied by the polylogarithms:
\be
\text{Li}_m(z)=\int_0^z\frac{dt}{t}\text{Li}_{m{-}1}(t)
\ee
which completes our argument.
\end{example}

\subsection{Duality of canonical integrals and the Aomoto form}
Now consider the special case of convex polytopes in projective space. We wish to discuss an important duality of canonical integrals.

Consider convex polytopes $\A_1,\A_2$ in $\P^m(\R)$ satisfying $\A_2\subset \A_1$. It follows that the dual polytopes (see Section~\ref{sec:dual}) satisfy $\A^*_1\subset \A^*_2$. While the polytope pair provides a canonical integral $\omega_{\A_2}(\A_1)$, the dual polytope pair provides a \defn{dual canonical integral} $\omega_{\A^*_1}(\A^*_2)$. We claim that these two integrals are identical:
\be
\omega_{\A_2}(\A_1)=\omega_{\A^*_1}(\A^*_2)
\ee
This is the \defn{duality of canonical integrals}, which can be derived easily by recalling the dual volume formulation of the canonical rational function in Section~\ref{sec:dualpolytopeform}, and arguing that both the canonical integral and its dual can be expressed by the same double integral as follows:
\be
\omega_{\A_2}(\A_1)=\omega_{\A^*_1}(\A^*_2)=\frac{1}{m!}\int_{Y\in\A_2}\int_{W\in \A^*_1}\frac{\lb Yd^mY\rb\lb Wd^mW\rb}{(Y\cdot W)^{m{+}1}}
\ee
where we let $Y\in\P^m(\R)$ and let $W$ denote vectors in the dual projective space.

The double integral appearing on the right is said to be in \defn{Aomoto form}, and was used to express polylogarithmic functions as an integral over a pair of simplices~\cite{aomoto}. For extensions to polytopes, see~\cite{arkani_yuan}.

Furthermore, we expect the duality of canonical integrals to hold for positive geometries in projective space with non-linear boundaries. However, in such case, the canonical rational function must be defined by the volume of the dual region, which for a general positive geometry can be obtained by approximating it by polytopes and taking the limit of their duals. The details are discussed in Section~\ref{sec:beyond}, where it is emphasized that this formulation of the canonical rational function does {\sl not} necessarily match the definition in Section~\ref{sec:positive} except for polytopes.


\section{Positively convex geometries}
\label{sec:convex}
Let $(X,X_{\geq 0})$ be a positive geometry.  The canonical form $\Omega(X,X_{\geq 0})$ may have zeros, and the zero set may intersect $X_{> 0}$.  Furthermore, in some cases the set of poles of $\Omega(X,X_{\geq 0})$ may also intersect $X_{>0}$.  If neither the poles nor the zeros intersect $X_{>0}$, then the form $\Omega(X,X_{\geq 0})$ must be uniformly oriented on each connected component of the interior.  In this case, our sign conventions for $\Res$ (see~\eqref{eq:res}) and orientation inheritance by boundary components (see Section~\ref{app:assumptions1}) ensures that $\Omega(X,X_{\geq 0})$ is {\sl positively} oriented on the interior relative to the orientation of $X_{>0}$, which can be proven easily by induction on dimension. In this case, we say that the canonical form is \defn{positively oriented}, or simply \defn{positive}. 

As we shall explain below, the positivity of the canonical form has some relation to the usual notion of {\sl convexity} of the underlying positive geometry. 
We define a positive geometry $(X,X_{\geq 0})$ to be \defn{positively convex} if its canonical form $\Omega(X,X_{\geq 0})$ is positive.  We remark that if $(X,X_{\geq 0})$ and $(Y,Y_{\geq 0})$ are positively convex geometries then so are $(X,X^-_{\geq 0})$ and $(X \times Y, X_{\geq 0} \times Y_{\geq 0})$.  Also, boundary components of positively convex geometries are again positively convex.

Consider $(\P^2,\A)$ where $\A$ is a convex quadrilateral.  As shown in Section \ref{sec:numerator}, the canonical form $\Omega(\A)$ has four linear poles, and one linear zero.  The poles are the four sides of $\A$, and the zero is the line $L$ that passes through the two intersection points of the opposite sides.  It is an elementary exercise to check that $L$ never intersects $\A$.  Thus the convex quadrilateral $\A$ is a positively convex geometry.  Indeed, we have:
\begin{equation}
\mbox{Every convex projective polytope $(\P^m,\A)$ is a positively convex geometry.}
\end{equation}
This follows immediately from either \eqref{eq:volume} or \eqref{eq:laplace} which give positive integral formulae for $\aOmega(\A)(Y)$ for $Y\in\Int(\A)$.

\begin{figure}
\centering
\includegraphics[width=5cm]{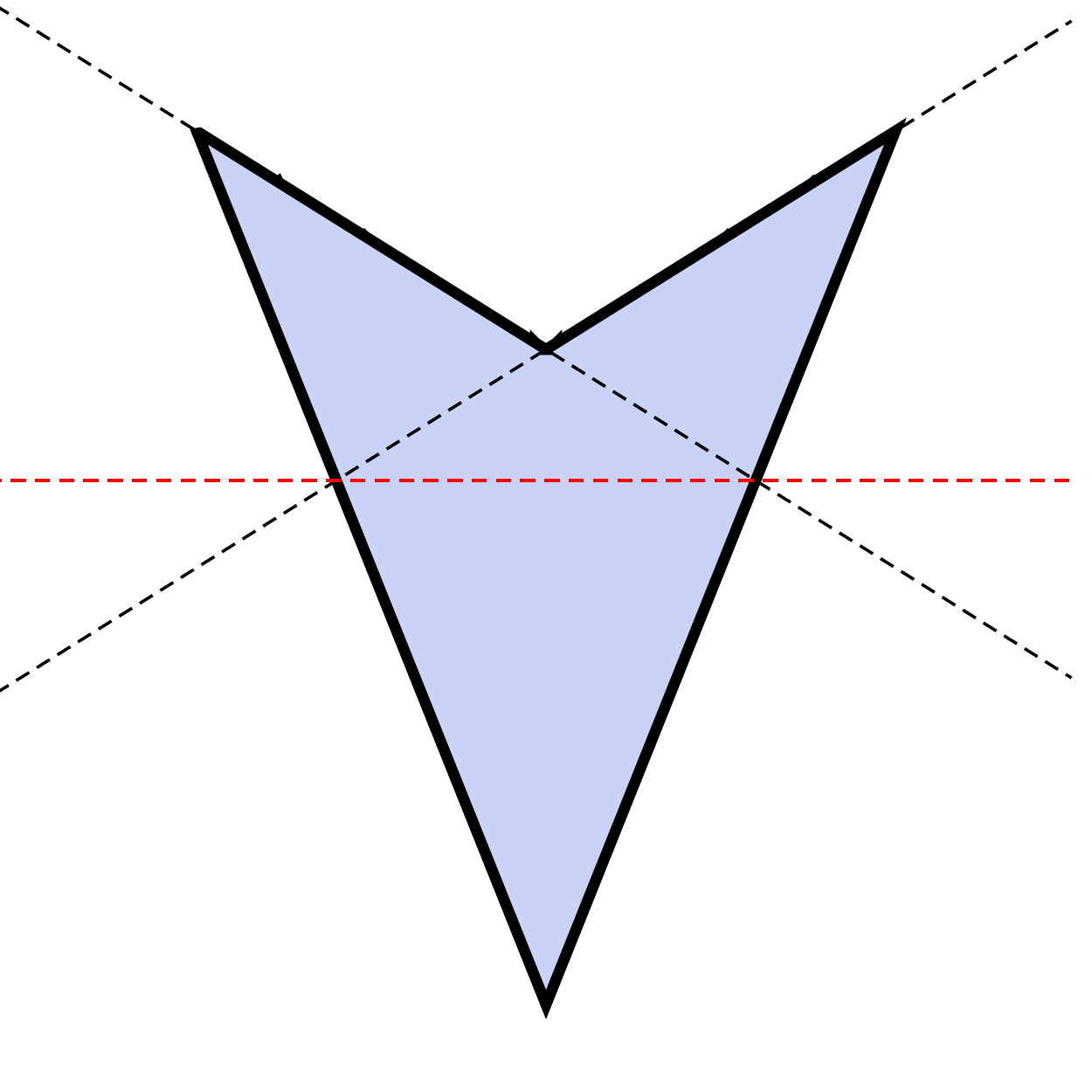}
\caption{A non-convex quadrilateral. The black dashed lines denote the two boundary components that pass through the geometry's interior. The red dashed line $L$ is where the canonical form vanishes.  The form is negative inside the small triangle bounded completely by dashed lines.}
\label{fig:nonconvex}
\end{figure}

On the other hand, suppose instead that $\A'$ is a nonconvex quadrilateral, as shown in Figure~\ref{fig:nonconvex}.  Again $\Omega(\A')$ has four linear poles and one linear zero.  However, in this case, the line $L$ passes through the interior of $\A'$, and indeed two of the poles also pass through the interior.  Thus $\A'$ is not a positively convex geometry.  Nevertheless, note that near most points of the boundary of $\A'$, the form $\Omega(\A')$ is positively oriented.  For a general positively convex geometry, the form is always positively oriented in a neighborhood of every 0-dimensional boundary component.

On the other hand, not every positive geometry (even connected ones) in $\P^2$ that is convex in the usual sense is also positively convex.  For example, consider the following semialgebraic subset of $\P^2(\R)$:
\begin{align*}
\A := &\mbox{\{triangle with vertices $(0,10), (-1,0), (1,0)$\} $\cup$} \\
&\mbox{ \{southern half disk with center $(0,0)$ and radius $1$\}}
\end{align*}
Since both the triangle and the southern half disk are positive geometries, by Section \ref{sec:triangulations}, $\A$ is itself a positive geometry.  One of the poles of $\Omega(\A)$ is along the boundary $C = \{x^2 + y^2 - 1 = 0\}$, and this pole passes through the triangle.  It is then clear that $\Omega(\A)$ is not positive everywhere in the interior of $\A$.  Indeed, the other poles of $\Omega(\A)$ are the two non-horizontal sides $S_1,S_2$ of the triangle, and the numerator of $\Omega(\A)$ is the line $L$ that passes through the intersection points $P_1, P_2$ of $S_1,S_2$ with $C$ above the horizontal axis.  Thus $\Omega(\A)$ is negative exactly within the region bounded by $C$ and $L$.

Let us now discuss some other examples of positively convex geometries.  Recall from Section \ref{sec:one} that a positive geometry in dimension one is a disjoint union of intervals, say $\A = \bigcup_{i=1}^r [a_i,b_i]$ where $a_1 < b_1 < a_2 < b_2 < \cdots < a_r < b_r$.  The canonical form is
$$
\Omega(\A) = \sum_i \left(\frac{1}{x-a_i} - \frac{1}{x-b_i}\right) dx.
$$
Given $x_0 \in (a_j,b_j)$, the positivity of $\Omega(\A)$ is equivalent to the inequality
\begin{align*}
&\sum_i \left(\frac{1}{x_0-a_i} - \frac{1}{x_0-b_i}\right)  \\
&=  \left(\frac{1}{x_0-a_j} - \frac{1}{x_0-b_{j-1}}\right) + \left(\frac{1}{x_0-a_{j-1}} - \frac{1}{x_0-b_{j-2}}\right) + \cdots + \frac{1}{x_0-a_1} \\
&+ \left(-\frac{1}{x_0-b_j} + \frac{1}{x_0-a_{j+1}}\right) + \left(-\frac{1}{x_0-b_{j+1}} + \frac{1}{x_0-a_{j+2}}\right)+ \cdots - \frac{1}{x_0-b_r} > 0
\end{align*}
where each bracketed term is positive.

We do not know of a simplex-like positive geometry that is not positively convex.  Indeed all the the positive geometries in Section \ref{sec:gensimplices} are positively convex.  For the simplices and generalized simplices in projective spaces that we construct this can be seen directly from the canonical form.  Alternatively, it is easy to see that the boundaries never intersect the positive part.  For the positive parts of Grassmannians, toric varieties, flag varieties and cluster varieties, one may argue as follows.  For each such positive geometry $(X,X_{\geq 0})$ of dimension $D$, there is a complex torus $T = (\C^*)^D \subset X$ such that the positive part $X_{>0}$ can be identified with the positive part $\R_{>0}^D \subset (\C^*)^D$, and furthermore the canonical form $\Omega(X_{\geq 0})$ can be identified with the standard holomorphic and non-vanishing form $\Omega_T$ on $T$.  If $x_1,\ldots,x_D$ are coordinates on $T$, then $\Omega_T = \prod_i d\log x_i$.  Since $\Omega(X_{\geq 0})$ is holomorphic and non-vanishing on $X_{>0}$, it follows that $(X,X_{\geq 0})$ is positively convex.  Note that for toric varieties, there is a unique choice of torus $T$, while for Grassmannians, flag varieties, and cluster varieties, there are many choices for the torus $T$.

Finally, we make the following conjecture:
\begin{conjecture}\label{conj:posconvex}
The Amplituhedron $\A(k,n,m;2^L)$ is a positively convex geometry for even $m$.
\end{conjecture}
The cases $m = 1,2$ (for $L=0$) are discussed in Section \ref{sec:Amplituhedron12}. Conjecture \ref{conj:posconvex} would follow immediately from the existence of a dual Amplituhedron (see Section~\ref{sec:dualAmp}). Note that the simple example of a Grassmann polytope in Section \ref{sec:Grassmann} is {\it not} positively convex, a fact which we numerically verified.

As we remarked in Section \ref{sec:1loop}, the $k=1$, $L=1$ positive loop Grassmannian $G_{> 0}(1,n;2)$ is also positively convex. 

\section{Beyond ``rational" positive geometries}
\label{sec:beyond}
We now consider some simple observations suggesting that our notion of positive geometry should likely be thought of as a \defn{rational} type of more general positive geometries. As we stressed in the introduction, with our current definitions the disk is {\sl not} a proper positive geometry; it does not have zero-dimensional boundaries and its associated canonical form vanishes. But of course we should be able to think of the disk as a polygon in the limit of an infinite number of vertices. In fact, let us consider more generally the limit of the polygon form as the polygon boundary is taken
to approximate a conic $Q_{IJ} Y^I Y^J = 0$. Without doing any
computations, it is easy to guess the structure of the answer. In the
smooth limit, the form should only have singularities when $Y$ sits on
the conic, i.e. when $YY\cdot Q \to 0$. But then by projective weights,
the form must look like
\begin{equation}
\Omega'(\A) = C_0 \times \frac{\langle Y d Y d Y \rangle (\det{Q})^{3/4}}{(YY\cdot Q)^{3/2}}
\end{equation}
for some constant $C_0$. Here we added an apostrophe since $\Omega'(\A)$ is not the canonical from defined in Section~\ref{sec:positive}. We will shortly calculate the form and
determine that the constant is $C_0= \pi$. Note that the disappearance of
the ``zero-dimensional boundaries" in the limit is associated with
another novelty: the form is not rational, but has branch-cut
singularities. This clearly indicates that the notion of positive
geometry we are working with needs to be extended in some way; what we
have been seeing are the ``rational parts" of the forms, which indeed
vanishes for the circle, but there is a more general structure for the
forms associated with more interesting analytic structures.

Let us now compute the form for the region bounded by a conic, which without loss of
generality we can take to be the unit disk $\mathcal{D}^2$. We can take the external data
to have the form

\begin{equation}
Z(\theta) = \left(\begin{array}{c} 1 \\ \cos \theta \\ \sin
\theta \end{array} \right)
\end{equation}
for $0\leq \theta<2\pi$.

Let us first do the computation in the most obvious way using approximation by triangulation. Consider a polygon $\A_N$ with vertices $Z(\theta_i)$ where $\theta_i=i \delta\theta$ for $i=0,\ldots, N{-}1$ and $\delta\theta\deff 2\pi/N$. Let $Z_*$ be an arbitrary vertex not on the circle. Then the canonical rational function for the polygon can be triangulated by
\be\label{eq:polygonApprox}
\aOmega(\A_N) = \sum_{i=0}^{N{-}1}[Z_*,Z(\theta_i),Z(\theta_{i{+}1})]
\ee
where the $i$'s are labeled mod $N$.

Now consider the canonical form for one of these triangles with corners
$Z_*, Z(\theta), Z(\theta + \delta \theta)$ for some small $\delta\theta$:
\begin{equation}
\frac{1}{2}\frac{\langle Z_* Z(\theta) Z(\theta + \delta \theta)
\rangle^2}{\langle Y Z_* Z(\theta) \rangle \langle Y Z_* Z(\theta +
\delta \theta) \rangle \langle Y Z(\theta) Z(\theta + \delta \theta)
\rangle} = \frac{1}{2}\delta \theta \frac{\langle Z_* Z(\theta) 
\dot{Z}(\theta) \rangle^2}{\langle Y Z(\theta)  \dot{Z}(\theta)
\rangle \langle Y Z_* Z(\theta) \rangle^2}
+O(\delta\theta^2)
\end{equation}
where the dot denotes differentiation in $\theta$. 

Thus, in the limit of large $N$,~\eqref{eq:polygonApprox} becomes:
\be
\aOmega'(\mathcal{D}^2)=\lim_{N\rightarrow \infty}\aOmega(\A_N) = \frac{1}{2}\int_0^{2\pi} d \theta \frac{\langle Z_* Z \dot{Z}\rangle^2}{\langle
Y Z_* Z \rangle^2 \langle Y Z \dot{Z} \rangle}
\ee

Note that very nicely this expression is invariant under any co-ordinate change in $\theta$. This tells us that the
result only depends on the shape of the circle, not on the particular
way in which it is approximated by a polygon. (Of course this is
trivially expected from the expression of the form as the area of the
dual polygon which approximates the dual conic, and we will perform
the computation in that way momentarily).

Let us put $Z_*=(1,0,0)$ and $Y=(1,x,y)$. We are left with
\begin{equation}
\aOmega'(\mathcal{D}^2) = \frac{1}{2}\int_0^{2\pi} d \theta \frac{1}{(1 - x \cos \theta - y \sin \theta)(y \cos\theta - x \sin \theta)^2}
\end{equation}
Note that for real $(x,y)$, the factor $(y \cos \theta - x \sin \theta)$ inevitably has a zero on the integration contour. This is easy to understand: the factor $\langle Y Z_* Z \rangle$ inevitably has a zero when $Y,Z_*,Z$ are collinear, which for any $Y$ on the disk's interior must occur for some $Z$ on the circle. Of course since the final form is independent of $Z_*$, this singularity is spurious. In order to conveniently deal with this, we will give $y$ a small imaginary part; by a rotation we can put the real part of $y$ to zero, so $y=-i\epsilon$. If we further put $z = e^{i \theta}$, we are left with a contour integral around the origin
\begin{equation}
\aOmega'(\mathcal{D}^2)=-i \oint dz \frac{4 z^2}{(x z^2 - 2 z +x)(x (z^2 - 1) -\epsilon(1 + z^2))^2}
\end{equation}
This integral can be trivially performed by residues. The poles associated with the second factor in the denominator are both pushed away from the circle, assuming that $Y$ lies in the circle (so $|x|<1$). 
Note that the product of the roots of the first factor is $1$, so one of the poles of the first factor is inside the circle and the other is outside. There is thus only a single pole inside the circle, and we obtain 
\begin{equation}
\aOmega'(\mathcal{D}^2)=\frac{\pi}{(1 - x^2-y^2)^{3/2}}
\end{equation}
where we have put the $y$ component back in.


We can of course also obtain the same result as a volume integral~\eqref{eq:volume} over the dual of the unit disk $\mathcal{D}^{2*}$ which is the unit disk in the dual space. Again, by putting $Y=(1,x,0)$ for $|x|<1$, and the integration variable $W=(1,r\cos\theta,r\sin\theta)$, this gives us the integral:
\begin{equation}
\Vol(\mathcal{D}^{2*})=\int_0^{2\pi} \int_0^1\frac{r dr d \theta}{(1 + r x {\rm cos} \theta)^3} = \frac{\pi}{(1 - x^2)^{3/2}} 
\end{equation}
which can be performed by first integrating over $r$, then doing the angular integral by residues via $z=e^{i \theta}$. We can re-express this in projective terms to precisely get 
\begin{equation}
\aOmega'(\mathcal{D}^2) = \pi \frac{ (\det Q)^{3/4}}{( YY\cdot Q )^{3/2}}
\end{equation}
as anticipated.

Let us next look at an obvious extension to the ``pizza slice" geometry $\mathcal{T}(\theta_1,\theta_2)$ from Example~\ref{ex:pizza}.  Note that the the integrand is exactly the same as in our original triangulation above, but we will instead be integrating over the arc $A(\theta_1,\theta_2)$ from $z_{1} = e^{i \theta_1}$ to $z_2=e^{i\theta_2}$, representing the intersetion of the pizza slice boundary with the circle. The reference $Z_*$ is still the same. We will not need the regulator so we will put $\epsilon \to 0$. 
Now, the integrand is
\begin{equation}
I(z) \deff\frac{1}{x^3}\frac{4 iz^2}{(z-r_+)(z-r_-)(z^2 - 1)^2}, \, \text{ with } r_{\pm} = \frac{1 \pm \sqrt{1 - x^2}}{x}
\end{equation}
Note that it has non-vanishing residues at $z=r_{\pm}$, but the residues at $z=\pm 1$ vanish. Instead we have double-poles at $z=\pm 1$. It is therefore natural to split this integrand into two parts: one part that only has simple poles with the correct residues at $z=r_{\pm}$, and another piece that does not have any simple poles at all, that matches the rest of the singularities. Thus we can write 
\begin{equation}
I = I^{{\rm log}} + I^{{\rm rational}}
\end{equation}
with
\begin{equation}
I^{{\rm log}}\deff \frac{i}{2(1 - x^2)^{3/2}} \left(\frac{1}{z - r_+} - \frac{1}{z - r_-}\right), \, \, 
I^{{\rm rational}} \deff -\frac{ 	i}{x^2(1 - x^2)} \frac{(x + 2 z + x z^2)}{(1 - z^2)^2}
\end{equation} 
The term $I^{{\rm log}}(z)$ has logarithmic singularities and (naturally) integrates to give a logarithm, and we get 
\begin{equation}
\frac{i}{2(1 - x^2)^{3/2}} {\rm log}\left(\frac{(z_2 - r_+)(z_1 - r_-)}{(z_1 - r_+)(z_2- r_-)}\right)
\end{equation}

We can express this more elegantly in projective terms as 
\begin{equation}
\aOmega^{{\rm log}}(\mathcal{T}(\theta_1,\theta_2)) = \frac{i ({\rm det} \, Q)^{3/4}}{2(YY\cdot Q )^{3/2}} \times {\rm log} [L_+,L_-,P_1,P_2]
\end{equation}

Here $L_{\pm}$ are the corners of the pizza slice on the conic $Q$, and $P_{1,2}$ are the two points on the conic that are on the tangent lines passing through $Y$. (Note that when $Y$ is outside the conic these tangent points are real, while when $Y$ is inside the conic they are complex). We can think of these as four points on the $\P^1$ defined by the conic itself, and $[a,b,c,d]$ is the cross-ratio of these four points on the conic. 



Now let us look at $I^{{\rm rational}}(z)$; since by construction it has vanishing residues, it must be expressible as the derivative of a rational function; indeed $I^{{\rm rational}} = \frac{-i}{x^2 (1-x^2)} \partial_z  \left(\frac{(1 + x z)}{(1 - z^2)}\right)$. Hence this term integrates to a rational function of $Y$, which of course reproduces our usual canonical form associated with the pizza-slice geometry, written in projective terms as 
\begin{equation}
\aOmega^{{\rm rational}}(\mathcal{T}(\theta_1,\theta_2)) = \frac{Y \cdot L}{(YY\cdot Q) (X \cdot W_1)(X \cdot W_2)}
\end{equation}
where $L$ is the line that kills the two ``bad" singularities as we have seen a number of times in our earlier discussion, and $W_{1,2}$ define the two linear boundaries of the pizza.

We can draw a few conclusions from these simple computations. To repeat, by the definitions of this paper, the circle should be assigned a vanishing canonical form, and this is intuitively reasonable, since it does not have any zero-dimensional boundaries so we cannot have an associated form with logarithmic singularities on all boundaries. But clearly the circle is a limiting case of a series of positive geometries, and the limit evades this logic in a nice way: the form simply is not rational, but has a branch cut singularity. Moving on to the pizza slice, our definitions {\it do} call this a positive geometry and give it an associated rational form, but (as with the circle) we clearly would not get the same result as we would from the continuous limit of a polygon. Here we observed that the pizza slice form naturally splits into two pieces: a ``rational" part, which precisely matches the form naturally associated with the positive geometry in our sense, plus a new ``transcendental" piece. 

The integrals associated with the circle and pizza slices  were simple enough to tempt us into direct computations, but the way in which the results split into rational and transcendental pieces was somewhat mysterious. Is this split canonical? And why did the ``rational part" of this computation magically reproduce our usual definition of the (rational) canonical form? 

We will answer these questions by doing the computation again, more conceptually,  not just for the case of a circle but for any positive geometry $\A$ bounded by a parametrized curve and two  lines emanating from a point $Z_*$. Suppose we have some 
polynomially parametrized curve $Z(t)$, where $t$ is some reparametrization of $\theta$. We will then consider the integral 
\begin{equation}
\aOmega(\A)=\int_{t_1}^{t_2} d t F(t),  \, {\rm with} \, F(t) = \frac{\langle Z_* Z \dot{Z} \rangle^2}{2\langle Y Z \dot{Z} \rangle \langle Y Z_* Z \rangle^2}
\end{equation}
Let us first examine the singularities of a general function 
\begin{equation}
F(t) = \frac{R(t)}{Q(t) P(t)^2} 
\end{equation} 
where $R(t),Q(t),P(t)$ are generic polynomials. We will assume that the degrees are such that $F$ decays at least as fast as $t^{-2}$ at large $t$.  
Let $q_i$ be the roots of $Q(t)$, and $p_\alpha$ be the roots of $P(t)$. Clearly, $F(t)$ has double-poles as $t \to p_\alpha$. Looking at the singular pieces as $t \to p_\alpha$ we have 
\begin{equation}
F(t \to p_\alpha) = \frac{s_\alpha}{(t - p_\alpha)^2} + \frac{r_\alpha}{t - p_\alpha} + \cdots
\end{equation}
with
\begin{equation}
s_\alpha = \frac{R(p_\alpha)}{Q(p_\alpha) \dot{P}(p_\alpha)^2}, \;\;\;\; r_\alpha = s_\alpha \left(\frac{\dot{R}}{R}(p_\alpha) - \frac{\dot{Q}}{Q}(p_\alpha) -  \frac{\ddot{P}}{\dot{P}}(p_\alpha)\right)
\end{equation}
Similarly we have simple poles for $F(t)$ as $t \to q_i$, 
\begin{equation}
F(t \to q_i) = \frac{r_i}{t - q_i} + \cdots,  \;\;\;\; {\rm with} \, r_i = \frac{R(q_i)}{\dot{Q}(q_i) P(q_i)^2}
\end{equation}

Thus it is natural to express $F(t)$ as a sum of two pieces, $F^{dp}(t)$ with only double-poles, and $F^{sp}(t)$ with only simple poles, 
\begin{equation}
F(t) = F^{dp}(t) + F^{sp}(t)
\end{equation}
where 
\begin{equation}
F^{dp}(t) = \sum_\alpha \frac{s_\alpha}{(t - p_\alpha)^2}, \, F^{sp} = \sum_\alpha \frac{r_\alpha}{t - p_\alpha} + \sum_i \frac{r_i}{t - q_i}
\end{equation}

When we integrate $F(t)$, this decomposition canonically divides the result into a ``rational" and ``logarithmic" part just as we saw in our pizza-slice example: 
\begin{equation}
\int_{t_1}^{t_2} dt F(t) = \aOmega^{{\rm rational}}(\A) + \aOmega^{{\rm log}}(\A)
\end{equation}
where 
\begin{eqnarray}
\aOmega^{{\rm rational}}(\A) &\deff & \sum_\alpha s_\alpha \left(\frac{1}{t_1 - p_\alpha} - \frac{1}{t_2 - p_\alpha}\right) \nonumber \\ \aOmega^{{\rm log}}(\A) &\deff & \sum_\alpha r_\alpha {\rm log} \left(\frac{t_2 - p_\alpha}{t_1 - p_\alpha}\right) + \sum_i r_i  {\rm log} \left(\frac{t_2 - q_i}{t_1 - q_i}\right)
\end{eqnarray}

Let us now specialize to the particular case of our $F(t)$, and look first at the rational part. We find
\begin{equation}
\aOmega^{{\rm rational}}(\A) = \sum_{ p_\alpha } \frac{\langle Z_* Z \dot{Z} \rangle^2}{2\langle Y Z \dot{Z} \rangle \langle Y Z_* \dot{Z} \rangle^2} \left(\frac{1}{t_1 - p_\alpha} - \frac{1}{t_2 - p_\alpha} \right)
\end{equation}
where the $p_\alpha$ are roots of $\langle Y Z_* Z(t) \rangle = 0$. Remarkably, we can recognize this sum over roots as a push-forward onto $\A$! Note first that since $\langle Y Z_* Z \rangle = 0$, we can 
write $Y = Z_* + u Z$ for some scalar $u$, and thus 
\begin{equation}
\aOmega^{{\rm rational}}(\A) = \sum_{p_\alpha} \frac{1}{2\langle Z_* Z \dot{Z} \rangle} \frac{1}{u^2} \left(\frac{1}{t_1 - p_\alpha} - \frac{1}{t_2 - p_\alpha} \right)
\end{equation}
Now, clearly, the map 
\begin{equation}
(t,u)\rightarrow \Phi(t,u) \deff Z_* + u Z(t) 
\end{equation}
is a morphism from $[t_1,t_2]\times [0,\infty]\subset \P(\R)^2$ to $\A$. The canonical form on $(t,u)$ space is 
\begin{equation}
\Omega([t_1,t_2]\times [0,\infty]) = dt \left(\frac{1}{t - t_2} - \frac{1}{t - t_1} \right) \times \frac{du}{u}
\end{equation}
Since $\langle Y d Y d Y \rangle = 2u du dt \langle Z_* Z \dot{Z} \rangle$, the push-forward by $\Phi$ gives us that 
\begin{equation}
\Phi_*\left(\Omega([t_1,t_2]\times[0,\infty ])\right)
= \langle Y d Y dY \rangle \sum_{p_\alpha} \frac{1}{2\langle Z_* Z \dot{Z} \rangle} \frac{1}{u^2} \left(\frac{1}{t_1 - p_\alpha} - \frac{1}{t_2 - p_\alpha} \right)
\end{equation}
which we immediately recognize as $\Omega^\text{rational}(\A)$. Hence, by Heuristic~\ref{heuristic}:
\be
\Omega(\A)=\Omega^{\text{rational}}(\A)
\ee

Let us next look at the logarithmic part of the integral. Here too there is a small surprise: the residues $r_\alpha$ as $t \to p_\alpha$ vanish:
\begin{equation}
r_\alpha = s_\alpha \left(2 \frac{\langle Z_* Z \ddot{Z} \rangle}{\langle Z_* Z \dot{Z} \rangle} - \frac{\langle Y Z \ddot{Z} \rangle}{\langle Y Z \dot{Z} \rangle} - \frac{\langle Y Z_* \ddot{Z}\rangle}{\langle Y Z_* \dot{Z} \rangle}\right) = 0, \, {\rm when} \, \langle Y Z_* Z \rangle = 0
\end{equation}
which can easily be seen by putting $Y = Z_* + u Z$. The logarithmic term then simply becomes 
\begin{equation}
\aOmega^{{\rm log}}(\A) = \sum_{q_i } \frac{\langle Z_* Z \dot{Z} \rangle^2}{2\langle Y Z \ddot{Z} \rangle\langle Y Z_* Z \rangle^2 } {\rm log}\left(\frac{t_2 - q_i}{t_1 - q_i}\right)
\end{equation}
where $q_i$ are the roots of $ \langle Y Z(t) \dot{Z}(t) \rangle = 0$. Note that since $(Z(t) \dot{Z}(t))$ is the line tangent to the parametrized curve at $t$, the roots of $\langle Y Z \dot{Z}\rangle =0$ are exactly the points on the curve whose tangents pass through $Y$. Also note that the prefactors of the logarithms are in general
algebraic functions of the variables defining the parametrized curve
of the geometry.

These observations clearly suggest that while the notion of ``positive geometry" we have introduced in this paper is perfectly well-defined and interesting in its own right, it should be generalized in an appropriate way to cover the case of possibly transcendental canonical forms for a more complete picture.

\section{Outlook}
\label{sec:conclusion}
We have initiated a systematic investigation of (pseudo-)positive geometries and their associated canonical forms. These concepts arose in the context of the connection between scattering amplitudes and the positive Grassmannian/Amplituhedron, but as we have seen they are also natural mathematical ideas worthy of study in their own right. Even our preliminary investigations have exposed many new venues of exploration. 

At the most basic level, we would like a complete understanding of positive geometries in projective spaces, Grassmannians, and toric, cluster and flag varieties. A complete classification in projective space seems within reach. In this case, amongst other things we need to demand that the boundary components have geometric genus zero, and that the boundary components of the boundary components have geometric genus zero, and so on.

Perhaps most pressingly we would like to prove that the Amplituhedron itself is a positive geometry! A proof of this fact seems within reach using familiar facts about the Amplituhedron. There is first a purely geometrical problem of showing that the Amplituhedron can be triangulated by images of ``cells" of the positive loop Grassmannian under the Amplituhedron map $Z$. With this in hand, we simply need to prove that the images of these cells are themselves positive geometries. At this moment, ``cells" refers to positroid cells at tree level; at loop level, we suspect that a generalization exists, which we showed for $k=1,L=1,m=2$.

We have also defined the notion of a ``positively convex" geometry, where the canonical form has no zeros or poles on the interior. As we have stressed, this is a highly non-trivial property of the canonical form for Amplituhedra of even $m$, which is not manifest term-by-term in the computation of the form using triangulations (e.g. BCFW). Finding more examples of positively convex geometries in more general settings should shed more light on these fascinating spaces. It would be particularly interesting to classify positively convex geometries in the context of general Grassmann polytopes; these are likely to exist given that cyclic polytopes are a special case of general convex polytopes.

The most important question is to give a natural and intrinsic way of defining the canonical form. Two directions here seem particularly promising--Is there a generalization of the Newton-polytope ``push-forward" method from polytopes to general positive geometries? The examples we have for $m=2$ Amplituhedra are encouraging but we do not yet have a general picture. Similarly, does the ``volume of dual polytope" or ``contour integral over auxiliary geometry" ideas extend beyond polytopes to general positive geometries?  We again gave analogs of these constructions for some Amplituhedra but do not understand the general structure. It would again be particularly interesting to look for such a description for general Grassmann polytopes.

Finally, we have a number of indications suggesting that our notion of
positive geometry should be thought of as a special case of
``rational" positive geometry, which must be extended in some way to
cases where the canonical form is not rational. We encountered
transcendental forms in two places--in looking at the most obvious
``dual Grassmannian integral" representation for the very simplest
Amplituhedron forms, and in the ``continuous" limit of polygons. There are also
other settings where such an extension seems called for--for instance,
there exist deformations of the positive Grassmannian/Amplituhedron
representation of scattering amplitudes to continuous helicities for
the external particles, that are natural from the point of view of
integrable systems~\cite{spectral}, which modify scattering forms away from being
rational. Demanding the existence of rational canonical forms on
positive geometries, with logarithmic singularities recursively
defined by matching lower-degree canonical forms on lower dimensional
boundaries, has led us to the simultaneously rigid and rich set of
objects we have studied in this paper. So, what replaces this
criterion for forms which are not rational? Finding the natural home
in which to answer this question, and identifying the associated
generalization beyond  rational positive geometries, is a fascinating
goal for future work.

\section*{Acknowledgments}
We would like to thank Michel Brion, Livia Ferro, Peter Goddard, Song He, Nils Kanning, Allen Knutson, Tomasz Lukowski, Alexander Postnikov, David Speyer, Hugh Thomas, Jaroslav Trnka, Lauren Williams and Ellis Yuan for stimulating discussions. NA-H was supported by the DOE under grant DOE DE-SC0009988. YB was supported by NSERC PGS D and Department of Physics, Princeton University. TL was supported by NSF grant DMS-1464693 and a Simons Fellowship.

\appendix

\section{Assumptions on positive geometries}
\label{app:assumptions}
In this section we discuss several technical assumptions needed for positive geometries. We will let $(X,X_{\geq 0})$ denote an arbitrary positive geometry throughout.
\subsection{Assumptions on $X_{\geq 0}$ and definition of boundary components}
\label{app:assumptions1}
Let $X_{>0}$ be the interior of the subspace $X_{\geq 0} \subset X(\R)$.  Then $X_{>0}$ is also a semialgebraic set.  We assume that $X_{>0}$ is an open oriented real submanifold of $X(\R)$ of dimension $D$, and that $X_{\geq 0}$ is the closure of $X_{>0}$.  In particular, $X(\R)$ must be a dimension $D$ real algebraic variety. 
If $X_{>0}$ has multiple connected components, one may have many choices of orientations.

Let $\partial X_{\geq 0} \deff X_{\geq 0} \setminus X_{>0}$.  Let $\partial X = \overline{ \partial X_{\geq 0}}$ denote the \defn{Zariski closure} in $X$, defined to be
\be
\partial X = \{x \in X \mid p(x) = 0 \text{ if $p(y) = 0$ for all $y \in \partial X_{\geq 0}$} \}
\ee
where $p$ denotes a homogeneous polynomial. In other words, if a polynomial vanishes on $\partial X_{\geq 0}$, then it also vanishes on $\partial X$, and $\partial X$ is the largest set with this property.  The set $\partial X$ is a closed algebraic subset of $X$.  We let $C_1, C_2, \ldots, C_r$ denote the codimension one (that is, dimension $D{-}1$) irreducible components of $\partial X$.  Define $C_{i, \geq 0}$ to be the closure of the interior of $C_i \cap \partial X_{\geq 0}$ as a subset of the real algebraic variety $C_i(\R)$.

The \defn{boundary components} of $(X, X_{\geq 0})$ are $(C_i, C_{i, \geq 0})$ for $i = 1,2,\ldots,r$.  It is clear that $C_{i, \geq 0}$ is a semialgebraic set.  Axiom (P1) requires that $C_{i,\geq 0}$ has dimension $D{-}1$.

We now define an orientation on $C_{i,\geq 0}$. Let $U\subset \R^{D{-}1}\times\R$ be an open set so that $(x,z)\in U\cap (\R^{D{-}1}\times \R_{\geq 0})$ is a local chart for $X_{\geq 0}$ with $z=0$ mapped to the boundary. We assume that Euclidean charts are oriented in the standard way.
\begin{itemize}
\item
For $D=1$: If the chart is orientation preserving, then we assign $+1$ orientation to the point $C_{i,\geq 0}$; otherwise, we assign $-1$.

\item
For $D>1$: We can always choose a chart that is orientation preserving. The boundary therefore inherits locally the standard orientation of the first $D{-}1$ components.

\end{itemize}

We remark that ``reversing the orientation of $X_{>0}$" maps positive geometries to positive geometries.

\subsection{Assumptions on $X$}\label{app:assumptions2}

We must impose geometric conditions on the complex variety $X$ for Axiom (P2) to make sense.  For the applications we have in mind, $X$ may be singular.  We will usually make the assumption that $X$ is a normal variety, so in particular the singular locus is of codimension two or more.  A rational $D$-form $\omega$ on $X$ is defined to be a rational $D$-form on the smooth locus $X^\circ \subset X$.  If $C \subset X$ is a codimension one subvariety, then $X^\circ \cap C$ is open and dense in $C$, so the residue $\Res_C \Omega$ (see Appendix~\ref{app:residue}) makes sense as a rational $(D-1)$-form on $C$.

Axiom (P2) guarantees the following additional property: $X$ must not have nonzero holomorphic $D$-forms, for otherwise there is no possibility for $\Omega(X,X_{\geq 0})$ to be unique.  If $X$ satisfies this property, then the uniqueness assumption in Axiom (P2) is immediate: If two forms $\Omega_1(X,X_{\geq 0})$ and $\Omega_2(X,X_{\geq 0})$ satisfy the stated condition then their difference will have no residues and is thus a holomorphic $D$-form on $X$, which must vanish everywhere.  In practice, $X$ is typically also a rational variety: $X$ has an (Zariski) open subset isomorphic to an open subset of affine space.

While we will implicitly assume that $X$ is a normal variety throughout this work, there do exist examples where $X$ is non-normal but deserves to be called a positive geometry. Indeed, some of the conjectural positive geometries in this work may require loosening some of our technical assumptions.  We discuss some explicit examples in Sections \ref{sec:degenerate} and \ref{sec:genpoly}.

\subsection{The residue operator}\label{app:residue}

The residue operator $\Res$ is defined in the following way. Let $\omega$ be a meromorphic form on $X$. Suppose $C$ is an irreducible subvariety of $X$ and $z$ is a holomorphic coordinate whose zero set $z=0$ locally parametrizes $C$. Let us denote the other holomorphic coordinates collectively as $u$. Then a \defn{simple pole} of $\omega$ at $C$ is a singularity of the form
\be
\omega(u,z)= \omega'(u)\wedge \frac{dz}{z}+\cdots
\ee
where the $\cdots$ denotes terms smooth in the small $z$ limit, and $\omega'(u)$ is a non-zero meromorphic form defined locally on the boundary component. We define
\be\label{eq:res}
\Res_C\omega\deff \omega' \text{       (locally)}
\ee
If there is no such simple pole, then we define the residue to be zero.

\section{Near-equivalence of three notions of signed triangulation}
\label{app:proofTri}
We prove \eqref{eq:signed} by induction on $D = \dim X $.
The implications \eqref{eq:signed} clearly hold when $D = 0$.  Suppose now $D > 0$.  Let $\{X_{i,\geq 0}\}$ be an interior triangulation of the empty set.  Let $C \subset X$ be an irreducible subvariety of codimension one.  Let $(C,C_{i,\geq 0})$ be the boundary component of $(X,X_{i,\geq 0})$ along $C$, where we set $C_{i,\geq 0} = \emptyset$ if such a boundary component does not exist.  We claim that $\{C_{i,\geq 0}\}$ is an interior triangulation of the empty set.    Let $y \in \bigcup_i C_{i,\geq 0}$, and assume that $y$ does not lie on the boundary of any of the $\{C_{i,\geq 0}\}$.  A dense subset of such points (in $C(\R)$) will in addition not lie on any of the boundary components $C' \neq C$ of any of the $\{X_{i,\geq 0}\}$, and we assume $y$ is chosen to lie in this dense subset.  By assumption, \eqref{eq:signtriang} holds for all points $x \in  U \setminus C$, where $U$ is an open ball in $X(\R)$ containing $y$.  By shrinking $U$ if necessary, we may assume that $U \cap C$ is an open disk, and we fix orientations of both $U$ and $U\cap C$.  The submanifold $U \cap C$ divides $U$ into $U^+$ and $U^-$, which can be labeled so that  
\begin{align*}
&\{i \mid y \in C_{i,>0} \text{ and $C_{i,>0}$ is positively oriented at $y$}\} \\
= &\{i \mid U^+ \subset X_{i,>0} \text{ and $X_{i, >0}$ is positively oriented on $U^+$}\} \\
\cup & \{i \mid U^- \subset X_{i,>0} \text{ and $X_{i, >0}$ is negatively oriented on $U^-$}\}
\end{align*}
and
\begin{align*}
&\{i \mid y \in C_{i,>0} \text{ and $C_{i,>0}$ is negatively oriented at $y$}\} \\
= &\{i \mid U^- \subset X_{i,>0} \text{ and $X_{i, >0}$ is positively oriented on $U^-$}\} 
\\
\cup & \{i \mid U^+ \subset X_{i,>0} \text{ and $X_{i, >0}$ is negatively oriented on $U^+$}\}.
\end{align*}
We are making use of the assumptions on the local behavior of boundary components from Appendix \ref{app:assumptions1}.  The equality 
\begin{align*}
&\{i \mid y \in C_{i,>0} \text{ and $C_{i,>0}$ is positively oriented at $y$}\} \\
= &\{i \mid y \in C_{i,>0} \text{ and $C_{i,>0}$ is negatively oriented at $y$}\}
\end{align*}
then follows from \eqref{eq:signtriang} applied to points in $U^+$ and points in $U^-$ respectively.  It follows that $\{C_{i,\geq 0}\}$ is an interior triangulation of the empty set.  By the inductive hypothesis we have that $\{C_{i,\geq 0}\}$ is a boundary triangulation of the empty set.  Since this holds for all $C$, we conclude that $\{X_{i,\geq 0}\}$ is a boundary triangulation of the empty set. 
%
%

Let $\Omega = \sum_{i=1}^t \Omega(X_{i,\geq 0})$.  
Suppose first that $\{X_{i,\geq 0}\}$ is a boundary triangulation of the empty set.  
Let $C$ be an irreducible subvariety of $X$ of codimension one. Taking the residue of $\Omega$ at $C$, we obtain
\be
\Res_C \Omega = \sum_{i=1}^t\Res_C\Omega(X_{i,\geq 0})=\sum_{i=1}^t\Omega(C_{i,\geq 0}).
\ee
By the inductive hypothesis we have $\sum_{i=1}^t\Omega(C_{i,\geq 0}) = 0$, so we conclude that all residues of $\Omega$ are 0, and thus $\Omega= 0$ and $\{X_{i,\geq 0}\}$ is a canonical form triangulation of the empty set.  Conversely, suppose that $\{X_{i,\geq 0}\}$ is a canonical form triangulation of the empty set.  Then $\Omega = 0$, so $\Res_C \Omega = 0$ for any irreducible subvariety $C \subset X$ of codimension one.  Thus $\{C_{i,\geq 0}\}$ form a canonical form triangulation of the empty set, and by the inductive hypothesis they also form a boundary triangulation of the empty set.  We conclude that $\{X_{i,\geq 0}\}$ is a boundary triangulation of the empty set. 

This completes the proof of \eqref{eq:signed}.

\section{Rational differential forms on projective spaces and Grassmannians}\label{app:projform}

\subsection{Forms on projective spaces}

Let $\eta$ be a rational $m$-form on $\C^{m+1}$, which we write as
\be
\eta = \sum_{I=0}^{m} f_I(Y) dY^0 \wedge \cdots \wedge\widehat{dY^I} \wedge\cdots \wedge dY^{m}
\ee
where $\widehat{dY^I}$ denotes omission.  Here $f_I(Y)$ is a rational function.  Then $\eta$ is the pullback of a rational $m$-form on $\P^m$ if and only if 
\begin{enumerate}
\item[(1)]
$\eta$ is homogeneous of degree 0, that is, $f_I(Y)$ is homogeneous of degree $-m$, and
\item[(2)]
$\langle\eta, \E \rangle= 0$, where $\E = \sum_{I=0}^{m} Y^I \partial_I$ is the radial vector field, also called the Euler vector field.
\end{enumerate}

If we explicitly solve (2), then we learn that 
\be
\eta =m!\cdot g(Y) \sum_{I=0}^{m} (-1)^{I} Y^I dY^0 \wedge \cdots \wedge \widehat{dY^I}\wedge \cdots \wedge dY^{m}
\ee
for some rational function $g(Y)$ of degree $-m{-}1$.  We have the elegant formula  
\be \sum_{I=0}^{m} (-1)^{I} Y^I dY^0 \wedge \cdots \widehat{dY^I} \cdots \wedge dY^{m} =
\frac{1}{m!} \langle Y d^m Y\rangle.
\ee
Thus an arbitrary form on $\P^m$ can be written as
\be
\eta = g(Y) \omega
\ee
for some function $g(Y)$ of degree $-m{-}1$, where we have defined $\omega\deff \lb Yd^mY\rb$

The factor $\omega$ is called the \defn{standard measure} on projective space, and if $\eta$ is the canonical form of some pseudo-positive geometry $(\P^m,\A)$, then $g(Y)$ is called the \defn{canonical rational function} of $(\P^m,\A)$ and is usually denoted $\aOmega(\A)$.

An alternative but equivalent description of a form $\eta$ on projective space is that $\eta$ must be invariant under \defn{local} $\GL(1)$ action on $Y$. That is, for any scalar $\alpha(Y)$ possibly dependent on $Y$, the form must be invariant under $Y\rightarrow Y'=\alpha(Y)Y$. The word ``local" refers to the dependence of $\alpha$ on $Y$. In contrast, if $\alpha(Y)$ is constant, then it would be called a \defn{global} transformation.

We argue that the factor $\omega$ is \defn{locally covariant}. That is, it scales by some power of $\alpha(Y)$, and the exponent of the scaling is its \defn{$\GL(1)$ weight}. Indeed, $dY'=Yd\alpha+\alpha dY$, hence 
\be\label{eq:local}
\lb Y'd^mY'\rb=\alpha^{m{+}1}\lb Yd^m Y\rb
\ee
where contraction with $Y'$ annihilates the $d\alpha$ term.

It follows that $\eta$ is \defn{locally invariant} under $\GL(1)$ (i.e. covariant with weight $0$). Local invariance is needed for $\eta$ to be well-defined on projective space.

Now suppose that $\theta$ is a rational $(m+1)$-form on $\C^{m+1}$ which is homogeneous of degree 0.  To obtain a form on projective space from $\theta$, we define
$
\eta = \langle \theta, \E \rangle.
$
It is not hard to see that $\eta$ is a $m$-form that satisfies (1), (2) from above.  Indeed, since $Y^I$ and $dY^I$ have degree 1, the vector field $\partial_I$ has degree $-1$, so that $\eta$ is homogeneous of degree 0.  Part (2) follows immediately from the fact that differential forms are alternating functions of vector fields, or from the calculation that
\be
\langle dY^0 \wedge \cdots \wedge dY^{m}, \E \rangle = \sum_{I=0}^{m} (-1)^{I} Y^I dY^0 \wedge \cdots \widehat{dY^I} \cdots \wedge dY^{m} = \frac{1}{m!} \langle Y d^m Y\rangle.
\ee
Thus we can canonically pass between $(m{+}1)$-forms, homogeneous of degree 0, on $\C^{m+1}$ and $m$-forms on $\P^m$ by replacing $d^{m+1}Y$ by $\langle Y d^m Y\rangle/m!$.  An alternative way of thinking of contracting against $\E$ is to divide by the measure of $\GL(1) = \C^*$.  In other words, we have the equality
\be
\frac{d^{m+1}Y}{\rm Vol\; GL(1)} = \omega
\ee



\subsection{Forms on Grassmannians}
We now extend our discussion to the Grassmannian $G(k,k{+}m)$.

Let $M(k,k{+}m)$ denote the space of $k \times (k{+}m)$ matrices.  For $Y \in M(k,k{+}m)$, we write $Y_1,Y_2,\ldots,Y_k$ for the rows of $Y$, and $Y_s^I$ for the component at row $s$ and column $I$.  We have an action of $\GL(k)$ on $M(k,k{+}m)$ with quotient $G(k,k{+}m)$.  For $s,t = 1,2,\ldots,k$ define a vector field
$$
\E_{st} = \sum_{I=1}^{k{+}m} Y_{s}^I \frac{\partial}{\partial Y_t^I}.
$$
This is the infinitesimal vector field corresponding to the action of $\exp(te_{st}) \in \GL(k)$, where $e_{st}$ is the $0,1$-matrix with a 1 in row $s$ and column $t$. 

A $km$-form $\eta = \eta(Y)$ on $M(k,k{+}m)$ is the pullback of a form on $G(k,k{+}m)$ if it satisfies the two conditions
\begin{enumerate}
\item[(1)] It invariant under $\GL(k)$.  That is, for any $g \in \GL(k)$, we have $g^* \eta = \eta$.
\item[(2)] We have $\langle \eta, \E_{st} \rangle = 0$ for any $s,t$.
\end{enumerate}
Condition (1) is equivalent to the condition that $\L_{\E_{st}} \eta = 0$ for any $s,t$, where $\L$ denotes the Lie-derivative.  By Cartan's formula, we have $\L_{\E_{st}} \eta = \langle d\eta, \E_{st} \rangle + d(\langle \eta, \E_{st} \rangle)$.  By (2), this is equivalent to the condition
$$
\langle d\eta, \E_{st} \rangle = 0.
$$

Let $d^{k \times (k{+}m)}Y$ denote the natural top form on $M(k,k{+}m)$. If we contract $d^{k \times (k{+}m)}Y$ against all the vector fields $\E_{st}$, we get up to a scalar the $km$-form
\be\label{eq:grassStandardMeasure}
\omega \deff \langle Y_1 Y_2 \cdots Y_k d^{m}Y_1\rangle \langle Y_1 Y_2 \cdots Y_k d^{m}Y_2\rangle \cdots \langle Y_1 Y_2 \cdots Y_k d^{m}Y_k\rangle.
\ee
While this form makes the $\GL(k{+}m)$ symmetry manifest, the $\GL(k)$ symmetry is hidden. Nonetheless, both are guaranteed by our derivation. An alternative representation of the measure making both manifest is given in Section~\ref{sec:wilson}.

An alternative way to think of this form is as
$$
\frac{d^{k\times (k{+}m)}Y}{\rm Vol\;GL(k)}=\omega
$$

As for projective space, we argue that $\omega$ is \defn{locally covariant} under $\GL(k)$ action. Consider an action of the form $Y\rightarrow L(Y)Y$, where $L(Y)\in \GL(k)$ acts on the rows of $Y$. The differential of $Y$ transforms as $dY\rightarrow L(L^{{-}1}dL\;Y+dY)$. Thus, in order to achieve local covariance, $\omega$ must vanish whenever any $dY_s^I$ is replaced by any $Y_{s'}^{I'}$, which is evidently the case. It follows therefore that $\eta$ is also locally invariant under $\GL(k)$, which is necessary for $\eta$ to be well-defined on the Grassmannian.

Furthermore, it is easy to see directly that $\langle \omega, \E_{st} \rangle = 0$ holds: to compute $\langle \omega, \E_{st} \rangle$, we replace the $t$-th factor $\langle Y_1 Y_2 \cdots Y_k d^{m}Y_t\rangle$ by $\langle Y_1 Y_2 \cdots Y_k Y_s d^{m{-}1}Y_t\rangle$, which is clearly 0.

Finally, any form $\eta$ on $G(k,k{+}m)$ can be written in the form
\be
\eta = g(Y)\omega
\ee
for some function $g(Y)$ of degree $-m{-}k$ in the Pl\"ucker coordinates of $Y$. Similar to case for the projective space, $\omega$ is called the \defn{standard measure} on the Grassmannian; and if $\eta$ is the canonical form for some positive geometry $(G(k,k{+}m),\A)$, then $g(Y)$ is called the \defn{canonical rational function} of $\A$ and is usually denoted $\aOmega(\A)$.

\subsection{Forms on $L$-loop Grassmannians}

The preceding discussion can be generalized to $L$-loop Grassmannians. Rather than fleshing out all the details, we simply write down the equivalent of the \defn{standard measure} $\omega$. The \defn{canonical rational function} is defined in a similar fashion. For our purposes, the main positive geometry of interest in the $L$-loop Grassmannian is the $L$-loop Amplituhedron $\A$, whose canonincal rational function $\aOmega(\A)$ is also called the \defn{amplitude}.

Let $\mathcal{Y}=(Y,Y_1,\ldots,Y_L)\in G(k,k+m;\k)$ (see~\eqref{eq:ampspace}). The standard measure $\omega$ on $G(k,k+m;\k)$ is given by:
\be\label{eq:measure}
\omega=
\prod_{s=1}^k\lb Yd^mY_s\rb\cdot \prod_{i=1}^L\prod_{s=1}^{k_i}\lb YY_id^{m{-}k_i}Y_{is} \rb
\ee
where $Y_s$ denotes the rows of $Y$ for $s=1,\ldots k$, and $Y_{is}$ denotes the rows of $Y_i$ for $s=1,\ldots, k_i$ and $i=1,\ldots, L$.

We leave it to the reader to argue that the standard measure is \defn{locally covariant} under the group $\G(k;\k)$ defined in Section \ref{sec:loopgrass}.

\section{Cones and projective polytopes}
\label{app:cones}

Let $V = \R^{m+1}$.  A subset $C \subset V$ is called a convex cone if it is closed under addition and scalar multiplication by $\R_{\geq 0}$.  Thus a convex cone always contains the 0-vector.  We say that a cone $C$ is \defn{pointed} if it does not contain a line.  Alternatively, $C$ is pointed if whenever vectors $v$ and $-v$ both lie in $C$ then we have $v = 0$.

A \defn{polyhedral cone} is the nonnegative real span 
\be
\Cone(Z) \deff \sp_{\R_{\geq 0}}(Z_1,Z_2,\ldots,Z_n) = \left\{\sum_i a_i Z_i \mid a_i \in \R_{\geq 0}\right\}
\ee
of finitely many vectors $Z_1,\ldots,Z_n$.  A rational polyhedral cone $C$ is one such that the generators $Z_i$ can be chosen to have integer coordinates.  The dimension $\dim C$ of $C$ is equal to the dimension of the vector space that $C$ spans.

From now on, let $C$ be a pointed polyhedral cone. 
We will often assume that the vectors $Z_i$ are not redundant, that is, none of them can be removed while still spanning $C$.  In this case, we call the $Z_i$ the generators, or edges, of $C$.  The interior of the cone $C = \Cone(Z)$ is then given by
\be
\Int(C) = \sp_{\R_{> 0}}(Z_1,Z_2,\ldots,Z_n) = \left\{\sum_i a_i Z_i \mid a_i \in \R_{> 0}\right\}
\ee

A face $F$ of $C$ is the intersection $F = C \cap H$ with a linear hyperplane $H \subset V$ such that $C$ lies completely on one side of $H$.  Thus if $H$ is given by $v \cdot \alpha = 0$ for some $\alpha \in \R^{m+1}$, then we must have $C \cdot \alpha \geq 0$ or $C \cdot \alpha \leq 0$.  Every face of $C$ is itself a rational pointed polyhedral cone with generators given by those $Z_i$ that lie on $H$.  Every face $F'$ of $F$ is itself a face of $C$.  A face $F$ is called a \defn{facet} of $C$ if $\dim F = \dim C - 1$.  The 0-vector is the unique face of $C$ of dimension $0$.  We do not allow $C$ to be a face of itself.

A cone $C$ is called \defn{simplicial} if it has $\dim C$ generators.  A \defn{triangulation} of a polyhedral cone $C$ is a collection $C_1,C_2,\ldots,C_r$ of simplicial cones, all with the same dimension as $C$, such that
\begin{enumerate}
\item
We have $C = \bigcup_i C_i$.
\item
The intersection $F = C_i \cap C_j$ of any two cones is a face of both cones $C_i$ and $C_j$. 
\end{enumerate}
We sometimes, but not always, require the generators of $C_i$ to be a subset of the generators of $C$.

\medskip

A \defn{projective polytope} is the image $\A = \Conv(Z)$ of a pointed polyhedral cone $C = \Cone(Z)$ (after removing the origin) in projective space $\P^m$.  The faces (resp. facets) of $\A$ are defined to be the images of the faces (resp. facets) of $C$.  The dimension of a face of $\A$ is one less than the corresponding face of $C$.  By definition, the empty face is a face of $\A$.  The interior $\Int(\A)$ is the image in $\P^m$ of the interior of the cone $C$.

Projective polytopes are the basic objects that we consider in Section \ref{sec:polytope}.  In some cases we will use inequalities to define projective polytopes, and it may not be clear that these inequalities are well-defined on projective space.  
However, these inequalities are well-defined when lifted to cones and $\R^{m+1}$.  Finally, we remark that in Section \ref{sec:polytope}, a projective polytope $\A$ usually refers to a polytope with an orientation, giving the interior of $\A$ the structure of an oriented manifold.  This orientation is always suppressed in the notation.  Facets $F$ of a projective polytope $\A$ acquire a natural orientation from the orientation of $\A$.

A triangulation of a projective polytope is simply the image of a triangulation of the corresponding cone.

\section{Monomial parametrizations of polytopes} \label{app:diffeo}

%
%
%

We shall use the language of oriented matroids, reviewed in Appendix \ref{sec:matroids}.

To simplify the notation, we set $r = m+1$ in this section.  It is convenient to work with cones instead of polytopes.  Let $z_1,\ldots,z_n \in \R^r$ and $Z_1,\ldots,Z_n \in \R^r$ be vectors spanning $\R^r$, and assume that the cone $C = \Cone(Z)$ spanned by $Z_i$ is pointed.  Assume that $z$ and $Z$ define the same oriented matroid $\M$, that is, 
\be \label{eq:chirotope}
\langle z_{i_1} \cdots z_{i_r} \rangle \text{ and } \langle Z_{i_1} \cdots Z_{i_r} \rangle \text{ have the same sign.}
\ee
Let $a \cdot b$ denote the standard inner product of vectors $a,b\in \R^r$.  Consider the map $\phi: \R^r \to \R^r$ given by
\be
\phi(u) = \sum_{i=1}^n e^{ z_i \cdot u} Z_i.
\ee
\begin{theorem}\label{thm:diffeo}
The map $\phi$ is a diffeomorphism of $\R^r$ onto the interior of the cone $\Cone(Z)$ spanned by $Z_1,Z_2,\ldots,Z_n$.
\end{theorem}

Before proving Theorem \ref{thm:diffeo}, we first deduce a number of corollaries.  Suppose that the $z_i$ are integer vectors.  Let $\tPhi: (\C^\ast)^r \to \C^r$ be the rational map given by
\be
\tPhi(X) = \sum_{i=1}^n X^{z_i} Z_i.
\ee

\begin{corollary}\label{cor:tPhi}
The restriction of $\tPhi$ to $\R_{>0}^r$ is a diffeomorphism of $\R_{>0}^r$ with $\Int(C) \subset \R^r \subset \C^r$.
\end{corollary}
\begin{proof}
The map $\tPhi|_{\R_{>0}^r}$ is the composition of the map $\phi$ with the map $(x_1,\ldots,x_r) \mapsto (\log x_1,\ldots,\log x_r)$.
\end{proof}

The set of vectors $z = \{z_1,z_2,\ldots,z_n\}$ is called \defn{graded} if there exists a vector $a \in \Q^{r}$ such that $a \cdot z_i = 1$ for all $i$.  If $z_i = (1,z'_i)$ where $z'_i \in \ZZ^{r-1}$, then clearly $z$ is graded.  Now define a rational map $\Phi: (\C^\ast)^r \to \P^{r-1}$ by
\be
\Phi(X) = \sum_{i=1}^n X^{z_i} Z_i.
\ee
\begin{corollary}\label{cor:diffeo}
The restriction of $\Phi$ to $\R_{>0}^{r-1}$ is a diffeomorphism of $\R_{>0}^{r-1}$ with $\Int(\A) \subset \P^{r-1}$.
\end{corollary}
\begin{proof}
Since the set $z$ is graded, we can change coordinates so that $z_i = (a,z'_i)$ for $z'_i \in \ZZ^{r-1}$.  By replacing $X_r$ with $X_r^a$, we may assume that $a = 1$, so that $z_i = (1,z'_i)$.

Use coordinates $(t,X,t)$ on $\C \times \C^{r-1} \cong \C^{r}$.  Consider the map $\tPhi:\C^r \to \C^r$ given by
\be
\tPhi(t,X) = \sum_{i=1}^n  t^1 X^{z'_i} \, Z_i = t \left( \sum_{i=1}^n X^{z'_i} Z_i \right).
\ee
Now restrict $\tPhi$ to $\R_{>0}^r$ and apply Corollary \ref{cor:tPhi}.  For any $t \in \R_{>0}$, the image of $\tPhi(t,X)$ in $\P^{r-1}$ is equal to $\Phi(X)$.  Fixing $t = 1$, we obtain the claimed diffeomorphism.
\end{proof}

We now provide two distinct proofs of \ref{thm:diffeo}.

\begin{proof}[First proof of Theorem \ref{thm:diffeo}]
The basic structure of our argument is similar to that in \cite[p.84--85]{Fulton}, but the details are significantly more complicated.  

The case $r =1$ is straightforward.  We assume that $r > 1$, and inductively that the result for $r-1$ is known.

Let us show that the map $\phi$ is a local isomorphism.  We compute the Jacobian of the map.  This is given by the determinant of the $r \times r$ matrix
\be
A(u) = \sum_{i=1}^n e^{\ip{z_i,u}} z_i^T Z_i 
\ee
where $z_i^T$ is the transpose of $z_i$.  For a fixed value of $u \in \R^r$, the vectors $e^{{z_1 \cdot u}} z_1,e^{{z_2 \cdot u}} z_2$, $\ldots,e^{{z_n \cdot u}} z_n$ have the same oriented matroid $\M$.  By the Cauchy-Binet identity the determinant $J(u) = \det(A(u))$ is a sum of the products $\ip{e^{{z_1 \cdot u}}z_{i_1} \cdots e^{{z_n \cdot u}}z_{i_r}} \ip{Z_{i_1} \cdots Z_{i_r}}$, over all $1 \leq i_1 < i_2 < \cdots < i_r \leq n$.  By our assumption \eqref{eq:chirotope}, we obtain $J(u) > 0$.  Thus $\phi$ is a local diffeomorphism.

Let us next show that $\phi$ is one-to-one.  It is enough to show that $\phi$ is one-to-one when restricted to any line $L = \{at+b \mid t \in \R \}$ in $\R^r$.  We have
\be
\phi(at+b) = \sum_{i=1}^n e^{t \,  z_i \cdot a  +  z_i \cdot  b} Z_i = \sum_{i=1}^n \beta_i e^{t \, z_i \cdot a } Z_i,
\ee
where $\beta_i = e^{\ip{z_i,b}} \in \R_{>0}$.  The signs of the vector $({z_1 \cdot a},{z_2 \cdot a}, \ldots, {z_n \cdot a})$ is a signed covector of the oriented matroid $\M(z)$.  It follows that there is a vector $A \in \R^r$ so that $({Z_1 \cdot A},{Z_2 \cdot A}, \ldots, {Z_n \cdot A})$ has the same signs (see Proposition \ref{prop:OM}).  The one-variable function $f(t) := \phi(at+b) \cdot A = \sum_{i=1}^n \beta_i e^{t \, {z_i \cdot a}} {Z_i \cdot A}$ has positive derivative, and is thus injective.  It follows that $\phi$ itself is one-to-one.

Let us show that the image of $\phi$ contains points arbitrarily close to any point on a facet of $C$.  Let $F$ be a facet of $C$, and after renaming let $Z_1,Z_2,\ldots,Z_s$ be the $Z_i$ lying in this facet.  Since $Z$ and $z$ have the same oriented matroid, we see that $z_1,z_2,\ldots,z_s$ are exactly the $z_i$ lying on a facet of $\Cone(z)$.  Thus there exists $a \in \R^r$ such that $z_1 \cdot a = z_2 \cdot a = \cdots = z_s \cdot a = 0$ and $z_i \cdot a < 0$ for $i > s$.


Let $H = \R^{r-1} \subset \R^r$ be a linear hyperplane not containing $a$.  Thus $H$ maps isomorphically onto the quotient $\R^r/\R.a$.  Then for any $b \in H$, we have that 
\be \label{eq:limit}
\phi(at + b) = \sum_{i=1}^s e^{z_i \cdot b}Z_i + \sum_{j=s+1}^n \beta_j e^{t\, z_j \cdot a} Z_j.
\ee
We have $\lim_{t \to \infty} \phi(at+b) = \sum_{i=1}^s e^{{z_i \cdot b}}Z_i$, and since $Z_1,\ldots,Z_s$ span a pointed cone in $\R^{r-1}$, by the inductive hypothesis we know that the map $\phi': H \to F$ given by
\be
\phi_F(b) = \sum_{i=1}^s e^{{z_i \cdot b}}Z_i
\ee
is a diffeomorphism onto the interior of $F$.  It follows that the image of $\phi$ contains points arbitrarily close to any point on a facet of $C$.  

We now argue that for each $p_0 \in \Int(F)$, the image of $\phi$ contains $B_{p_0} \cap \Int(C)$, for some open ball $B_{p_0}$ centered at $p_0$.  Let $b_0 =\phi_F^{-1}(p_0) \in H$.  There exists a neighborhood $U$ of $b_0$, two positive real numbers $\delta > \alpha > 0$, and a vector $Z^* \in \sp(Z_{s+1},\ldots,Z_n)$ such that for any ray $I = [t_0,\infty) \subset \R$ with $t_0$ sufficiently large, we have that for $t \in I$ and $b \in U$, 
\be \label{eq:estimate}
|\phi(at+b) - (\phi_F(b) + e^{-\alpha t} Z^*)| < e^{-\delta t_0}.
\ee
Here, $\alpha = \min_{j \geq s+1}(-{z_j \cdot a})$, and $Z^* = \sum_{j  \mid \alpha = -{z_j \cdot a} } \beta_j Z_j$.

Let $\eta: (\Int(F) + \R_{>0} Z^*) \to \R \times H$ be the map $\eta(p +  \gamma Z^*) = (-(\log \gamma)/\alpha,\phi_F^{-1}(p))$.  The composition $\phi' = \eta \circ \phi$ thus sends $I \times U$ to $\R \times H$.  If $\phi'(t,b) = (s,c)$, then \eqref{eq:estimate} implies that $|t-s| < C \cdot t_0$ and $|b-c| < e^{-\delta t_0}$, for some constant $C$.  Furthermore, the injectivity of $\phi$ implies the injectivity of $\phi'$.  Applying the following claim to appropriate closed (topological) balls inside $I \times U$, we see that the image of $\phi'$ contains $[t',\infty) \times U'$ for some $t' > t$ and $U' \subset U$ a neighborhood containing $b_0$.  

\medskip

\noindent {\bf Claim.} Suppose $ \bar B^r \subset \R^r$ is a closed ball and $f: \bar B^r \to \R^r$ is an injective continuous function, smooth when restricted to either $S^{r-1}$ or $B^r$.  Then $f(S^{r-1})$ separates $\R^r$ into two components, and $f(B^r)$ is a homemorphism onto the bounded component.

\noindent {\it Proof of claim.} The first claim is the Jordan--Brouwer separation theorem.  By the invariance of domain theorem, $f|_{B^r}$ maps $B^r$ homeomorphically to an open subset of $\R^r$.  It is clear that the boundary of $f(B^r)$ is equal to $f(S^{r-1})$.  Thus $f(B^r)$ must equal the bounded component of $\R^r \setminus f(S^{r-1})$.  \qed


\medskip

We have thus shown that the image of $\phi$ contains a neighborhood of every point in $\Int(F)$ (intersected with $\Int(C)$).

Finally, we show that $\phi$ is surjective onto the interior $\Int(C)$.  Let $Z_i$ be an edge of $C$.  We shall show a weak form of convexity: 

\medskip

\noindent {\bf Claim.} The image of $\phi$ intersected with any line parallel to $Z_i$ is either connected or empty.  

\noindent {\it Proof of claim.} Consider a linear projection $\pi: \R^r \to \R^{r-1}$ with kernel equal to the span of $Z_i$, where $Z_i$ is one of the edges of $C$.  The images $\pi(Z_1),\ldots,\pi(Z_{i-1}),\pi(Z_{i+1}),\ldots,\pi(Z_n) \in \R^{r-1}$ has oriented matroid $\M/i$, the contraction of $\M$ by $i$.  Since $Z_i$ is an edge of $C$, we have that $\pi(C)$ is a pointed cone.

Let $\kappa: \R^r \to \R$ be the map $\kappa(u) = {z_i \cdot u}$.  Let $H = {\rm ker} \; \kappa$ be the kernel of $\kappa$.  For $t \in \R$, consider the function $g_t: H \cong \R^{r-1} \to \R^{r-1}$ given by
\be
g_t(h) = \pi \circ \phi (h + tz_i) = \sum_{j\neq i} e^{{z_j\cdot(h + tz_i)}} \pi(Z_i) = \sum_{j\neq i}  e^{{z'_j\cdot h}} (\alpha_i\pi(Z_i))
\ee
where $\alpha_i = e^{z_j\cdot (tz_i)} > 0$, and $z'_j \in \R^{r-1} \cong \R^r/\R.z_i$ are the images of $z_j$ under the natural quotient map.  But the oriented matroid of $z'_1,\ldots,z'_{i-1},z'_{i+1},\ldots,z'_n$ is equal to $\M/i$ as well, so by the same result for $r-1$ we deduce that $g_t$ maps $\R^{r-1}$ bijectively onto the interior of $\pi(C)$. 

For a point $p \in \Int(\pi(C))$, the restriction of $\kappa$ to $(\pi \circ \phi)^{-1}(p)$ is therefore a bijection, and we deduce that $(\pi \circ \phi)^{-1}(p)$ is diffeomorphic to $\R$.  As a consequence, for any $p \in \R^{r-1}$, the set $\pi((\pi \circ \phi)^{-1}(p))$ is either connected or empty.  Let $S = \Image(\phi)$.  Then $\pi((\pi \circ \phi)^{-1}(p))$ is the intersection of $S$ with a line parallel to $Z_i$.  Varying $p$, we obtain the claim. \qed

For an edge $Z_i$ of $C$ and a point $q \in \R^r$, let $L_i(q) = \{q + t Z_i \mid t \in \R\}$ be the line through $q$ parallel to $Z_i$.  Let $q$ be in the image of $\phi$.  We now show that $(q - C) \cap \Int(C) \subset \Image(\phi)$.  Consider the line $L_i(q)$. Let us first suppose that $L$ intersects $C$ at a point $p \in \Int(F)$ in the interior of a facet $F$.  Since $\Image(\phi)$ contains all points in some neighborhood of $p$, and $\Image(\phi) \cap L$ is connected, we see that $q + tZ_i \in \Image(\phi)$ for $t < 0$.  Now suppose that $L$ intersects $C$ in the interior of a face $F'$.  First suppose $F'$ has one dimension less than that of a facet.  Since $\Image(\phi)$ contains an open neighborhood of $q$, we can find some edge $Z_j$ of $C$ and some $\epsilon$ such that $q + \epsilon Z_j$ and $q- \epsilon Z_j$ both lie in $\Image(\phi)$, but now $L_i(q+\epsilon Z_j)$ and $L_i(q - \epsilon Z_j)$ both intersect $C$ in the interior of facets.  By an induction on the codimension of the face $F'$, we obtain that $(q - C) \cap \Int(C) \subset \Image(\phi)$, for any $q \in \Image(\phi)$.

Since $\Cone(z)$ is pointed, there exists a vector $u \in \R^r$ such that ${z_i\cdot u} > 0$ for all $i$.  It follows that for any $M > 0$, there exists a point $q \in \Image(\phi)$ such that $q = \sum_{i=1}^s \alpha_i Z_i$ where $\alpha_i > M$ for all $i$.  It follows that $\Image(\phi) = \Int(C)$.
\end{proof}

\begin{proof}[Second proof of Theorem \ref{thm:diffeo}]
We will prove Corollary~\ref{cor:tPhi}, which is equivalent. Let $\tilde{\Phi}_{>0}$ denote the restriction $\tilde{\Phi}|_{\R^r_{>0}}:\R_{>0}^r\rightarrow \Int(C)$. Note that the basic ideas behind this theorem were first presented in the proof of Khovanskii's theorem~\cite{fewnomials}.

Let $Y\in \Int(C)$. We need to argue that the system of equations $Y=\tilde{\Phi}_{>0}$ consists of exactly one positive root (i.e. a root for which each component is positive). 

The idea is to apply induction on the number of vertices $n$, starting at $n=r$ which is essentially trivial. For higher $n$, we smoothly deform the polytope to a lower vertex polytope, and argue that the positive roots do not bifurcate throughout the deformation. The result therefore follows from the induction hypothesis.

For convenience, let us define the function $\Psi(X)\deff \sum_{i=1}^n(C_i(X)-C_{0i})Z_i$, where $Y=\sum_{i=1}^nC_{0i}Z_i$ for some constants $C_{0i}>0$. It therefore suffices to show that the system of equations $\Psi(X)=0$ contains exactly one positive root.

We first observe that by setting any one of the variables $X_a$ to zero, we necessarily land on a boundary component of the polytope. It follows that no real solution $X$ exists that intersects $\partial \R_{>0}^r$. We will use this fact later on.

For $n=r$, the $Z_i^I$ matrix is invertible, so it suffices to show that the system $C_i(X)=C_{0i}$ for $i=1,\ldots ,n$ has exactly one positive root. Taking the log of both sides, we find
\be
\sum_{a=1}^{r}z_{ai}\log X_a=\log C_{0i}
\ee
This is a linear system of equations with \defn{invertible} matrix $z_{ai}$, and therefore has a unique solution. Positivity of the solution is clear.

For $n>r$, let us introduce a new variable $t$ and consider the system of equations
\be
\Psi(t;X) = t (C_1(X)-C_{01}) Z_1 + \sum_{i=2}^n (C_i(X)-C_{0i}^*) Z_i
\ee
We want to show that $\Psi(t;X)=0$ has exactly one positive root at $t=1$. At $t=0$, the induction hypothesis gives us exactly one positive root. Also, as we argued earlier, this system does not have roots in $\partial \R_{>0}^r$ (i.e. As $t$ evolves from $0$ to $1$, roots will not move through the boundary of the positive region.) Therefore it suffices to show that the system does not bifurcate in the positive region as $t$ evolves. Indeed, the toric Jacobian (i.e. the Jacobian with respect to $u$) is non-vanishing.
\be
J(t;u) &=& \det_{aI}\left(X_a \frac{\partial \Psi^I(t;X)}{\partial X_a}\right) = \det\left(\sum_{i=1}^nC'_i z_{i}^TZ_i
\right)\\
&=&\sum_{1\le i_1<\ldots <i_r\le n}C'_{i_1}\ldots C'_{i_r}\lb z_{i_1}\ldots z_{i_r}\rb\lb Z_{i_1}\ldots Z_{i_r}\rb
\ee
where $C'_1 = t C_1$ and $C'_i = C_i$ for $i> 1$. This Jacobian computation was essentially done in the first proof, but now with an extra factor of $t$.

It follows that the system does not bifurcate as $t$ evolves from $0$ to $1$, which completes our argument.
\end{proof}

\section{The global residue theorem}\label{app:GRT}

We provide a very brief review of the \defn{global residue theorem} (or GRT)~\cite{griffithsharris}, which is an extension of Cauchy's theorem to multiple variables. Applications of the GRT in amplitude physics is discussed extensively in~\cite{Grassmannian}. We remind the reader of the statement and provide an intuitive reason it should be true that differs from standard approaches to its proof, which also naturally connects to the ``push-forward" formula for canonical forms.  

Let $z\in \C^n$ denote complex variables, and let $g(z)$ and $f_a(z)$ for $a=1,...,n$ be polynomials. Consider the roots of the system of equations
\be
\{f_a(z)=0\}
\ee
which we assume to be a finite set of isolated points $z_i$. Furthermore, let us define the holomorphic top form
\be
\omega\deff \frac{g(z)d^nz}{\prod_{a=1}^nf_a(z)}
\ee
We define the \defn{global residue} of $\omega$ at $z_i$ to be
\be\label{eq:globalresidue}
\Res_{z_i}\omega = \frac{g(z_i)}{\det\left(\frac{\partial f_a}{\partial z_b}\right)(z_i)}
\ee
where the determinant is applied on the matrix indexed by $(a,b)$.

The global residue theorem states the following:

\begin{theorem}\label{thm:GRT}(global residue theorem) The sum of all the global residues of $\omega$ vanishes provided that $\left(\sum_{a=1}^n\deg f_a\right)-\deg g>n$. Namely,
\be\label{eq:GRT}
\sum_{i}\Res_{z_i}\omega=0
\ee
\end{theorem}

For $n=1$, this is simply the statement that the sum of all residues of a holomorphic function on the complex plane is zero, provided that there is no pole at infinity.

Very often, there may be many (irreducible) polynomial factors $F_1,...,F_N$ for $N\ge n$ that appear in the denominator of $\omega$:
\be
\omega=\frac{g(z)d^nz}{\prod_{A=1}^NF_A(z)}
\ee
In that case, we can group the the polynomials into $n$ index subsets $S_1,...,S_n$ that form a partition of the index set $\{1,...,N\}$, then let $f_a=\prod_{A\in S_a}F_A(z)$ and apply the theorem. The global residue theorem thereby involves summing over all the roots of any system of $n$ polynomials $F_{A_1}=...=F_{A_n}=0$ for which $A_a\in S_a$ for every $a=1,...,n$. It is easy to show that the global residue at any root $z_i$ of $n$ such polynomials looks like
\be
\Res_{z_i}\omega=\frac{g(z_i)}{\left(\det\left(\frac{\partial F_{A_a}}{\partial z_b}\right)\prod_{A\not\in\{A_1,...,A_n\}}F_A\right)(z_i)}
\ee

The intuition behind the theorem is as follows. Let us denote by $\sigma(\omega)$ the sum over all the global residues. Since $\sigma(\omega)$ is symmetric in the roots of the system of polynomial equations, it must therefore be a \defn{rational function} of the coefficients of the polynomials $g(z)$ and $f_a(z)$, by a standard argument in Galois theory.

Furthermore, by applying the ``pole cancellation" argument surrounding~\eqref{eq:noPoles}, it follows that $\sigma(\omega)$ has no poles anywhere. Indeed, the poles, if they exist, must appear when the Jacobian of $f_a$ vanishes at a root $z_i$. But similar to our argument surrounding~\eqref{eq:noPoles}, the roots bifurcate in the neighborhood of $z_i$, and the pairwise contribution cancel at leading order, thus giving no pole at the root $z_i$. The $\sigma(\omega)$ must therefore be \defn{polynomial} in the coefficients. 

Now suppose we introduce an extra parameter $t$ alongside the existing parameters $z_i$ so that each polynomial $f_a$ and $g$ becomes homogeneous, and we do so without changing the degree of the polynomials (i.e. the highest order term in each polynomial is unaffected). Now we imaging taking $t\rightarrow \infty$ in $\sigma(\omega)$. As $t$ becomes large, the roots $z_i$ also become large like $O(t)$. Each global residue~\eqref{eq:globalresidue} must therefore scale like $O(t^{-(\sum_a\deg f_a-\deg g-n)})$. It follows that
\be
\sigma(\omega)\sim O(t^{-(\sum_a\deg f_a-\deg g-n)})
\ee
which by assumption is a strictly negative power. It follows that $\sigma(\omega)$ must vanish, which completes our argument.

\section{The canonical form of a toric variety}
\label{app:toricform}
We follow the notation in \cite{Fulton}.  Fix a lattice $N \cong \ZZ^m$ and its dual lattice $M = N^\vee$.  Let $\sigma$ be a \defn{pointed rational polyhedral cone} in the vector space $N_\Q \cong N \otimes_\ZZ \Q$.  Thus $\sigma$ is the cone generated by elements in $N$, and $\sigma$ does not contain any line in $N_\Q$.  The dual cone $\sigma^\vee \subset M_\Q$ is given by
\be
\sigma^\vee = \{u \in M_\Q \mid \ip{u,v} \geq 0 \text{ for all } v \in \sigma\}.
\ee
For example, if $\sigma = \{0\}$, then $\sigma^\vee = M_\Q$.  The cone $\sigma^\vee$ is pointed if and only if $\sigma$ is $d$-dimensional.  

The semigroup $S_\sigma$ is the set $\sigma^\vee \cap M$, endowed with the multiplication inherited from $M$.  The semigroup algebra $\C[S_\sigma]$ is the commutative ring generated by symbols $\chi^u$ for $u \in S_\sigma$, and multiplication given by $\chi^u \cdot \chi^{u'} = \chi^{u+u'}$.  The element $\chi^0$ is the identity element of $\C[S_\sigma]$.  By definition, the spectrum $\Spec(\C[S_\sigma])$ is a (normal) \defn{affine toric variety} $A_\sigma$.

A \defn{fan} $\Delta$ in $N$ is a finite set of pointed rational polyhedral cones $\sigma$ in $N_\Q$ such that 
\begin{enumerate} \item Each face of a cone in $\Delta$ is also in $\Delta$;
\item
The intersection of two cones in $\Delta$ is a face of each.
\end{enumerate}
By gluing the affine toric varieties $\{A_\sigma \mid \sigma \in \Delta\}$, one obtains the \defn{normal toric variety} $X(\Delta)$ associated to the fan $\Delta$.

The toric variety $X(\Delta)$ contains a dense algebraic torus:
$$
T = (\C^\ast)^m = \Spec(\C[x_1^{\pm 1}, \ldots, x_m^{\pm 1}])
$$
We define a rational top-degree form $\Omega_\Delta$ on $X(\Delta)$ as the push-forward of the canonical form \be
\Omega_T = \frac{dx_1}{x_1} \wedge \cdots \wedge \frac{dx_m}{x_m}
\ee
on $T$.  If we change coordinates on $T$, the canonical form $\Omega_\Delta$ will change by $\pm 1$.  To fix a choice of sign, we fix an orientation on $M_\Q$.  If $\sigma^\vee$ is a full-dimensional cone in $M_\Q$ and $\tau^\vee \subset \sigma^\vee$ is a face of $\sigma^\vee$, then the orientation on $M_\Q$ naturally induces an orientation on the linear span $\sp(\tau^\vee)$.

The \defn{toric divisors} $D_i$ of $X_\Delta$ are indexed by the edges or one-dimensional cones $\tau_i$ of the fan $\Delta$ (\cite[Section 3.3]{Fulton}).  The divisor $D_i$ is the toric variety $X(\Delta')$ where $\Delta' \in N_\Q/\sp(\tau_i)$ is the image of the fan $\Delta$.

\begin{proposition}\label{prop:toricform}
The canonical form $\Omega_\Delta$ has a simple pole along each toric divisor and no other poles.  The residue $\Res_{D} \Omega_\Delta$ along a toric divisor is equal to the canonical form of the divisor $D$.
\end{proposition}
\begin{proof}
The toric variety $X(\Delta)$ is the union of the dense torus $T \subset X(\Delta)$ and its toric divisors.  Since $\Omega_{\Delta}|_T = \Omega_T$ has no poles in $T$, the only poles are along the toric divisors.

Let the toric divisor $D$ correspond to the edge $\tau$ of $\Delta$.  The rational forms $\Omega_\Delta$ and $\Res_{D} \Omega_\Delta$ are determined by their restriction to any open subset (for example, the dense torus).  We may thus replace $X(\Delta)$ by an open affine toric subvariety that contains a dense subset of $D$.  So we can and will assume that $\Delta$ consists of the two cones $\{0\}$ and $\tau$, so $X(\Delta) = A_\tau$.  After a change of basis of $N$, we may assume that $\tau = \sp(e_m)$ is the span of the last basis vector.  Thus $A_\tau \cong (\C^\ast)^{m-1} \times \C$ and $D  \cong (\C^\ast)^{m-1} \times \{0\}  \subset A_\tau$.  We have
\begin{align}
\Res_{x_m = 0}\left(\frac{dx_1}{x_1} \wedge \cdots \wedge \frac{dx_m}{x_m}\right) &= \frac{dx_1}{x_1} \wedge \cdots \wedge \frac{dx_{m-1}}{x_{m-1}} = \Omega_D 
\end{align}
the canonical form of the toric divisor $D$.
\end{proof}

\section{Canonical form of a polytope via toric varieties}
\label{app:push-forwardproof}
The aim of this section is to prove Theorem \ref{thm:push-forward}.

We shall use the following result roughly stating that taking residues and pushing forward commute.
\begin{proposition}[{\cite[Proposition 2.5]{polar}}]\label{prop:KR}
Suppose that $f: M \to N$ is a proper and surjective map of complex manifolds of the same dimension.  Let $\eta$ be a meromorphic form with only first order poles on a smooth hypersurface $V$.  Suppose that $f(V)$ is a smooth hypersurface in $N$.  Then $f_*\eta$ has first order poles on $f(V)$ and
\be
\Res_{f(V)} f_* \eta = f_* \Res_{V} \eta.
\ee
\end{proposition}

Let $X':= X(z)$ and let $\Omega' :=\pi_*(\Omega_{X'})$.  We shall check that $\Omega'$ satisfies the recursive defining property of the canonical form $\Omega(\A)$.  The statement is trivial when $m = 0$.  We proceed by induction on $m$.

Let $g: X \to X'$ be a toric resolution of singularities.  Thus $X$ is a smooth complete toric variety, and the map $g$ has the following properties:
\begin{enumerate}
\item
The map $g$ restricts to an isomorphism $g: T \cong T'$ from the dense torus $T \subset X$ to the dense torus $T' \subset X'$.
\item
We have $g_*(\Omega_{X}) = \omega_{X'}$ and $g^{-1}(\Omega_{X'}) = \Omega_X$.
\item
The boundary divisor $X \setminus T = \bigcup D \subset X$ is a normal crossings divisor.
\item
For every torus orbit closure $V$ of $X$, the image $g(V)$ is a torus orbit closure of $X'$.  Thus if $D \subset X$ is an irreducible toric divisor, then either $\dim g(D) < \dim D$, or $g(D) = D'$, where $D'$ is some toric divisor of $X$.
\end{enumerate}

\def\tpi{\tilde \pi}

The rational map $\pi:\P^{n-1} \to \P^m$ has indeterminacy locus $\ker(Z) \subset \P^{n-1}$ consisting of lines in $\R^n$ sent to 0 by the linear map $Z: \R^n \to \R^{m+1}$.  Let $\tpi = \pi \circ g: X \to \P^m$.  Then $\pi$ has indeterminacy locus $L = \{x \in X \mid g(x) \in \ker(Z)\}$.  By {\it elimination of indeterminacy}, there exists a sequence of blowups $f_i:X_{i} \to X_{i-1}$ along smooth centers $C_{i-1} \subset X_{i-1}$, for $i = 1,2,\ldots,r$, satisfying:
\begin{enumerate}
\item $X_0 = X$,
\item the composition $h_i: = f_{1} \circ f_{2} \circ \cdots \circ f_i: X_{i} \to X_0$ restricts to an isomorphism over $X \setminus L$,
\item there is a regular map $\pi_r: X_r \to \P^m$ such that $\pi_r = \pi \circ h_r$ as rational maps.
\end{enumerate}

All the varieties $X_i$ are smooth, and the maps $f_i$ are all proper and birational.  The map $\pi_r$ is proper and surjective.  

Let $\partial X$ denote the toric boundary of $X$.  It is the union of the toric divisors of $X$.  We shall make the technical assumption that 
\begin{align} \label{eq:tech}
\begin{split}
&\mbox{the union of the exceptional divisor of $h_i$ with $h_i^{-1}(\partial X)$} \\ 
&\mbox{is a simple normal crossings divisor.}
\end{split}
\end{align}  See \cite[p.142, Main Theorem II]{Hironaka}.

The rational form $\Omega_X$ has simple poles along each irreducible toric divisor $D \subset X$ (see Appendix \ref{sec:toric}).  We now investigate the poles of $h_r^*(\Omega_X)$.

\medskip 

\noindent {\bf Claim.} The rational form $h_r^*(\Omega_X)$ has simple poles along the strict transforms $\tilde D \deff \overline{h_r^{-1}(D \cap U)}$ of each irreducible toric divisor $D$, and no other poles.
\begin{proof}[Proof of Claim]
We shall prove the statement for the form $h_i^*(\Omega_X)$ on $X_i$, proceeding by induction on $i$.

The indeterminacy locus $L \subset X$ is closed.  We first check that $L$ does not contain any torus fixed point of $X$.  It is enough to show that $\pi: X' \to \P^m$ is defined at the torus fixed points of $X'$.  This is clear: a torus fixed point of $X'$ is sent by $\pi$ to the point $Z_k \in \P^m$, where $\sp_{\R_{\geq 0}}(Z_k)$ is an edge of the cone $\Cone(Z)$.  (Since $z$ is graded, we must have $z_k \neq 0$, and thus $Z_k \neq 0$ as well.)

Let $U = X \setminus L$.  Then for each irreducible toric divisor of $D$, we have that $U \cap D$ is open and dense in $D$; since $\Omega_X$ has a simple pole along $D$, it is clear that the strict transform $h_i^*(\Omega_X)$ 
has a simple pole along the strict transform $\tilde D$.  

Let $E_j \subset X_j$ denote the (irreducible) exceptional divisor of the map $f_j: X_j \to X_{j-1}$.  We denote by $\tilde E_j \subset X_i$ the strict transform of $E_j$ under $f_{j+1} \circ \cdots \circ f_i$.  Since $h_i$ is an isomorphism away from $\bigcup_{j \leq i} \tilde E_j$, it suffices to show that $h_r^*(\Omega_X)$ has no pole along each $\tilde E_j$.  Since the pullback of a holomorphic form is holomorphic, we are reduced to showing that $h_i^*(\Omega_X)$ has no pole along $E_i$.

We now consider the blowup of $X_{i-1}$ along the smooth center $C_{i-1}$.  Suppose $C_{i-1}$ lies in the divisors $\tilde D_1,\tilde D_2,\ldots,\tilde D_t$ of $X_{i-1}$ (and not in any of the other $\tilde D_j$).  By our assumption \eqref{eq:tech}, the $\tilde D_j$ have normal crossings, so the intersection $R = \tilde D_1 \cap \tilde D_2 \cap \cdots \cap \tilde D_t$ has codimension $t$.  The intersection $U \cap (D_1 \cap \cdots \cap D_t)$ is dense in $(D_1 \cap \cdots \cap D_t)$ (since torus fixed points of $X$ do not lie in $L$).  Since $h_{i-1}(C_{i-1}) \cap U = \emptyset$, we have that $C_{i-1}$ has codimension at least one in $R$.

%

Thus ${\rm codim}(C_{i-1}) = t+s$ where $s \geq 1$.  Let $p \in C_{i-1}$.  In a neighborhood $V$ of $p$ in $X_{i-1}$, we may find local coordinates $x_1,x_2,\ldots,x_t,x_{t+1},\ldots,x_{t+s},y_1,y_2,\ldots,y_\ell$, so that $C_{i-1}$ is cut out by the ideal generated by $x_1,x_2,\ldots,x_{t+s}$, and $\tilde D_j$ is cut out by $x_j$ for $1,2,\ldots,t$.  We may thus write (using the inductive hypothesis)
\be
h_{i-1}^*(\Omega_X) = \frac{q}{x_1x_2 \cdots x_t} dx_1 \wedge \cdots \wedge dx_{t+s} \wedge dy_1 \wedge \cdots \wedge dy_\ell
\ee
where $q$ is a meromorphic function whose poles do not pass through $p$.

The blowup $B$ of $V$ along $C_{i-1} \cap V$ has local coordinates $w_1,\ldots,w_{t+s-1},x_{t+s},y_1,\ldots,y_\ell$, where $x_i = w_i x_{t+s}$ for $i = 1,2,\ldots,t+s-1$.  The exceptional divisor is given by $x_{t+s} = 0$ in these coordinates.  The pullback of $h^*_{i-1}(\Omega_X)$ to $B$ is given by
\begin{align}
&\frac{q}{w_1w_2 \cdots w_t x_{t+s}^t} d(w_1 x_{t+s}) \wedge d(w_2 x_{t+s}) \wedge \cdots \wedge d(w_{t+s-1} x_{t+s}) \wedge dx_{t+s} \wedge dy_1 \wedge \cdots \wedge dy_\ell \\ &= \frac{ x_{t+s}^{s-1}\; q}{w_1w_2 \cdots w_t} dw_1  \wedge dw_2 \wedge \cdots \wedge dw_{t+s-1} \wedge dx_{t+s} \wedge dy_1 \wedge \cdots \wedge dy_\ell.
\end{align}
Since $s -1 \geq 0$, the blowup $B$ has no pole (but sometimes a zero) along the exceptional divisor.  This completes the proof of the claim.
\end{proof}

Let $D_i$ be an irreducible toric divisor of $X$ and $\tilde D_i \subset X_r$ its strict transform.  Suppose that $g(D_i)$ is a divisor in $X'$.  Then $\pi(D_i)$ is the linear subspace of $\P^m$ spanned by some face $F$ of the polytope $\A$.  Apply Proposition \ref{prop:KR} with the choice $f = \pi_r$ and with $M, N, V$ being appropriate open subsets of $X_r$, $\P^m$, and $\tilde D_i$ respectively.
We compute
\begin{align*}
\Res_{\pi(D_i)} \, \Omega' &=
\Res_{\pi_r(\tilde D_i)} (\pi_r)_*(h_r^{-1}(\Omega_X)) 
\\
&= (\pi_r)_* \Res_{\tilde D_i}(h_r^{-1}(\Omega_X)) & \mbox{by Proposition \ref{prop:KR},}\\
&= \tpi_* \Res_{D_i}\Omega_X & \mbox{$h_r$ is an isomorphism over a dense subset of $D_i$,} \\
&= \tpi_* \Omega_{D_i} & \mbox{by Proposition \ref{prop:toricform},} \\
&= \pi_* \Omega_{g(D_i)} \\
&= \Omega(F) & \mbox{by the induction hypothesis.}
\end{align*}
We have used the fact that $h_r$ restricts to an isomorphism over a dense subset of $D_i$.

Now, if $D_i$ is such that ${\rm codim } \;g(D_i) \geq 2$, then $\tilde D_i$ clearly cannot contribute to poles of $\Omega' = (\pi_r)_*(h_r^{-1}(\Omega_X))$.  The claim above states that the only poles of $h_r^{-1}(\omega_X)$ are along the $\tilde D_i$, so we conclude that the only poles of $\Omega'$ are along linear subspaces spanned by the faces of $\A$.  We conclude that $\Omega' = \Omega(\A)$.

\section{Oriented matroids}
\label{sec:matroids}
A basic reference for oriented matroids is \cite{matroids}.
\def\sign{{\rm sign}}
Let $v_1,v_2,\ldots,v_n \in \R^k$ be a collection of vectors.  The \defn{oriented matroid} $\M = \M(v)$ keeps track of the (positive) linear dependencies between this set of vectors.  If $\alpha \in \R$, we let $\sign(\alpha) = +, 0, -$ depending on whether $\alpha > 0, \alpha = 0$, or $\alpha < 0$.

\medskip
\noindent {\it Chirotope.} The \defn{chirotope} of $\M$ is given by the function $\chi: \{1,2,\ldots,n\}^{k} \to \{+,0,-\}$ specified by
\be
\chi(i_1,i_2,\ldots,i_k) := \sign \ip{v_{i_1} v_{i_2} \cdots v_{i_k}}.
\ee

\medskip
\noindent {\it Signed vectors.} Define the space of linear dependencies $D(v) \subset \R^n$ by
\be
D(v) := \{(a_1,a_2,\ldots,a_n) \in \R^n \mid \sum_{i=1}^n a_iv_i =0\}.
\ee
The set of \defn{signed vectors} of the oriented matroid $\M$ is the set $\sign(D(v))$.

\medskip
\noindent {\it Signed covectors.} Define the space of value vectors $V(v) \subset \R^n$ by
\be
V(v) := \{(c \cdot v_1, \ldots, c \cdot v_n) \in \R^n \mid c \in \R^k\}.
\ee
The set of \defn{signed covectors} of the oriented matroid $\M$ is the set $\sign(V(v))$.

\medskip

\begin{proposition}\label{prop:OM}
Any one of (a) the chirotope, (b) the signed vectors, and (c) the signed covectors, of an oriented matroid $\M$ determines the other two.
\end{proposition}

Given the oriented matroid $\M = \M(v)$ and an index $i \in \{1,2,\ldots,n\}$, we can consider the \defn{deletion} $\M\backslash i$, which is the oriented matroid associated to $v_1,v_2,\ldots,\widehat{v_i},\ldots,v_n$.  We also have the \defn{contraction} $\M/i$ which is the oriented matroid associated to $\bar v_1,\bar v_2,\ldots,\widehat{v_i},\ldots,\bar v_n$, where $\bar v_j$ is the projection of $v_j \in V$ to a linear hyperplane $H \subset V$ which does not contain $v_i$.  A basic result is that $\M \backslash i $ and $\M/i$ depend only on $\M$ and not on the vectors $v_1,\ldots,v_n$.
\section{The Tarski-Seidenberg theorem}
\label{app:tarski}
A {\it semialgebraic} subset of $\R^n$ is a finite union of sets that are cut out by finitely many polynomial equations $p(x) = 0$ and finitely many polynomial inequalities $q(x) > 0$.  Let $\pi: \R^n \to \R^{n-1}$ be the projection map that forgets the last coordinate.  The Tarski-Seidenberg theorem states that if $S \subset \R^n$ is a semialgebraic set, then so is $\pi(S)$.  We now explain why this implies that the Amplituhedron, and more generally, any Grassmann polytope is semialgebraic.

Let $f: \R^n \to \R^m$ be a rational map.  That is 
$$
f(x) = \left(\frac{p_1(x)}{q_1(x)},\frac{p_2(x)}{q_2(x)},\ldots,\frac{p_m(x)}{q_m(x)}\right)
$$
for polynomials $p_i,q_i$ in $n$ variables.  Suppose $S \subset \R^n$ is a semialgebraic set such that $f$ is well defined on $S$ (so $q_i(x)$ is non-zero everywhere in $S$).  We show that $f(S)$ is semialgebraic.  Write $(x,y)$ for a typical point in $\R^n \times \R^m$.  Define $\Gamma_S = \{(x,y)\} \subset \R^n \times \R^m$ by the $m$ polynomial equations
$$
y_1 q_1(x) -p_1(x) = 0, \ldots, y_m q_m(x) - p_m(x) = 0
$$
and the condition that $x \in S$.  It is clear that $\Gamma_S$ is semialgebraic, since the condition $x \in S$ is semialgebraic.  Also $f(S)$ is equal to $\pi_m(\Gamma_S)$ where $\pi_m :\R^n \times \R^m \to \R^m$ is the projection to the last $m$ factors.  Thus applying the Tarski-Seidenberg theorem $n$ times, we conclude that $f(S)$ is semialgebraic.

Now suppose $S \subset \P^n$ is semialgebraic.  Thus $S$ is a finite union of sets that are cut out by finitely many homogeneous polynomial equations $p(x) = 0$ and finitely many homogeneous polynomial inequalities $q(x) > 0$.  Note that to make sense of the inequalities, we first define a semialgebraic cone in $\R^{n+1}$, then we project down to $\P^n$.  Cover $\P^n$ with $n+1$ charts $U_0,\ldots,U_n$ each isomorphic to $\R^n$.  Then $S \cap U_i$ is semialgebraic in $\R^n$.  Now let $g: \P^n \to \P^m$ be a rational map, and suppose $g$ is well-defined on $S$.  Thus
$$
g(x) = \left(p_0(x),\ldots,p_m(x)\right)
$$
for homogeneous polynomials $p_i$ in $n+1$ variables all of the same degree.  Cover $\P^m$ with charts $V_0,\ldots,V_m$, so for example $V_0 = (1,y_1,y_2,\ldots,y_m)$.  Then for any $i = 0,1,\ldots,n$ and any $j = 0,1,\ldots,m$, the rational function $g$ restricts to a rational function $f_{ij}: U_i \to V_j$.  For example $f_{00}(x) = (p'_1(x)/p'_0(x),\ldots,p'_m(x)/p'_0(x))$ where $p'(x_1,\ldots,x_n) = p(1,x_1,\ldots,x_n)$.  By the above argument, $f_{ij}(U_i \cap S)$ is semialgebraic in $V_j$.  Since the definition of semialgebraic allows finite unions, we deduce that $g(S)$ is semialgebraic in $\P^m$.

Now let $G_{\geq 0}(k,n) \subset G(k,n)$ be the nonnegative Grassmannian.  Embed $G(k,n)$ inside $\P^d$ where $d = \binom{n}{k}-1$.  By definition $G_{\geq 0}(k,n)$ is a semialgebraic set inside $\P^d$.  Let $Z: \R^{n} \to \R^{k+m}$ be a linear map, which we assume to be of full rank.  This induces a map $Z_G: G(k,n) \to G(k,k+m)$, and also a map $Z_\P: \P^d \to \P^c$, where $c = \binom{k+m}{k}-1$.  The map $Z_\P$ is in fact linear, with coefficients given by minors of the matrix $Z$.  In particular, $Z_\P$ is a rational map in the above sense, and thus $Z_\P(G_{\geq 0}(k,n))$ is a semialgebraic set.

\newpage{\pagestyle{empty}\cleardoublepage}

\end{document}